\documentclass[conference]{IEEEtran}
\usepackage{cite}

\ifCLASSINFOpdf
   \usepackage[pdftex]{graphicx}
\else
   \usepackage[dvips]{graphicx}
\fi
\usepackage[cmex10]{amsmath}
\usepackage{amssymb}
\usepackage{algorithm}
\usepackage{algpseudocode}

\usepackage{mdwmath}
\usepackage{mdwtab}

\usepackage{caption}
\usepackage[caption=false,font=footnotesize]{subfig}
\usepackage{footnote}

\usepackage{url}

\usepackage{tikz}
\tikzstyle{vertex}=[auto=left,circle,draw=black,fill=black!25,minimum size=20pt,inner sep=0pt]

\newtheorem{thm}{Theorem}
\newtheorem{definition}{Definition}

\newtheorem{lemma}{Lemma}

\newcommand{\ER}{Erd\"os-R\'{e}nyi }

\begin{document}
%
\title{Locally Estimating Core Numbers}

\author{\IEEEauthorblockN{Michael P. O'Brien, Blair D. Sullivan}
\IEEEauthorblockA{Department of Computer Science\\
North Carolina State University\\
Raleigh, North Carolina, 27607\\
Email: \{mpobrie3,blair\_sullivan\}@ncsu.edu}}


%


\maketitle

\begin{abstract}
Graphs are a powerful way to model interactions and relationships in data from a wide variety of application domains.
In this setting, entities represented by vertices at the ``center'' of the graph are often more important than those associated with vertices on the ``fringes''.
For example, central nodes tend to be more critical in the spread of information or disease and play an important role in clustering/community formation.
Identifying such ``core'' vertices has recently received additional attention in the context of {\em network experiments}, which analyze the response when a random subset of vertices are exposed to a treatment (e.g.
inoculation, free product samples, etc).
Specifically, the likelihood of having many central vertices in any exposure subset can have a significant impact on the experiment.

We focus on using {\em $k$-cores and core numbers} to measure the extent to which a vertex is central in a graph.
Existing algorithms for computing the core number of a vertex require the entire graph as input, an unrealistic scenario in many real world applications.
Moreover, in the context of network experiments, the subgraph induced by the treated vertices is only known in a probabilistic sense. 
We introduce a new method for estimating the core number based only on the properties of the graph within a region of radius $\delta$ around the vertex, and prove an asymptotic error bound of our estimator on random graphs.
Further, we empirically validate the accuracy of our estimator for small values of $\delta$ on a representative corpus of real data sets.
Finally, we evaluate the impact of improved local estimation on an open problem in network experimentation posed by Ugander et al.

\end{abstract}


%
\IEEEpeerreviewmaketitle

\section{Introduction}
In a graph modeling complex interactions between data instances, the connectivity of the vertices often yields useful insights about the properties of the represented entities.  For example, in a social network, it is important to distinguish between a person who is a member of a relatively large, tight-knit community, and a person who exists on the periphery without much involvement with any cohesive communities.  This type of connectivity is often not well-correlated with the vertex degree (raw number of connections), leading to the introduction of several more sophisticated metrics. Here, we focus on the {\em core number}~\cite{batagelj:decomp}. To build intuition, consider the setting of a graph representing friendships in a social network -- the vertices with large core numbers form a set where people tend to have many common friends, whereas the subset of vertices with small core numbers form a group where most people do not know one another.  The applications of core numbers are numerous, well studied, and permeate a variety of domains, such as community detection~\cite{giatsidis:community,matula:ordering}, virus propagation~\cite{kitsak:spreader}, data pruning~\cite{bader:protein}, and graph visualization~\cite{alvarez-hamelin:viz}.

Existing algorithms for computing core numbers run in linear time with respect to the number of vertices and edges of the graph~\cite{batagelj:decomp}, but require the entire graph as input and simultaneously calculate the core number for every vertex.  There are a number of scenarios in which current approaches are insufficient.  First, one might only be concerned with the properties of a small subset of query vertices\footnote{For example, vertices which have some common metadata (e.g. age, physical location, etc.) of interest to the user that is not represented explicitly in the graph}. If the query vertices constitute a small fraction of the entire graph, it would be much more time- and memory-efficient to locally estimate the core numbers of those specific vertices rather than using the global algorithm, regardless of the fact that it is linear.   Second, the entire graph may be unknown and/or infeasible to obtain due to scale, privacy, or business strategy reasons.  For example, suppose one wanted to investigate patterns of phone calls between cell phone users.  It would be possible to contact a small set of people and ask them to voluntarily log their calls for one month.  However, without access to the telecommunication companies' private records, expanding the domain to the national or global level would not be possible.  Finally, current methods for determining core numbers cannot be applied to solve problems in the domain of network experimentation, as we describe below.

A {\em network treatment experiment} is a randomized experiment in which subjects are divided into two groups: those receiving treatment and those receiving none (or a placebo).  Network treatments differ from other randomized experiments in that the effects of the treatment (or lack thereof) on a given subject are assumed to be dependent on the experiences of other subjects in the experiment.  Randomly assigning subjects into two groups (treated versus untreated) is equivalent to randomly partitioning the vertices of a graph into two sets.  Previous work has given methods to calculate degree probabilities of a randomly partitioned graph~\cite{ugander:exposure}, but an analogous algorithm for the core numbers is an open problem.  Given the results of Kitsak et al.\ on the importance of core numbers in spreading information~\cite{kitsak:spreader}, an algorithm to predict the likelihood of a given subject having a large core number would help researchers better understand the impact of their experimental design.  For example, researchers conducting market testing on a product that relies on social interaction (such as a new social networking site, online game, etc.) would have a greater ability to see whether early access to their product will generate widespread excitement among the test subjects.  Additionally, if certain groups of vertices (say, females participating in the experiment) are more likely to be core exposed than the others, we can reduce the bias of the estimate of the treatment effect by upweighting the probability that underrepresented vertices (males) are treated.

All of the challenges mentioned above could be addressed by estimating the core number of a vertex based on local graph properties.  The existence of an accurate non-global estimator is intuitively well-grounded, as it has been shown that addition and deletion of edges can only affect the core numbers of a limited subset of vertices~\cite{sariyuce:streaming}.  This suggests that in spite of the fact that computing core numbers exactly requires knowledge of the whole graph, the core number of a given vertex may often depend on a much smaller subgraph.  Moreover, even if a core number estimate has a small error, it may still be useful in applications.  In particular, it is often sufficient to delineate between vertices with a ``large'' core number and those with a ``small'' core number.  That is to say, while a vertex in the 50-core is substantially more well-connected than one in the 2-core, it may be functionally identical to a vertex in the 51-core in downstream analysis.

This work introduces a new local estimator of the core number at a specified vertex, which allows a user to tune the balance between accuracy and computational complexity by varying the size of the local region around the query vertex that it considers. We prove that in an \ER random graph, the error of our approximation at each vertex asymptotically almost surely grows arbitrarily slowly with the size of the graph.  We also empirically evaluate the estimates with respect to the actual core numbers on a representative corpus of real-world graphs of varying sizes.  The results on these graphs demonstrate that high accuracy can be achieved even when considering only a small local region.  Finally, we show how our estimators can be applied to address the aforementioned open problem in network treatment experiments.  Specifically, we give an algorithm to tighten the upper bound on the core number of a vertex given in \cite{ugander:exposure}, and evaluate the impact empirically on a sample experiment.

\section{Background and Definitions}
In this paper, all graphs are simple, undirected, and unweighted.  Unless otherwise specified, $G$ denotes a graph with vertex set $V$ and edge set $E$, where $n=|V|$.  The number of edges incident to a vertex $v$ is the degree of $v$, denoted $d(v)$.  We also assume that $\forall v\in V$, $d(v)>0$ (isolated vertices are not relevant to our algorithms and analysis).  The notation $G[S]$ denotes the subgraph of $G$ \emph{induced} on the vertices $S\subseteq V$.  In other words, $G[S] = (S,F)$ is the graph with vertex set $S$ and edge set $F = \{(u,v)\in E |\, u,v\in S\}$.  Given a function $f$, we say a set of vertices $u_1,u_2,\dots$ is \emph{$f$-ordered} if for all $i$, $f(u_i)\leq f(u_{i+1})$.

To quantify the idea of a ``central" vertex, we formalize the notion of being in a well-connected subgraph:
\begin{definition}[\hspace{-0.05em}\cite{batagelj:decomp}]
The {\em $k$-core} of $G$, denoted $C_k$, is the maximal induced subgraph of $G$ with minimum degree at least $k$.
\end{definition}
Clearly, a graph with minimum degree at least $k$ also has minimum degree at least $k-i$ for $i = 1, 2, \ldots$, so $C_k\subseteq C_{k-1}$.  Thus, for sake of specificity, we measure the degree to which a vertex is central by the deepest core in which it participates:
\begin{definition}
The {\em core number} of $v$, denoted $k(v)$, is the largest $k \geq 0$ such that $v\in C_k$.
\end{definition}
The term {\em core structure} will be used broadly to describe all properties of or relating to the $k$-cores of $G$.  A common global metric measures the depth of this structure:
\begin{definition}
The {\em degeneracy} of $G$ is the largest $k$ for which $|C_k|>0$; a graph with degeneracy at least $D$ is said to be {\em $D$-degenerate}.
\end{definition}

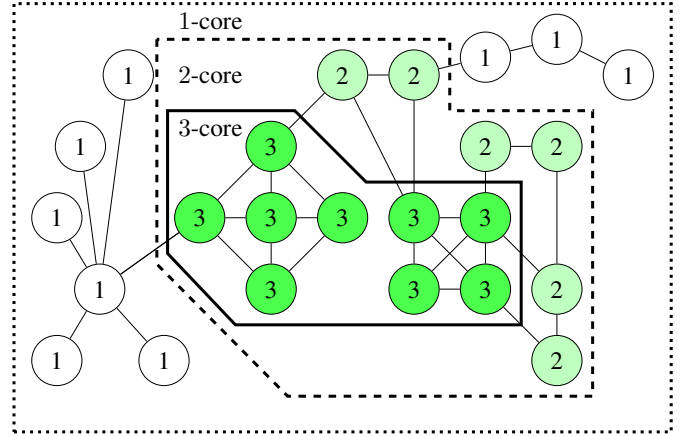
\begin{figure}[!h]
	\centering
	\resizebox {\columnwidth} {!} {\begin{tikzpicture}
	\node[vertex,fill=green!0] (n11)  at (1,1)  {1};
	\node[vertex,fill=green!0] (n21)  at (2.5,1)  {1};
	\node[vertex,fill=green!70] (n41)  at (4,2)  {3};
	\node[vertex,fill=green!70] (n61)  at (6,2)  {3};
	\node[vertex,fill=green!70] (n71)  at (7,2)  {3};
	\node[vertex,fill=green!25] (n81)  at (8,1)  {2};
	\node[vertex,fill=green!0] (n91)  at (9,5)  {1};
	
	\node[vertex,fill=green!0] (n12)  at (1,3)  {1};
	\node[vertex,fill=green!0] (n22)  at (1.6,2)  {1};
	\node[vertex,fill=green!70] (n32)  at (3,3)  {3};
	\node[vertex,fill=green!70] (n42)  at (4,3)  {3};
	\node[vertex,fill=green!70] (n52)  at (5,3)  {3};
	\node[vertex,fill=green!70] (n62)  at (6,3)  {3};
	\node[vertex,fill=green!70] (n72)  at (7,3)  {3};
	\node[vertex,fill=green!25] (n82)  at (8,2)  {2};
	\node[vertex,fill=green!0] (n92)  at (8,5.5)  {1};

	\node[vertex,fill=green!0] (n13)  at (1.33,4)  {1};
	\node[vertex,fill=green!0] (n23)  at (2,5)  {1};
	\node[vertex,fill=green!70] (n43)  at (4,4)  {3};
	\node[vertex,fill=green!25] (n53)  at (5,5)  {2};
	\node[vertex,fill=green!25] (n63)  at (6,5)  {2};
	\node[vertex,fill=green!25] (n73)  at (7,4)  {2};
	\node[vertex,fill=green!25] (n83)  at (8,4)  {2};
	\node[vertex,fill=green!0] (n93)  at (7,5.25)  {1};

	\foreach \from/\to in {n11/n22,n21/n22,
	n32/n22,
	n91/n92,
		n12/n22,
		n13/n22,
		n23/n22}
	\draw(\from) -- (\to);

	\foreach \from/\to in {n22/n32,n92/n93}	
	\draw(\from) -- (\to);
	
	\foreach \from/\to in {n93/n63}	
	\draw(\from) -- (\to);
	\foreach \from/\to in {n71/n81,n73/n72,n81/n82,
		n83/n73,n83/n82,
		n63/n62,n63/n53}	
	\draw(\from) -- (\to);

	\foreach \from/\to in {n82/n72,n53/n43,n53/n62}	
	\draw(\from) -- (\to);
	\foreach \from/\to in {n41/n42,n41/n32,n41/n52,
		n32/n42,n32/n43,
		n42/n43,
		n52/n42,n52/n43,
		n61/n62,n61/n72,n61/n71,
		n62/n72,n62/n71,
		n72/n71}	
	\draw(\from) -- (\to);
	
	\draw[solid,very thick](2.55,3) -- (2.55,2.5) -- (3.5,1.5) -- (7.5,1.5) -- (7.5,3.5) -- (5.33,3.5) -- (4.33,4.5) -- (2.55,4.5) -- (2.55,3);
	\draw[dashed,very thick](2.4,5.5) -- (2.4,2.35) -- (4.23,0.5) -- (8.5,0.5) -- (8.5,4.5) -- (6.5,4.5) -- (6.5,5.5) -- (2.4,5.5);
	\draw[dotted,very thick](0.4,0) -- (9.6,0) -- (9.6,6) -- (0.4,6) -- (0.4,0);
	\node at (3.15,4.25) {$3$-core};
	\node at (3.15,5) {$2$-core};
	\node at (3.15,5.75) {$1$-core};

\end{tikzpicture}}
	\caption{Sample graph with its $k$-cores and core numbers labeled. \label{fig:core_demo}
}
\end{figure}
\vskip-0.25cm

The core number of a vertex $v$ in a graph $G$ can be found by performing Algorithm \ref{alg:core_decomp} (from~\cite{batagelj:decomp}), which finds a \emph{core decomposition} of $G$.  Beginning with $i=0$, the algorithm deletes vertices with degree at most $i$ (and their incident edges) until there are no such vertices remaining.  The removal of edges incident to a vertex may cause its neighbors whose degrees were initially larger than $i$ to have their degree reduced to at most $i$.  In this case, those neighbors would be also be deleted in the degree $i$ phase.  Once all vertices remaining in $G$ have $d(v)> i$, $i$ is incremented by 1 and the process repeats until $G$ no longer has any vertices.  The core number of a vertex is the value of $i$ when it is removed from the graph and the degeneracy of $G$ is the value of $i$ when the last vertex is removed.  Since each vertex and edge is removed exactly once, a core decomposition can be completed in $O(|V|+|E|)$ time.
\begin{algorithm}
\caption{Core decomposition \label{alg:core_decomp}}
\begin{algorithmic}[1]
	\Statex \textsc{input}:  Graph $G = (V,E)$
	\Statex \textsc{output}:  $k(v)$  $\forall v \in V$
	\State $i\leftarrow 0$
	\While{$|V|>0$}
		\While{$\exists v: d(v)\leq i$}
			\State $k(v)\leftarrow i$
			\State $V\leftarrow V\backslash \{v\}$
			\State $E\leftarrow E\backslash \{(u,v)|u\in V\}$
		\EndWhile
		\State $i\leftarrow i+1$
	\EndWhile
\end{algorithmic}
\end{algorithm}

\\
The basic $k$-core decomposition has been tailored to meet additional constraints.  Montressor et al.\ \cite{montressor:distributeddecomp} and Jakma et al.\ \cite{jakma:distributeddecomp} proposed methods by which the core decomposition could be computed in parallel.  Li et al.\ \cite{sariyuce:streaming} and, later, Sar\'{i}y\"{u}ce et al.\  \cite{li:dynamic} described ways to update the core numbers of vertices in a dynamic graph without recomputing the full core decomposition each time a vertex or edge is added.  Finally, Cheng et al.\ \cite{cheng:decomp} gave an alternate implementation of the core decomposition for systems with insufficient memory to store the entire graph at once.

\section{Local Estimation and Theory}
In order to estimate core numbers efficiently without knowledge of the entire graph, we will restrict the domain of our algorithms to a localized subset around a vertex: 
\begin{definition}
The {\em $\delta$-neighborhood} of a vertex $v$, denoted $N_\delta(v)$, is the set of vertices at distance \footnote{We use the typical shortest-path distance function throughout} at most $\delta$ from $v$.
\end{definition}
The estimation algorithms will vary the size of their input by allowing $\delta$ to range from zero to $\Delta$, the {\em diameter} of $G$ (the maximum distance among all pairs of vertices).
\subsection{Neighborhood-based estimation}
A relatively na\"ive approach to local estimation would be to compute a core decomposition of the subgraph of $G$ induced on the $\delta$-neighborhood of $v$ and use the resulting core number of $v$ as the estimate:
	\begin{definition} \label{def:kbreve}
		Let the {\em induced estimator}, $\breve k_\delta(v)$, be the core number of $v$ in $G[N_\delta(v)]$.
	\end{definition}
By increasing $\delta$, the induced subgraph captures a progressively larger fraction of the graph, improving the estimate until $\breve k_\delta(v) = k(v)$ (which is guaranteed to happen at $\delta = \Delta$, but could happen for significantly smaller $\delta$).  Note that when $\delta = 1$, the estimator will only be close to the core number if the neighbors of $v$ are highly interconnected (so there is a subtle relationship to clustering coefficient).

\begin{lemma} \label{lemma:lower_bound}
	Let $G = (V,E)$ be a graph. For all $ v\in V$, \\
	$0 = \breve k_0(v) \leq \breve k_1(v) \leq \breve k_2(v) \leq \cdots \leq \breve k_\Delta(v) = k(v)$
\end{lemma}
\begin{proof}
	First we will establish the boundary conditions on our inequality.  Because $N_\Delta(v) = G$, $\breve k_\Delta(v) = k(v)$.  Since $N_0(v) = \{v\}$, $\breve k_0 = 0$.
	Additionally, for all $\delta \in [1,\Delta]$, $N_{\delta-1}(v) \subseteq N_\delta(v)$.  This implies that for all $u\in V$, the degree of $u$ in the subgraph induced by the $(\delta-1)$-neighborhood of $v$ can be no greater than the degree of $u$ in the $\delta$-neighborhood of $v$.  Thus $v$ cannot participate in a deeper core with respect to the $(\delta-1)$-neighborhood than it does with respect to the $\delta$-neighborhood, which makes $\breve k_{\delta -1}\leq \breve k_\delta$.
\end{proof}

In order to make a more sophisticated estimation, let us consider the information gained as $\delta$ increases.  At $\delta=0$, we assume that $d(v)$ is known.  Because the $k$-core requires all vertices to have minimum degree $k$, $k(v)$ can not be greater than $d(v)$.  Thus, $d(v)$ itself can be an estimate of $k(v)$.
\\
Expanding out to $\delta =1$ allows information about $v$'s immediate neighbors to be utilized.  Suppose the core numbers of the neighbors of $v$ were known.  It would then be possible to compute $k(v)$ precisely using the following two lemmas (previously shown by Montresor et al):
\begin{lemma}[\hspace{-0.05em}\cite{montressor:distributeddecomp}]\label{lemma:neighbors_in_core}
	A vertex $v$ is in the $j$-core of a graph $G$ if and only if $v$ has at least $j$ neighbors in the $j$-core.
\end{lemma}

We now give a closed-form algebraic expression for the largest $j$ satisfying Lemma~\ref{lemma:neighbors_in_core}.

\begin{lemma}[\hspace{-0.05em}\cite{montressor:distributeddecomp}]
\label{lemma:alt_k_core_definition}
Let $u_1,u_2,\dots$ be the $k$-ordered neighbors of $v$. Then $$k(v) = \max_{1\leq i \leq d(v)}\left(\operatorname{min}\left( k(u_i), d(v)-i+1 \right)\right).$$
\end{lemma}
\begin{proof}
By the definition of core number and Lemma~\ref{lemma:neighbors_in_core}, $k(v)$ is the largest $j$ in $0\leq j\leq d(v)$ 
so that $v$ has at least $j$ neighbors with core number at least $j$. For each $u_i$ ($1\leq i\leq d(v)$), $v$ has $d(v)-i+1$ neighbors with core numbers at least $k(u_i)$, since $k(u_i)\leq k(u_{i+1})$. If $k(u_i)\leq d(v)-i+1$, $v$ has at least $k(u_i)$ neighbors in the $k(u_i)$-core.   Otherwise, $v$ has at least $d(v)-i+1$ neighbors in the $(d(v)-i+1)$-core.  This shows that the core number of $v$ must be at least the minimum of $k(u_i)$ and $d(v)- i + 1$ for every $i$ (and thus the maximum over $i$). Equality follows easily by contradiction.
\end{proof}

  Note that $k(v)$ is the maximum value of the minimum of two functions of $i$:  $k(u_i)$ and $d(v)-i+1$.  With respect to $i$, $k(u_i)$ is monotonically non-decreasing and $d(v)-i+1$ is monotonically decreasing. 
From a geometric perspective, then, the maximum of their minimums occurs at the intersection of the curves $k(u_i)$ and $d(v)-i+1$ (as stylized in Figure~\ref{fig:cheng_geometric}).
	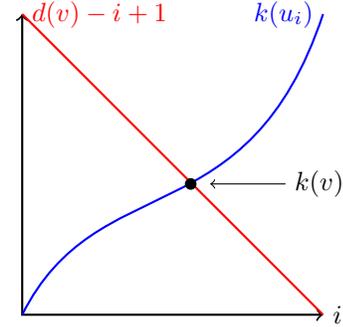
\begin{figure}[!h]
		\centering
			\begin{tikzpicture}
			\draw [<->,thick] (0,4) node (yaxis) [above] {}
				|- (4,0) node (xaxis) [right] {$i$};
			\draw[red,thick] (4,0) -- (0,4) node [right] {$d(v)-i+1$};
			\draw[blue,thick] (0,0) .. controls (1, 2) and (3,1) .. (4,4) node [left] {$k(u_i)$};
			\draw[<-] (2.5,1.74) -- (3.5,1.74) node [right] {$k(v)$};
			\filldraw[black] (2.24,1.74) circle (2pt);
			\end{tikzpicture}
		\caption{$k(v)$ is the $y$-value at intersection of two functions.
		\label{fig:cheng_geometric}}
	\end{figure}
	\vskip-0.5cm
	
Although the core numbers of the neighbors of a vertex may not be known {\it a priori}, the reasoning behind Lemma~\ref{lemma:alt_k_core_definition} gives useful insight into the behavior of $k(v)$.  As shown by Cheng et al.~\cite{cheng:decomp}, an upper bound on $k(v)$ can be achieved if an upper bound on $k(u)$ is known for all $u\in N_1(v)$.

\begin{thm}[\hspace{-0.05em}\cite{cheng:decomp}]\label{thm:cheng_result}
Let $G = (V,E)$ be a graph, $v \in V$, and $\psi$ any function satisfying $\psi(u)\geq k(u) \forall u\in V$.  Let $v$'s neighbors be $\psi$-ordered.  Then $$k(v) \leq \max_{1\leq i \leq d(v)}\left(	\min\left( \psi(u_i), d(v)-i+1 \right)\right).$$
\end{thm}
\begin{proof}
Substituting $\psi(u_i)$ for $k(u_i)$ in the expression from Lemma~\ref{lemma:alt_k_core_definition} can only increase the right hand side, giving an upper bound on $k(v)$. 
\end{proof}

We base our second estimator on the idea of incorporating iterative upper bounds on the core numbers of a vertex's neighbors:
	\begin{definition} \label{def:khat}
	Let the {\em propagating estimator}, $\hat k_\delta$, be the estimator of $k(v)$ given by the recurrence
	\begin{equation*}
	\hat k_{\delta}(v) = 
	\begin{cases}
		\underset{1\leq i \leq d(v)}{\operatorname{max}}\left(
		\operatorname{min}( \hat k_{\delta-1}(u_i), d(v)-i+1 )
		\right) & \text{if } \delta>0 \\
		d(v) & \text{if } \delta = 0
	\end{cases}
	\end{equation*}
	where $u_1,u_2,\dots$ are the $\hat k_{\delta-1}$-ordered neighbors of $v$.
	\end{definition}
	
Pseudocode for computing $\hat k_\delta(v)$ is given in Algorithm \ref{alg:k_hat}.  Essentially, Algorithm \ref{alg:k_hat} first computes the coarsest upper bound ($\hat k_0$) for those vertices at distance at most $\delta$ from $v$.  Those estimates are used in conjunction with Theorem \ref{thm:cheng_result} to compute a slightly finer upper bound, $\hat k_1$, for those vertices at distance at most $\delta -1$ from $v$.  This process ``propagates" inwards towards $v$ until its immediate neighbors have $\hat k_{\delta-1}$ values, which are used as the upper bounds in formulating $\hat k_\delta(v)$.  The computational complexity of finding $\hat k_\delta$ is linear in the product of $\delta$ and the number of edges in $N_\delta(v)$ (see Theorem \ref{thm:khat_complexity}). Since Algorithm \ref{alg:core_decomp} is also linear with respect to the number of edges in the graph, the computational complexity of computing $\hat k_\delta$ is comparable to that of $\breve k_\delta$ for small $\delta$.

\begin{algorithm}
\caption{Algorithm for computing $\hat k_\delta(v)$ \label{alg:k_hat}}
\begin{algorithmic}[1]
	\Statex \textsc{input}:  Graph $G$, vertex $v$
	\Statex \textsc{output}:  $\hat k_\delta(v)$
	\If{$\delta=0$}
		\State \Return $d(v)$
	\Else
		\For{$u\in N_1(v)$}
			\State Compute $\hat k_{\delta-1}(u)$
		\EndFor
		\State $\hat k_{\delta-1}$-order $N_1(v)$ \label{alg:k_hat:sorting}
		\State $\hat k_\delta\leftarrow d(v)$
		\For{$i\in\{1\dots|N_1(v)|\}$} \label{alg:k_hat:forloop}
			\State $j\leftarrow \min(\hat k_{\delta-1}(u_i),d(v)-i+1)$
			\If{$j>\hat k_\delta$}
				\State $\hat k_\delta \leftarrow j$
			\EndIf
		\EndFor
		\State \Return $\hat k_\delta$
	\EndIf
\end{algorithmic}
\end{algorithm}

\begin{thm} \label{thm:khat_complexity}
For a given vertex $u$, $\hat k_\delta(u)$ can be computed in $O(\delta\cdot |E_\delta(u)|)$ time using Algorithm~\ref{alg:k_hat}, where $E_\delta(u)$ is the edge set of $G[N_\delta(u)]$.
\end{thm}
\begin{proof}
Fix $\delta\in \{1,2,\dots\}$ and $u\in V$.  Assume $\forall u' \in N_1(u)$, $\hat k_{\delta-1}(u')$ is known.  Since $\hat k_\delta$ can only take integer values in the interval $[0,\max_{v\in V}d(v)]$, the sorting (line \ref{alg:k_hat:sorting}) can be done in $O(d(u))$ time using a bucket sort.  Once sorted, each neighbor of $u$ is visited once (line \ref{alg:k_hat:forloop}) to find the minimum, which can also be done in $O(d(u))$ time.  Thus computing $\hat k_\delta(u)$ from $\{\hat k_{\delta-1}(u') | u' \in N_1(u)\}$ has complexity $O(d(u))$.

Using dynamic programming, we can compute and store $\hat k_1(w)$ $\forall w\in N_{\delta-1}(v)$, which in turn can be used to compute $\hat k_2(w)$ $\forall w\in N_{\delta-2}(v)$, and so on (through $\delta-1$ iterations) until we have computed $\hat k_\delta(v)$.  The $j$th such iteration requires $\sum_{u\in N_{\delta-j}} O(d(u)) = O(E_{\delta-j}(v))$ operations.  Since $N_{\delta-1} \subseteq N_\delta$, the total running time is $O(\delta\cdot |E_\delta|)$.
\end{proof}

Unlike $\breve k_\delta(v)$, the estimate $\hat k_\delta(v)$ is a decreasing upper bound on $k(v)$ as $\delta$ increases: 

\begin{thm} \label{lemma:upper_bound}
	$\forall v\in V$, $k(v) \leq \hat k_\delta(v) \leq \hat k_{\delta-1}(v) \cdots \leq \hat k_1(v) \leq \hat k_0(v) = d(v)$ for any $\delta \geq 1$.
\end{thm}
\begin{proof}
We first prove that $k(v)\leq \hat k_\delta(v)$ for all $\delta\geq 0$ by induction on $\delta$.  The base case $k(v) \leq \hat k_0(v) = d(v)$ holds, since 
core number is always bounded by degree. Assume $k(v)\leq\hat k_\delta(v)$. Then $\hat k_\delta$ is an upper bound $\psi$ as in Theorem~\ref{thm:cheng_result}, and substituting the right hand side with Definition~\ref{def:khat}, we have $k(v)\leq\hat k_{\delta+1}(v)$.

We now prove that $\hat k_{j+1}(v) \leq \hat k_{j}(v)$ for all $j\geq 0$ and $\forall v \in V$ by induction on $j$.  Combining Definition~\ref{def:khat} with Lemma~\ref{lemma:alt_k_core_definition}, we see that $\hat k_\delta(v)$ is the maximum of the minimum of the functions $\hat k_{\delta-1}(u_i)$ and $d(v)-i+1$ of $1\leq i\leq d(v)$.  Since the maximum of $d(v)-i+1$ is $d(v)$, $\hat k_\delta(v)\leq d(v)$.  Because $\hat k_0(v) = d(v)$ the base case $\hat k_1(v)\leq \hat k_0(v)$ is satisfied. Suppose that for some $j \geq 0$, $\hat k_{j}(v) \leq \hat k_{j-1}(v)$ for all $v \in V$. By Theorem~\ref{thm:cheng_result}, $v$ has at least $\hat k_{j}(v)$ neighbors that satisfy $\hat k_{j-1}(u) \geq \hat k_{j}(v)$ for $u\in N_1(v)$.  Each such $u$ also satisfies $\hat k_{j}(u)\leq \hat k_{j-1}(u)$, meaning $v$ can have no more than $\hat k_{j}(v)$ neighbors that satisfy $\hat k_{j}(u)\geq \hat k_{j}(v)$.  Thus $\hat k_{j+1}(v)\leq \hat k_j(v)$.
\end{proof}

\subsection{Structures leading to error}\label{sec:tree_error}
One natural question is whether either $\hat k_\delta$ or $\breve k_\delta$ has bounded error (is a constant-factor approximation of the core number). Unfortunately, there are extremal constructions forcing unbounded error for both estimators; both are based on $T_{j,\ell}$, the complete $j$-ary tree with $\ell$ levels (labelled so that level $i$ has $j^{i-1}$ vertices), rooted at a vertex $v$ (Figure \ref{fig:tree}).  

\begin{lemma}\label{lem:khat_tree}
For all $\delta \geq 0$ and integers $x \geq 1$, there exists a graph $G$ and vertex $v$ so that $\hat k_\delta(v) - k(v) = x-1$. 
\end{lemma}
\begin{proof} 
First note that since $T_{j,\ell}$ is a tree, it is 1-degenerate. We show that the root vertex $v$ has $\hat k_\delta$ estimators with unbounded error. For any vertex $u$ with level number $i$ in $[2, \ell-\delta-1]$, every vertex not equal to $v$ in $u$'s $\delta$-neighborhood has degree $j+1$.  As a result, $d(u) = \hat k_1(u) = \cdots = \hat k_{i-1} = j+1$, and $\hat k_{i} = \cdots = \hat k_{\delta} = j$ (since $v$ has degree $j$ and will have propagated inwards).
Thus, for any $\delta \leq \ell-1$, we have $\hat k_\delta(v)-k(v)=j-1$. 
\end{proof}

\begin{figure}
\centering
\begin{tikzpicture}
    \node[vertex,draw,fill=blue!40] (n11) at (4.5,6)  {$v$};
    \node[vertex,draw,fill=blue!40] (n21) at (2.5,5)  {};
    \node[vertex,draw,fill=blue!40] (n22) at (6.5,5)  {};
    \node[vertex,draw,fill=blue!40] (n31) at (1.5,4)  {};
    \node[vertex,draw,fill=blue!40] (n32) at (3.5,4)  {};
    \node[vertex,draw,fill=blue!40] (n33) at (5.5,4)  {};
    \node[vertex,draw,fill=blue!40] (n34) at (7.5,4)  {};
    \node[vertex,draw,fill=blue!40] (n41) at (1,3)  {};
    \node[vertex,draw,fill=blue!40] (n42) at (2,3)  {};
    \node[vertex,draw,fill=blue!40] (n43) at (3,3)  {};
    \node[vertex,draw,fill=blue!40] (n44) at (4,3)  {};
	\node[vertex,draw,fill=blue!40] (n45) at (5,3)  {};
	\node[vertex,draw,fill=blue!40] (n46) at (6,3)  {};
	\node[vertex,draw,fill=blue!40] (n47) at (7,3) {};
	\node[vertex,draw,fill=blue!40] (n48) at (8,3) {};
	\node[vertex,draw,fill=red!40] (n51) at (3,1) {$w_1$};
	\node[vertex,draw,fill=red!40] (n52) at (6,1) {$w_2$};

    \foreach \from/\to in {n11/n21,n11/n22,
    	n21/n31,n21/n32,n22/n33,n22/n34,
    	n31/n41,n31/n42,n32/n43,n32/n44,
    	n33/n45,n33/n46,n34/n47,n34/n48}
    \draw(\from) -- (\to);
	\foreach \from/\to in {n51/n41,n51/n42,n51/n43,n51/n44,n51/n45,n51/n46,n51/n47,n51/n48,
		n52/n41,n52/n42,n52/n43,n52/n44,n52/n45,n52/n46,n52/n47,n52/n48}
      \draw[dashed](\from) -- (\to);
\end{tikzpicture}
\caption{$T_{2,4}$ is in blue.  $T'_{2,4}$ is $T_{2,4}$ plus $w_1$, $w_2$ (red), and the dashed edges. \label{fig:tree}}
\end{figure}
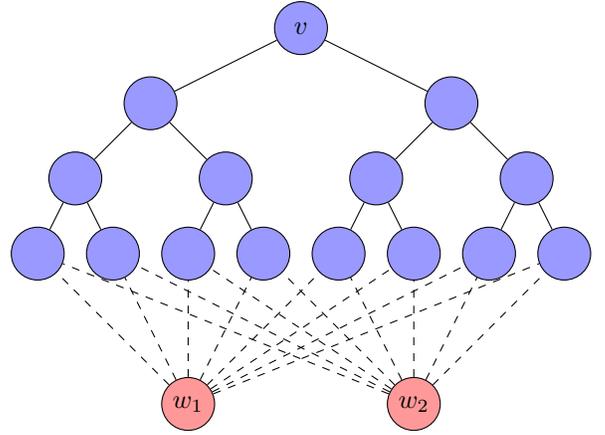

\begin{lemma} For all $\delta \geq 0$ and integers $x \geq 1$, there exists a graph $G$ and vertex $v$ so that $\breve k_\delta(v) - k(v) = x-1$. 
\end{lemma}
\begin{proof}
Consider the graph $T'_{j,\ell}$ created by adding vertices $w_1,\dots, w_j$ to $T_{j,\ell}$ and then connecting each of them to each of the leaves of $T_{j,\ell}$ (see Figure \ref{fig:tree}).  Then the root $v$ is the vertex of minimum degree in $T'_{j,\ell}$ and has core number $j$.  Any induced subgraph of $T'_{j,\ell}$ that does not include at least one $w_i$ is a tree, making it 1-degenerate.  Thus $k(v) - \breve k_\delta(v)= j-1$ whenever $\ell \geq \delta + 1$.
\end{proof}

Note that in $T_{j,\ell}$, $\breve k_\delta(v) = k(v)$ for any $\delta>0$.  Likewise, in $T'_{j,\ell}$, $\hat k_\delta(v) = k(v)$ for any $\delta \geq 0$.  Despite the fact that the errors of both estimators can theoretically be arbitrarily large, structures causing egregious errors (like $T_{j,\ell}$ and $T'_{j,\ell}$) are unlikely to occur in real world networks; we provide evidence to support this claim in the next sections.

\subsection{Expected behavior on random graphs}
In order to better understand the errors generated when approximating core number with $\hat k_\delta$, we analyze its behavior on a well-studied random graph model.

\begin{definition}[\hspace{-0.05em}\cite{erdos:random}]
{\em \ER random graphs}, denoted $\mathcal{G}(n,p)$, are the family of graphs with $n$ vertices constructed by placing an edge between each pair of vertices uniformly at random with probability $p$.
\end{definition}

To avoid confusion, we use $\mathcal{G}(n,p)$ to denote the set of all \ER random graphs with $n$ vertices and edge probability $p$ and $G(n,p)$
for a specific instance. Since all graphs on $n$ vertices occur  in $\mathcal G(n,p)$ with non-zero probability (when $p \in (0,1)$), analysis 
typically focuses on whether a graph property is very likely (or unlikely) to occur as the size of an \ER random graph grows large. In keeping with prior work (\cite{pittel:sudden,janson:normality,luczak:size}), we assume the average degree is fixed, letting $(n-1)p = \bar d$, a constant. Under this assumption, we using the following notion of ``very likely'':

\begin{definition}
A random event $X$ is said to happen {\em asymptotically almost surely} (a.a.s.) if $\lim_{n\to \infty} \mathbb{P}[X] = 1$.
\end{definition}

Specifically, we focus our attention on the growth of the error term $\hat k_1(v)-k(v)$ as $n\to \infty$ by deriving probabilistic expressions for $\hat k_1(v)$ and $k(v)$ for any $v\in G(n,p)$, then demonstrating how each term grows with $n$ compared to a function in $\omega(1)$ (recall a function $f(n)$ is $\omega(1)$ if $\lim_{n\to\infty}\frac{1}{f(n)} = 0$).
\begin{thm}\label{thm:gnp_behavior}
Suppose $\epsilon(n)$ is $\omega(1)$.  Then for any $v\in G(n,p)$, 
$\hat k_1(v)-k(v)$ is $O(\epsilon(n))$ asymptotically almost surely.
\end{thm}
\begin{proof} 
Fix a vertex $v$, and let $S_{>\kappa}$, $S_{<\kappa}$, and $S_{=\kappa}$ be defined to be the subsets of $N_1(v)$ with vertices of degree greater than, equal to, or less than $\kappa$, respectively. We first evaluate $\mathbb{P}[\hat k_1(v) = \kappa | d(v) = d]$.
By Definition \ref{def:khat}, if $\hat k_1(v) = \kappa$, $v$ has at least $\kappa$ neighbors $u$ with $\hat k_0(u) \geq \kappa$ but less than $\kappa + 1$ with $\hat k_0(u) \geq \kappa+1$ (or else $\hat k_1(v) > \kappa$).  Therefore, $\hat k_1(v) =\kappa$ implies $|S_{>\kappa}|\leq \kappa$ and $|S_{=\kappa}|+|S_{>\kappa}|\geq \kappa$, and
	\begin{multline*}
		\mathbb{P}[\hat k_1(v) = \kappa | d(v) = d] = \\
		\sum_{i=0}^{\kappa} \sum_{j=0}^{d-\kappa} \frac{d!}{i!j!x!} \mathbb{P}[(|S_{>\kappa}| = i)\land(|S_{<\kappa}| = j)\land(|S_{=\kappa}| =x)],
	\end{multline*}
where $x=d-i-j$.

	As $n$ tends to infinity, the probability that any two neighbors of $v$ have an edge between them approaches $0$.  Therefore, the degrees of  $v$'s neighbors can be treated as independent, identical distributions in the limit.  Since $n$ is large and $\bar{d}$ is fixed, this distribution is asymptotically Poisson with mean $\bar{d}$.  If $\zeta_{\lambda}(k)$ and $Z_{\lambda}(k)$ denote the Poisson probability mass function and cumulative distribution function, respectively, with mean $\lambda$ (evaluated at $k$), then:
	\begin{multline}\label{eq:gnp_k_hat_given}
		\mathbb{P}[\hat k_1(v) = \kappa | d(v) = d] = \\
		\sum_{i=0}^{\kappa} \sum_{j=0}^{d-\kappa} \frac{d!}{i!j!x!}
		(1-Z_{\bar d}(\kappa-1))^i
		Z_{\bar d}(\kappa-2)^j
		\zeta_{\bar d}(\kappa-1)^{x}.
	\end{multline}
By computing Equation \ref{eq:gnp_k_hat_given} at each value of $d$ for which $\kappa$ is a possible value for the core number, we have:
\begin{equation}\label{eq:gnp_k_hat}
	\mathbb{P}[\hat k_1(v) = \kappa] =
	\sum_{d=\kappa}^{n-1} \zeta_{\bar{d}}(d)\cdot\mathbb{P}[\hat k_1(v) = \kappa|d(v) = d].
\end{equation}
Pittel et al.~\cite{pittel:sudden} demonstrated that in $\mathcal{G}(n,p)$, the proportion of vertices in the $k$-core is a.a.s.\ a function of $\bar d$ but not of $n$.  Moreover, for any vertex $v$, $k(v)$ is a.a.s.\ bounded by a constant (equivalently, in $\Theta(1)$). If $\hat k_1(v)$ were also bounded by a constant, then the error term $\hat k_1(v)-k(v)$ would be a.a.s.\ $O(1)$.  However, since the Poisson random variables in Equations \ref{eq:gnp_k_hat_given} and \ref{eq:gnp_k_hat} are only parameterized by $\bar{d}$ and not by $n$, the proportion of vertices in $G(n,p)$ with $\hat k_1 = \kappa$ is a.a.s.\ convergent to some non-zero constant.
Thus, the probability of having an arbitrarily large value of $\hat k_1$ does not vanish as $n$ grows large for a fixed (constant) $\kappa$.

Let $\kappa$ be a function of $n$ in $\omega(1)$.  By Stirling's approximation of the factorial,
\[
	\zeta_{\bar{d}}(\kappa) \approx \frac{e^{-\bar{d}}}{\sqrt{2\pi \kappa}}\left(\frac{e\bar{d}}{\kappa}\right)^\kappa.
\]
Then as $n$ grows large, $\zeta_{\bar d}(\kappa) \to 0$ and $Z_{\bar d}(\kappa)\to 1$.  In Equation \ref{eq:gnp_k_hat_given}, the probability that a neighbor $u$ of vertex $v$ has degree at least $\kappa$ (that is, $u\in S_{>\kappa} \cup S_{=\kappa}$) is asymptotically zero, and consequently $\mathbb{P}[\hat k_1(v)=\kappa]$ also vanishes in the limit. This implies that a.a.s.\ $\mathbb{P}[\hat k_1(v)=\kappa]\in O(\epsilon(n))$ for any error function $\epsilon(n) \in \omega(1)$.
Using the result of~\cite{pittel:sudden}  that $k(v)\in \Theta(1)$, we have a.a.s.\ $\hat k_1(v)-k(v) \in O(\epsilon(n))$.
\end{proof}

\section{Experimental Results}
In the previous section, the behavior of the propagating estimator $\hat k_\delta$ was analyzed from a theoretical perspective.  In order to enhance this picture, we present computational results on a corpus of real data.
\subsection{Methods}
The estimators $\hat k_\delta$ and $\breve k_\delta$ were evaluated on nine real-world graphs that appear in the following section\footnote{We also tested several additional graphs, which gave qualitatively similar results, and are thus omitted for length.  The results can be found in the arXiv version of this paper.}.  Not only do these graphs cover a variety of domains, but they also are structurally dissimilar. The graphs vary in size, density, core structure, and diameter (see Table~\ref{tab:data_summary} and Figure~\ref{fig:core_dist_real}). 
\begin{table}[!t]
	\centering
	{
	\begin{tabular}{|c|c|c|c|c|c|}
	\hline
	\bf Graph & \bf $|V|$ & \bf $|E|$ &  \bf $\operatorname{max} d(v)$ & \bf $D$ & \bf $\Delta$  \\\hline

	\textsc{Amazon}\cite{snap}  & 334863 & 925872 &  549 & 6& 47\\
	Co-purchases &&&&& \\\hline
	\textsc{AS}\cite{snap} & 6214 & 12232 & 1397 & 12 & 9 \\
	Autonomous systems &&&&& \\\hline
	\textsc{ca-AstroPh}\cite{snap}  & 17903 & 196972 &  504 & 56 & 14 \\
	Academic citations &&&&& \\\hline
	\textsc{DBLP}\cite{snap}  & 317080 & 1049866 &  343 & 113& 23 \\
	Academic citations &&&&& \\\hline
	\textsc{Enron}\cite{snap}  & 33696 & 180811 &  1383 & 43 & 13\\
	Email correspondence &&&&& \\\hline
	\textsc{Facebook}\cite{traud:facebook}  & 36371 & 1590655 &  6312 & 81 & 7\\
	Facebook friendship &&&&& \\\hline
	\textsc{Gnutella}\cite{snap}  &  26498 & 65359 &  355 & 5 & 11\\
	Peer-to-peer filesharing &&&&& \\\hline
	\textsc{H.~sapiens}\cite{biogrid}  & 18625 & 146322 & 9777 & 47 & 10\\
	Protein-protein interation &&&&& \\\hline
	\textsc{WPG}\cite{iri}  &4941 & 6594 &  19 & 5 & 46\\
	Western US power grid &&&&& \\\hline
	\end{tabular}
	}
	\caption{
	Summary statistics for real-world graphs \label{tab:data_summary}
	}
\end{table} 

We computed the core number $k(v)$, as well as the values of $\hat k_\delta(v)$ and $\breve k_\delta(v)$ for each vertex\footnote{Code and data are available at \texttt{\url{https://dl.dropboxusercontent.com/u/32167511/core_number_estimate.zip}}}, letting $\delta$ vary from $0$ to $\Delta$. To compare the accuracy of the estimators among vertices, we normalize by the true core number at each vertex.  We refer to this metric as the {\em core number estimate ratio}.  When the estimator ($\hat k_\delta$ or $\breve k_\delta$) is exactly equal to the core number, the core number estimate ratio is 1, its {\em optimal} value.  Since $\hat k_\delta$ is an upper bound on $k$, its core number estimate ratio is always at least one and becomes {\em less optimal} the larger it gets; the opposite is true for $\breve k_\delta$, a lower bound.

\subsection{Results}
We first turn our attention to how often the estimators achieve optimal core number estimate ratios.  Figure~\ref{fig:unnormalized_correct} shows how the proportion of vertices with ratio one grows as $\delta$ increases from zero to four.  In all the graphs, the core number estimate ratio for $\hat k_\delta$ is optimal at least as often as that of $\breve k_\delta$ at $\delta \leq 1$.  Additionally, the proportion of vertices with optimal $\hat k_\delta$ estimate ratios is large in all the graphs (often upwards of $90\%$).  While the propagating estimator does not have as pronounced of an advantage over the induced estimator when $\delta =2$, the number of vertices with optimal $\hat k_\delta$ estimate ratios still grows noticeably.

\begin{figure}[!t]
	\centering
	\includegraphics[width=1.02\linewidth]{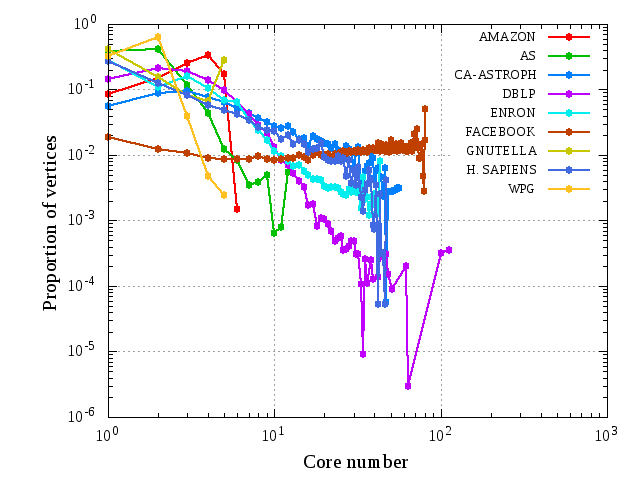}
	\caption{Core number distribution for the real world networks in Table~\ref{tab:data_summary}.\label{fig:core_dist_real}}
\end{figure}

\begin{figure}[!b]
	\centering
	\subfloat[\textsc{Amazon}]{\includegraphics[width=0.34\linewidth]{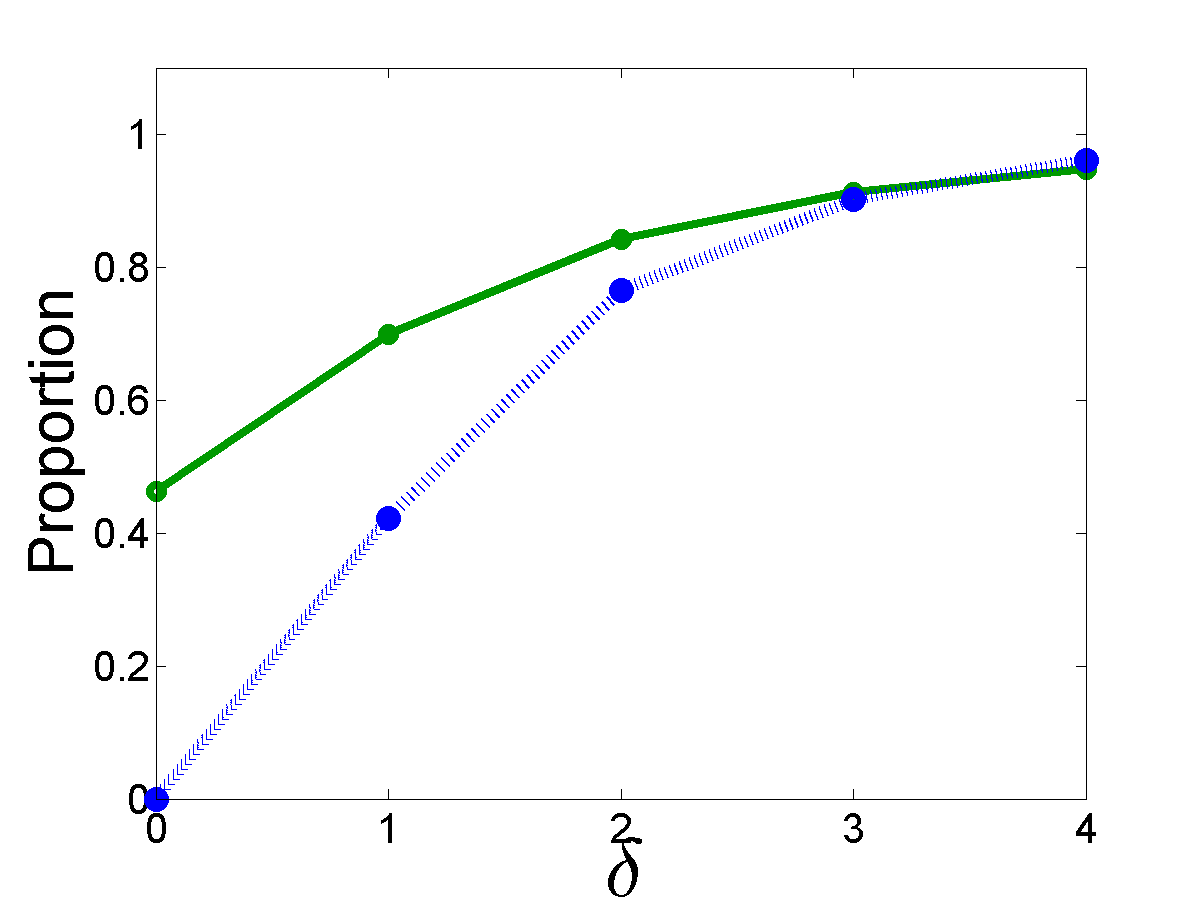}}
	\subfloat[\textsc{AS}]{\includegraphics[width=0.34\linewidth]{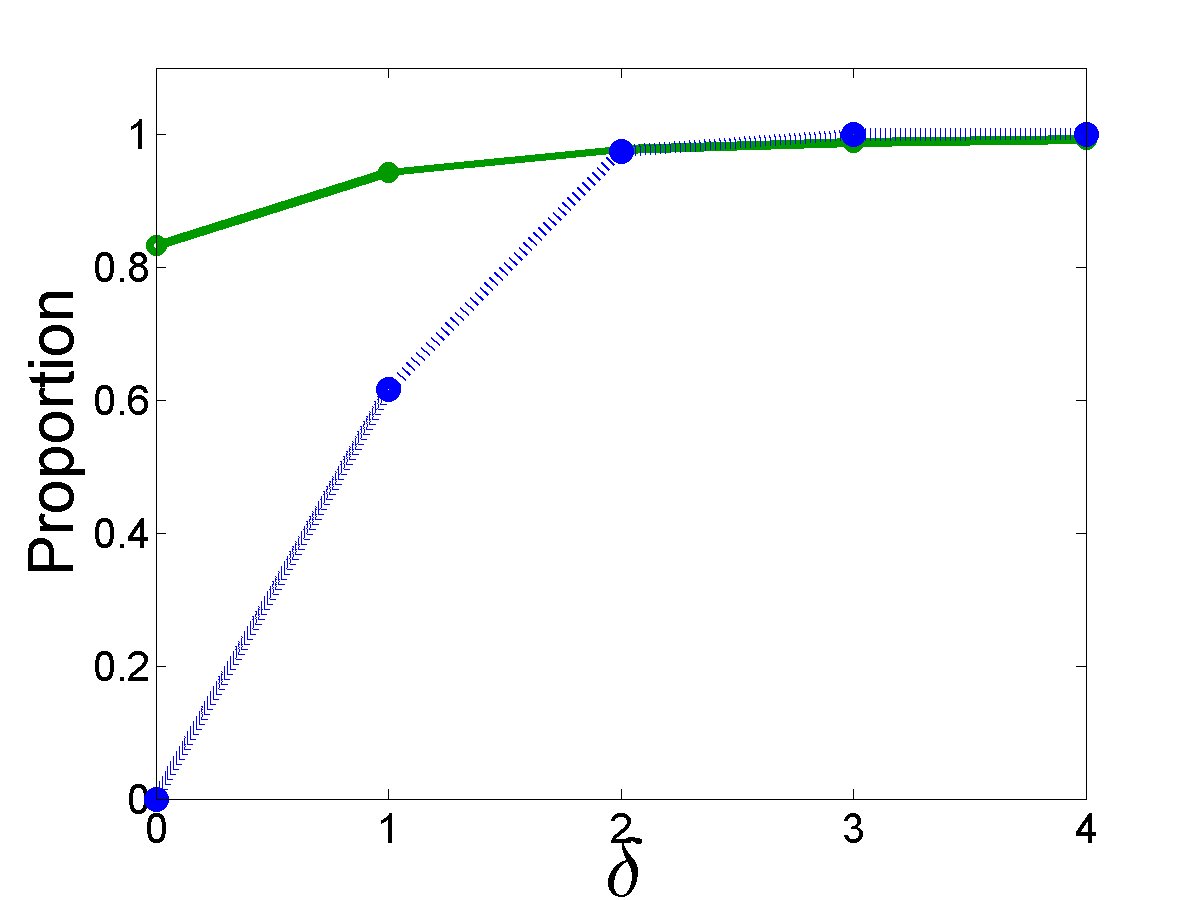}}
	\subfloat[\textsc{ca-AstroPh}]{\includegraphics[width=0.34\linewidth]{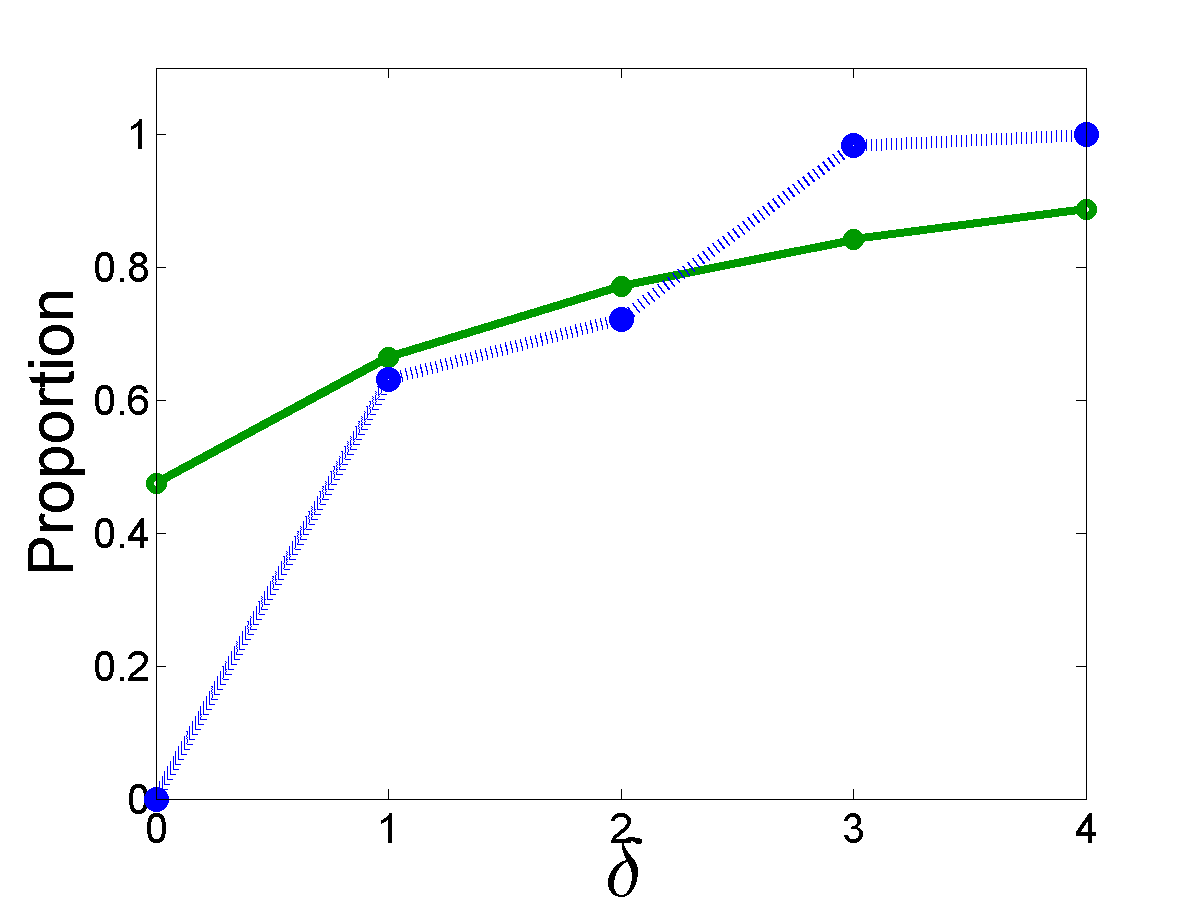}}
	\\
	\subfloat[\textsc{DBLP}]{\includegraphics[width=0.34\linewidth]{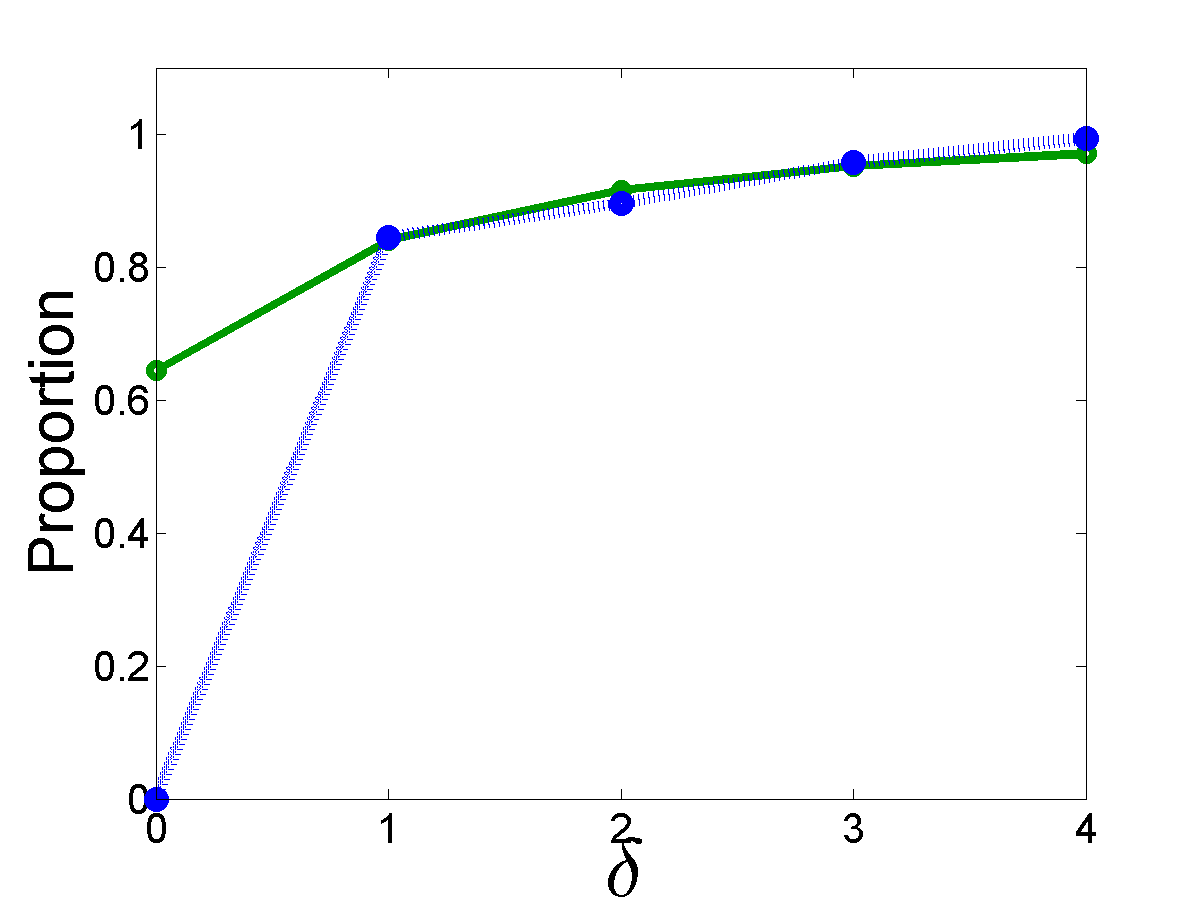}}
	\subfloat[\textsc{Enron}]{\includegraphics[width=0.34\linewidth]{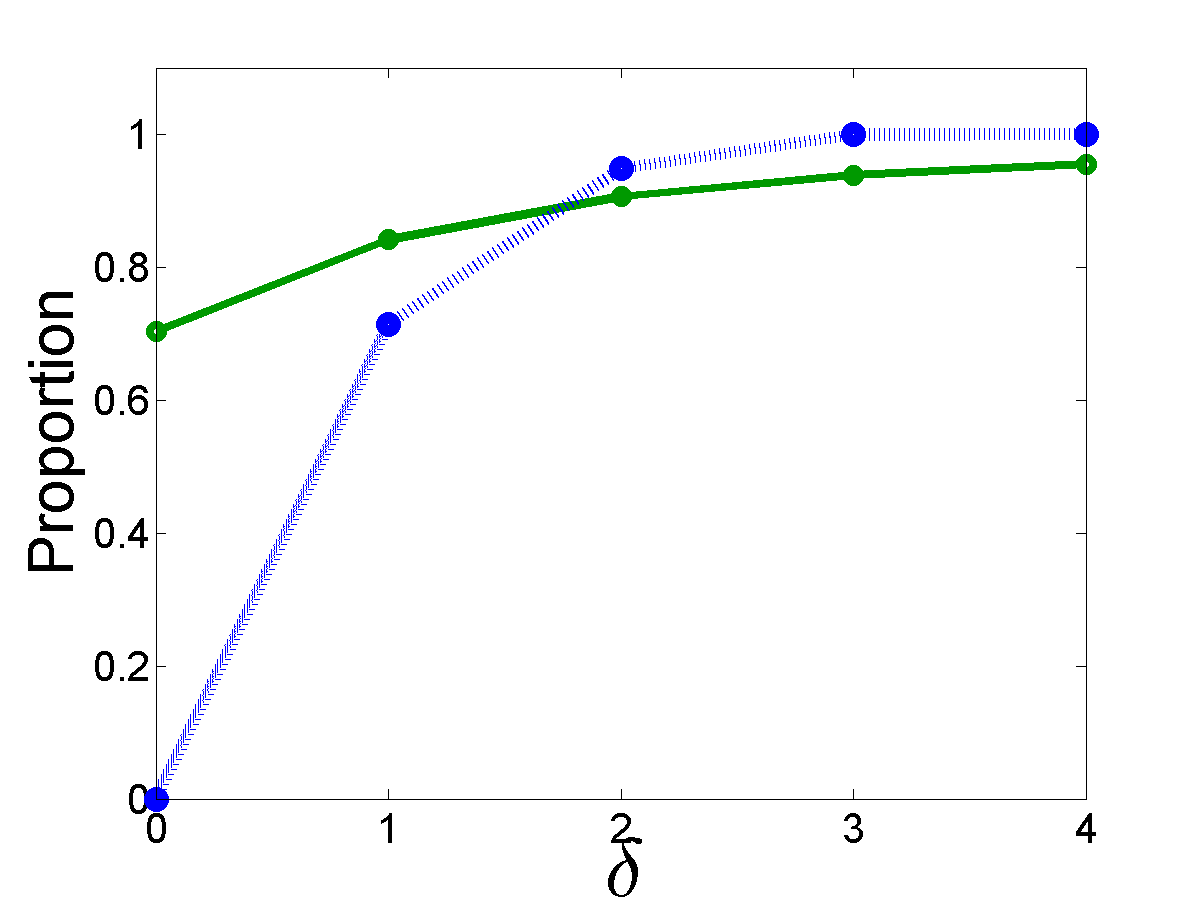}}
	\subfloat[\textsc{Facebook}]{\includegraphics[width=0.34\linewidth]{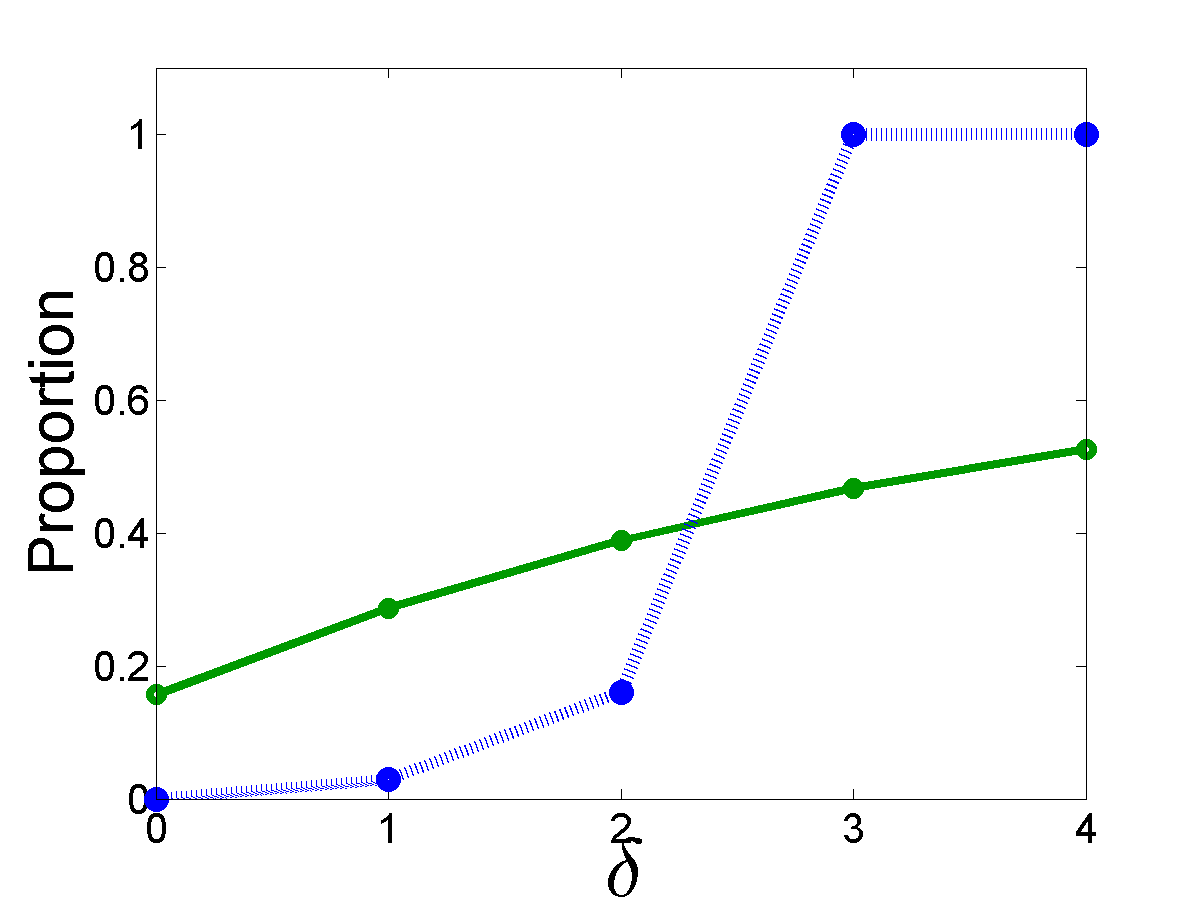}}
	\\
	\subfloat[\textsc{Gnutella}]{\includegraphics[width=0.34\linewidth]{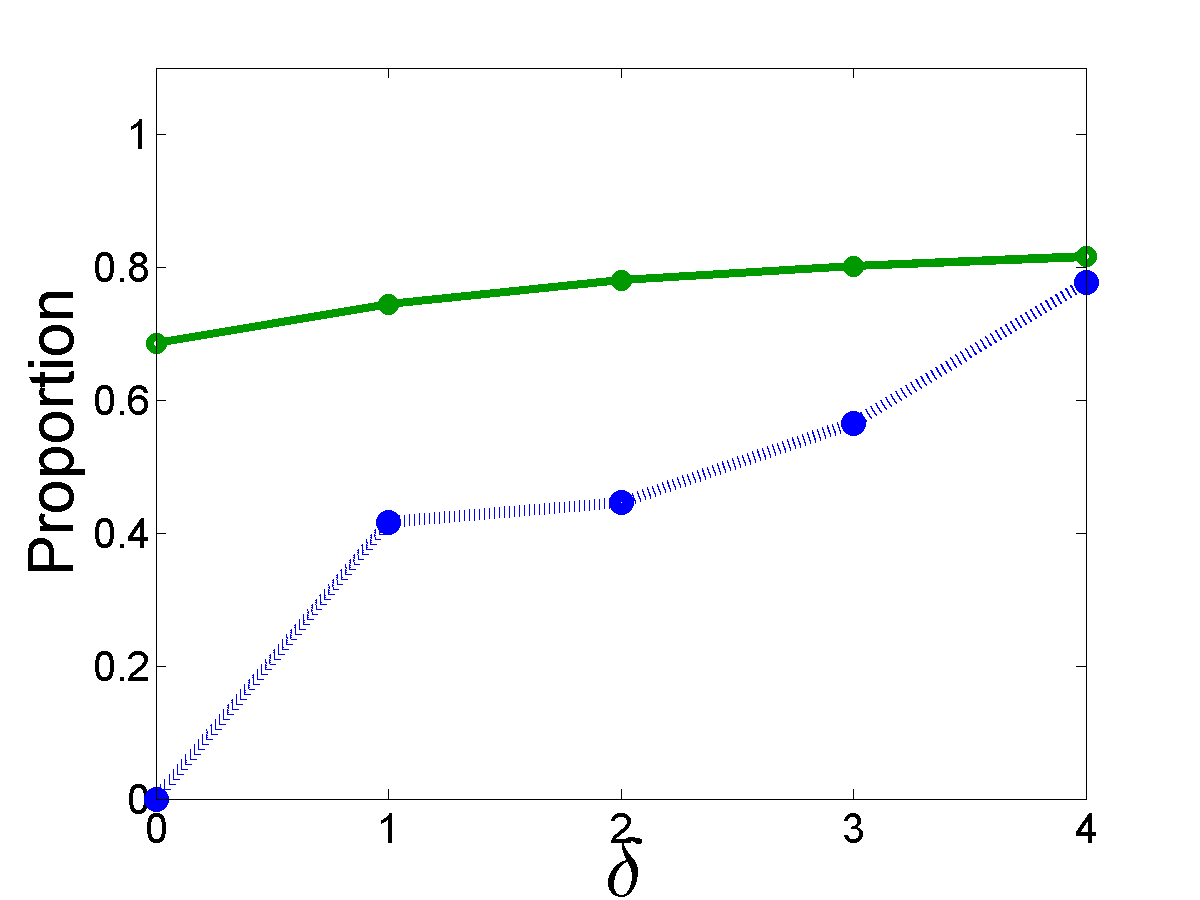}}
	\subfloat[\textsc{H.~sapiens}]{\includegraphics[width=0.34\linewidth]{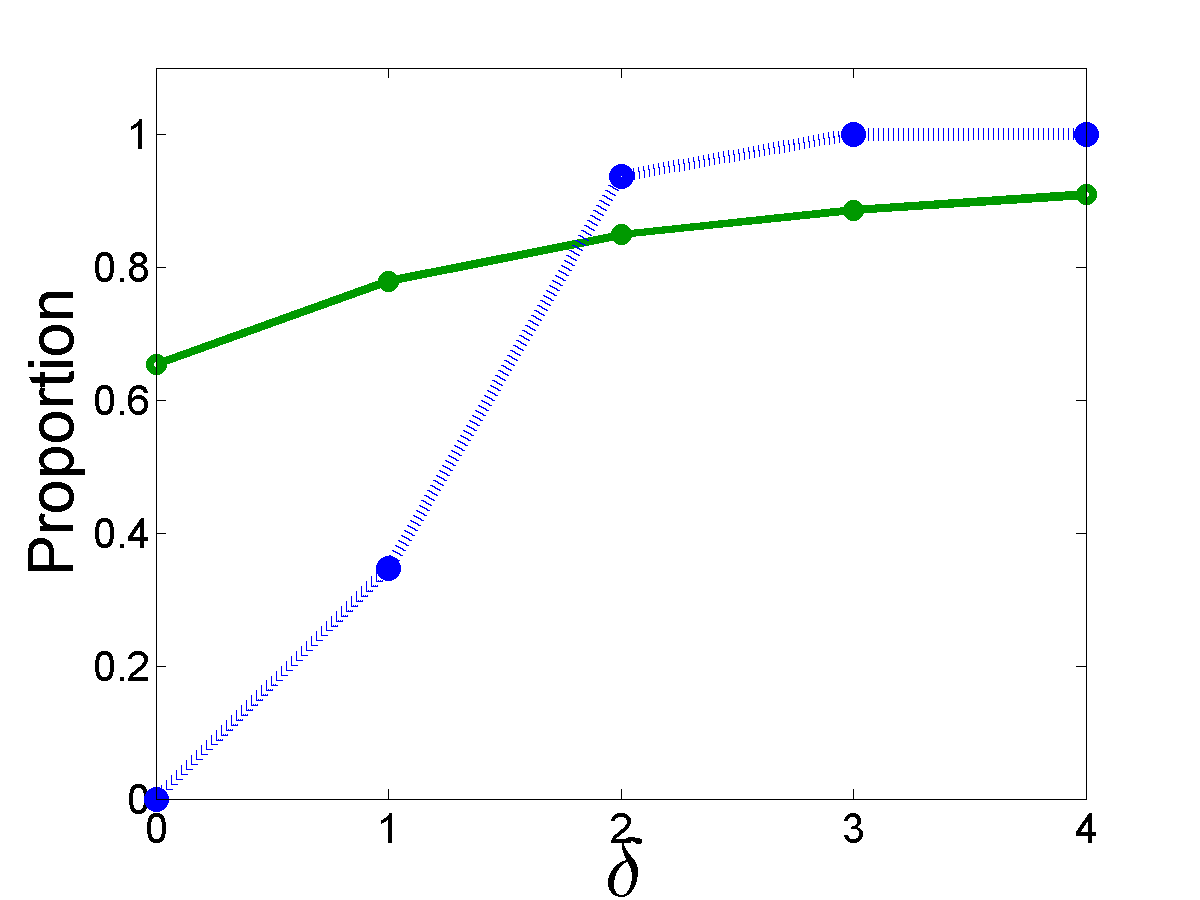}}
	\subfloat[\textsc{WPG}]{\includegraphics[width=0.34\linewidth]{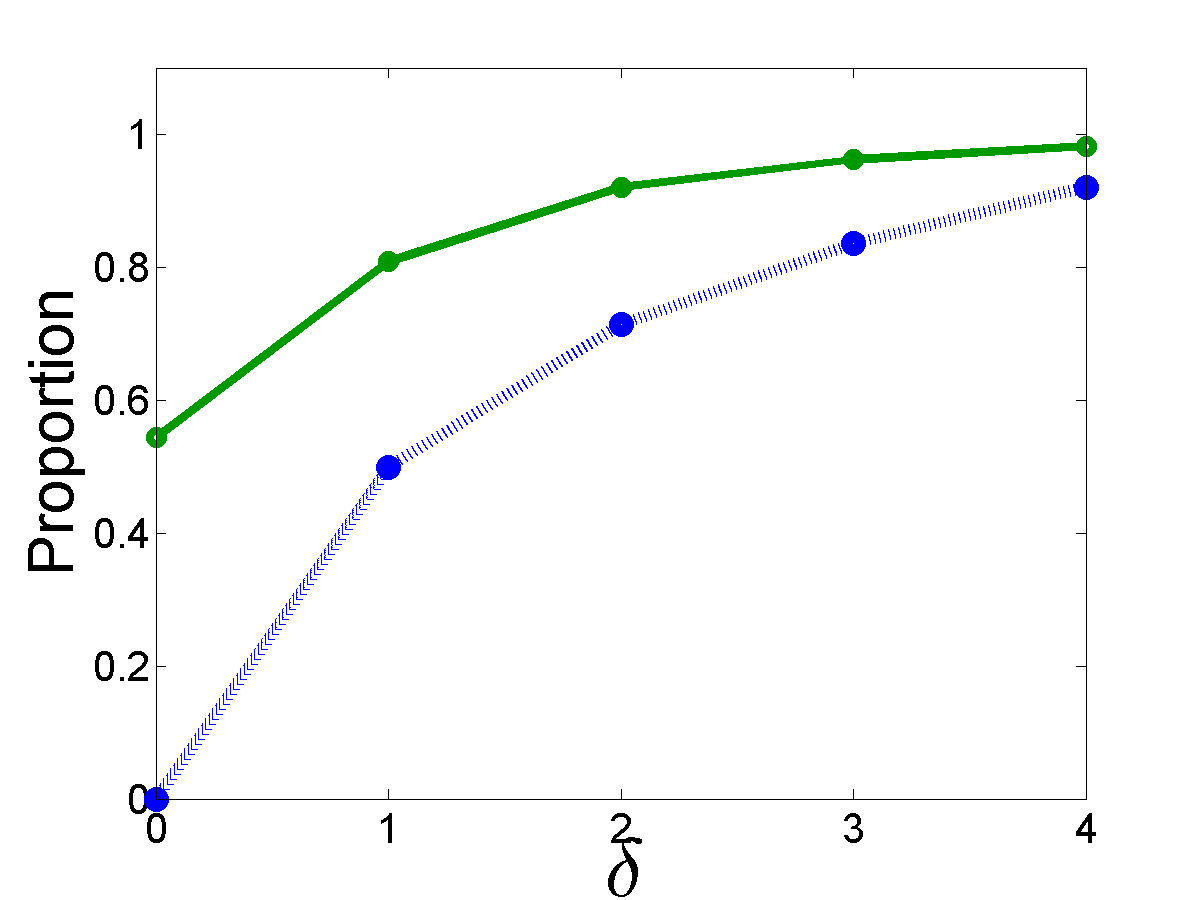}}
	\caption{Proportion of vertices with optimal core number estimate ratios for the propagating (solid green) and induced (dashed blue) estimators as a function of $\delta$.\label{fig:unnormalized_correct}}
\end{figure}

We also examined the distribution of core number estimate ratios among those vertices where the estimate was not exact.  The change in this distribution over the range $\delta =1$ to $\delta = 4$ is shown in Figure~\ref{fig:estimate_growth}, demonstrating that not only are the sub-optimal estimates closely centered around $1$, but also that increasing $\delta$ can significantly decrease the size of the ``tail'' of the distribution (thereby improving the core number estimates of those vertices with the least optimal ratios).

\begin{figure}[!b]
	\centering
	\subfloat[\textsc{Amazon}]{\includegraphics[width=0.50\linewidth]{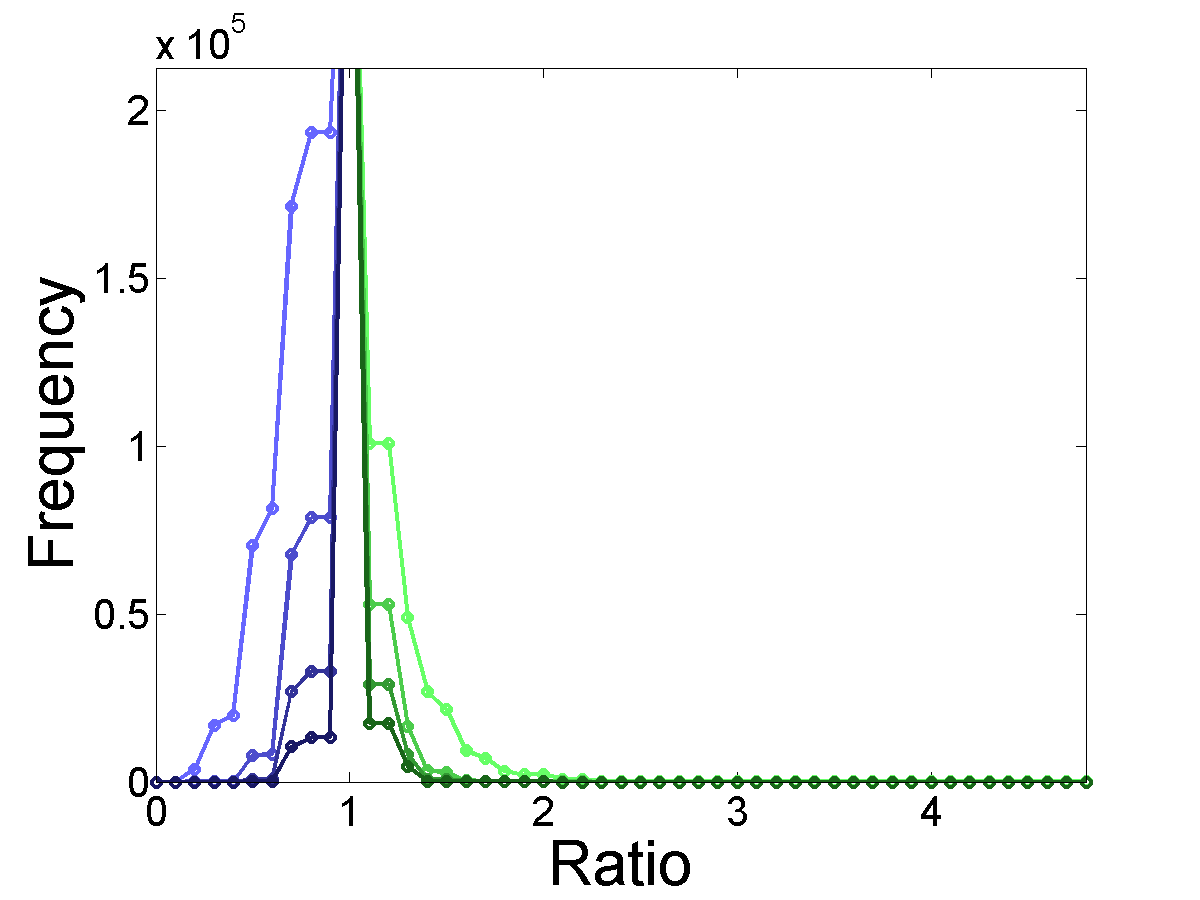}}
	\subfloat[\textsc{AS}]{\includegraphics[width=0.50\linewidth]{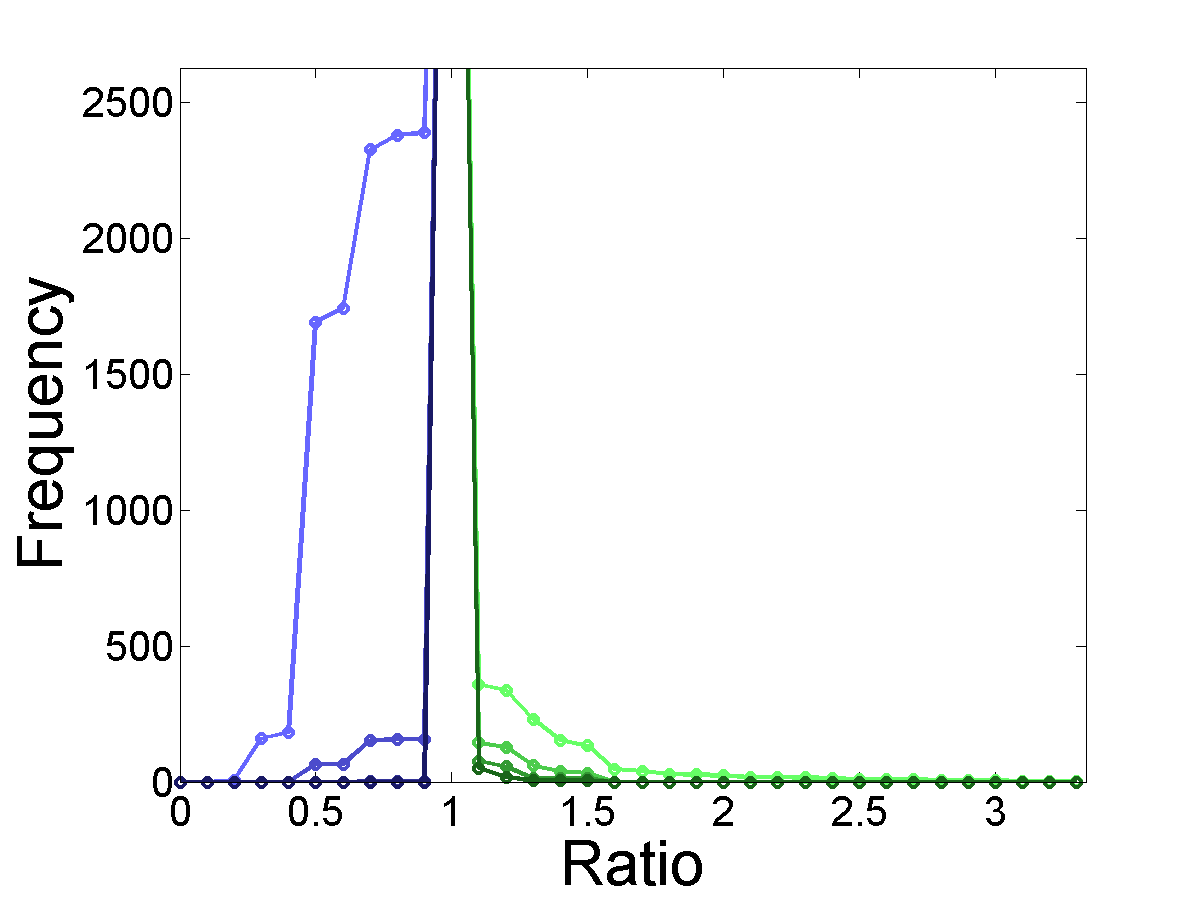}}
	\\
	\subfloat[\textsc{ca-AstroPh}]{\includegraphics[width=0.50\linewidth]{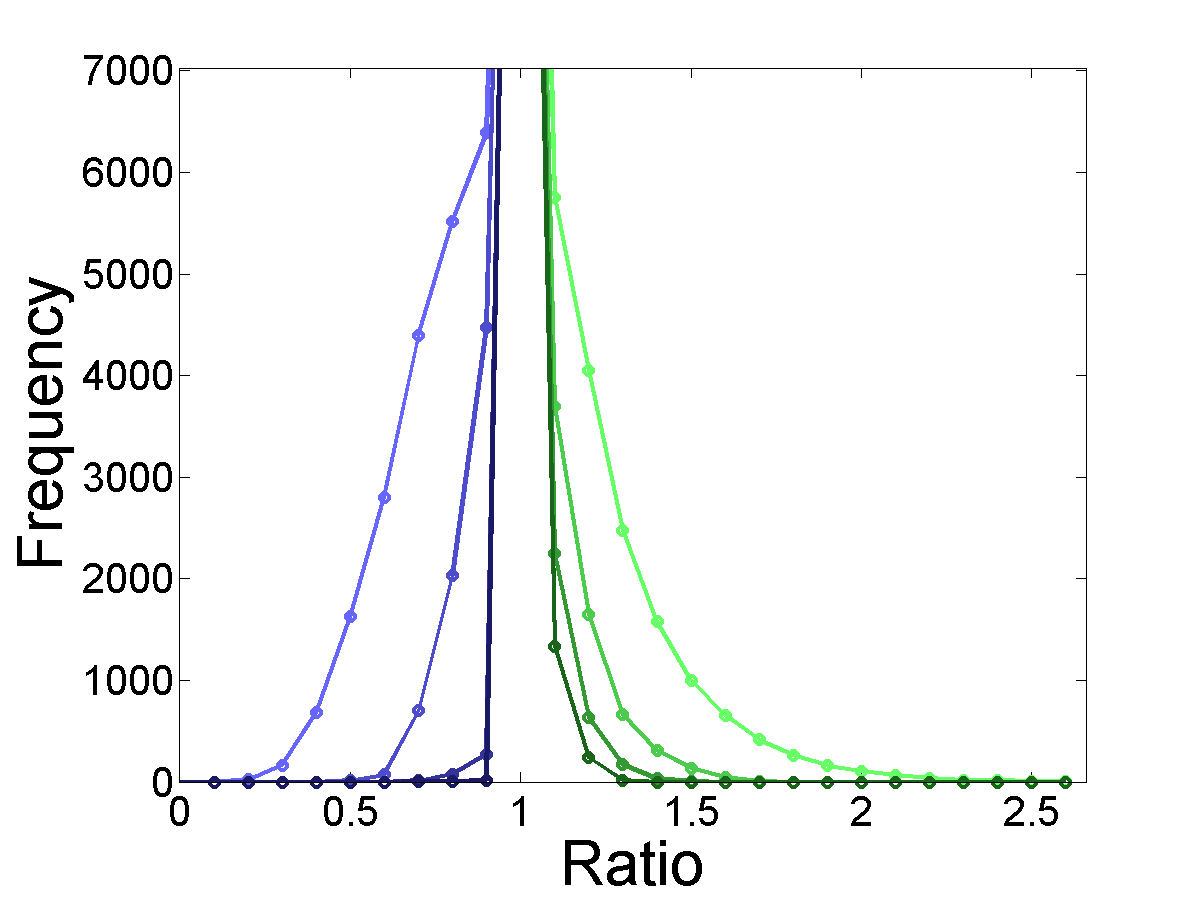}}
	\subfloat[\textsc{DBLP}]{\includegraphics[width=0.50\linewidth]{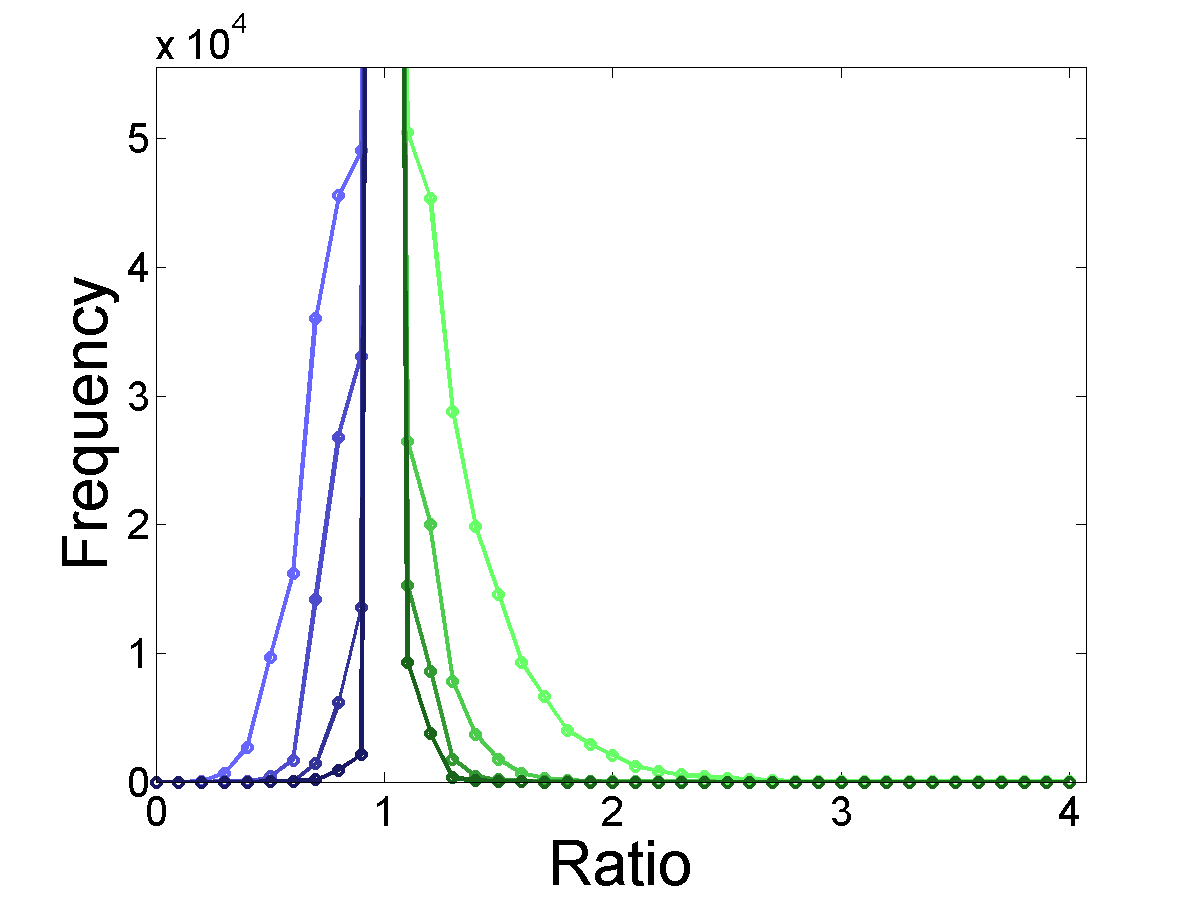}}
	\\
	\subfloat[\textsc{Enron}]{\includegraphics[width=0.50\linewidth]{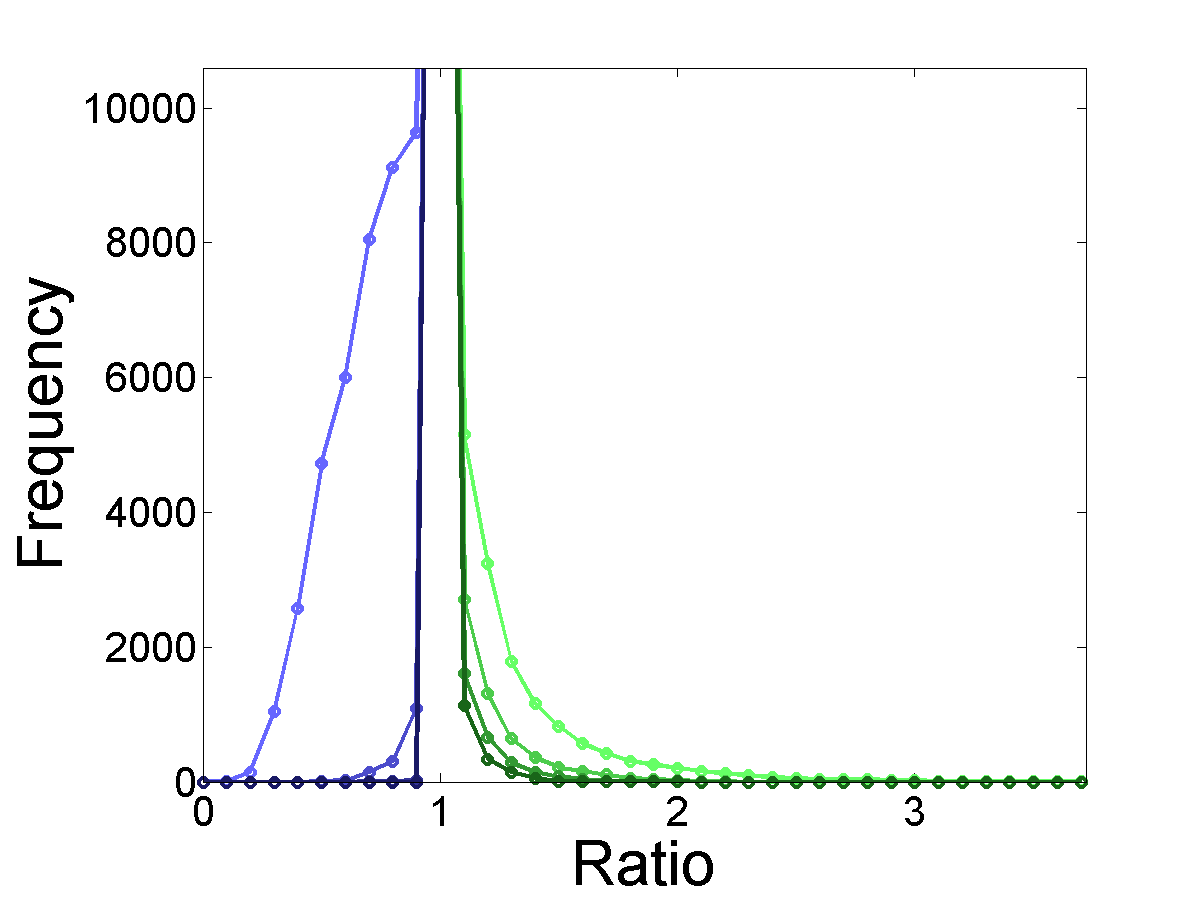}}
	\subfloat[\textsc{Facebook}]{\includegraphics[width=0.50\linewidth]{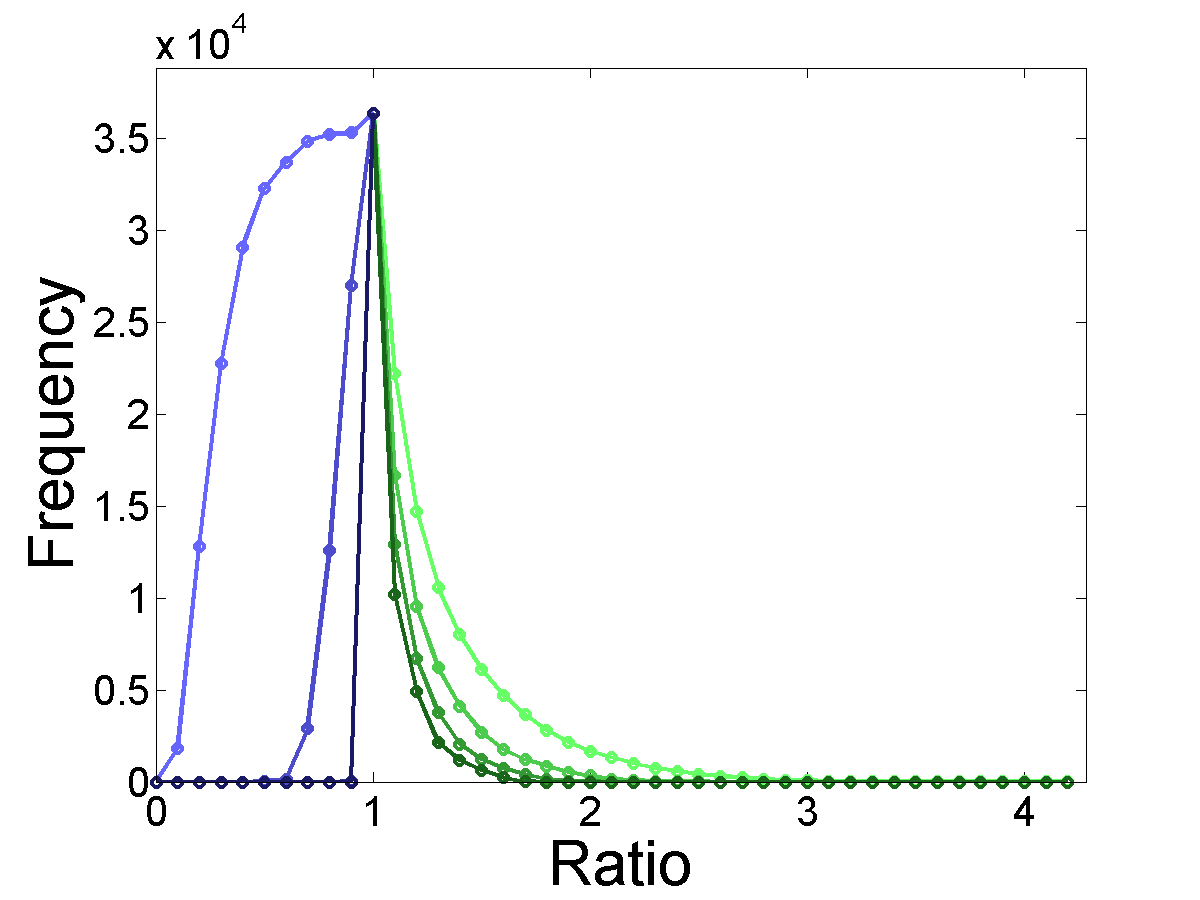}}
	\\
	\subfloat[\textsc{Gnutella}]{\includegraphics[width=0.50\linewidth]{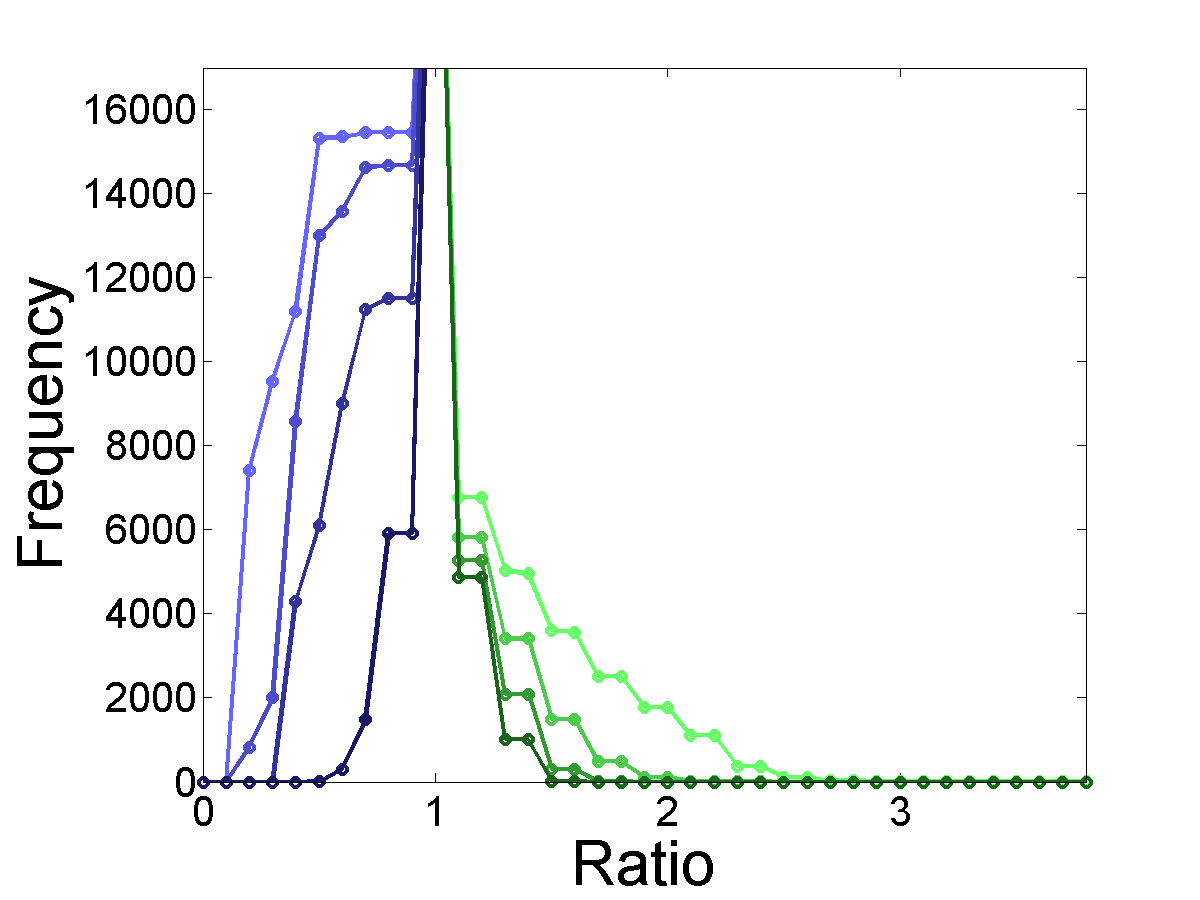}}
	\subfloat[\textsc{H.~sapiens}]{\includegraphics[width=0.50\linewidth]{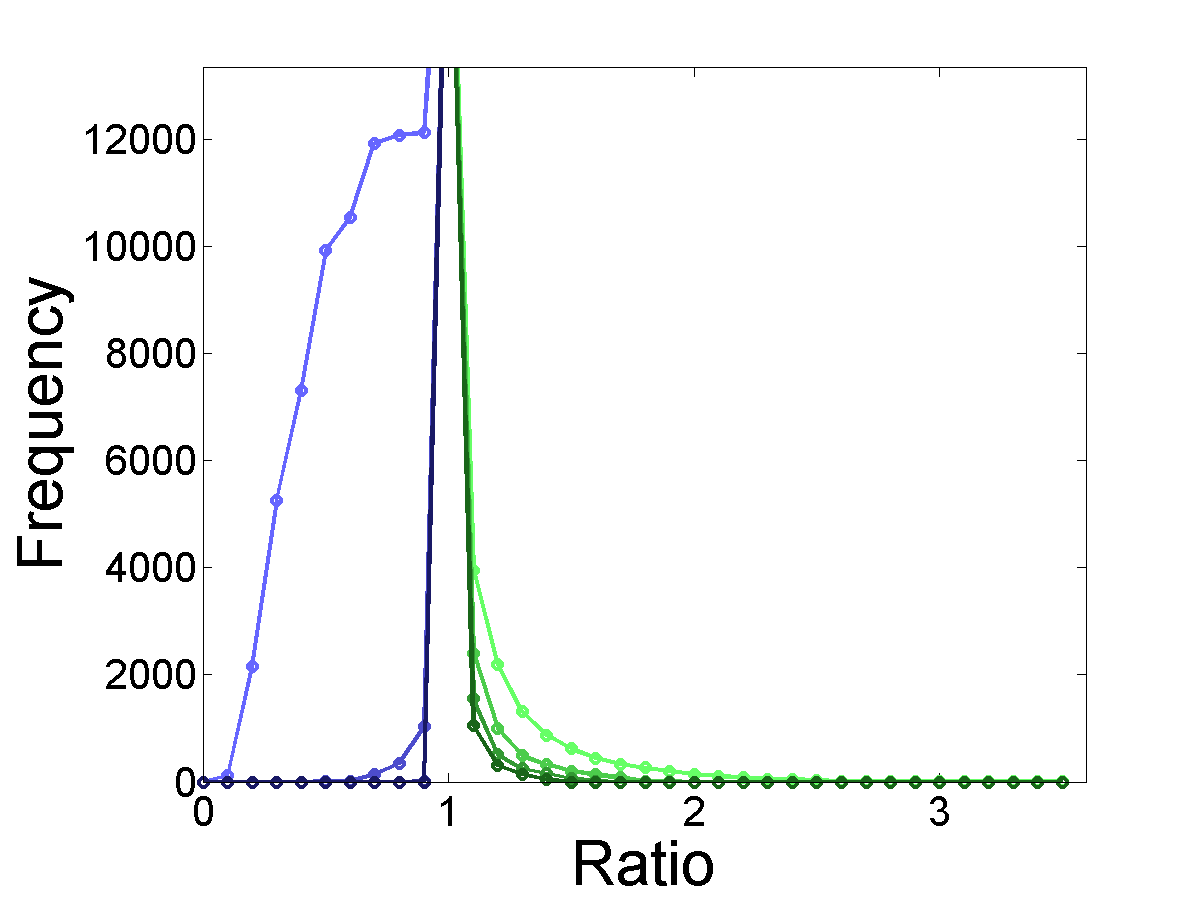}}
	\\
	\subfloat[\textsc{WPG}]{\includegraphics[width=0.50\linewidth]{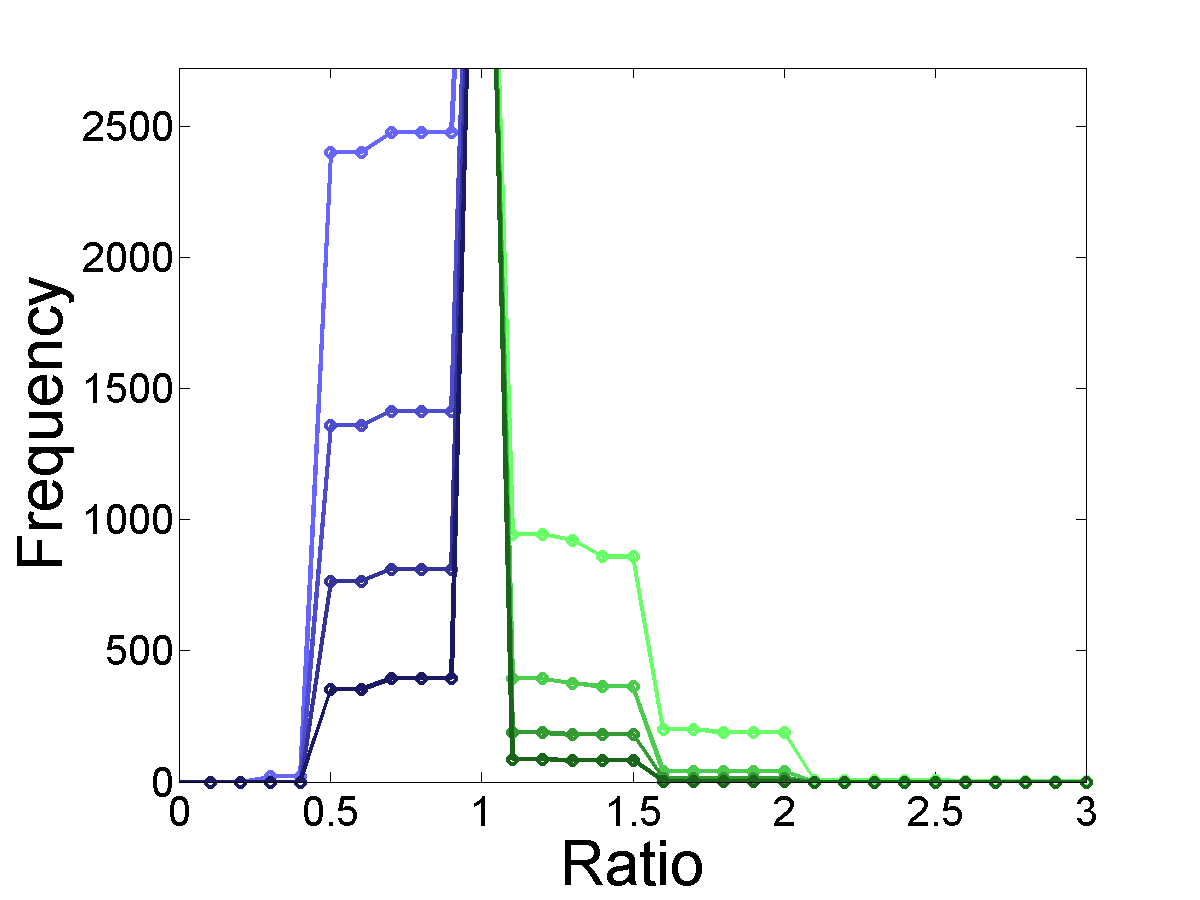}}
	\caption{Number of vertices with core number estimate ratios less optimal that a given threshold from $\delta = 1$ (lightest line) to $\delta = 4$ (darkest line).  $\hat k_\delta$ is shown in green while $\breve k_\delta$ is shown in blue.  Because the number of vertices with optimal ratios is frequently large (see Figure~\ref{fig:unnormalized_correct}), the vertices with optimal ratios may not appear within the limits of the plot in order better capture the distribution of those vertices with suboptimal ratios.\label{fig:estimate_growth}}
\end{figure}

Although Figures~\ref{fig:unnormalized_correct} and~\ref{fig:estimate_growth} suggest that $\hat k_\delta$ and $\breve k_\delta$ can accurately estimate the core numbers in real world graphs using only knowledge of the $\delta$-neighborhood with a small value of $\delta$, it is important to understand how the size of the $\delta$-neighborhood impacts the behavior of the estimates.  The purpose of having a localized estimate is to reduce the size of the input needed to compute the core number of a vertex.  If the average $\delta$-neighborhood encompasses most of the graph, then not only is this purpose defeated, but we also may not be able to judge whether the accuracy of the localized estimates is only due to having knowledge of the entire graph (as opposed to any theoretical merits of the algorithms themselves).  The mean and variance of the proportion of vertices in the $\delta$-neighborhood is shown in Table~\ref{tab:neighborhood_size}.  The rate of growth of $\delta$-neighborhood sizes varies significantly among the nine graphs, which suggests that picking a value of $\delta$ to maintain appropriately small $\delta$-neighborhoods is highly dependent on the structure of the graph.  Nonetheless, the average size of the $\delta$ neighborhood is below ten percent of the entire graph for all datasets at $\delta=1$ and in all but one (namely \textsc{Facebook}, which we know to be significantly different from the other networks) at $\delta=2$.

\begin{table}
	\centering
	\begin{tabular}{|c|cc|cc|cc|cc|}\hline
	 & \multicolumn{2}{c|}{$\delta=1$}& \multicolumn{2}{c|}{$\delta = 2$}& \multicolumn{2}{c|}{$\delta = 3$} & \multicolumn{2}{c|}{$\delta = 4$} \\
	\bf Graph & \bf Avg. & \bf Var. & \bf Avg. & \bf Var. & \bf Avg. & \bf Var. & \bf Avg. & \bf Var. \\\hline
	\textsc{Amazon} & $.00$ & $.00$ & $.00$ & $.00$ & $.00$ & $.00$ & $.00$ & $.00$   \\\hline
	\textsc{AS} & $.00$ & $.00$ & $.09$ & $.01$ & $.44$ & $.07$ & $.82$ & $.04$ 
 \\\hline
	\textsc{ca-AstroPh} & $.00$ & $.00$ & $.03$ & $.00$ & $.25$ & $.04$ & $.66$ & $.06$ 
 \\\hline
	\textsc{DBLP} & $.00$ & $.00$ & $.00$ & $.00$ & $.00$ & $.00$ & $.03$ & $.00$ 
\\\hline
	\textsc{Enron} & $.00$ & $.00$ & $.03$ & $.00$ & $.28$ & $.04$ & $.74$ & $.06$ 
 \\\hline
	\textsc{Facebook} & $0.00$ & $0.00$ & $0.21$ & $0.03$ & $0.89$ & $0.02$ & $1.00$ & $0.00$ 
 \\\hline
	\textsc{Gnutella} & $.00$ & $.00$ & $.00$ & $.00$ & $.02$ & $.00$ & $.15$ & $.02$ 
 \\\hline
	\textsc{H.~sapiens}& $0.00$ & $0.00$ & $0.31$ & $0.07$ & $0.81$ & $0.04$ & $0.98$ & $0.00$ 
 \\\hline
	\textsc{WPG} & $.00$ & $.00$ & $.00$ & $.00$ & $.00$ & $.00$ & $.01$ & $.00$ 
 \\\hline
	\end{tabular}
	\caption{Proportion of vertices in $N_\delta$.  Values less than $.01$ rounded to $0$.\label{tab:neighborhood_size}}
\end{table}

Another natural way to measure the relative amount of information in the $\delta$-neighborhood is to normalize by the diameter $\Delta$. Figure~\ref{fig:normalized_neighborhood_full} shows that the average proportion of vertices in the $\delta$-neighborhood is approximately uniform in all graphs when $\delta$ is expressed as a fraction of $\Delta$.  In particular, there is a significant increase in the rate of growth of the $\delta$-neighborhood size occurring when $\delta$ is approximately $20\%$ of the diameter. Thus, as one might expect, neighborhood-based core number estimates seem most appropriate when $\delta$ is a small fraction of the diameter.

\begin{figure}[!b]
	\centering
	\subfloat[\textsc{Amazon}]{\includegraphics[width=0.34\linewidth]{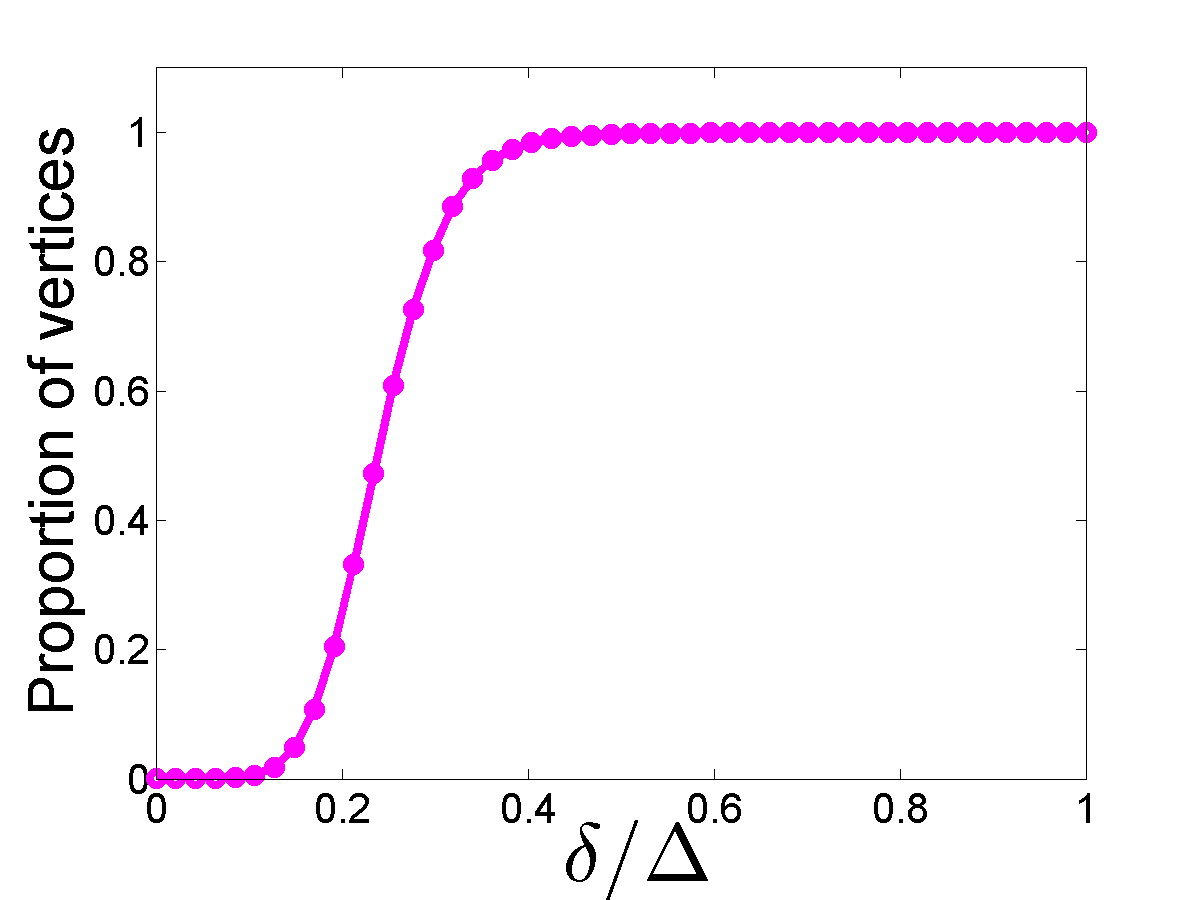}}
	\subfloat[\textsc{AS}]{\includegraphics[width=0.34\linewidth]{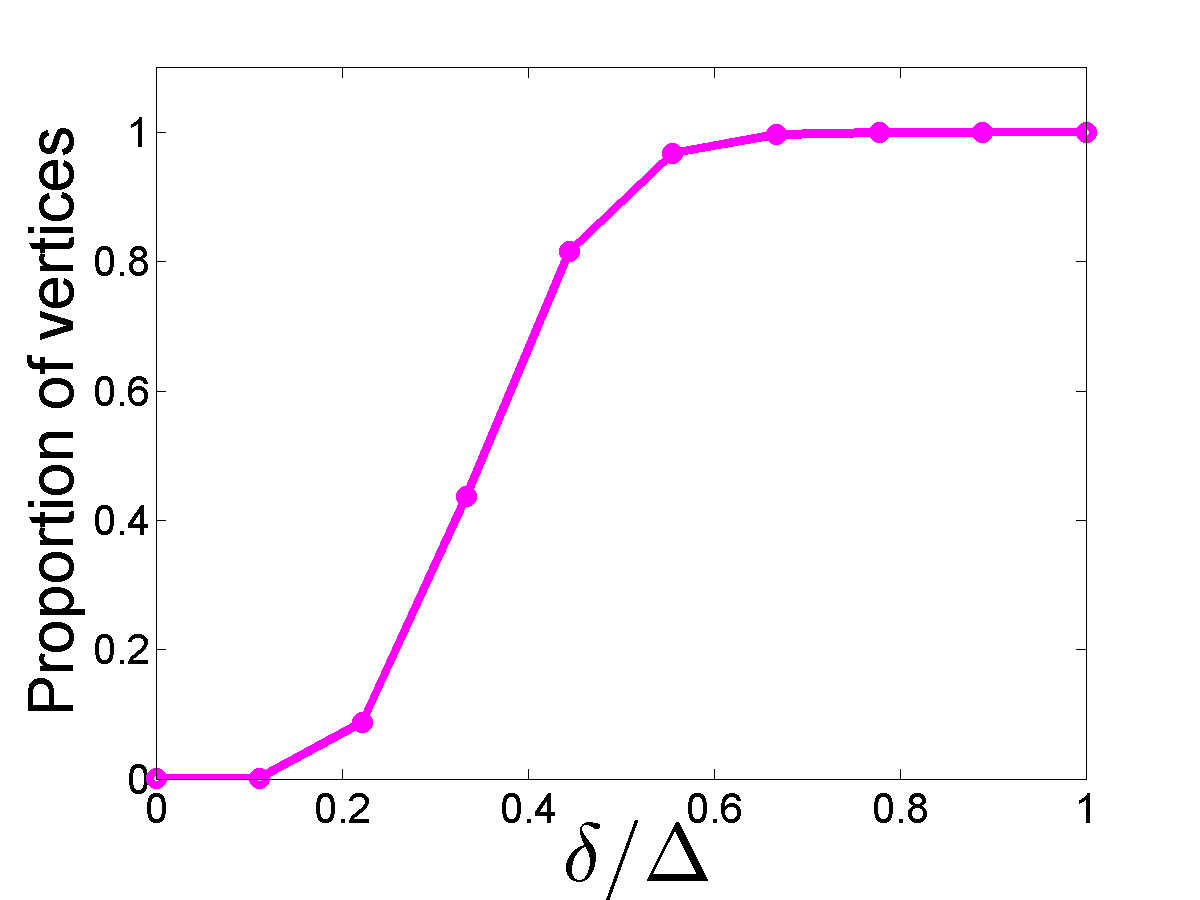}}
	\subfloat[\textsc{ca-AstroPh}]{\includegraphics[width=0.34\linewidth]{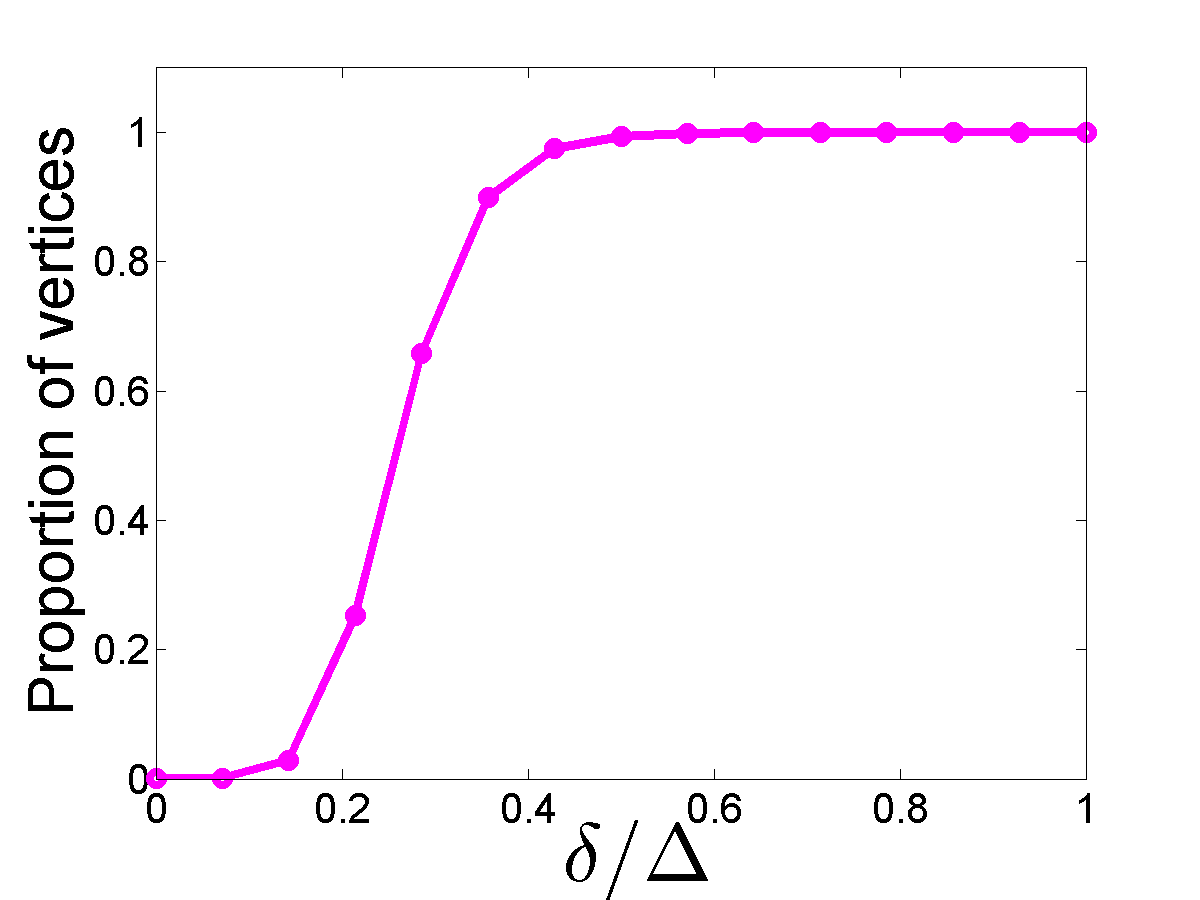}}
	\\
	\subfloat[\textsc{DBLP}]{\includegraphics[width=0.34\linewidth]{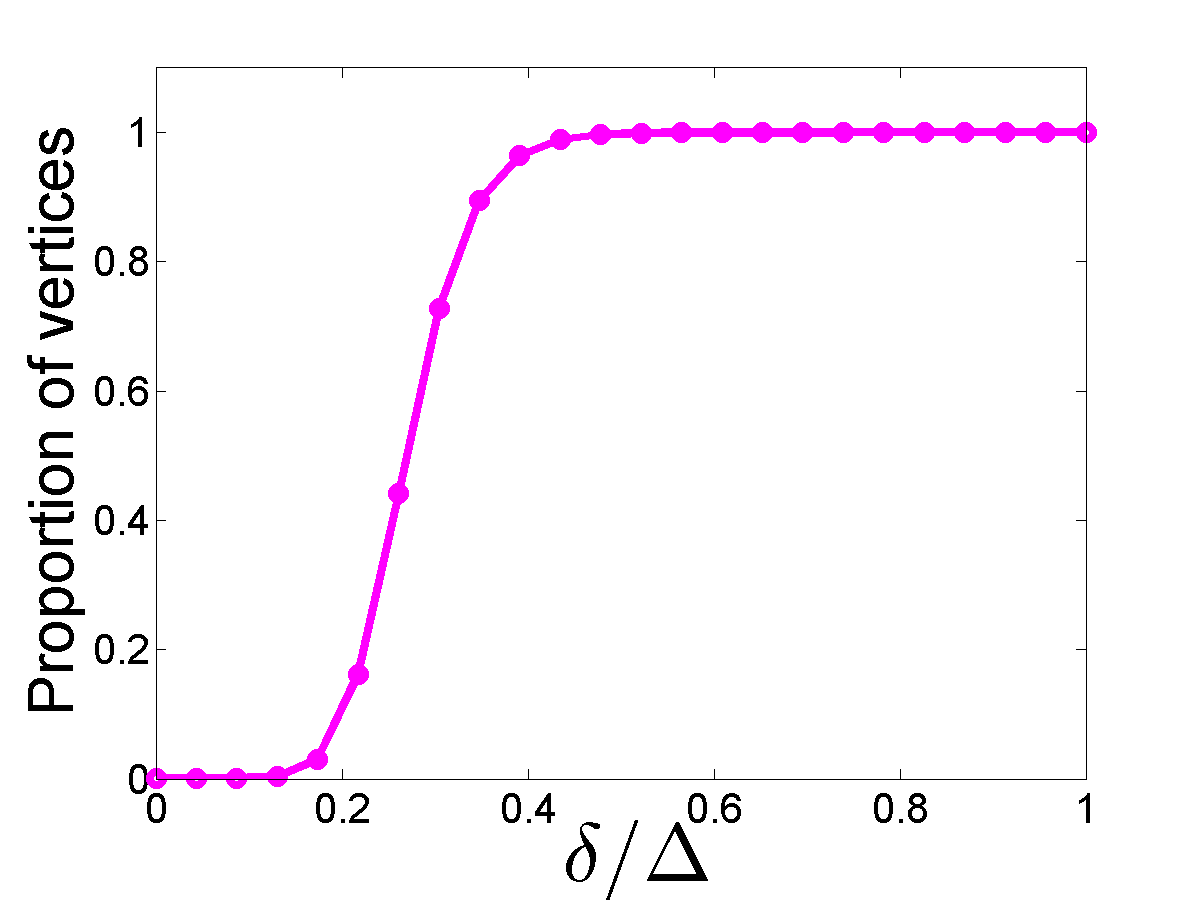}}
	\subfloat[\textsc{Enron}]{\includegraphics[width=0.34\linewidth]{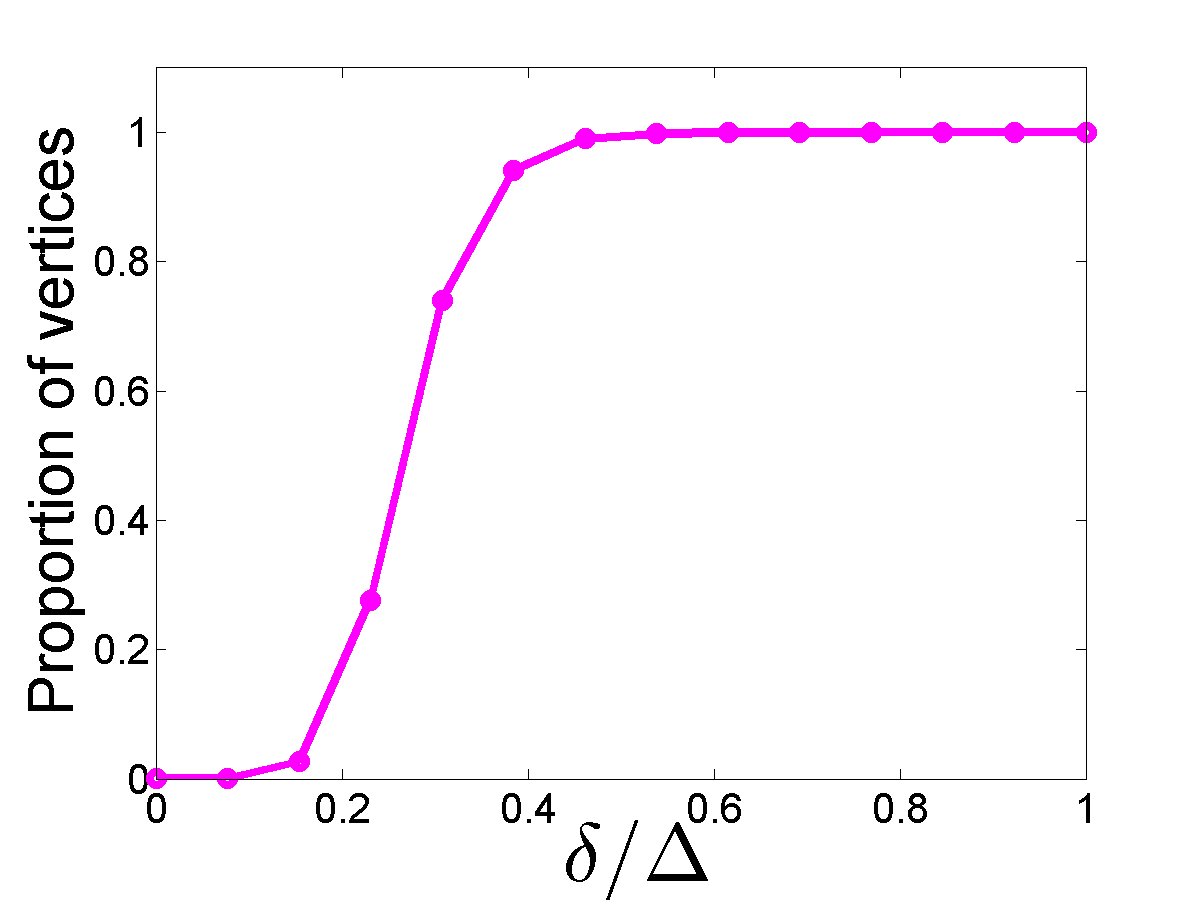}}
	\subfloat[\textsc{Facebook}]{\includegraphics[width=0.34\linewidth]{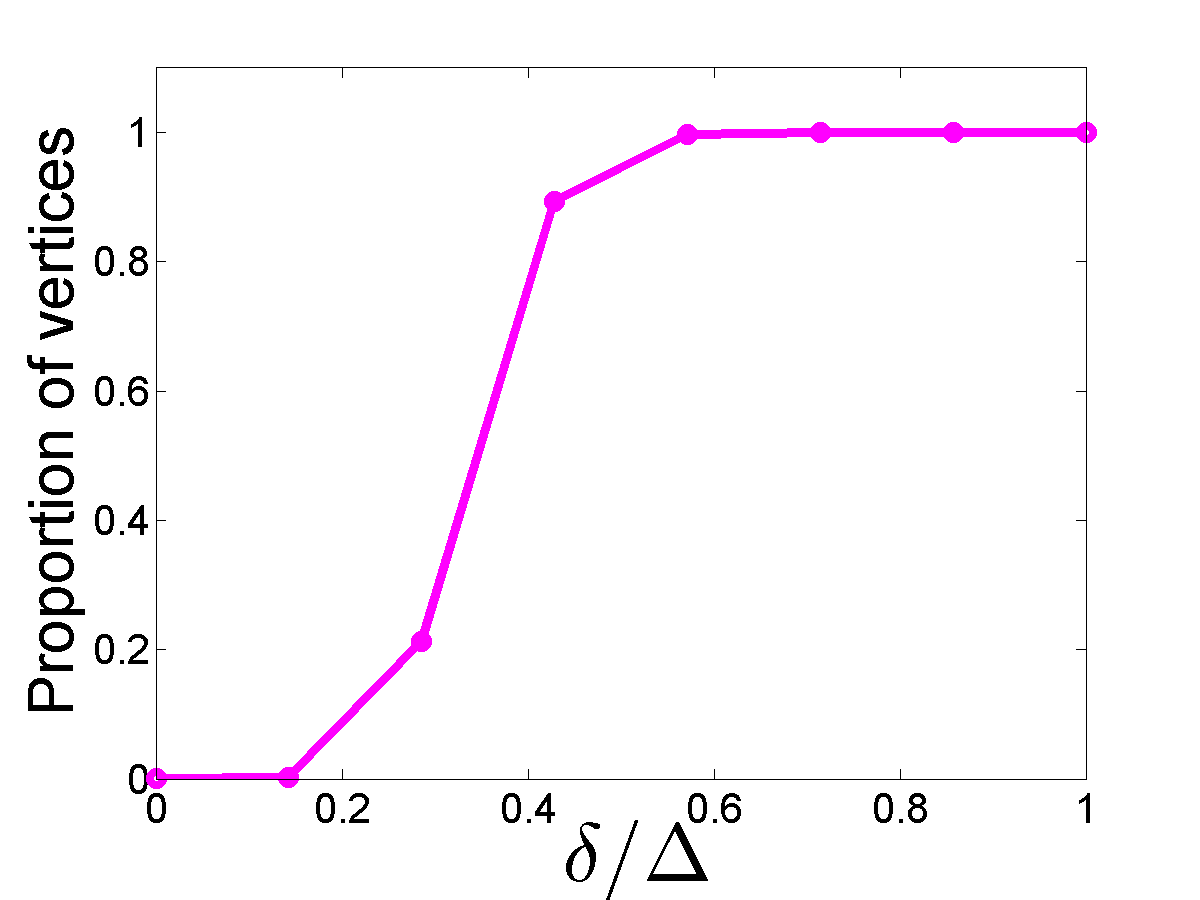}}
	\\
	\subfloat[\textsc{Gnutella}]{\includegraphics[width=0.34\linewidth]{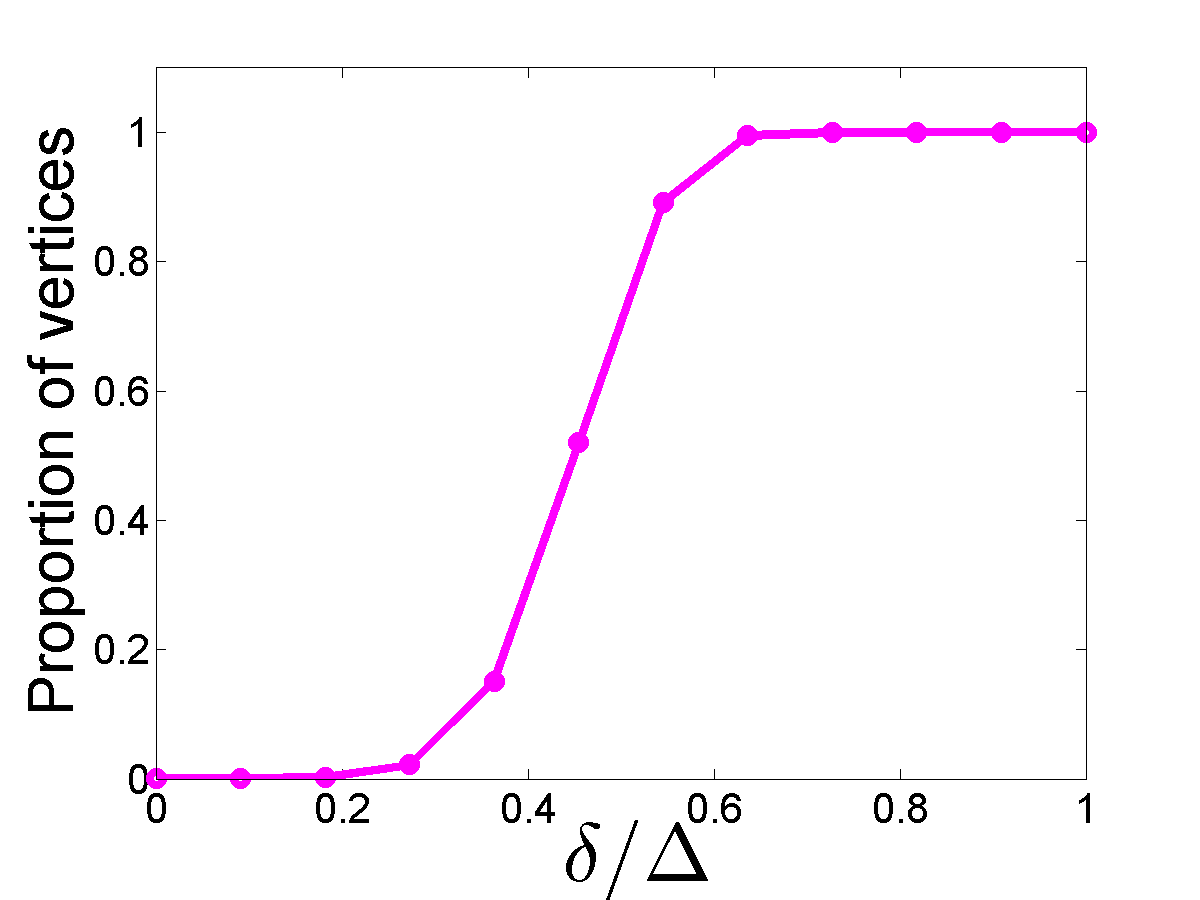}}
	\subfloat[\textsc{H.~sapiens}]{\includegraphics[width=0.34\linewidth]{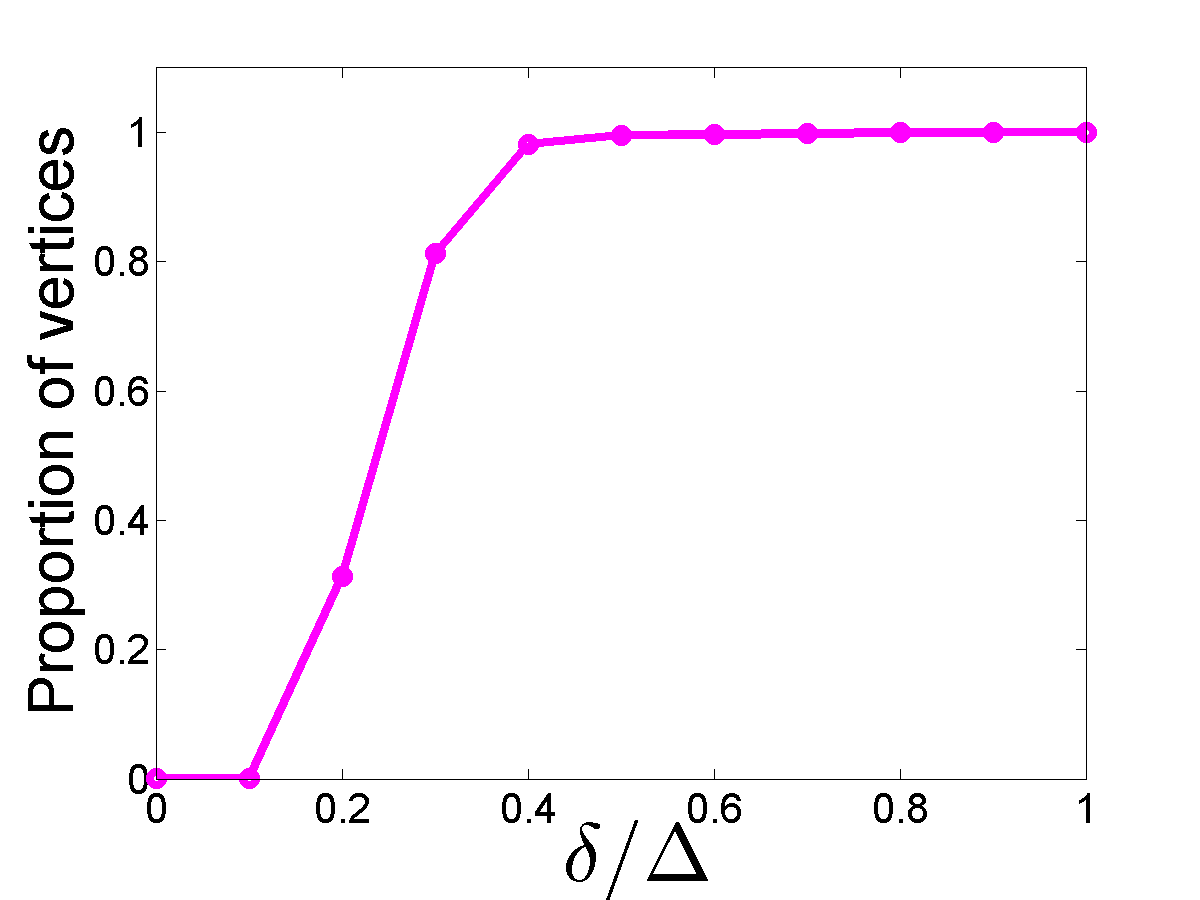}}
	\subfloat[\textsc{WPG}]{\includegraphics[width=0.34\linewidth]{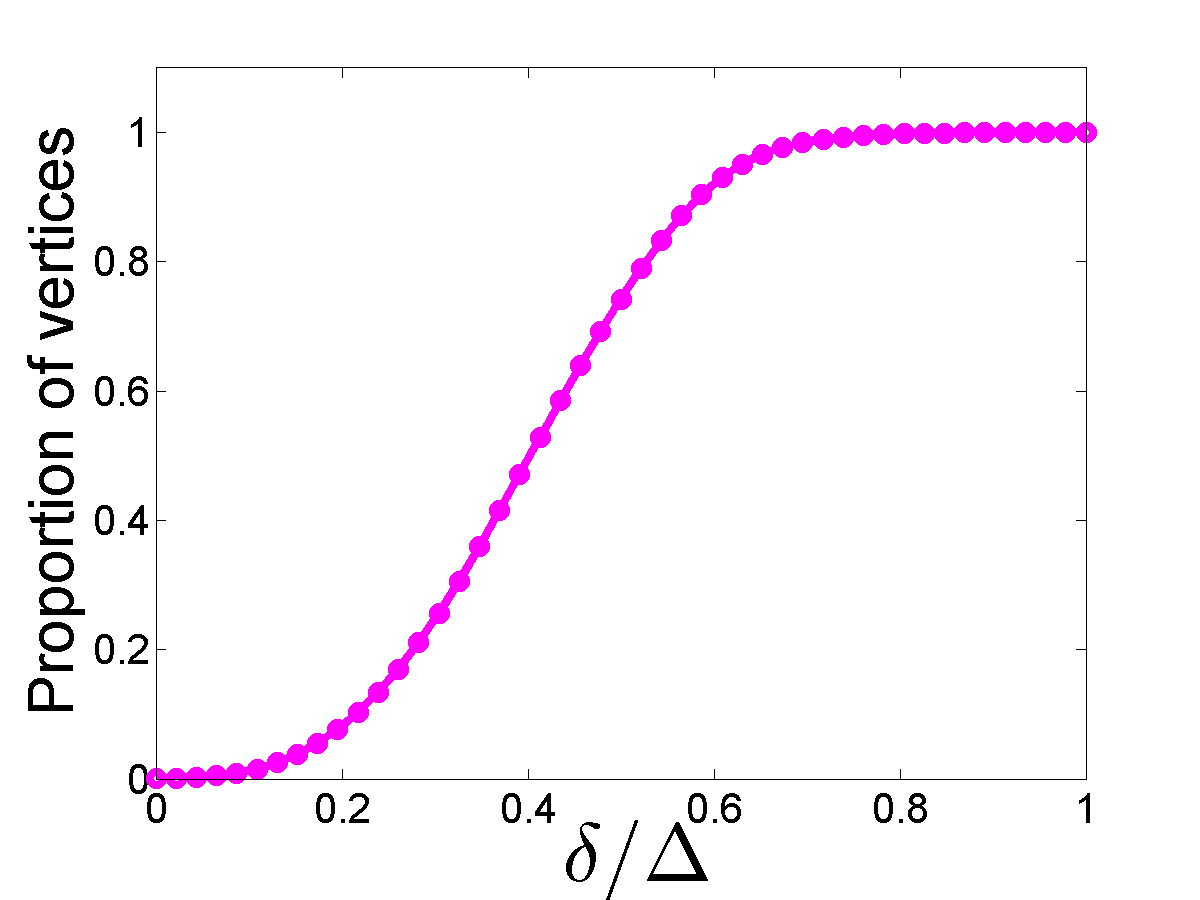}}
	\caption{Average proportion of vertices in $N_\delta$ as a function of $\delta/\Delta$.\label{fig:normalized_neighborhood_full}}
	
\end{figure}

To see the effect of this normalized setting on accuracy, consider Figure~\ref{fig:normalized_correct}.  After normalizing $\delta$ by the diameter, the variation between graphs is less pronounced.  Even when $\delta$ is small compared to $\Delta$, the optimum core number estimate ratio can be achieved.  It is worthy to note that \textsc{Facebook} is a stark exception in which a significant proportion of vertices cannot acheive an optimal core number estimate ratio even when $\delta = \Delta$.  This graph has many vertices with $d(v) \gg k(v)$ and a small diameter.  Thus, many vertices have very inaccurate $\hat k_0$-values, which propagate inwards and remain uncorrected due to the small number of refinements performed on the estimate. Ultimately, we conclude that $\hat k_\delta$ best achieves its goal of accurately estimating the core number using only a small local section of the graph when the graph has a large diameter and the ratio $\delta/\Delta$ is small (e.g. less than $0.2$).

\begin{figure}[!t]
	\centering
	\subfloat[\textsc{Amazon}]{\includegraphics[width=0.34\linewidth]{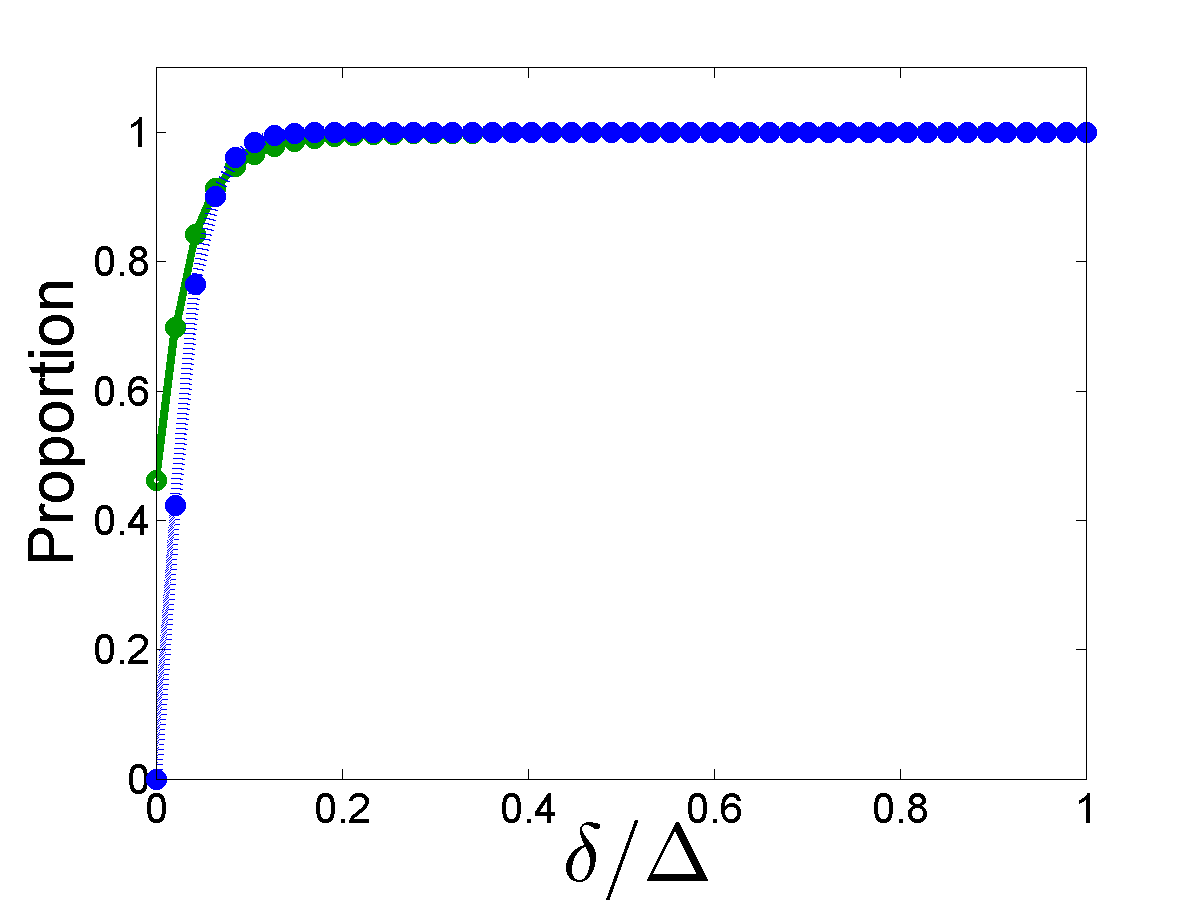}}
	\subfloat[\textsc{AS}]{\includegraphics[width=0.34\linewidth]{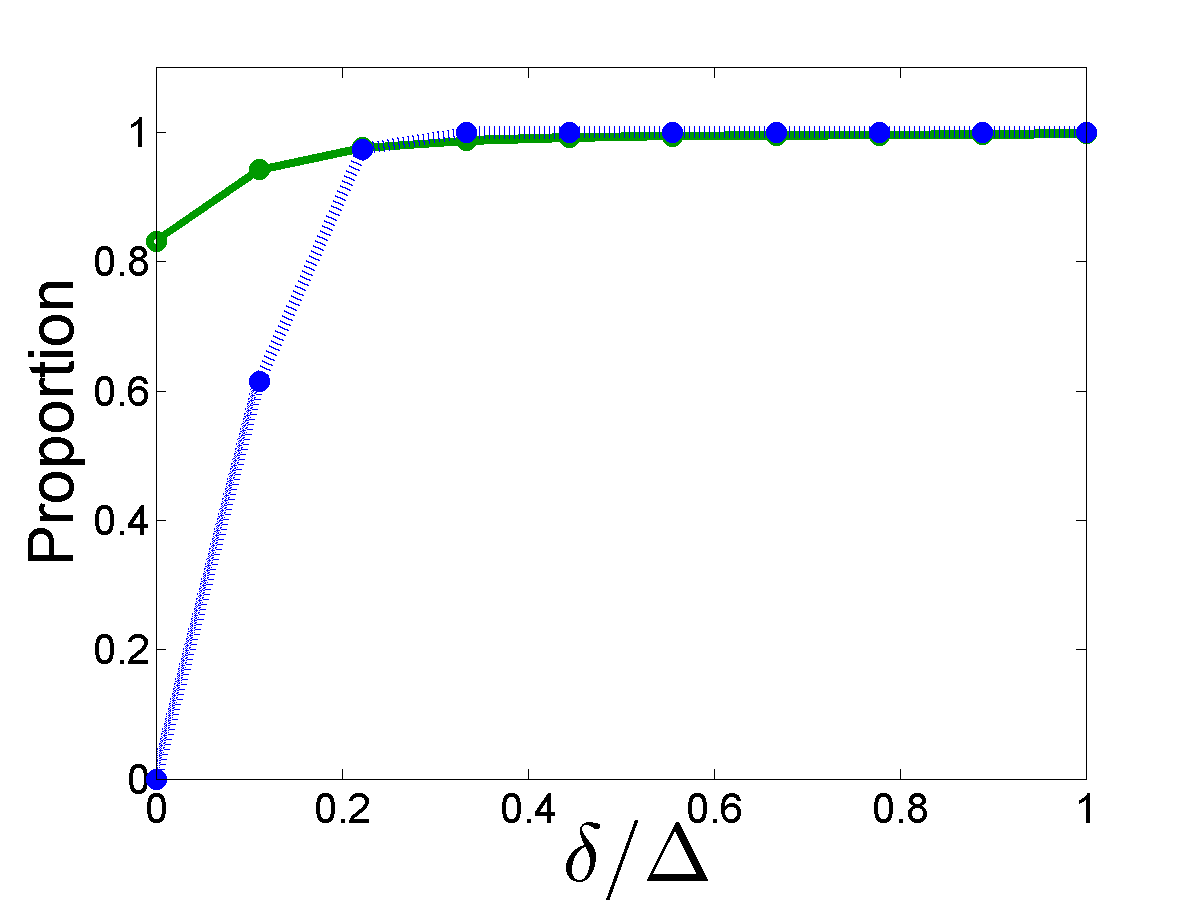}}
	\subfloat[\textsc{ca-AstroPh}]{\includegraphics[width=0.34\linewidth]{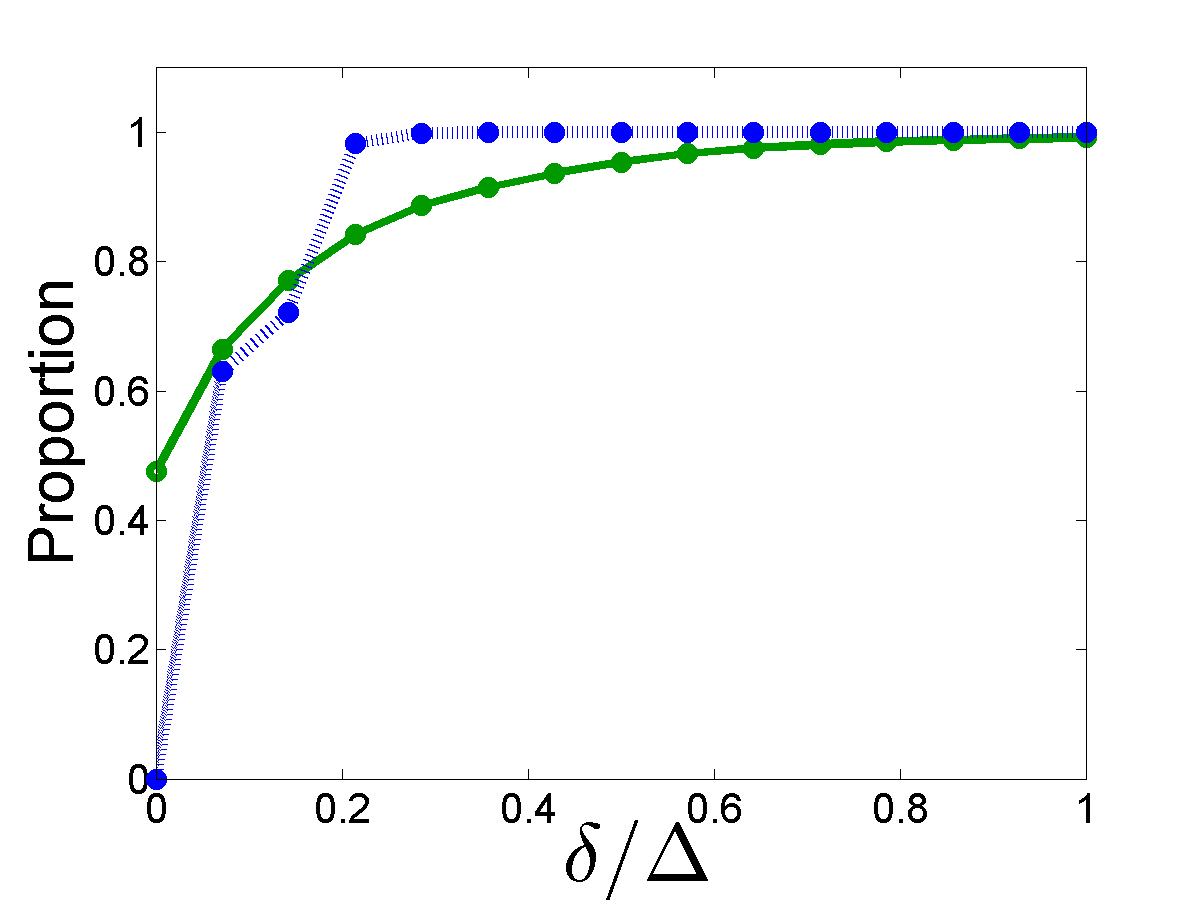}}
	\\
	\subfloat[\textsc{DBLP}]{\includegraphics[width=0.34\linewidth]{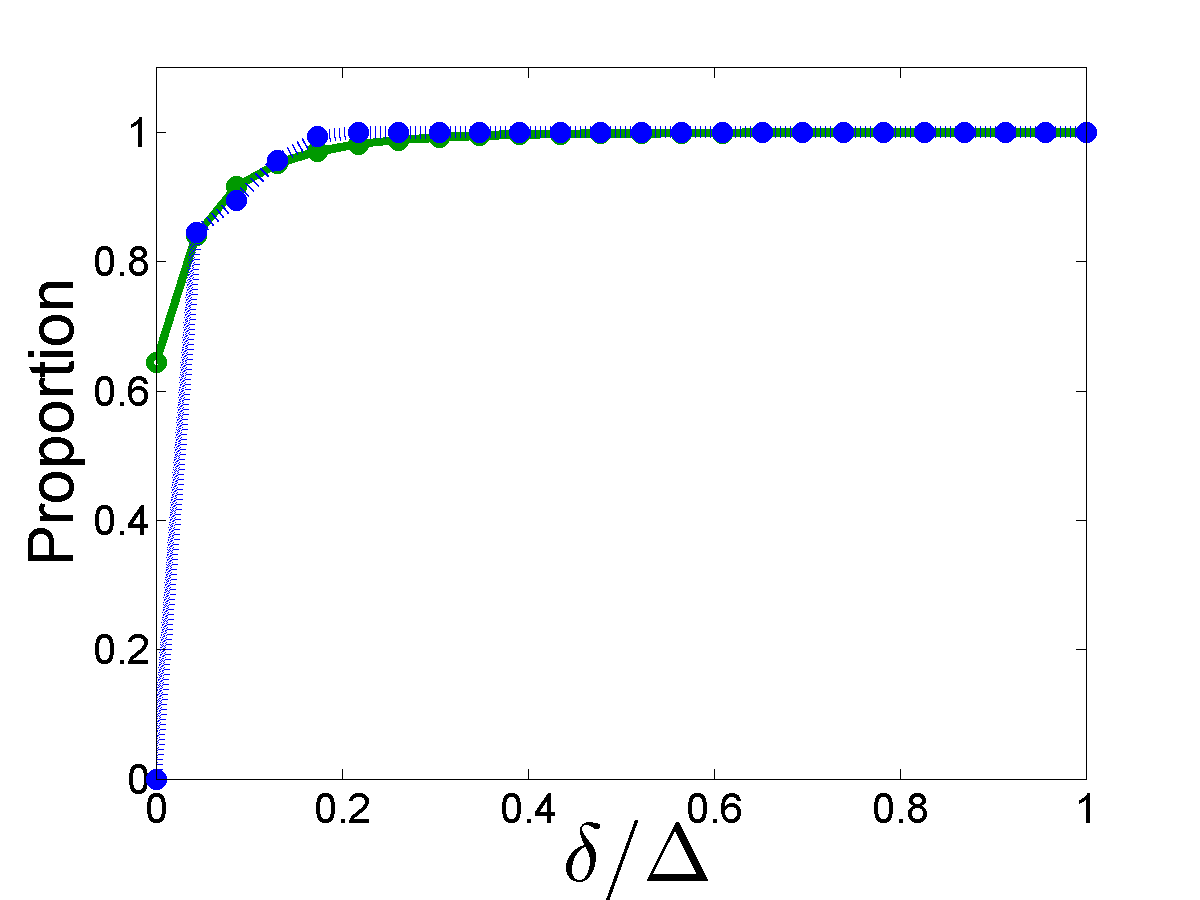}}
	\subfloat[\textsc{Enron}]{\includegraphics[width=0.34\linewidth]{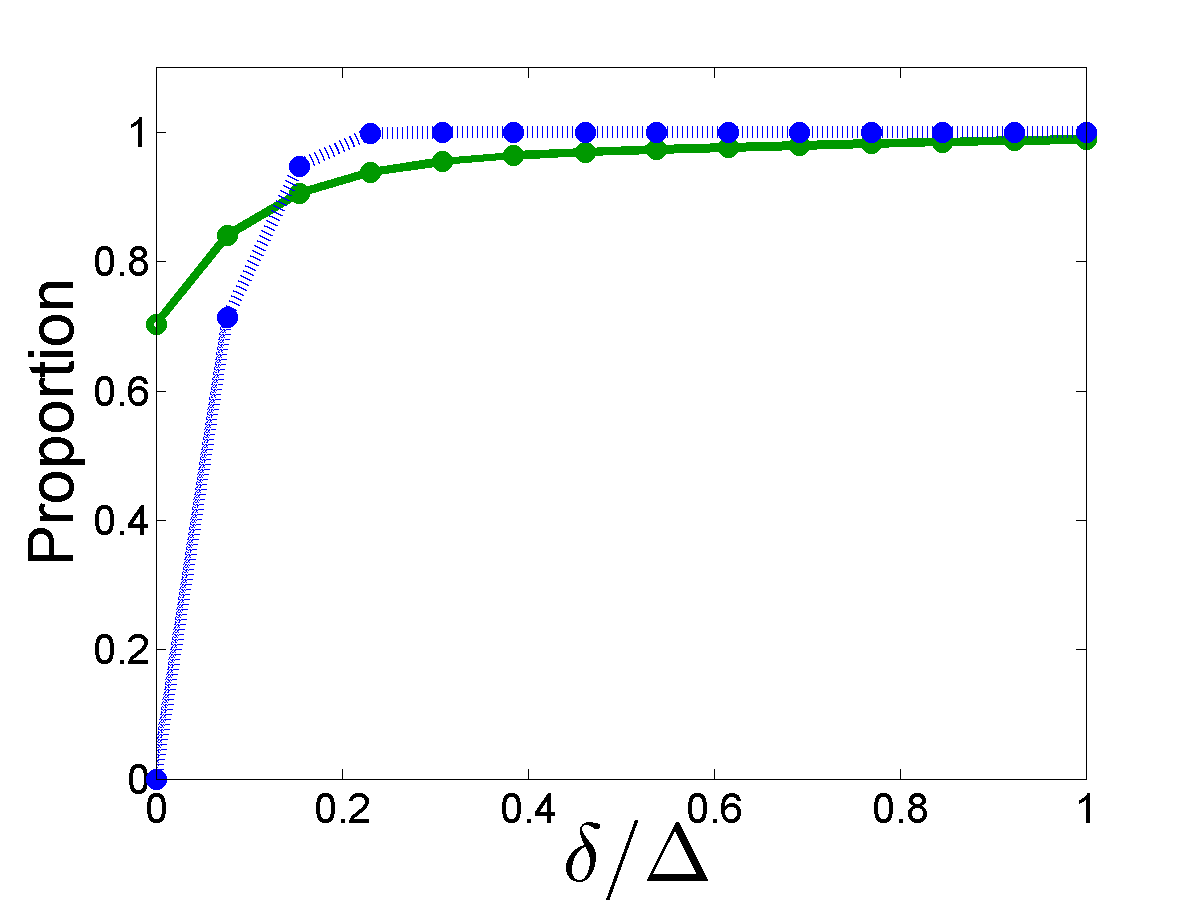}}
	\subfloat[\textsc{Facebook}]{\includegraphics[width=0.34\linewidth]{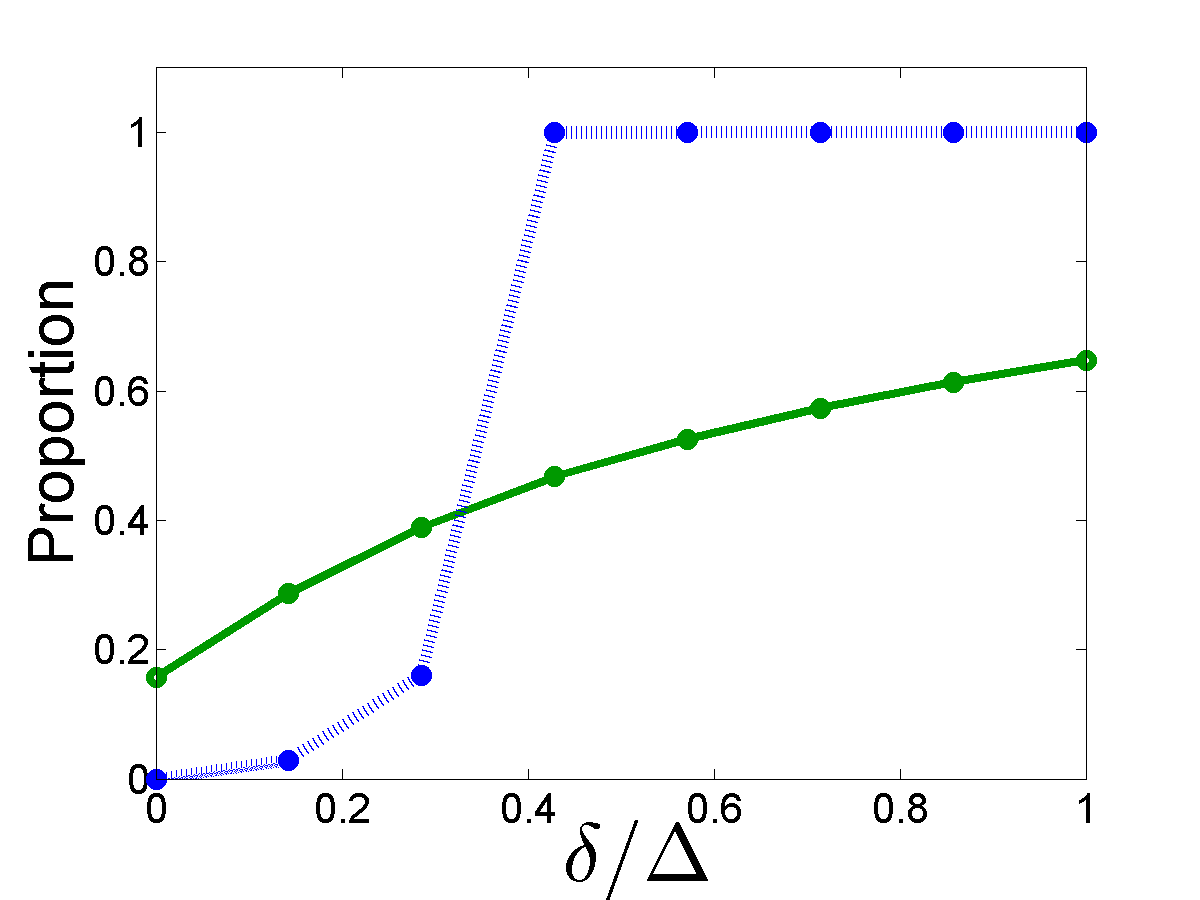}}
	\\
	\subfloat[\textsc{Gnutella}]{\includegraphics[width=0.34\linewidth]{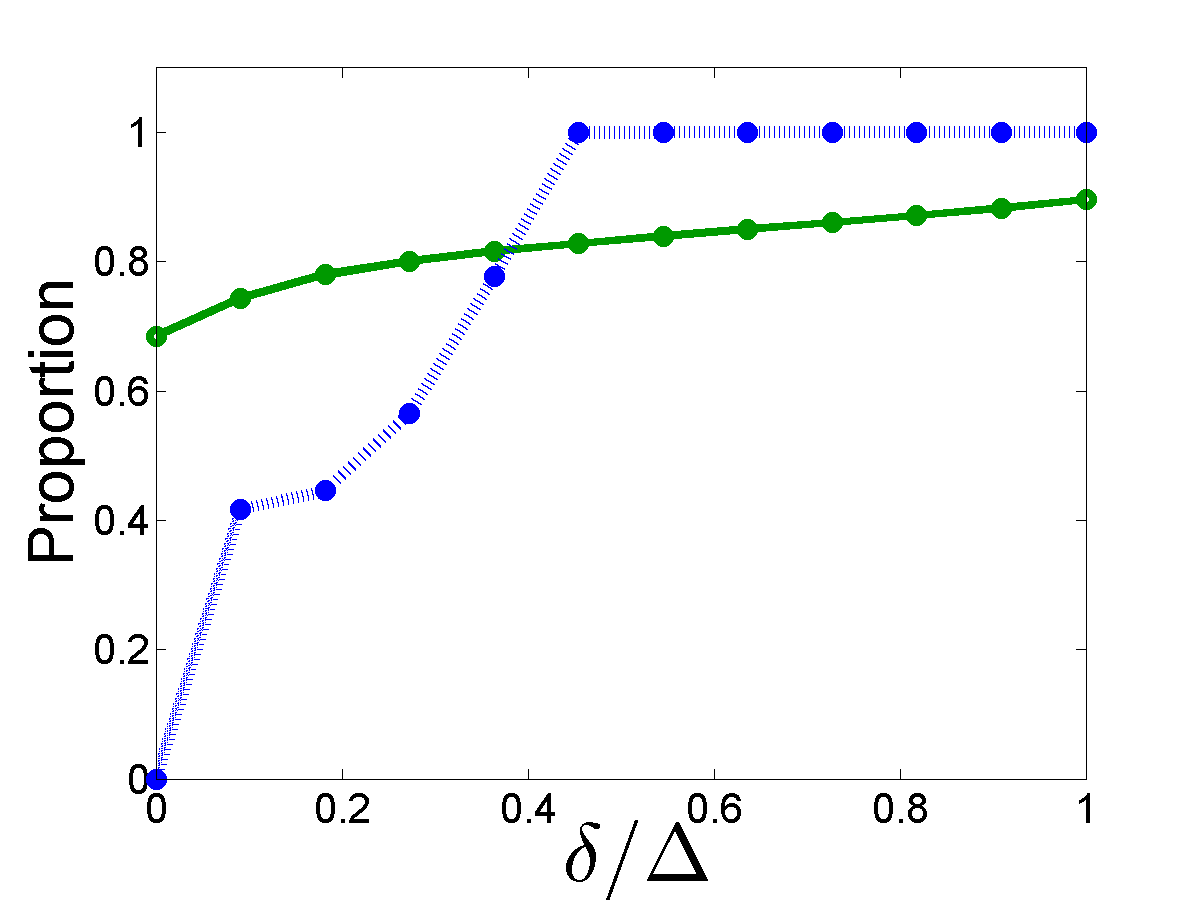}}
	\subfloat[\textsc{H.~sapiens}]{\includegraphics[width=0.34\linewidth]{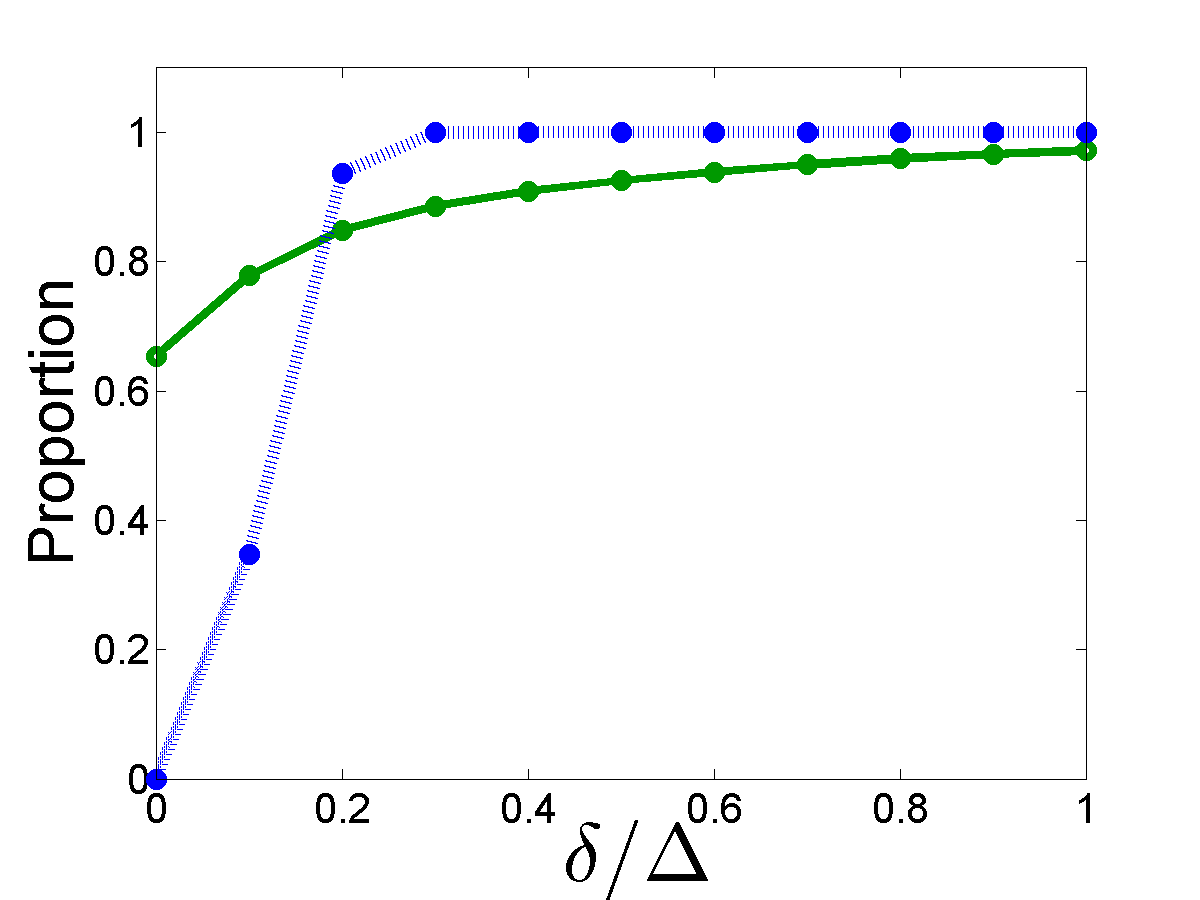}}
	\subfloat[\textsc{WPG}]{\includegraphics[width=0.34\linewidth]{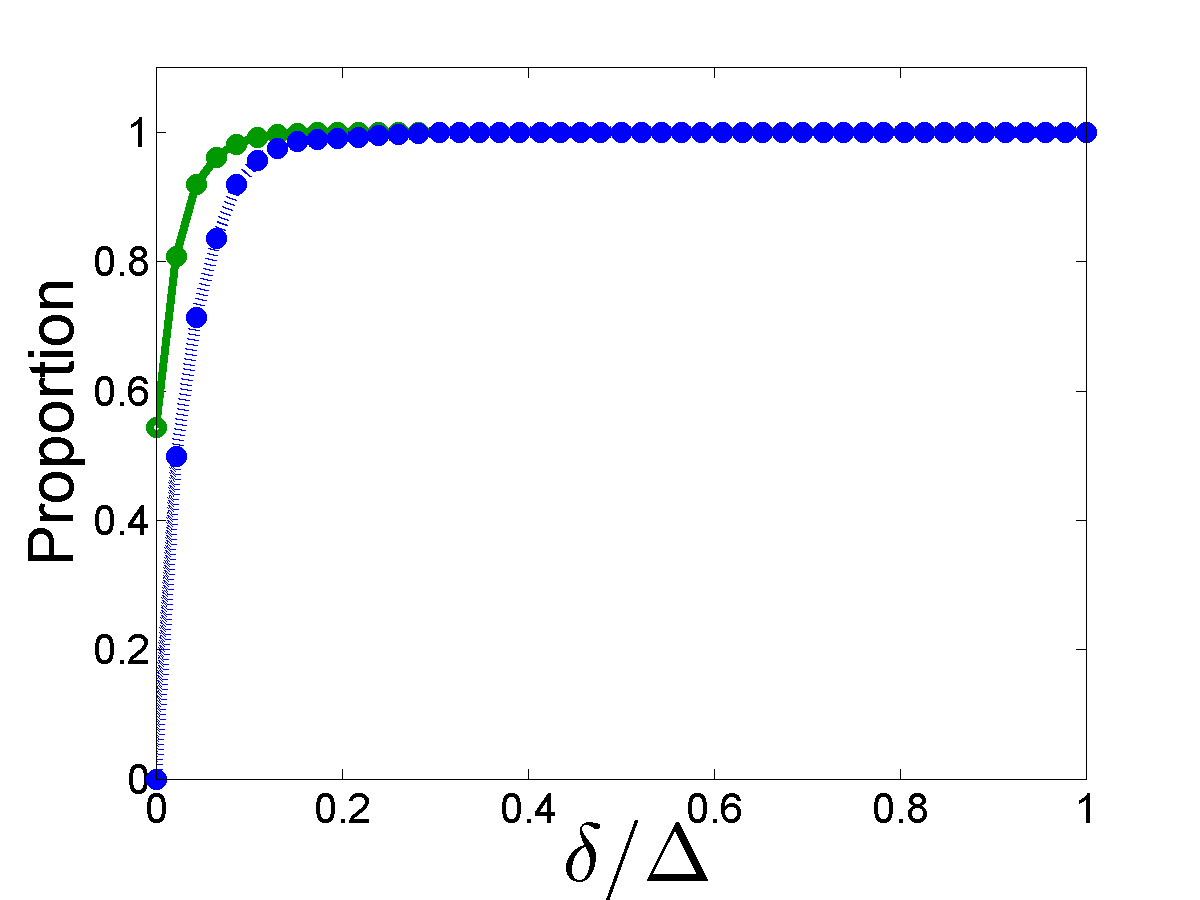}}
	\caption{Proportion of vertices with optimal core number estimate ratios for the propagating estimator (solid green) and the induced estimator (dashed blue) as a function of $\delta$.  The $x$-axis has been normalized by the diameter.\label{fig:normalized_correct}}
\end{figure}

\section{Application to Network Experimentation}
We now turn to the domain of network experiments and use the $\hat k_\delta$ estimator to address an open problem given in \cite{ugander:exposure}.
\subsection{Problem Statement}
Recall from the introduction that a \emph{network treatment experiment} is a random experiment in which some subjects are given a treatment and the rest are not.  It differs from other experiments in that the effects of the treatment are assumed to be dependent on interactions between subjects, which can be modeled by a graph. The general goal is to measure the subjects' experiences in a hypothetical universe where the entire graph is treated by observing the experience of the subject when only some of the graph is treated.  Ugander et al.~\cite{ugander:exposure} focused on local properties of the vertices to compare these two scenarios.  In particular, they identified two useful ways to concretely measure the experience of a subject via graph properties:
	\begin{definition}[\hspace{-0.05em}\cite{ugander:exposure}] \label{def:absolute_neighborhood_exposure}
	A vertex $v$ experiences {\em absolute $k$-degree exposure} if $v$ and at least $k$ of $v$'s neighbors receive treatment.
	\end{definition}
	
	\begin{definition}[\hspace{-0.05em}\cite{ugander:exposure}] \label{def:absolute_core_exposure}
	A vertex $v$ experiences {\em absolute $k$-core exposure} to a treatment condition if $v$ belongs to the $k$-core of the graph $G[V']$, where $V'$ is the set of treated vertices.
		\end{definition}
We will use $X^{(d)}_k(v)$ and $X^{(c)}_k(v)$ to denote the events that a vertex $v$ experiences absolute $k$-degree and absolute $k$-core exposure, respectively.

In order to reduce variance in later sections of their analysis, Ugander et al.\ first cluster the graph and then assign treatment randomly to the clusters (as opposed to individual subjects).  If a cluster is chosen to be treated, all vertices in the cluster receive treatment; otherwise none of them do.  Ugander et al.\ utilize a {\em $3$-net clustering} that is formed by growing balls of radius two centered at randomly selected vertices until every vertex is covered by some ball.
The procedures for computing the probabilities of $X^{(d)}_k(v)$ and $X^{(c)}_k(v)$ are independent of the method by which the graph was clustered, so we choose to omit further detail here and refer the reader to~\cite{ugander:exposure} for details.   

Once the graph is clustered, a recursive function can be used to compute the probability that vertex $v$ experiences absolute $k$-degree exposure.  We follow the notation of~\cite{ugander:exposure}. Let $s$ be the number of clusters that contain at least one vertex in $\{v\} \cup N_1(v)$, indexed $\{1, \ldots, s\}$ so that $v$ resides in the highest numbered cluster.  If $p$ is the probability that a cluster is treated and $\vec{w_v} = (w_{v,1},\dots w_{v,s})$ is the number of edges from $v$ to the vertices in each cluster, then
\begin{equation}\label{eq:neighborhood_exposure}
	\mathbb{P}[X^{(d)}_\kappa(v)] = pf(s-1,\kappa-w_{v,s};p,\vec{w_v}),
\end{equation}
	where the function $f(j,T;p,\vec{w_v})$ is defined as
	\begin{equation*}
		\begin{split}		
		f(1,0;p,\vec{w_v}) = & \, 1 \\
		f(1,T;p,\vec{w_v}) = & \, p\mathbf{1}[T\leq w_{v,1}] \\
		f(j,T;p,\vec{w_v}) = & \, pf(j-1,T-w_{v,j};p,\vec{w_v}) \\
			& +(1-p)f(j-1,T;p,\vec{w_v}),
		\end{split}
	\end{equation*}
	where $\mathbf{1}[B]$ denotes the {\em indicator function} (evaluates to $1$ if the Boolean expression $B$ is true and $0$ otherwise). 
	
	The function $f$ defined above recursively visits each cluster $j$ containing a neighbor of $v$ and considers the probability that $v$ is $T$-degree exposed in the first $j$ clusters conditioned on whether cluster $j$ receives treatment.  If $j$ is treated, $v$ needs to have $T-w_{v,j}$ treated neighbors in the first $j-1$ clusters; otherwise, it needs $T$ such neighbors.   It follows from Definition \ref{def:absolute_neighborhood_exposure} that if $v$ is $k$-degree exposed, cluster $s$ is necessarily treated.  This also implies that all of $v$'s neighbors in the same cluster are necessarily treated as well.  Thus, we are ultimately concerned with finding $\kappa-w_{v,s}$ treated neighbors in the remaining $s-1$ clusters that contain a neighbor of $v$.
Using dynamic programming, we can compute $\mathbb P[X^{(d)}_i(v)]$ for all $0 \leq i\leq \kappa$ in $O(s\kappa)$ time.

\subsection{Estimating $k$-core exposure probabilities}
	In~\cite{ugander:exposure}, Ugander et al.\ left computing the exact probability of absolute $k$-core exposure as an open problem, since the core decomposition requires knowledge of the entire graph.  They instead defer to the fact that the absolute $k$-core exposure probability is bounded from above by the absolute $k$-degree exposure probability and use the latter in lieu of the former.  This is problematic because there may not be a consistent relationship between $d(v)$ and $k(v)$.  For example, if $v$ is a vertex with degree $100$ and core number $10$, $\mathbb{P}[X^{(c)}_{20}]=0$ independent of $\mathbb{P}[X^{(d)}_{20}]$.  Although this is an extreme case, there are more general cases where the two probabilities are not correlated.  Specifically, vertices that require a large value of $\delta$ before $\breve k_\delta(v) = k(v)$ (as in Figure~\ref{fig:tree_treatment}) can have many treated neighbors without having a large core number.
\begin{figure}
	\centering
	\begin{tikzpicture}
	    \node[vertex,draw,fill=cyan!40] (n11) at (5,6)  {$v$};
	    \node[vertex,draw,fill=cyan!40] (n21) at (2,5)  {$u_1$};
	    \node[vertex,draw,fill=cyan!40] (n22) at (5,5)  {$u_2$};
	    \node[vertex,draw,fill=cyan!40] (n23) at (8,5)  {$u_3$};
	    \node[vertex,draw,fill=cyan!40] (n31) at (1,4)  {};
	    \node[vertex,draw,fill=cyan!40] (n32) at (2,4)  {};
	    \node[vertex,draw,fill=cyan!40] (n33) at (3,4)  {};
	    	\node[vertex,draw,fill=cyan!40] (n34) at (4,4)  {};
		\node[vertex,draw,fill=cyan!40] (n35) at (5,4)  {};
	    \node[vertex,draw,fill=cyan!40] (n36) at (6,4)  {};
	    \node[vertex,draw,fill=cyan!40] (n37) at (7,4)  {};
	    \node[vertex,draw,fill=cyan!40] (n38) at (8,4)  {};
	    \node[vertex,draw,fill=cyan!40] (n39) at (9,4)  {};

		\node[vertex,draw] (n41) at (2,1.5) {};
		\node[vertex,draw] (n42) at (5,1.5) {};
		\node[vertex,draw] (n43) at (8,1.5) {};
	
	    \foreach \from/\to in {n11/n21,n11/n22,n11/n23,
	    	n21/n31,n21/n32,n21/n33,
	    	n22/n34,n22/n35,n22/n36,
	    	n23/n37,n23/n38,n23/n39}
	    \draw(\from) -- (\to);
		\foreach \from/\to in {n41/n31,n41/n32,n41/n33,n41/n34,n41/n35,n41/n36,n41/n37,n41/n38,n41/n39,
			n42/n31,n42/n32,n42/n33,n42/n34,n42/n35,n42/n36,n42/n37,n42/n38,n42/n39,
			n43/n31,n43/n32,n43/n33,n43/n34,n43/n35,n43/n36,n43/n37,n43/n38,n43/n39}
	      \draw(\from) -- (\to);
	\end{tikzpicture}
	\caption{$T'_{3,3}$ with $13$ of $16$ vertices treated.  $k(v) = k(u_1) = k(u_2) = k(u_3) = 3$.  Although $v$, $u_1$, $u_2$, and $u_3$ have all of their neighbors treated, they only have core number $1$ with respect to the treated subgraph.\label{fig:tree_treatment}}
\end{figure}
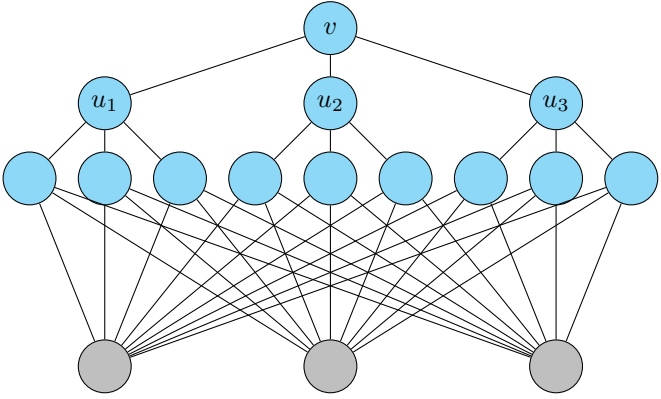 
Recall that $\hat k_0(v)$ is the degree of $v$.  As we have shown, even expanding the scope and computing $\hat k_1(v)$ can yield a considerably more accurate estimate of of the core number than the degree. Therefore, a tighter bound of the core exposure probability can be achieved by examining the degree exposure probability of $v$'s neighbors.  To capture this, we introduce a $\hat k_1$-related condition we call {\em neighbor-degree exposure}.
\begin{definition}
A vertex $v$ experiences {\em absolute $k$-neighbor-degree exposure} if at least $k$ of $v$'s neighbors experience absolute $k$-degree exposure.
\end{definition}

We denote the event that vertex $v$ is absolute $k$-neighbor-degree exposed with $\hat X_k(v)$.  Algorithm \ref{alg:k_hat_ugander} gives a method for computing $\mathbb{P}[\hat X_\kappa(v)]$, which can then be used to estimate (specifically, find an upper bound on) 
$\mathbb{P}[X^{(c)}_\kappa(v)]$.  

	\begin{algorithm}
	\caption{Absolute $k$-neighbor-degree exposure probability of $v$ \label{alg:k_hat_ugander}}
	\begin{algorithmic}[1]
	\Statex \textsc{input}:  Graph $G$, vertex $v$, clustering $C$, exposure probability $p$, desired exposure level $\kappa$
	\Statex \textsc{output}:  $\mathbb{P}[\hat X_\kappa(v)]$
	\State Let $C(x)$ denote the cluster containing $x$
	\State $\mathcal{C} \leftarrow \left(\bigcup_{u\in N_2(v)} C(u)\right) \backslash C(v)$
	\State $\hat p_\kappa \leftarrow 0$
	\For{$\mathcal{S}\subseteq \mathcal{C}$}\label{alg:k_hat_ugander:forloop}
		\State ${Y} \leftarrow \{u\in N_1(v): X^{(d)}_\kappa(v) \text{ is true in } G[\mathcal{S}\cup C(v)]\}$
		\If{$|{Y}|\geq k $}
			\State $\hat p_\kappa \leftarrow \hat p_\kappa+p^{|\mathcal{S}|+1}(1-p)^{|\mathcal{C}|-|\mathcal{S}|}$ \label{alg:k_hat_ugander:adding}
		\EndIf
	\EndFor
	\State \Return $\hat p_{\kappa}$
	\end{algorithmic}
	\end{algorithm}
	
The algorithm iterates through all subsets of clusters containing a vertex in $v$'s $2$-neighborhood and determines whether treating them yields a scenario where $v$ has $k$ neighbors that are absolute $k$-degree exposed.  If so, line \ref{alg:k_hat_ugander:adding} adds the probability of that configuration occurring to the final probability.  Because it enumerates all possible subsets of $\mathcal{C}$, Algorithm \ref{alg:k_hat_ugander} will run in $O( s\kappa \cdot d(v)\cdot 2^s)$ time in the worst case.  While this algorithm is exponential in the number of clusters, Ugander et al.\ assume that the graph satisfies some restricted growth conditions\footnote{Namely, $\exists c$ such that $|N_{\delta+1}(v)| \leq c\cdot |N_\delta(v)|$ $\forall v\in V$}.  In this case, the number of clusters that contain vertices from $N_2(v)$ does not grow with respect to the size of the graph~\cite{ugander:exposure} which bounds the running time at $O(\kappa\cdot d(v))$. 

In graphs failing the restricted growth requirements, the running time can still be improved.  Note that if treating a specific subset of clusters $\mathcal{S}$ on line \ref{alg:k_hat_ugander:forloop} does not yield $\kappa$ vertices in $N_1(v)$ that are $\kappa$-degree exposed, then treating any $\mathcal{S'}\subseteq \mathcal{S}$ also cannot yield at least $\kappa$ vertices in $N_1(v)$ that are $\kappa$-degree exposed.  Thus, if the subsets of $\mathcal{C}$ are enumerated in decreasing order of their sizes, we can prune the search space to avoid needless computation.  Moreover, the clustering algorithm can be biased towards selecting $3$-net clusterings that minimize $|\mathcal{C}|$.  For example, one possible bias would be to select the centers of the balls with probability proportional to their degrees.

We applied Algorithm \ref{alg:k_hat_ugander} and Equation \ref{eq:neighborhood_exposure} to the \textsc{WPG} data set and binned the data based on the difference $\mathbb{P}[\hat X_\kappa]-\mathbb{P}[X^{(d)}_\kappa]$ as shown in Figure \ref{fig:network_treatment}.  It is particularly noteworthy that multiple vertices have a neighbor-degree exposure probability of zero but a non-zero probability of degree exposure.  Moreover, many of those vertices have their degree exposure probability maximized (equal to $0.25$).  Thus, the empirical data confirms that absolute degree exposure probability may be a misleading estimate of absolute core exposure probability.
\begin{figure}
	\centering
	\subfloat[$\kappa=4$]{\includegraphics[width=0.5175\linewidth]{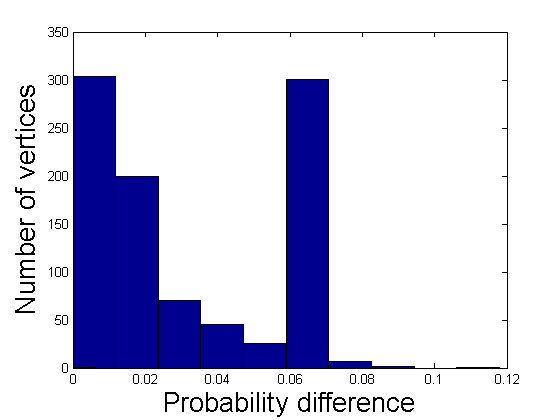}}
	\subfloat[$\kappa=5$]{\includegraphics[width=0.5175\linewidth]{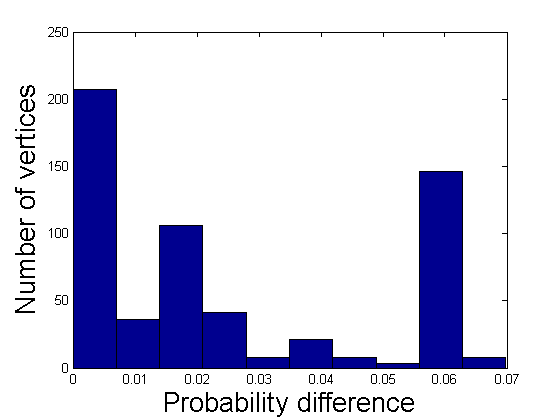}}
	\caption{$\mathbb{P}[X^{(d)}_\kappa]-\mathbb{P}[\hat X_\kappa]$ for the \textsc{WPG} graph at $p=0.25$.  Vertices with $\mathbb{P}[X^{(d)}_\kappa(v)]=0$ are omitted. \label{fig:network_treatment}}
\end{figure}

Finally, we consider a second approach for improving the approximation of $\mathbb{P}[X^{(c)}_k(v)]$ that, like $\hat k_\delta$, tightens an upper bound on $\mathbb{P}[X^{(c)}_\kappa(v)]$ by using the bounds on $\mathbb{P}[X^{(c)}_\kappa(u)]$ for $u$ in $ N_1(v)$.  We examine those vertices $u'$ in $N_1(v)$ that satisfy $\mathbb{P}[X^{(d)}_\kappa(u')] = 0$.  These vertices cannot contribute to $\mathbb{P}[\hat X_\kappa(x)]$, so we can disregard them when computing $\vec{w_v}$.  Thus, we can use Equation \ref{eq:neighborhood_exposure} with a modified $\vec{w_v}$ to get a tighter upper bound on $\mathbb{P}[X^{(c)}_\kappa]$.  Figure \ref{fig:pruned_exposure} shows that pruning can decrease the probability of a majority of the vertices (in fact, many probabilities decrease from $p$ to $0$).  This further bolsters our argument that $X^{(c)}_\kappa(v)$ is only weakly correlated with $X^{(d)}_\kappa(v)$, but using information from $v$'s neighbors can yield a much tighter upper bound at minimal additional cost.
\begin{figure}[!h]
	\centering
	\subfloat[\textsc{Amazon}]{\includegraphics[width=0.34\linewidth]{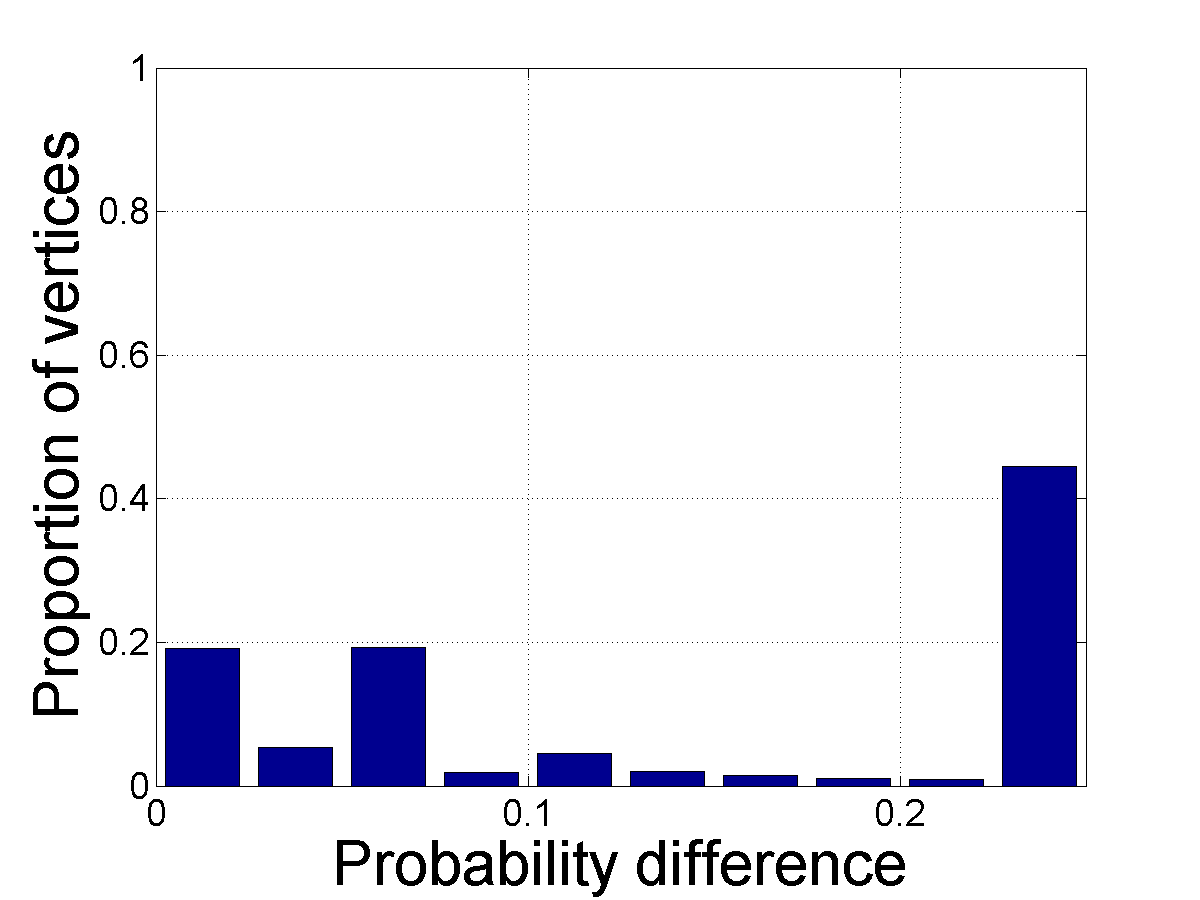}}
	\subfloat[\textsc{AS}]{\includegraphics[width=0.34\linewidth]{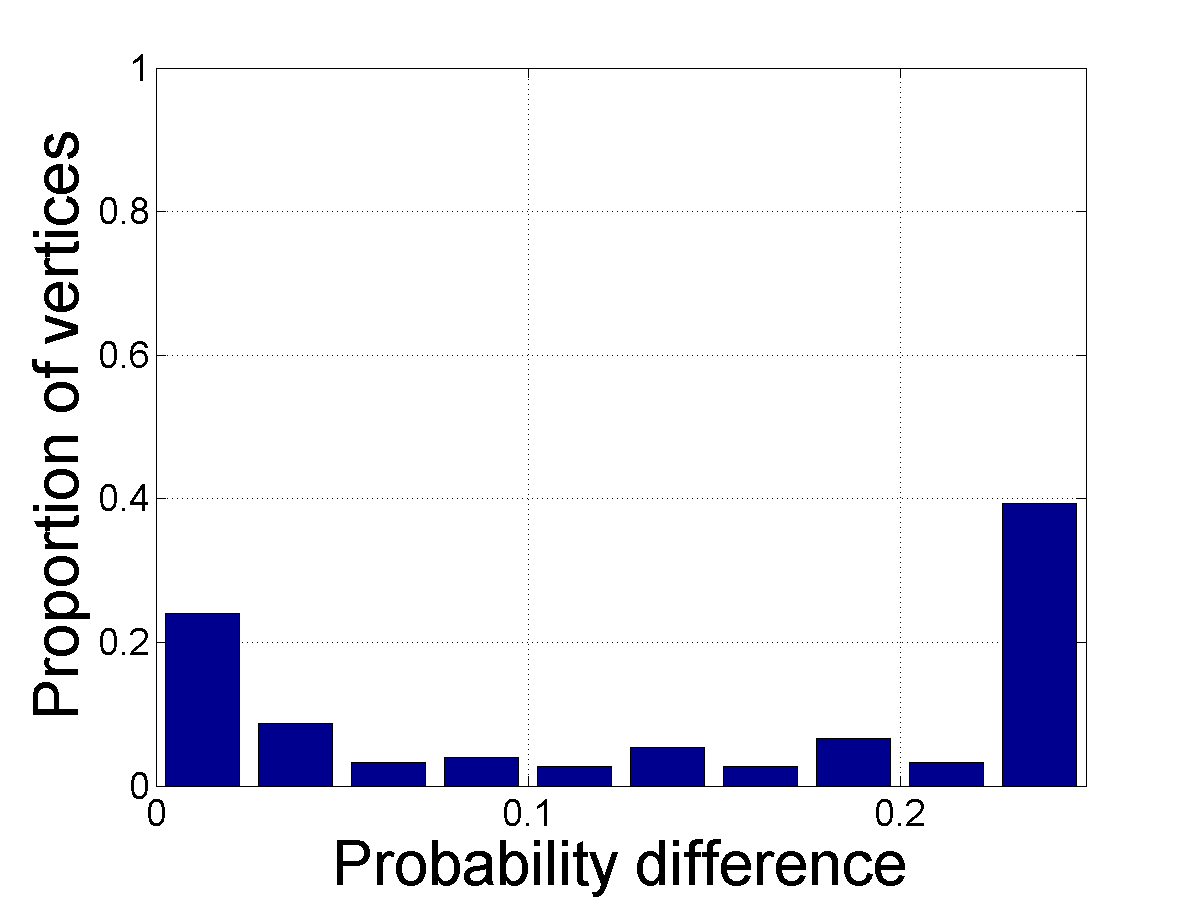}}
	\subfloat[\textsc{ca-AstroPh}]{\includegraphics[width=0.34\linewidth]{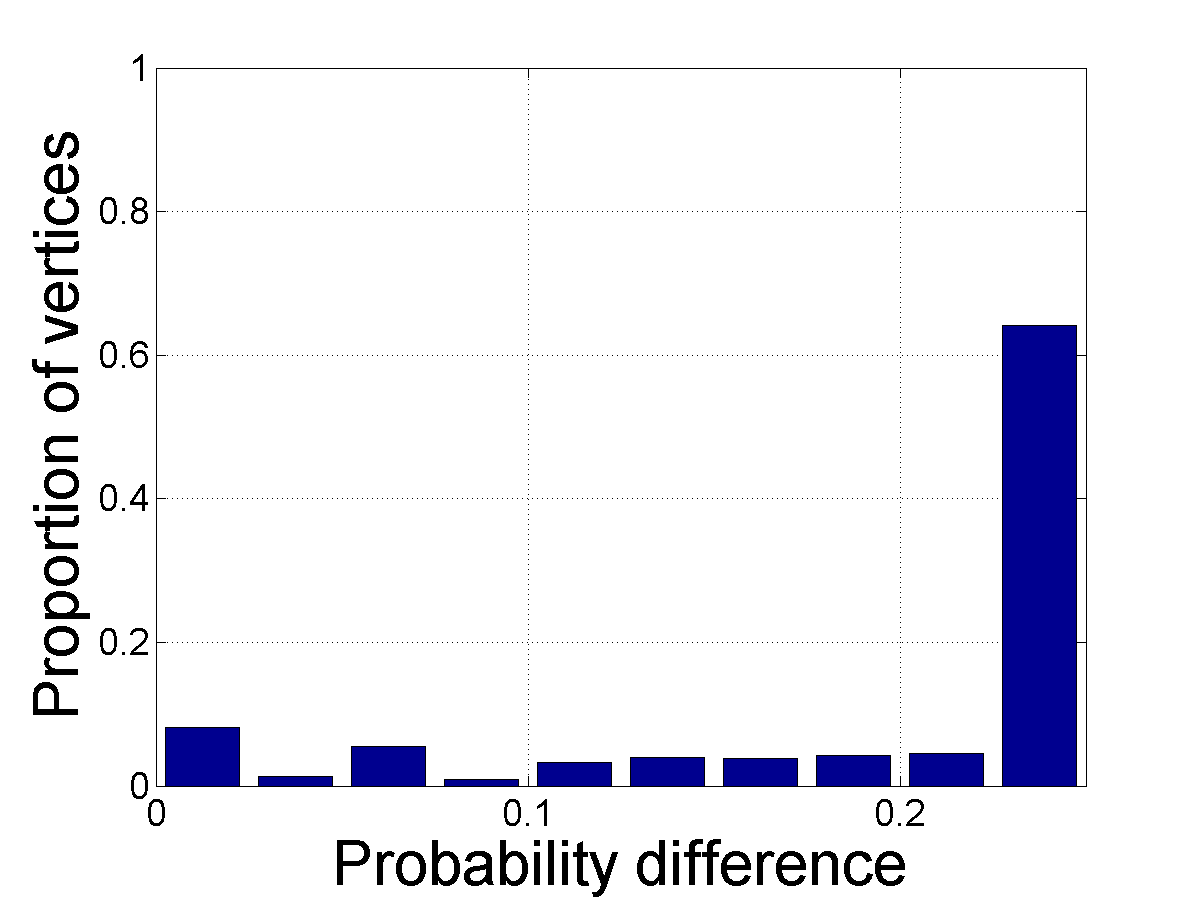}}
	\\
	\subfloat[\textsc{DBLP}]{\includegraphics[width=0.34\linewidth]{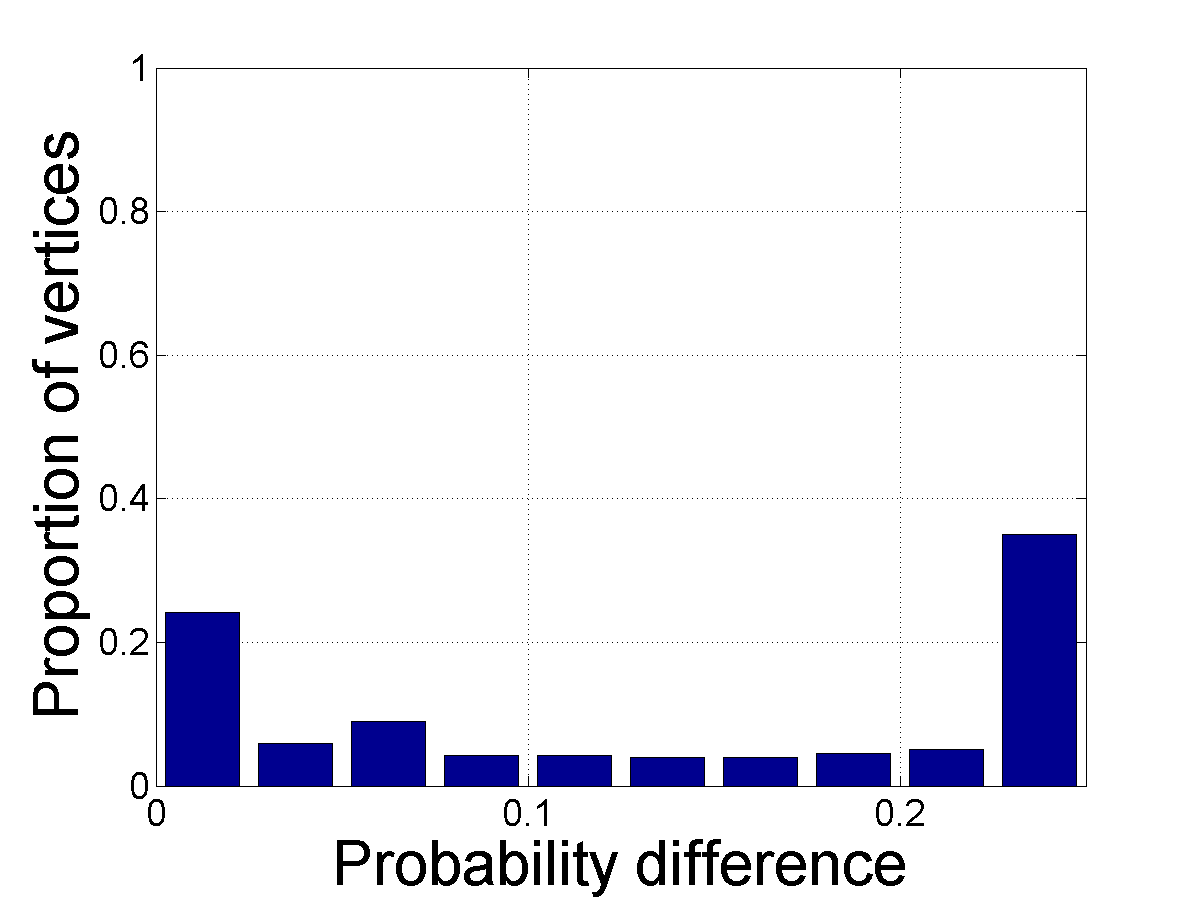}}
	\subfloat[\textsc{Enron}]{\includegraphics[width=0.34\linewidth]{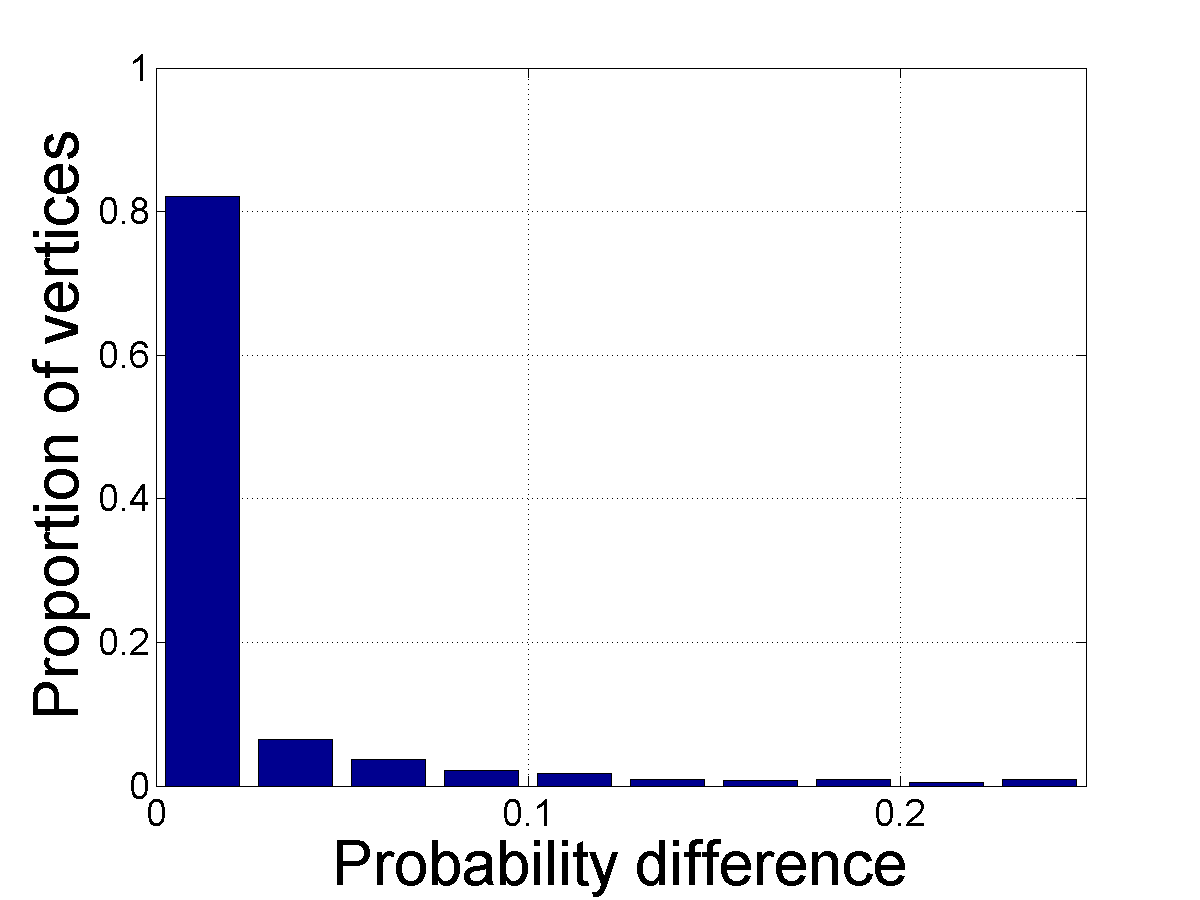}}
	\subfloat[\textsc{Facebook}]{\includegraphics[width=0.34\linewidth]{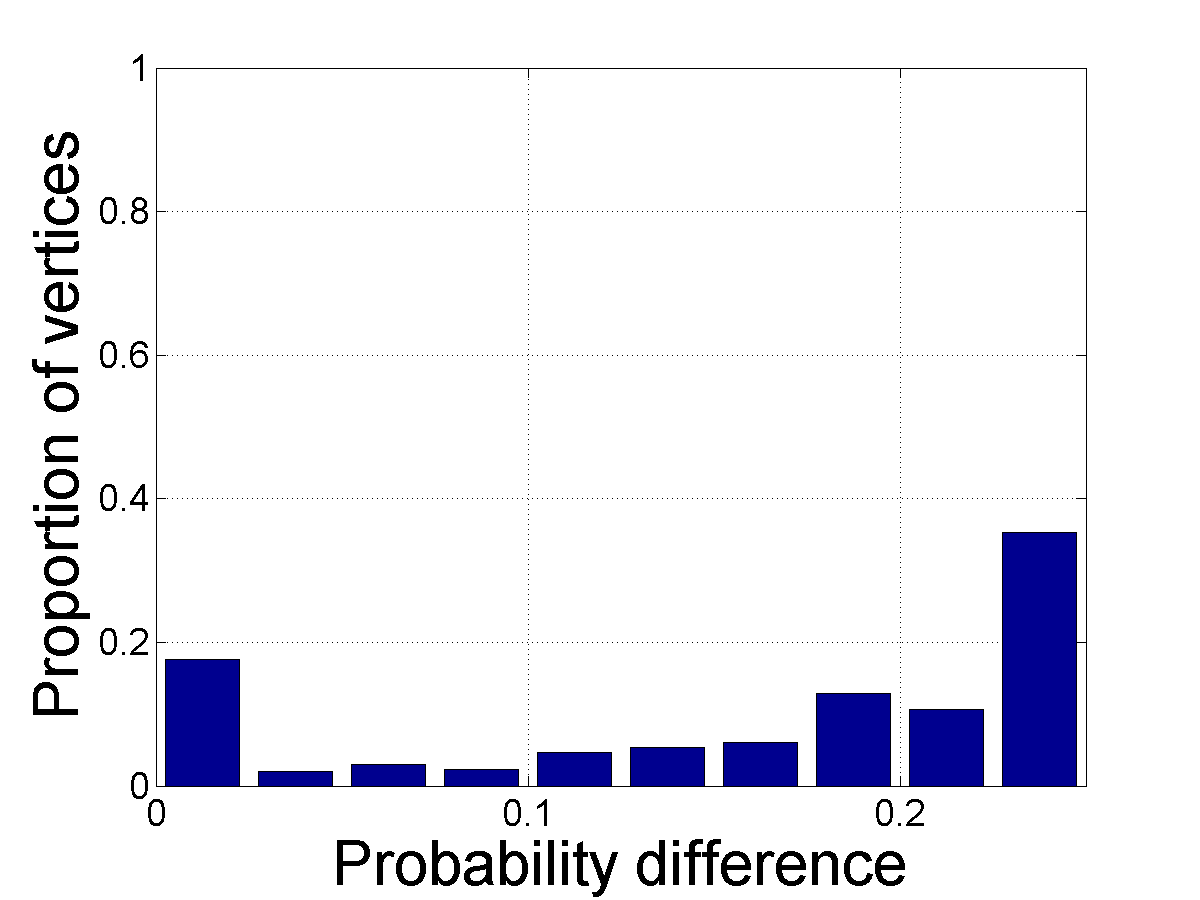}}
	\\
	\subfloat[\textsc{Gnutella}]{\includegraphics[width=0.34\linewidth]{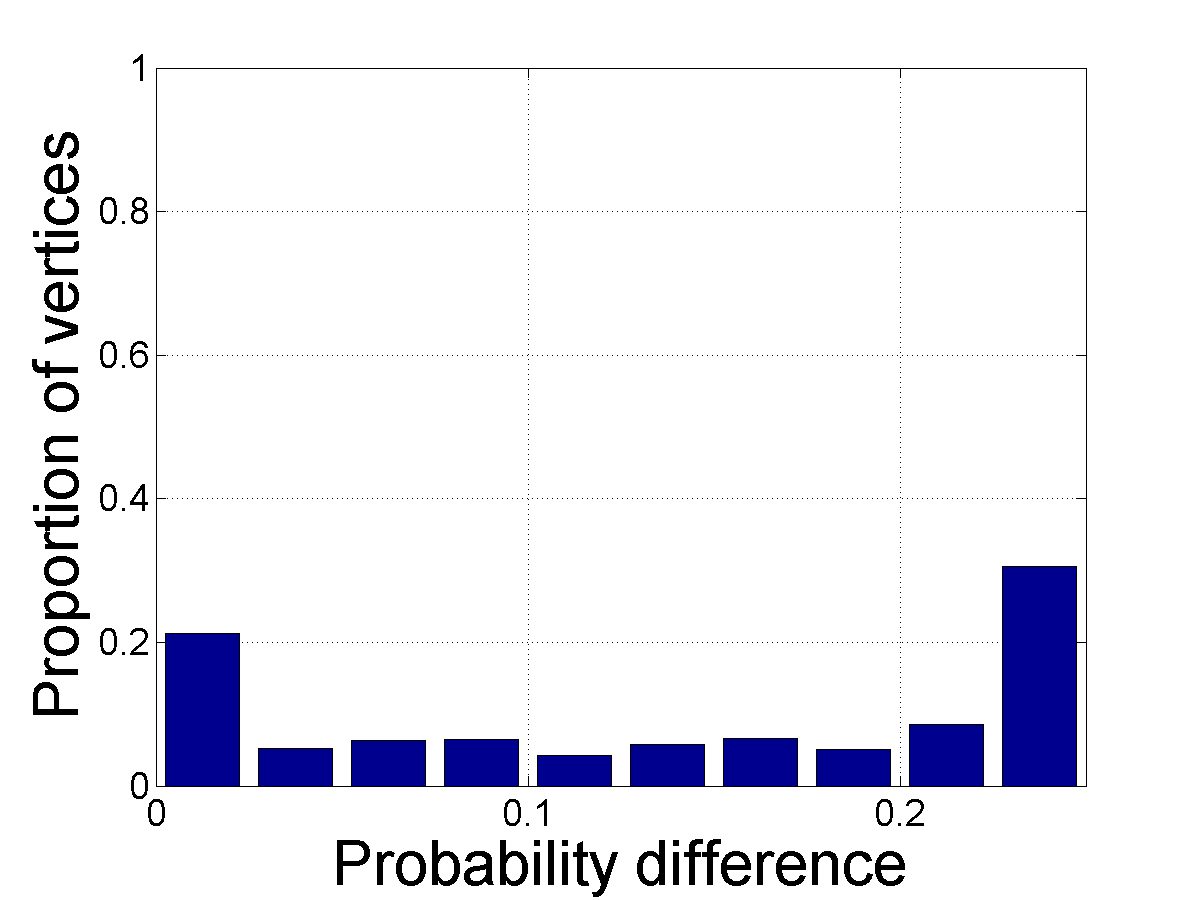}}
	\subfloat[\textsc{H.~sapiens}]{\includegraphics[width=0.34\linewidth]{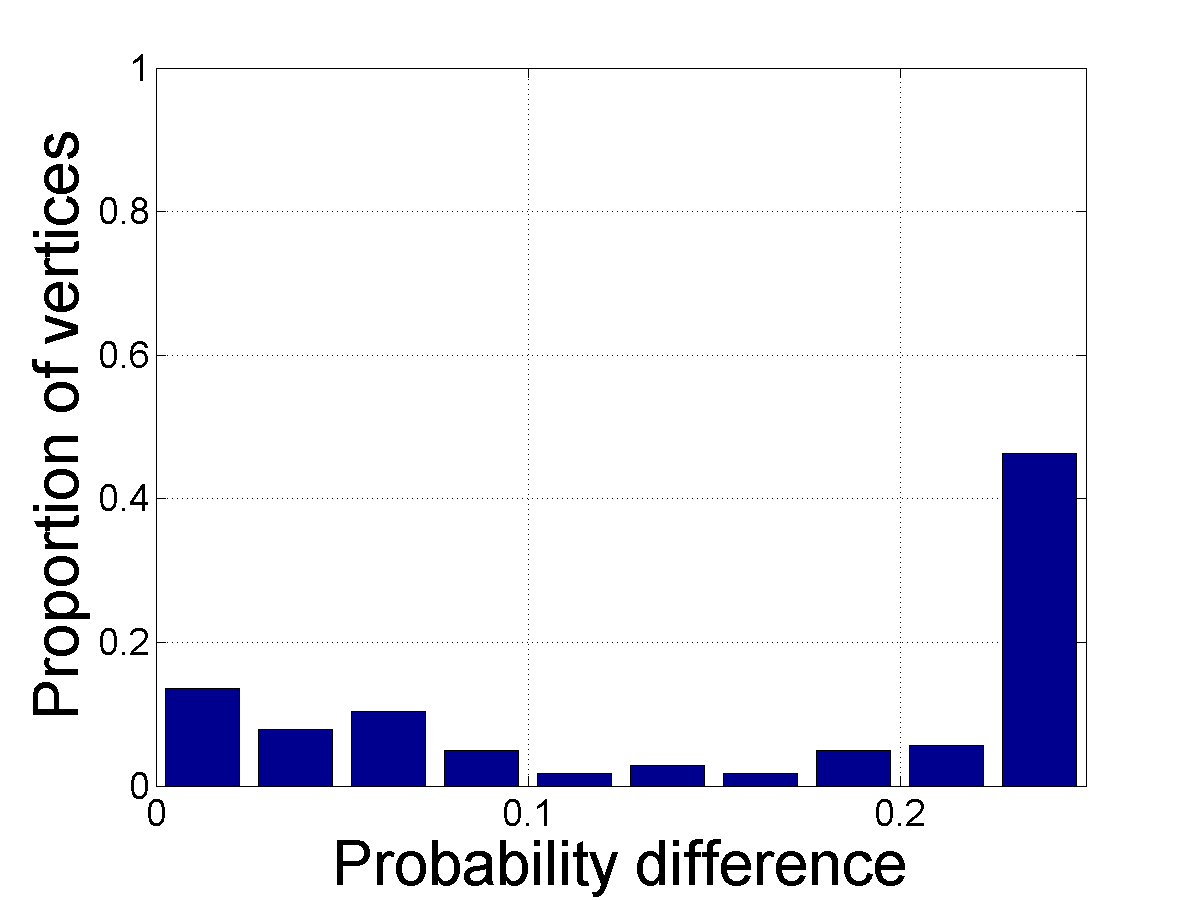}}
	\subfloat[\textsc{WPG}]{\includegraphics[width=0.34\linewidth]{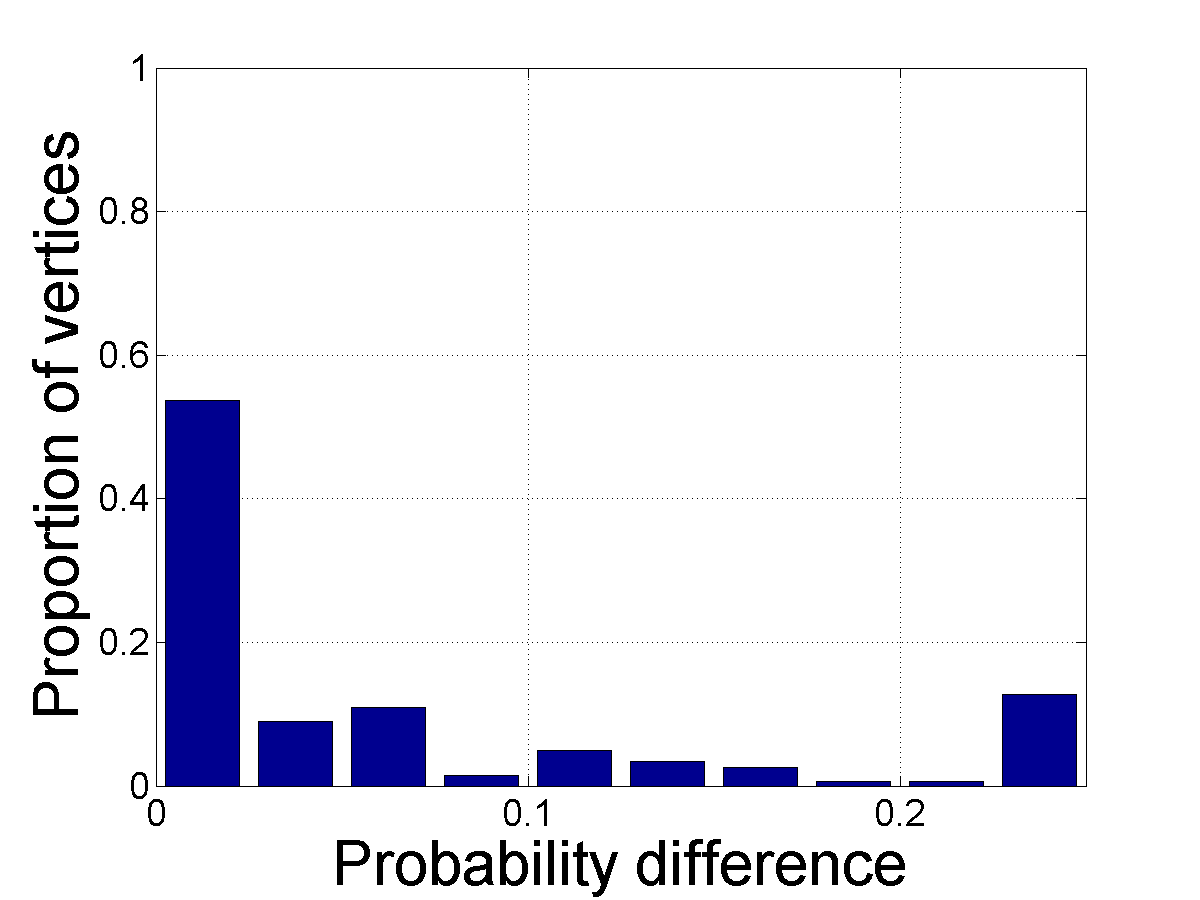}}
	\caption{Histogram of differences between $\mathbb{P}[X^{(d)}_\kappa]$ before and after pruning for $\kappa = 7$ and $p=0.25$.  The $x$-axis gives the difference in probability, while the $y$-axis gives the proportion of vertices occurring in that bin.  For clarity, only those vertices which are not pruned are considered in the plot. \label{fig:pruned_exposure}}
	\end{figure}

\section{Conclusions and Future Work}
We introduced  $\hat k_\delta$, a novel method of estimating the core number of a vertex that uses only the data available in the $\delta$-neighborhood of the vertex.  We formally proved that in an \ER graph, the error of $\hat k_1$ grows arbitrarily slowly with respect to the size of the graph.  After computing $\hat k_2$ on a representative corpus of real-world networks, we demonstrated that a high-accuracy estimate of the core number can be achieved using a limited subset of the graph.  Finally, we described two ways in which the estimators could be used to improve calculations in network treatment experiments.

There are a number of natural extensions to this research.  Algorithm \ref{alg:k_hat} computes $\hat k_{\delta-1}$ for each neighbor $u_i$ of $v$, which in turn requires calculating $\hat k_{\delta-1}$ and so forth.  However, since $\hat k_\delta(v)$ is geometrically the value at the intersection of the functions $d(v)-i+1$ and $k_{\delta-1}(u_i)$, refining the core number estimates at the ``first'' vertices ($u_1,u_2,\dots$) and ``last'' vertices ($u_d,u_{d-1},\dots$) may not affect where the curves intersect.  Thus, computational complexity could possibly be reduced by only refining the estimates of vertices near $u_i$.

There may also be use for $\hat k_\delta$ in graph property testing.  Property testing refers to using an easily computable graph property in order to give an estimate of a less tractable property.  For example, the {\em hyperbolicity} of a graph informally measures the extent to which a graph is tree-like 
\cite{cohen:hyperbolicity}.  As was discussed in Section~\ref{sec:tree_error}, tree-like structures with high degree but low degeneracy can lead to large errors in $\hat k_\delta$.  Therefore, a large error in $\hat k_\delta(v)$ may indicate that $v$ participates in a structure with low hyperbolicity.  Since the hyperbolicity is computed in $O(|V|^4)$ time, it would be significantly faster to indirectly flag such vertices by computing $\hat k_\delta(v)$ and $k(v)$ at every vertex and observing their difference.

\section*{Acknowledgments}
The authors thank Johan Ugander for introducing them to the problem of 
calculating the $k$-cores in a network experiment during a workshop at the Statistical and Applied Mathematical Sciences Institute (SAMSI)
 and for providing helpful comments and discussion that improved the manuscript.
This work was supported in part by the National Consortium for Data Science
Faculty Fellows Program and the Defense Advanced Research Projects Agency
under SPAWAR Systems Center, Pacific Grant N66001-14-1-4063. Any opinions,
findings, and conclusions or recommendations expressed in this publication are
those of the author(s) and do not necessarily reflect the views of DARPA, SSC
Pacific, or the NCDS.



\bibliographystyle{IEEEtran}
\bibliography{refs/2014_ICDM}
%
%
%
\appendix
\begin{table}[!h]
	\centering
	{
	\begin{tabular}{|c|c|c|c|c|c|}
	\hline
	\bf Graph & \bf $|V|$ & \bf $|E|$ &  \bf $\operatorname{max} d(v)$ & \bf $D$ & \bf $\Delta$  \\\hline
	\textsc{A.~thaliana}\cite{biogrid}  & 6854 & 16615 & 1308 & 15 & 14\\
	Protein-protein interation &&&&& \\\hline
	\textsc{Amazon}\cite{snap}  & 334863 & 925872 &  549 & 6& 47\\
	Co-purchases &&&&& \\\hline
	\textsc{AS}\cite{snap} & 6214 & 12232 & 1397 & 12 & 9 \\
	Autonomous systems &&&&& \\\hline
	\textsc{ca-AstroPh}\cite{snap}  & 17903 & 196972 &  504 & 56 & 14 \\
	Academic citations &&&&& \\\hline
	\textsc{DBLP}\cite{snap}  & 317080 & 1049866 &  343 & 113& 23 \\
	Academic citations &&&&& \\\hline
	\textsc{Enron}\cite{snap}  & 33696 & 180811 &  1383 & 43 & 13\\
	Email correspondence &&&&& \\\hline
	\textsc{Facebook}\cite{traud:facebook}  & 36371 & 1590655 &  6312 & 81 & 7\\
	Facebook friendship &&&&& \\\hline
	\textsc{Facebook 2}\cite{traud:facebook}  & 1657 & 61049 &  577 & 60 & 6\\
	Facebook friendship &&&&& \\\hline
	\textsc{Facebook 3}\cite{traud:facebook}  & 1446 & 59589 &  375 & 60 & 6\\
	Facebook friendship &&&&& \\\hline
	\textsc{Facebook 4}\cite{traud:facebook}  & 2672 & 65244 &  405 & 43 & 7\\
	Facebook friendship &&&&& \\\hline
	\textsc{Facebook 5}\cite{traud:facebook}  & 2250 & 84386 & 670 & 58 & 6\\
	Facebook friendship &&&&& \\\hline
	\textsc{Gnutella}\cite{snap}  &  26498 & 65359 &  355 & 5 & 11\\
	Peer-to-peer filesharing &&&&& \\\hline
	\textsc{H.~sapiens}\cite{biogrid}  & 18625 & 146322 & 9777 & 47 & 10\\
	Protein-protein interation &&&&& \\\hline
	\textsc{WPG}\cite{iri}  &4941 & 6594 &  19 & 5 & 46\\
	Western US power grid &&&&& \\\hline
	\end{tabular}
	}
	\caption{
	Summary statistics for real-world graphs \label{tab:data_summary}
	}
\end{table} 

\begin{table}[!h]
	\centering
	\begin{tabular}{|c|cc|cc|cc|}\hline
	 & \multicolumn{2}{c|}{$\delta=1$}& \multicolumn{2}{c|}{$\delta = 2$}& \multicolumn{2}{c|}{$\delta = 3$} \\
	\bf Graph & \bf Mean & \bf Max & \bf Mean & \bf Max & \bf Mean & \bf Max \\\hline
	\textsc{A.~thaliana}  & 6 & 1309 & 342 & 2760 & 1207 & 6113 \\\hline
	\textsc{Amazon} & 7 & 550 & 42 & 1397 & 156 & 5723 \\\hline
	\textsc{AS} & 5 & 1398 & 544 & 4330 & 2714 & 5939  \\\hline
	\textsc{ca-AstroPh} & 23 & 505 & 519 & 6065 & 4517 & 14659 \\\hline
	\textsc{DBLP} & 8 & 344 & 88 & 5417 & 1052 & 48236 \\\hline
	\textsc{Enron} & 13 & 1384 & 905 & 16745 & 9319 & 30572 \\\hline
	\textsc{Facebook} & 88 & 6313 & 7744 & 36371 & 32512 & 36371  \\\hline
	\textsc{Facebook2} & 75 & 578 & 1093 & 1578 & 1622 & 1656 \\\hline
	\textsc{Facebook3} & 83 & 376 & 1056 & 1406 & 1426 & 1445 \\\hline 
	\textsc{Facebook4} & 50 & 406 & 1038 & 2332 & 2427 & 2664  \\\hline
	\textsc{Facebook5} & 76 & 671 & 1293 & 2150 & 2171 & 2248 \\\hline
	\textsc{Gnutella} & 6 & 356 & 59 & 4282 & 574 & 15971 \\\hline
	\textsc{H.~sapiens} & 17 & 9778 & 5845 & 18625 & 15128 & 18625  \\\hline
	\textsc{WPG} & 4 & 20 & 10 & 61 & 23 & 142 \\\hline
	\end{tabular}
	\caption{Variations in size of $N_\delta$}
\end{table}
\begin{table}[!h]
	\centering
	\begin{tabular}{|c|cc|cc|cc|cc|}\hline
	 & \multicolumn{2}{c|}{$\delta=1$}& \multicolumn{2}{c|}{$\delta = 2$}& \multicolumn{2}{c|}{$\delta = 3$} & \multicolumn{2}{c|}{$\delta = 4$} \\
	\bf Graph & \bf Avg. & \bf Var. & \bf Avg. & \bf Var. & \bf Avg. & \bf Var. & \bf Avg. & \bf Var. \\\hline
	\textsc{A.~thaliana} & $.00$ & $.00$ & $.07$ & $.03$ & $.20$ & $.05$ & $.55$ & $.07$ \\\hline
	\textsc{Amazon} & $.00$ & $.00$ & $.00$ & $.00$ & $.00$ & $.00$ & $.00$ & $.00$   \\\hline
	\textsc{AS} & $.00$ & $.00$ & $.09$ & $.01$ & $.44$ & $.07$ & $.82$ & $.04$ 
 \\\hline
	\textsc{ca-AstroPh} & $.00$ & $.00$ & $.03$ & $.00$ & $.25$ & $.04$ & $.66$ & $.06$ 
 \\\hline
	\textsc{DBLP} & $.00$ & $.00$ & $.00$ & $.00$ & $.00$ & $.00$ & $.03$ & $.00$ 
\\\hline
	\textsc{Enron} & $.00$ & $.00$ & $.03$ & $.00$ & $.28$ & $.04$ & $.74$ & $.06$ 
 \\\hline
	\textsc{Facebook} & $.00$ & $.00$ & $.21$ & $.03$ & $.89$ & $.02$ & $1.00$ & $.00$ 
 \\\hline
	\textsc{Facebook2} & $.05$ & $.00$ & $.66$ & $.05$ & $.98$ & $.00$ & $1.00$ & $.00$ 
\\\hline
	\textsc{Facebook3} & $.06$ & $.00$ & $.73$ & $.04$ & $.99$ & $.00$ & $1.00$ & $.00$ 
\\\hline
	\textsc{Facebook4} & $.02$ & $.00$ & $.39$ & $.04$ & $.91$ & $.02$ & $.99$ & $.00$ 
\\\hline
	\textsc{Facebook5} & $.03$ & $.00$ & $.57$ & $.05$ & $.96$ & $.01$ & $1.00$ & $.00$ 
\\\hline
	\textsc{Gnutella} & $.00$ & $.00$ & $.00$ & $.00$ & $.02$ & $.00$ & $.15$ & $.02$ 
 \\\hline
	\textsc{H.~sapiens}& $.00$ & $.00$ & $.31$ & $.07$ & $.81$ & $.04$ & $.98$ & $.00$ 
 \\\hline
	\textsc{WPG} & $.00$ & $.00$ & $.00$ & $.00$ & $.00$ & $.00$ & $.01$ & $.00$ 
 \\\hline
	\end{tabular}
	\caption{Proportion of vertices in $N_\delta$.  Values less than $.01$ rounded to $0$.\label{tab:neighborhood_size}}
\end{table}
\pagebreak

\begin{figure}[H]
	\centering
	\includegraphics[width=0.62\linewidth]{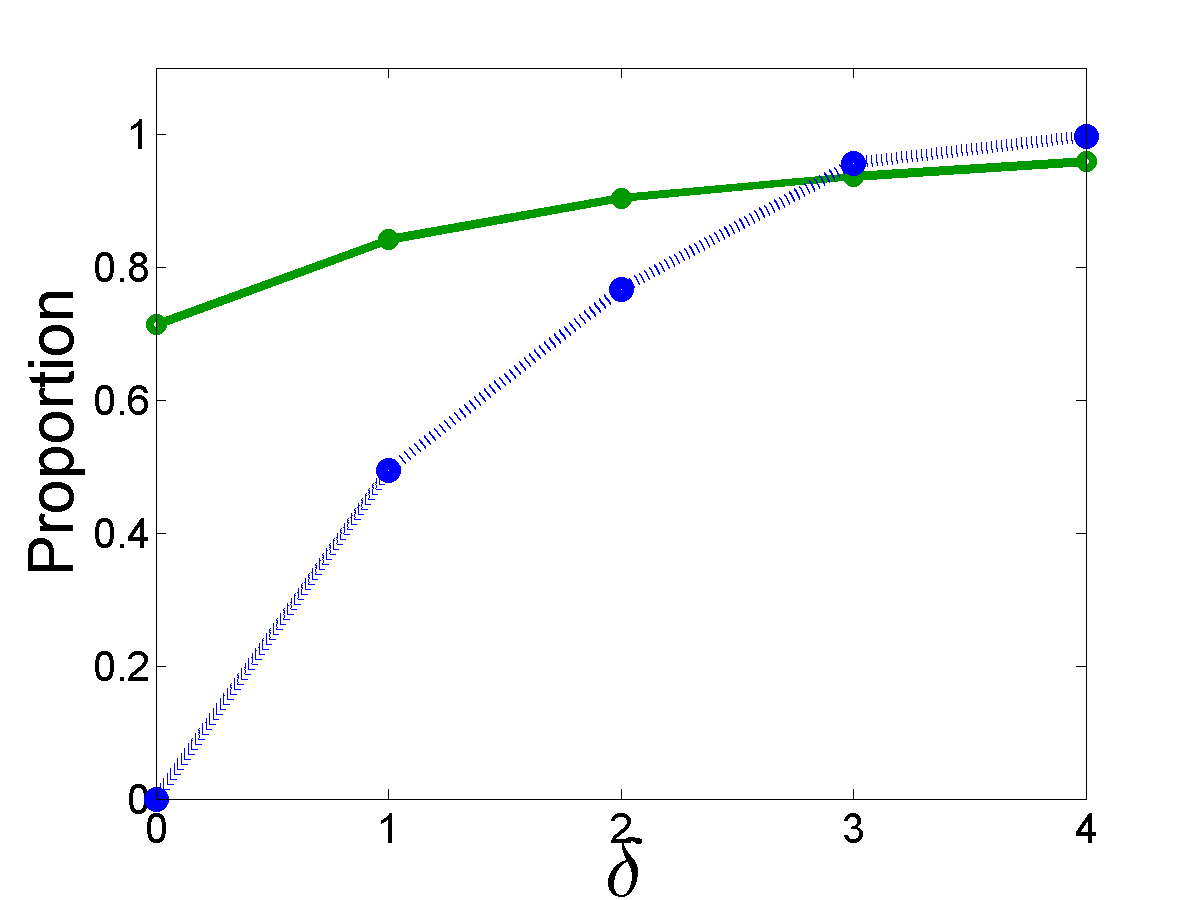}
	\caption{Proportion of vertices in \textsc{A.~thaliana} with optimal core number estimate ratios (see Figure~\ref{fig:unnormalized_correct}).}
\end{figure}
\begin{figure}[H]
	\centering
	\includegraphics[width=0.62\linewidth]{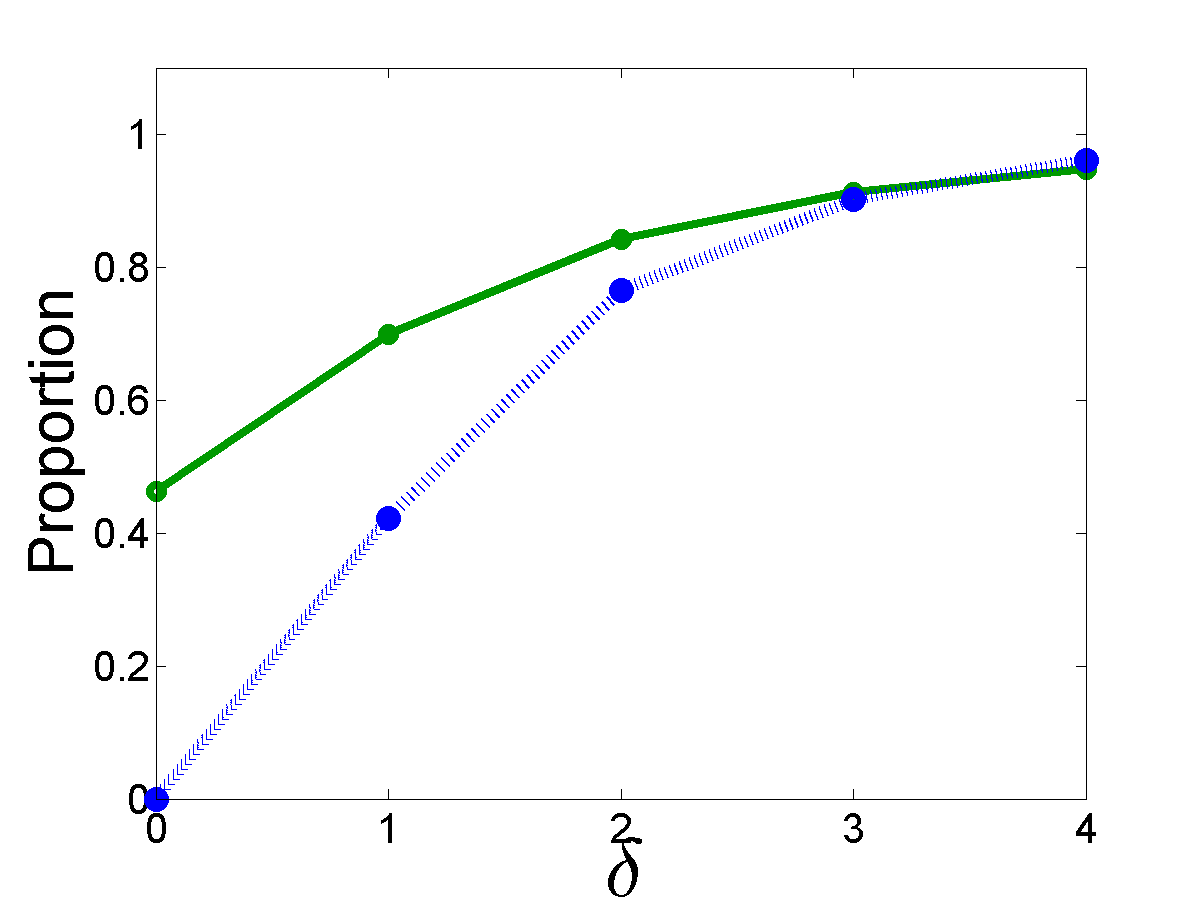}
\caption{Proportion of vertices in \textsc{Amazon} with optimal core number estimate ratios (see Figure~\ref{fig:unnormalized_correct}).}\end{figure}
\begin{figure}[H]
	\centering
	\includegraphics[width=0.62\linewidth]{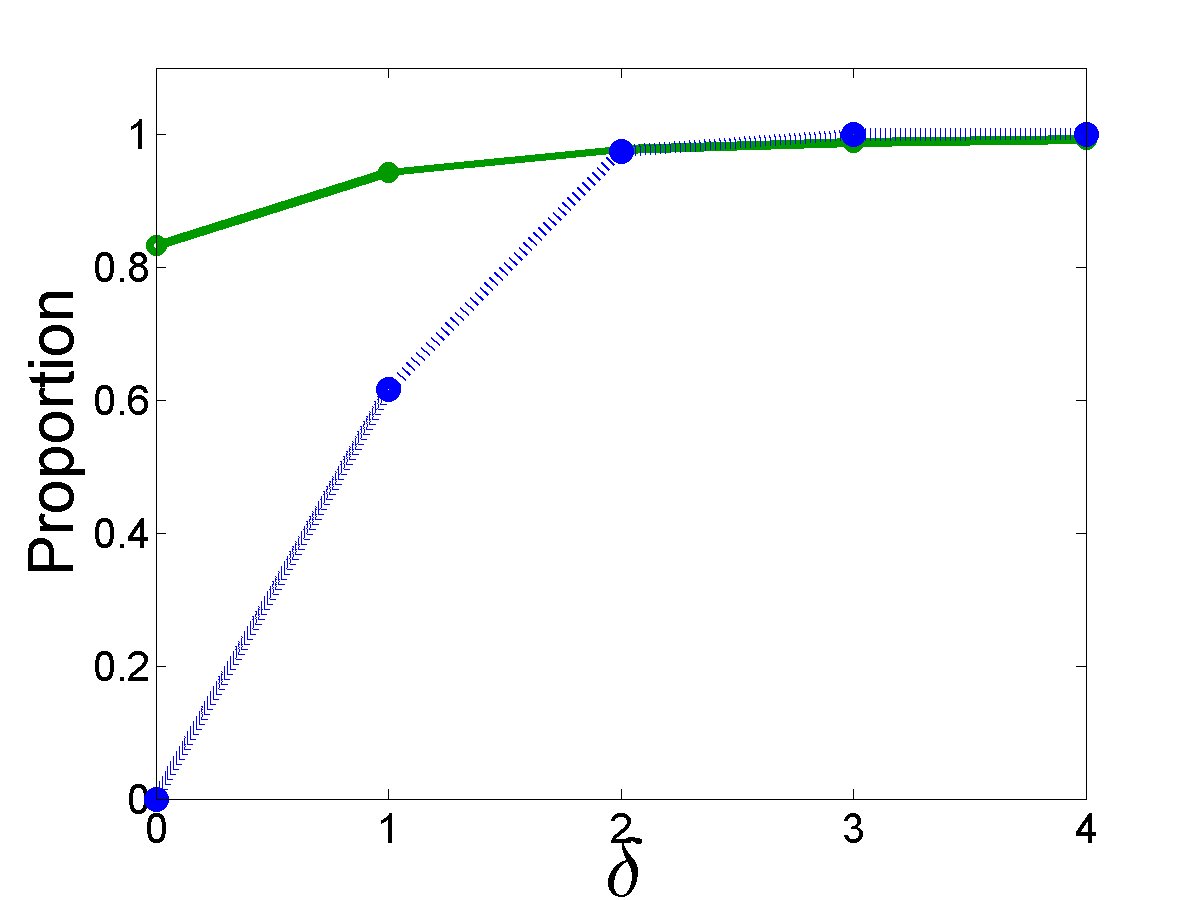}
\caption{Proportion of vertices in \textsc{AS} with optimal core number estimate ratios (see Figure~\ref{fig:unnormalized_correct}).}\end{figure}
\begin{figure}[H]
	\centering
	\includegraphics[width=0.62\linewidth]{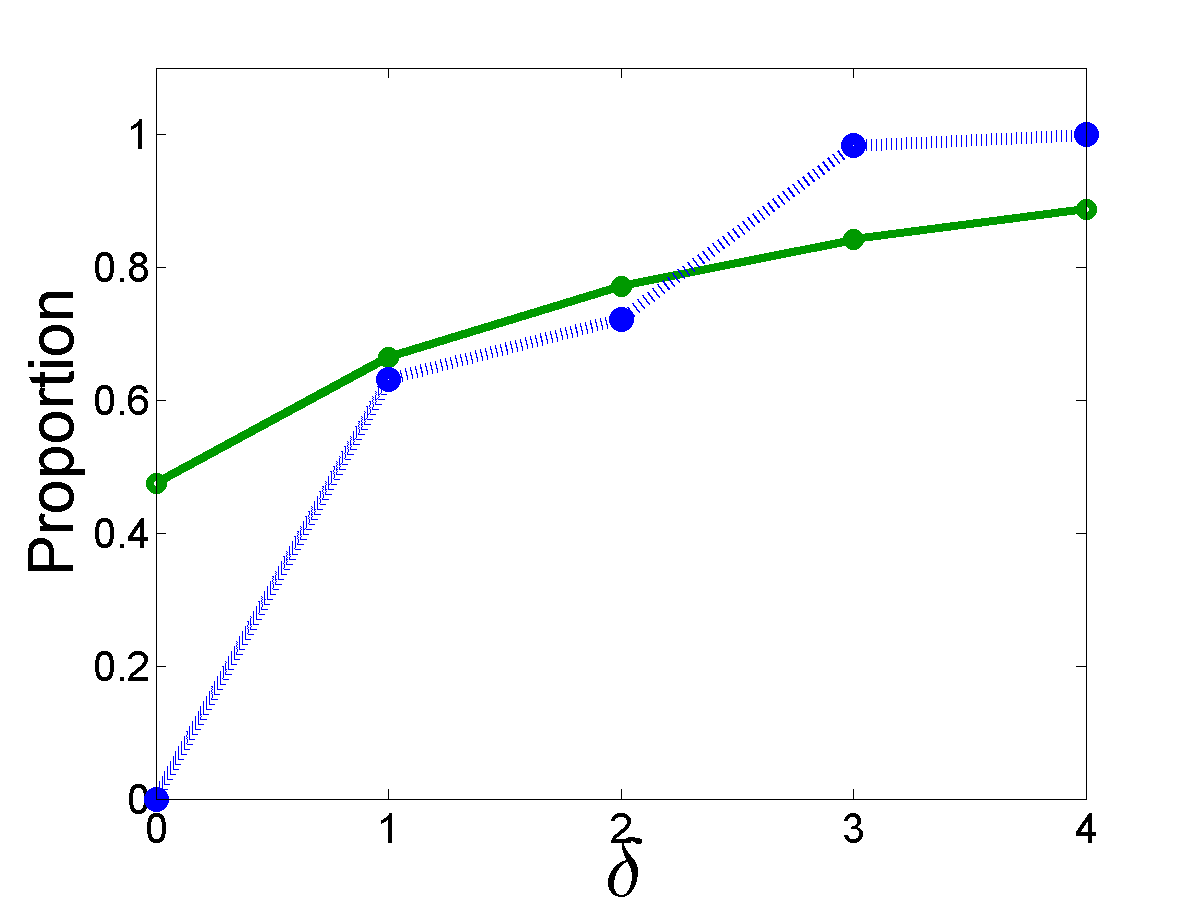}
\caption{Proportion of vertices in \textsc{ca-AstroPh} with optimal core number estimate ratios (see Figure~\ref{fig:unnormalized_correct}).}\end{figure}
\begin{figure}[H]
	\centering
	\includegraphics[width=0.62\linewidth]{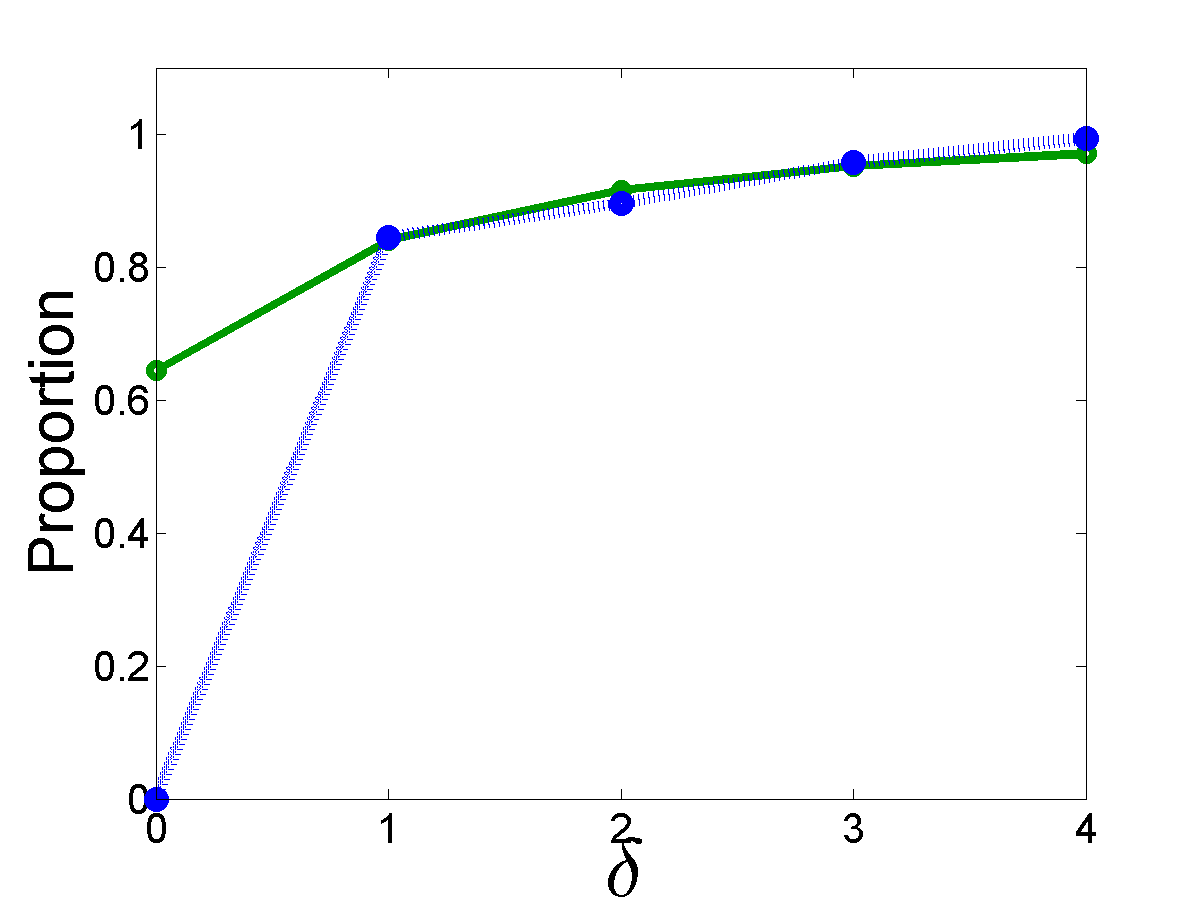}
\caption{Proportion of vertices in \textsc{DBLP} with optimal core number estimate ratios (see Figure~\ref{fig:unnormalized_correct}).}\end{figure}
\begin{figure}[H]
	\centering
	\includegraphics[width=0.62\linewidth]{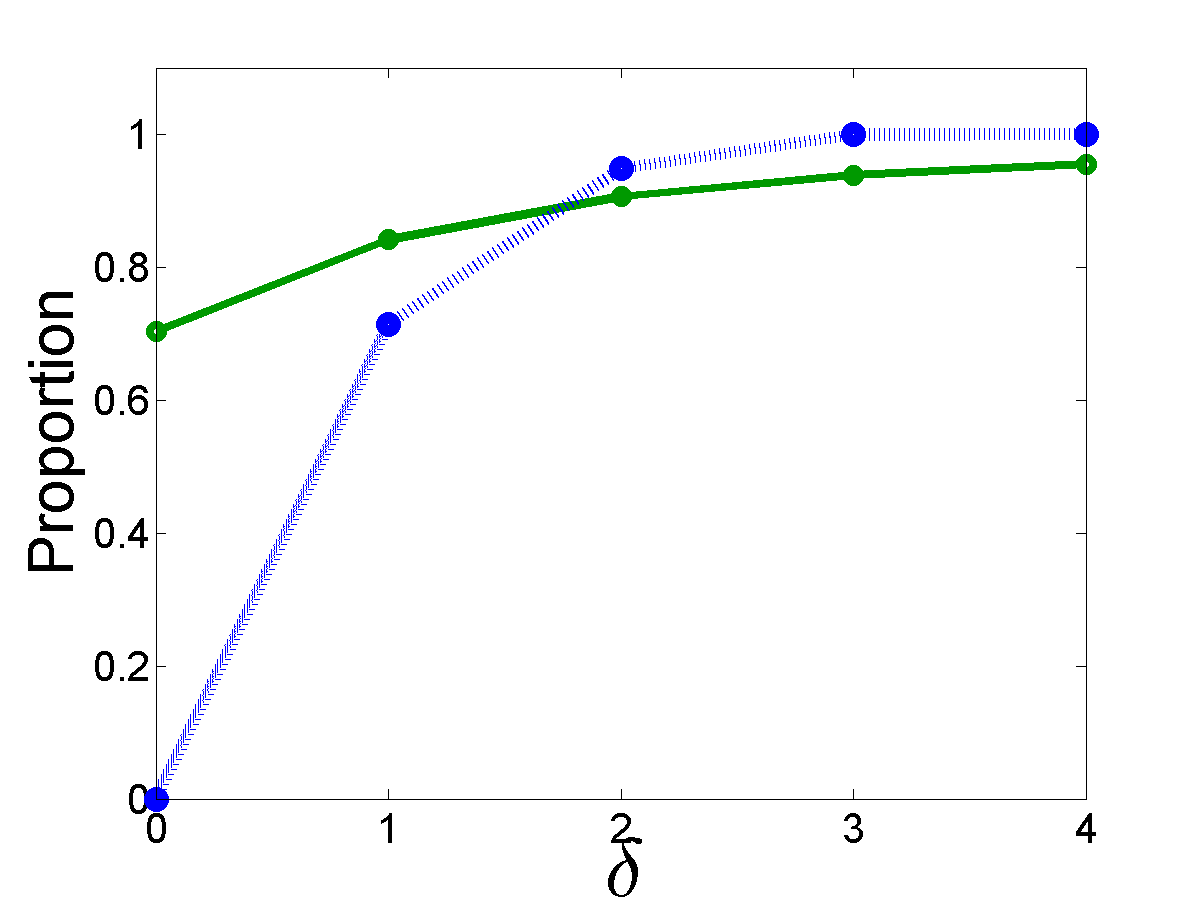}
\caption{Proportion of vertices in \textsc{Enron} with optimal core number estimate ratios (see Figure~\ref{fig:unnormalized_correct}).}\end{figure}
\begin{figure}[H]
	\centering
	\includegraphics[width=0.62\linewidth]{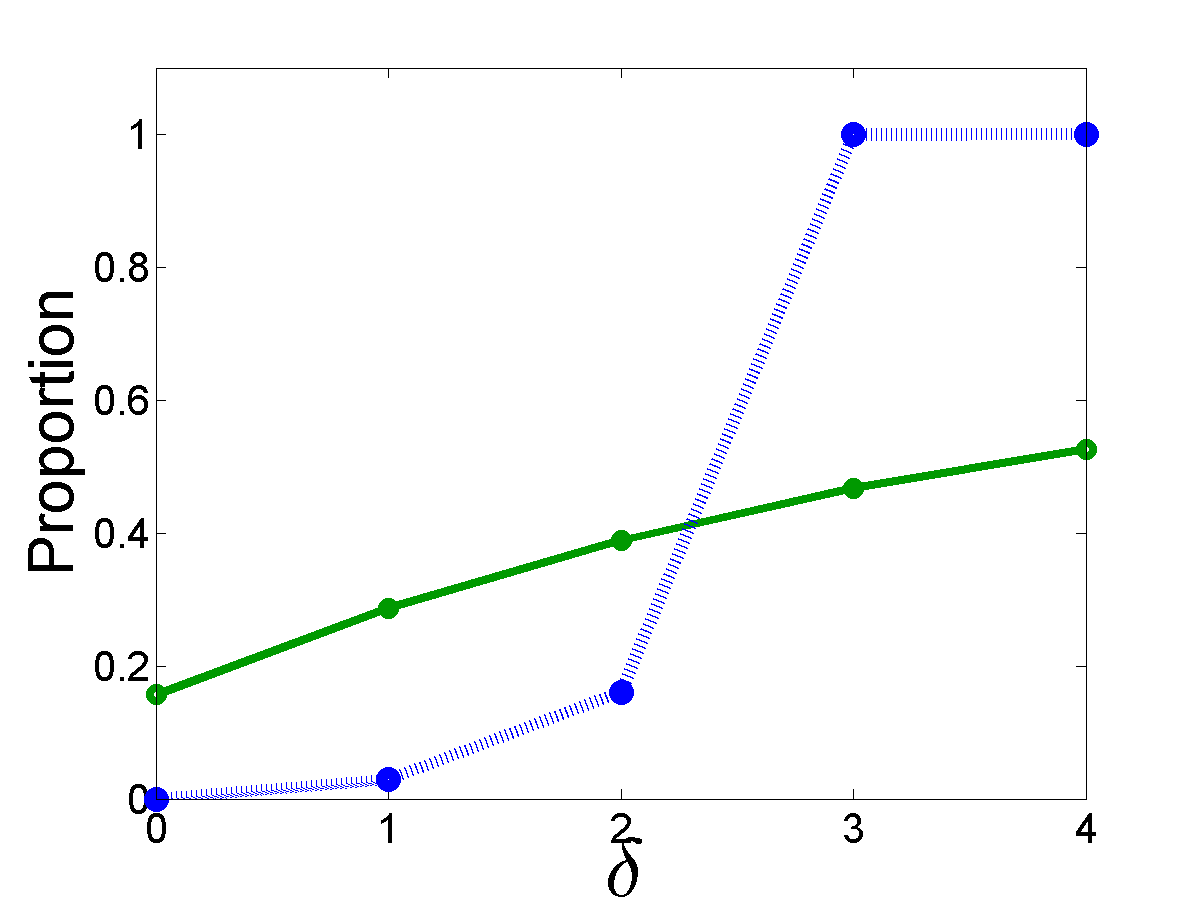}
\caption{Proportion of vertices in \textsc{Facebook} with optimal core number estimate ratios (see Figure~\ref{fig:unnormalized_correct}).}\end{figure}
\begin{figure}[H]
	\centering
	\includegraphics[width=0.62\linewidth]{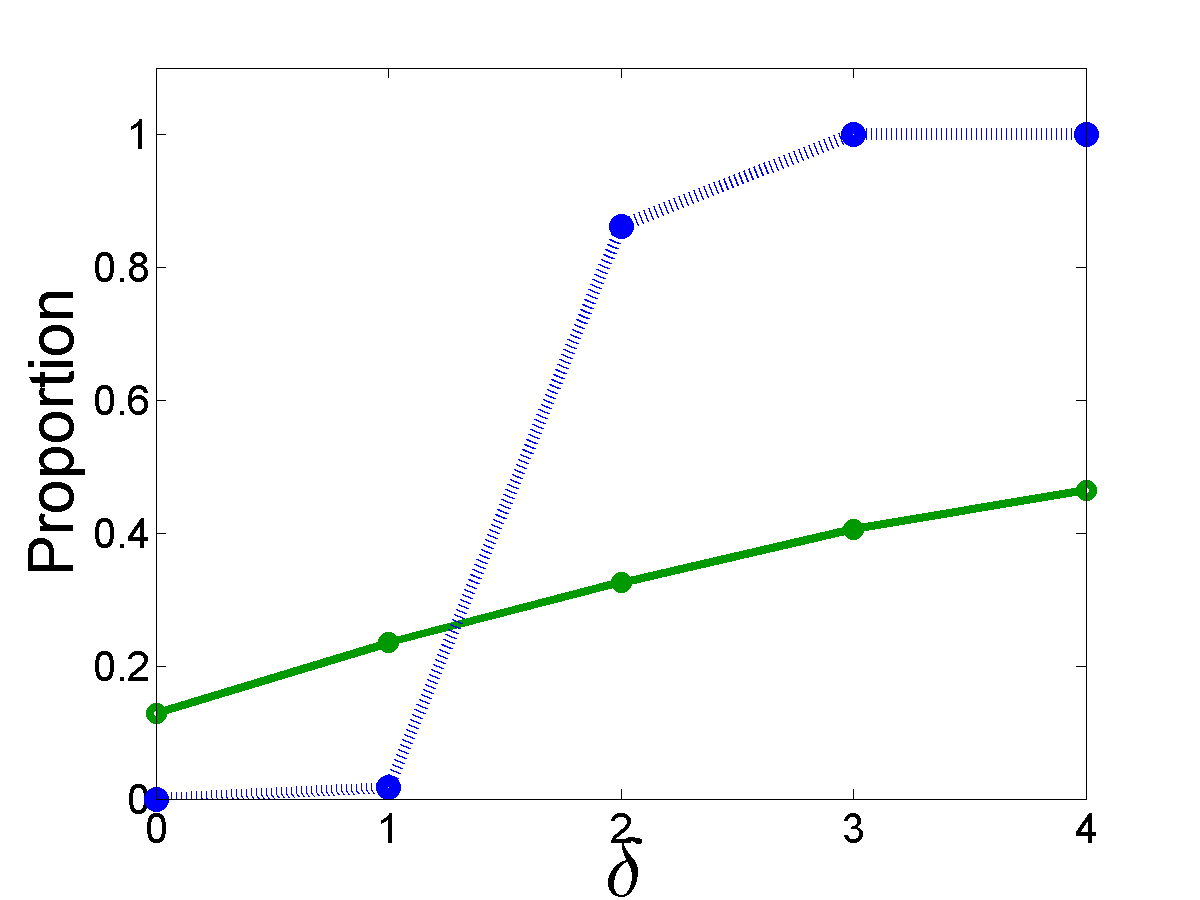}
\caption{Proportion of vertices in \textsc{Facebook2} with optimal core number estimate ratios (see Figure~\ref{fig:unnormalized_correct}).}\end{figure}
\begin{figure}[H]
	\centering
	\includegraphics[width=0.62\linewidth]{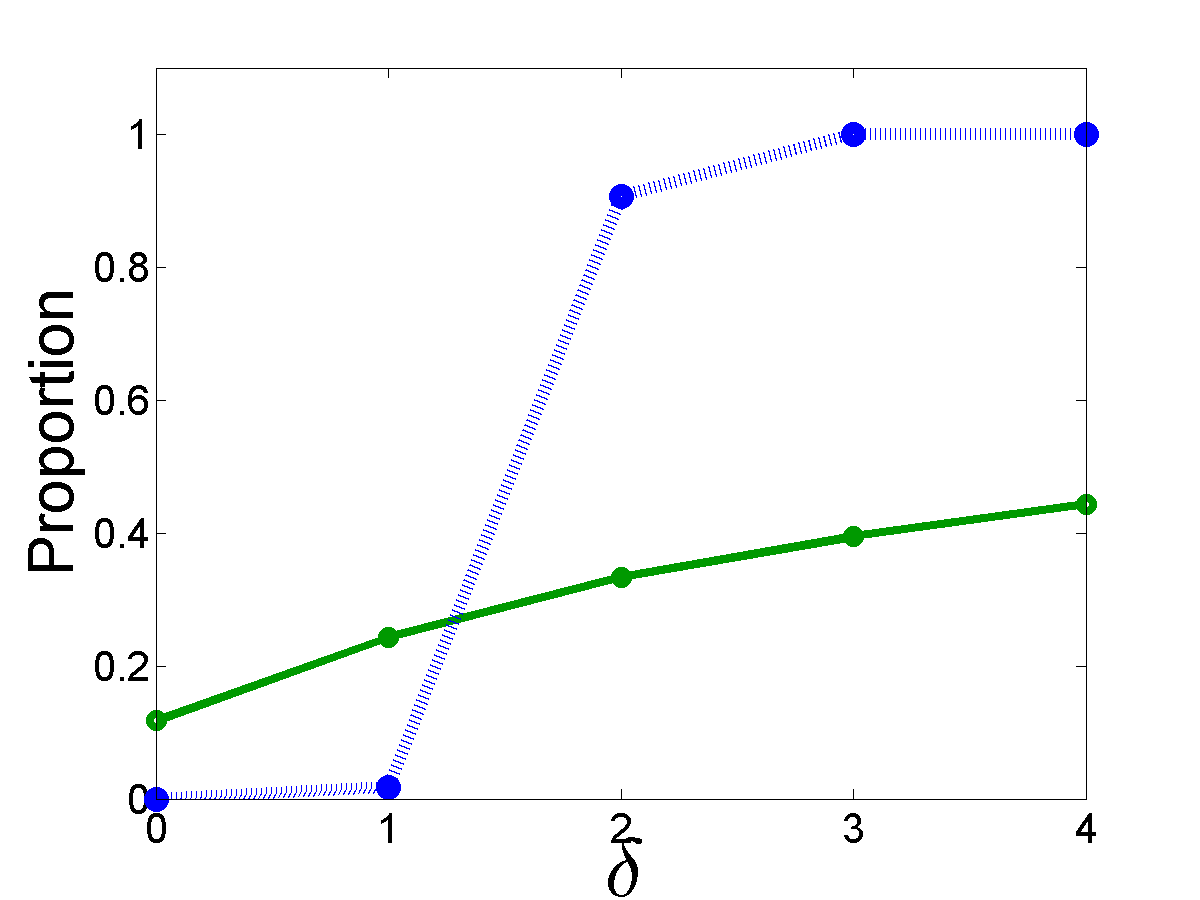}
\caption{Proportion of vertices in \textsc{Facebook3} with optimal core number estimate ratios (see Figure~\ref{fig:unnormalized_correct}).}\end{figure}
\begin{figure}[H]
	\centering
	\includegraphics[width=0.62\linewidth]{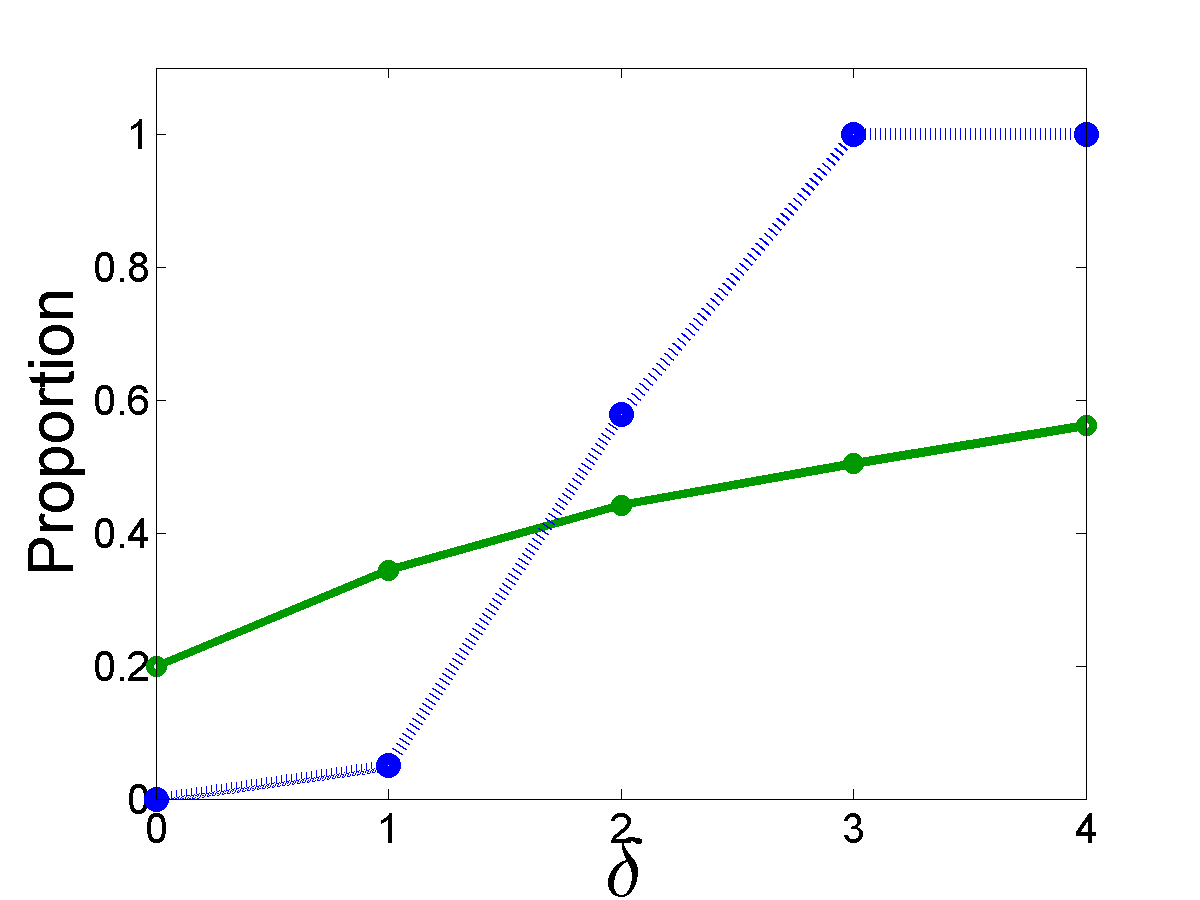}
\caption{Proportion of vertices in \textsc{Facebook4} with optimal core number estimate ratios (see Figure~\ref{fig:unnormalized_correct}).}\end{figure}
\begin{figure}[H]
	\centering
	\includegraphics[width=0.62\linewidth]{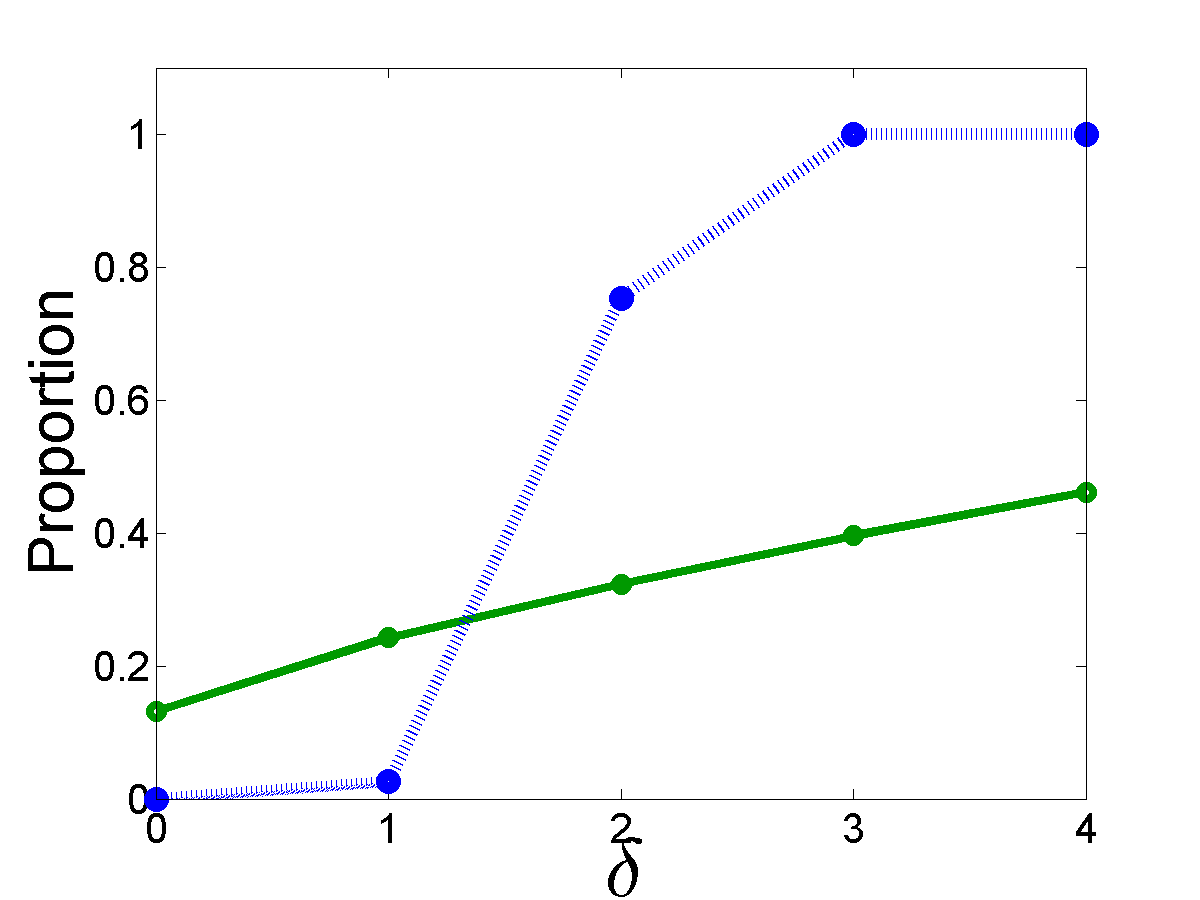}
\caption{Proportion of vertices in \textsc{Facebook5} with optimal core number estimate ratios (see Figure~\ref{fig:unnormalized_correct}).}\end{figure}
\begin{figure}[H]
	\centering
	\includegraphics[width=0.62\linewidth]{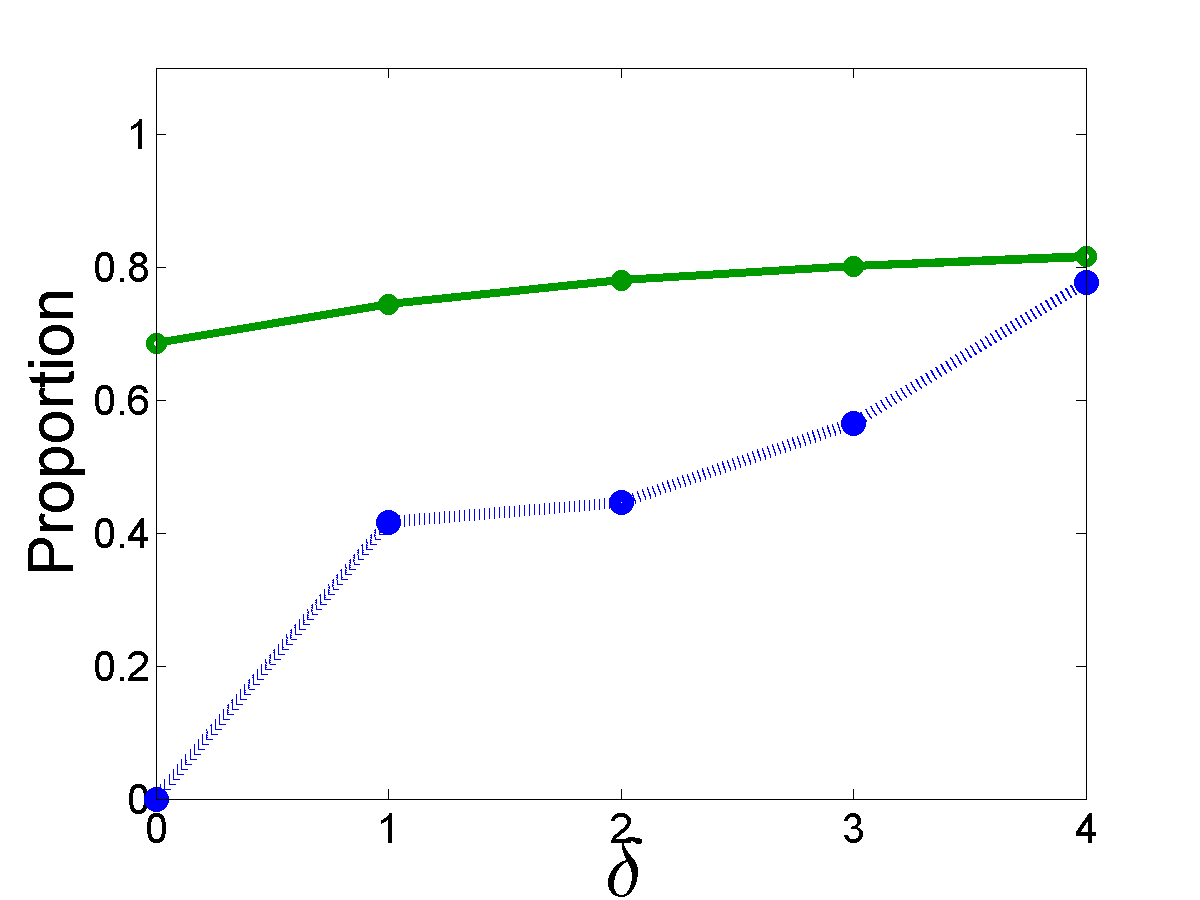}
\caption{Proportion of vertices in \textsc{Gnutella} with optimal core number estimate ratios (see Figure~\ref{fig:unnormalized_correct}).}\end{figure}
\begin{figure}[H]
	\centering
	\includegraphics[width=0.62\linewidth]{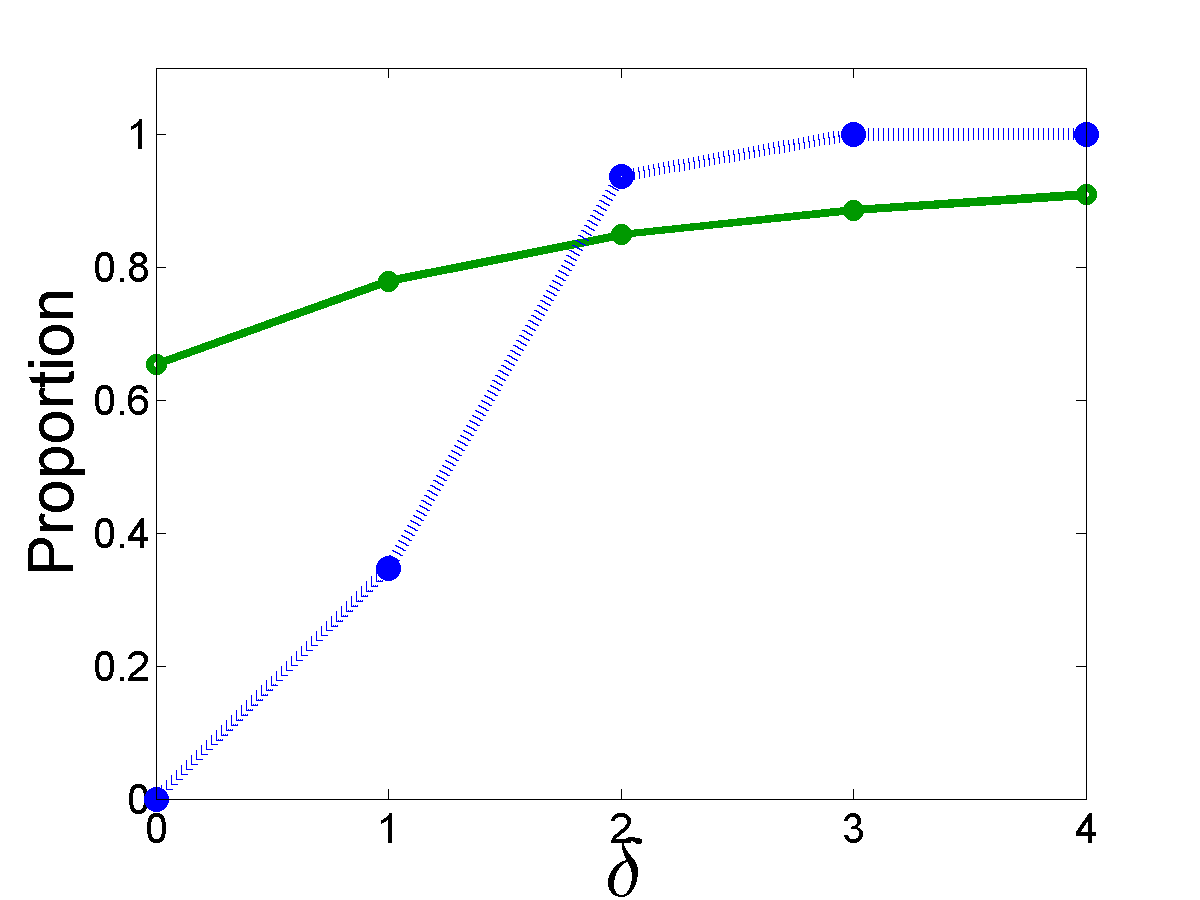}
\caption{Proportion of vertices in \textsc{H.~sapiens} with optimal core number estimate ratios (see Figure~\ref{fig:unnormalized_correct}).}\end{figure}
\begin{figure}[H]
	\centering
	\includegraphics[width=0.62\linewidth]{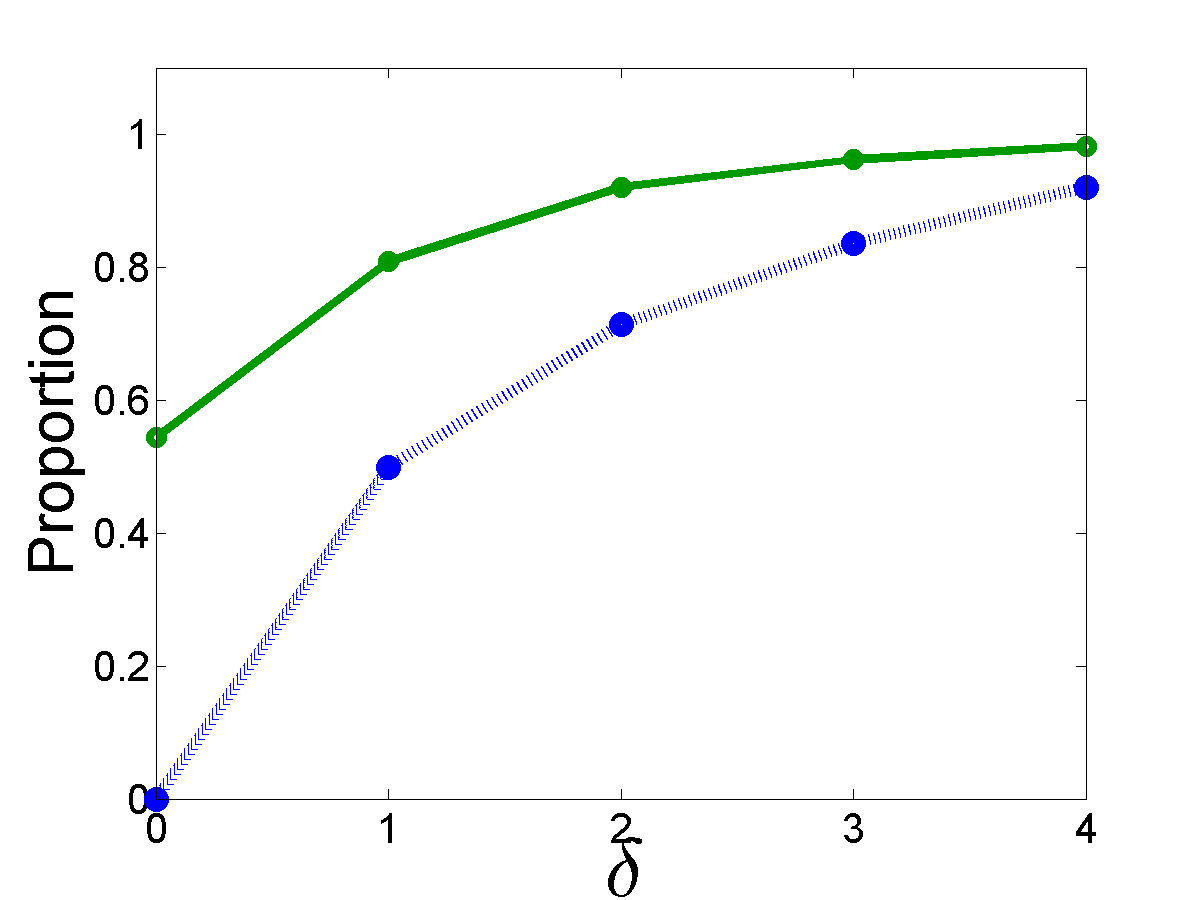}
\caption{Proportion of vertices in \textsc{WPG} with optimal core number estimate ratios (see Figure~\ref{fig:unnormalized_correct}).}\end{figure}

\begin{figure}[H]
	\centering
	\includegraphics[width=0.62\linewidth]{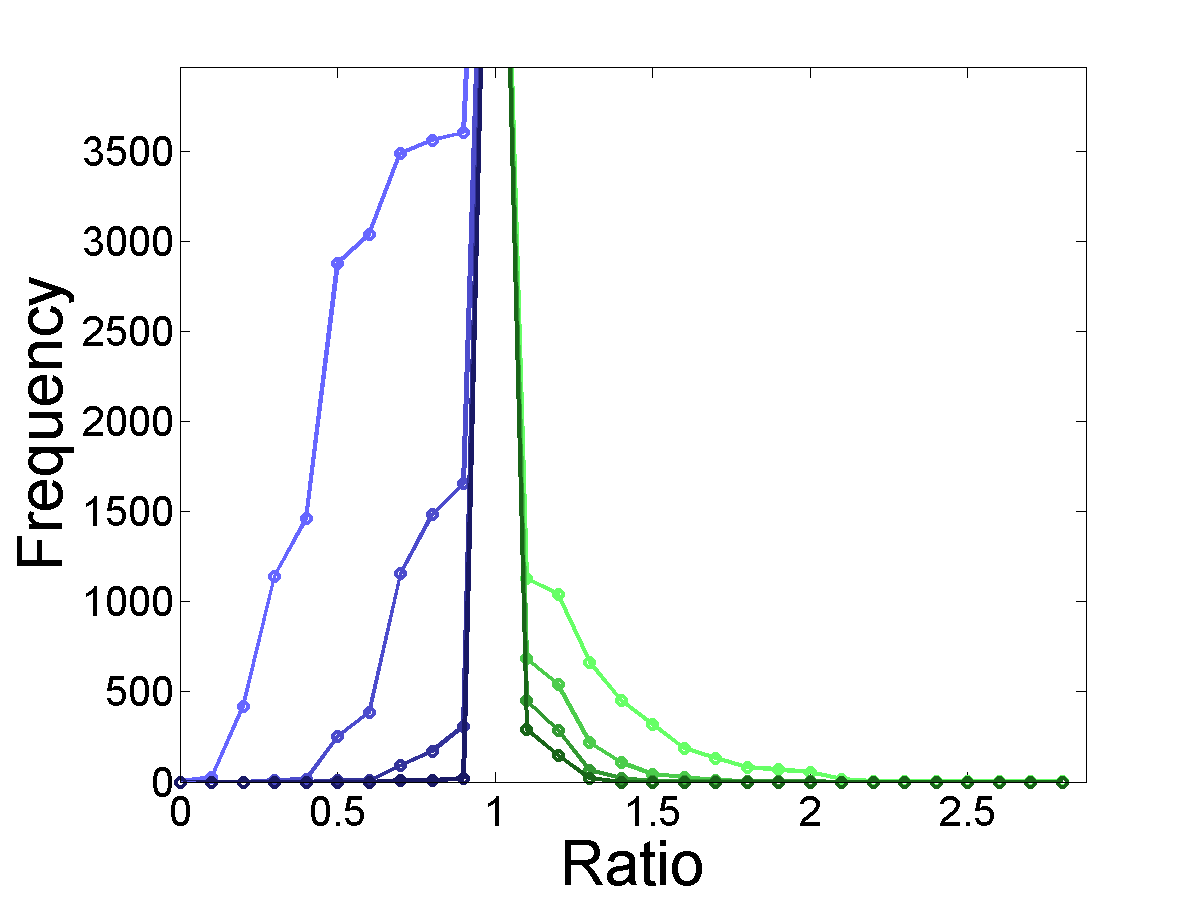}
\caption{Number of vertices in \textsc{A.~thaliana} with core number estimate ratios less optimal that a given threshold (see Figure~\ref{fig:estimate_growth}).}\end{figure}
\begin{figure}[H]
	\centering
	\includegraphics[width=0.62\linewidth]{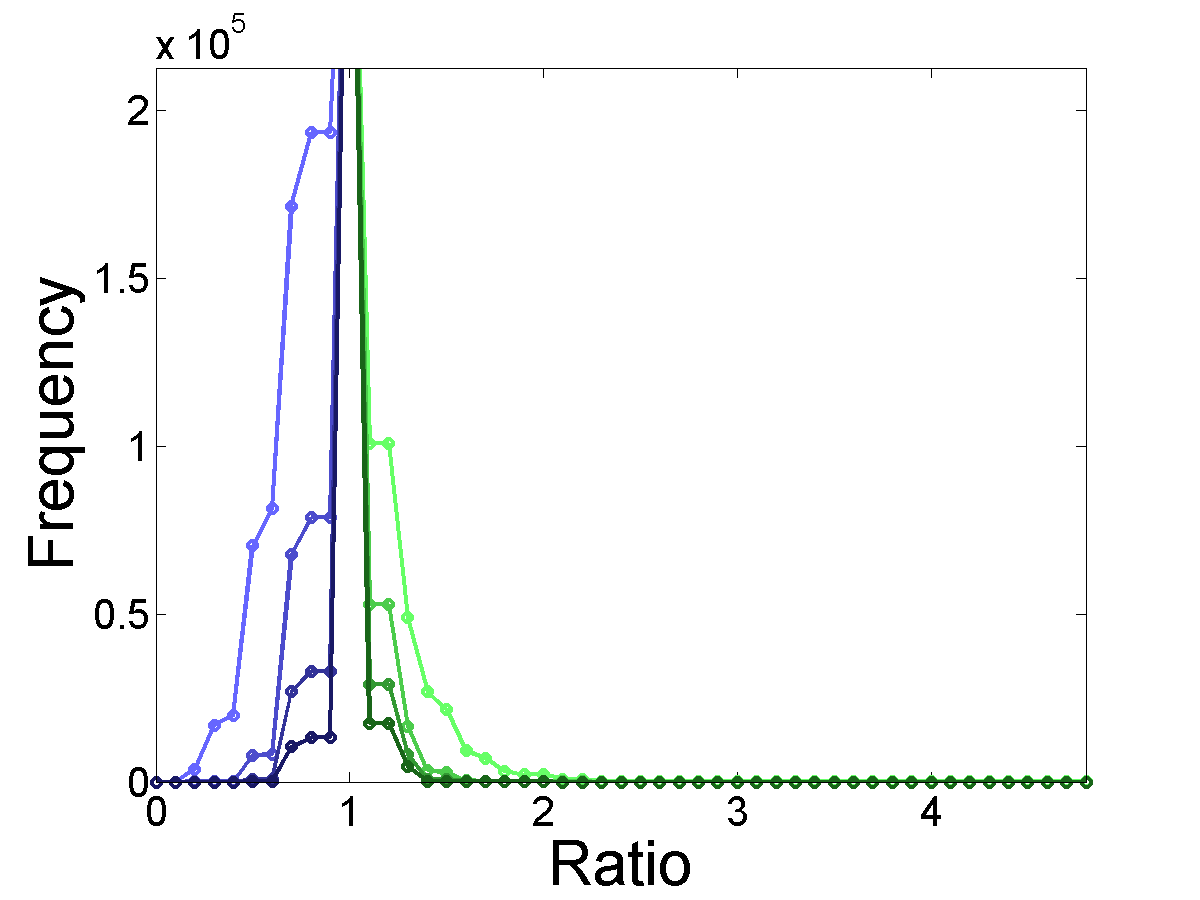}
\caption{Number of vertices in \textsc{Amazon} with core number estimate ratios less optimal that a given threshold (see Figure~\ref{fig:estimate_growth}).}\end{figure}
\begin{figure}[H]
	\centering
	\includegraphics[width=0.62\linewidth]{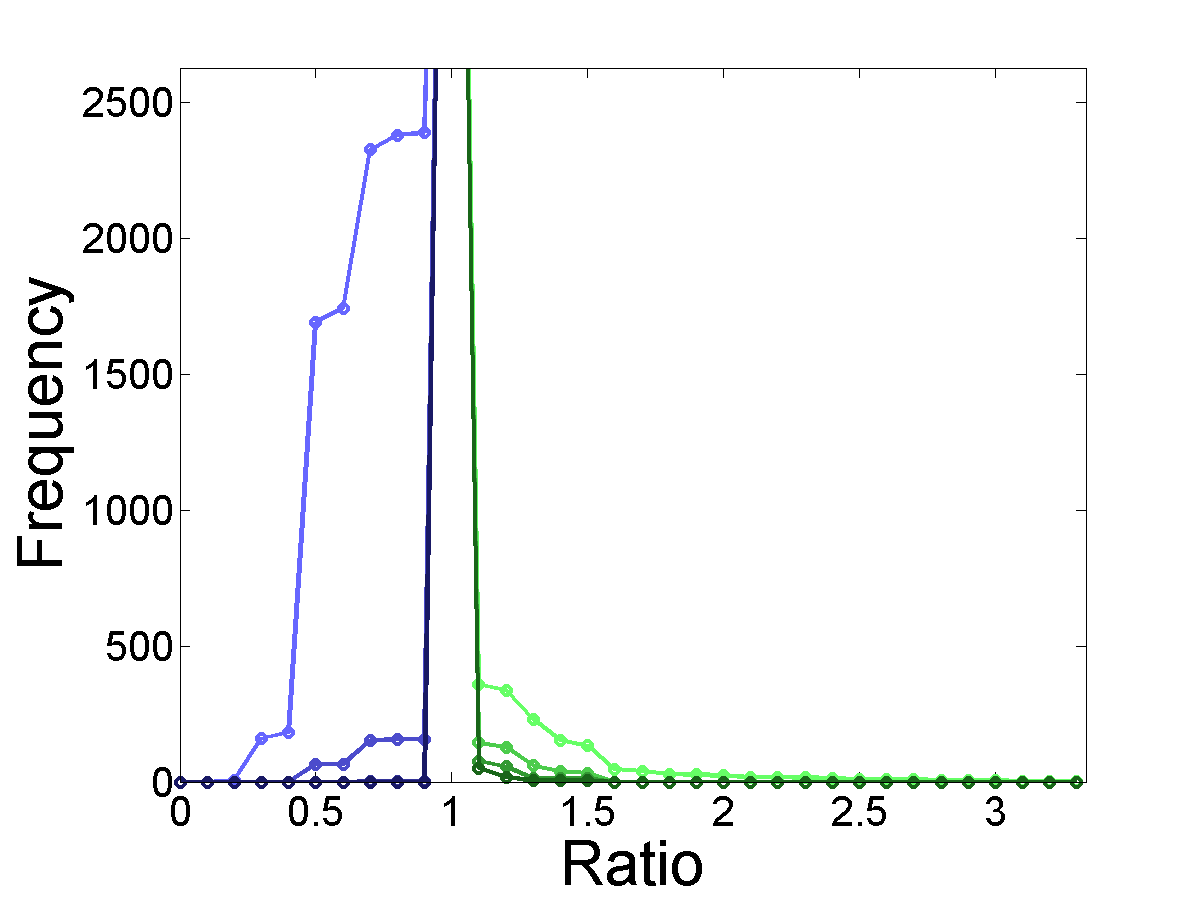}
\caption{Number of vertices in \textsc{AS} with core number estimate ratios less optimal that a given threshold (see Figure~\ref{fig:estimate_growth}).}\end{figure}
\begin{figure}[H]
	\centering
	\includegraphics[width=0.62\linewidth]{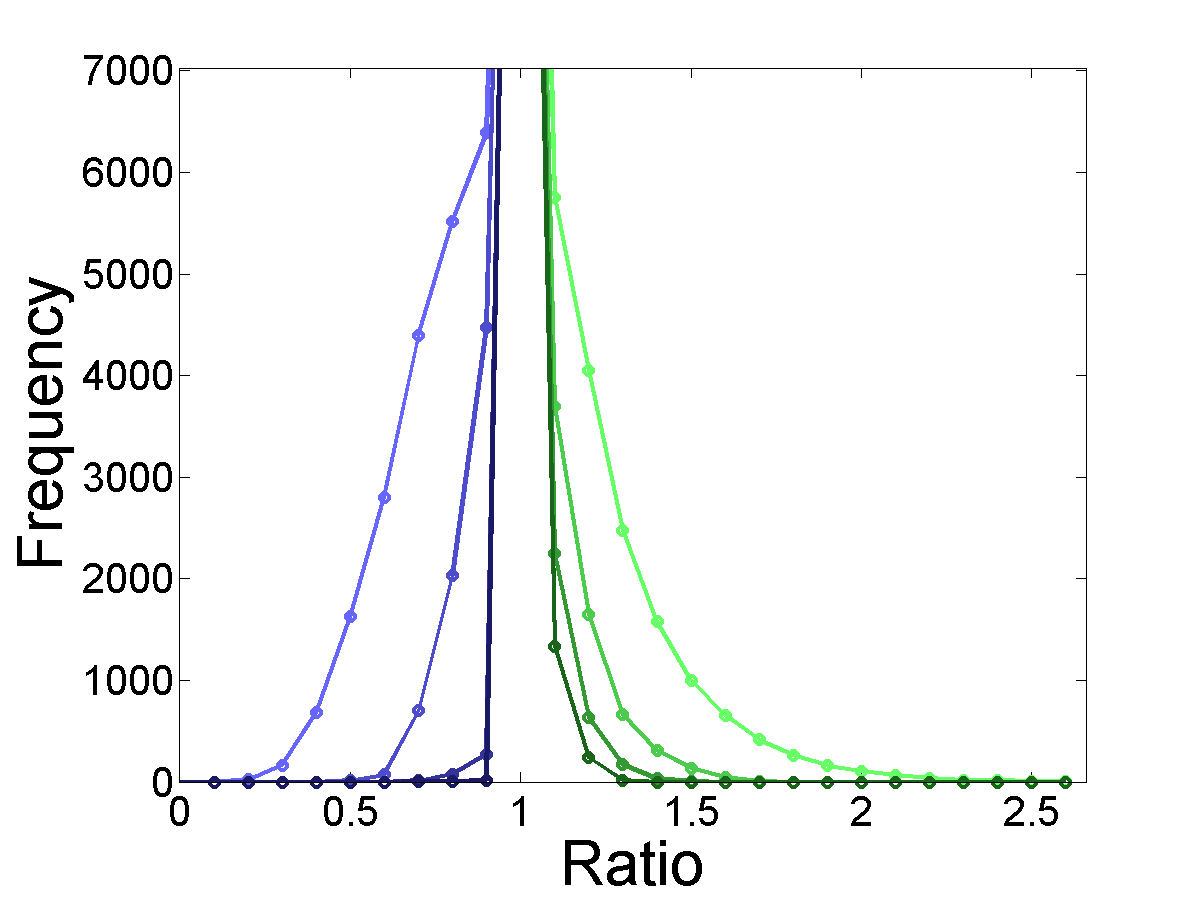}
\caption{Number of vertices in \textsc{ca-AstroPh} with core number estimate ratios less optimal that a given threshold (see Figure~\ref{fig:estimate_growth}).}\end{figure}
\begin{figure}[H]
	\centering
	\includegraphics[width=0.62\linewidth]{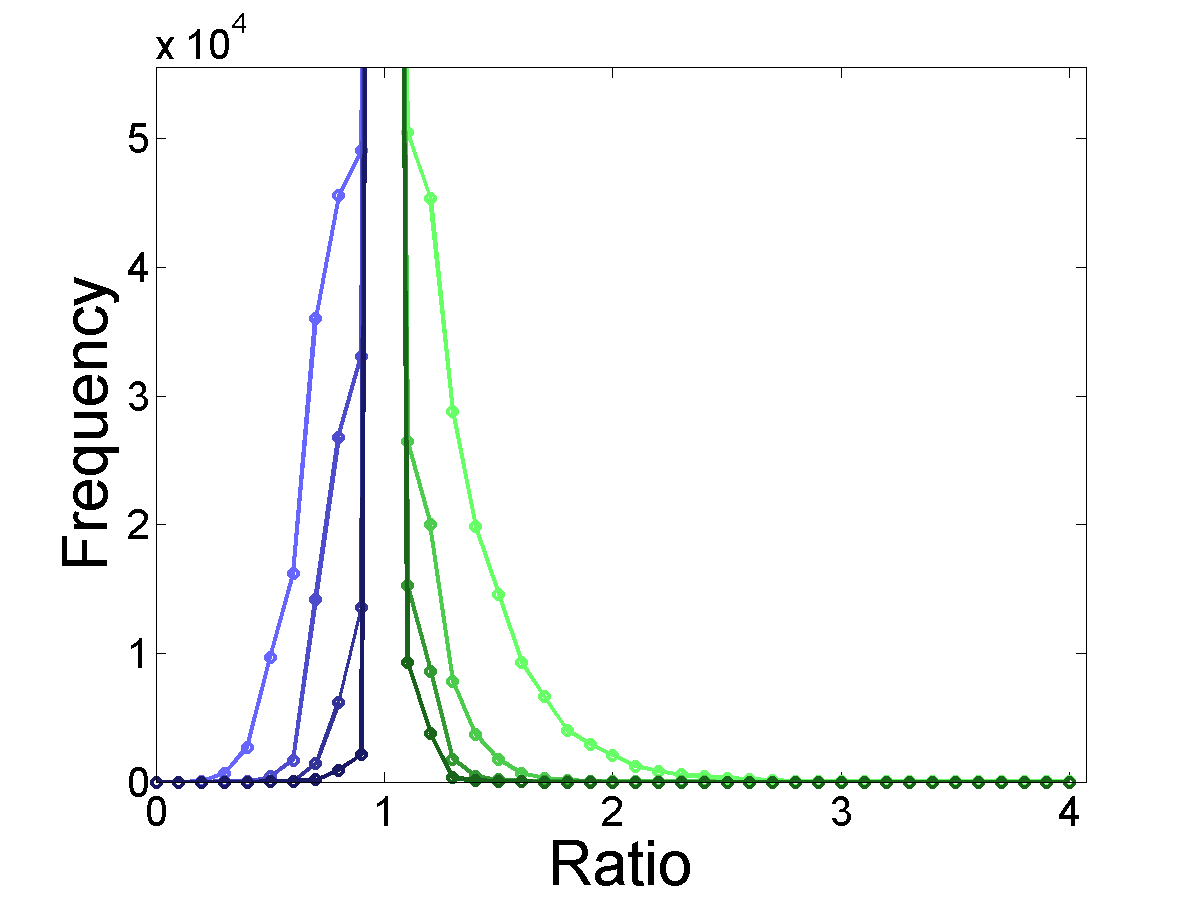}
\caption{Number of vertices in \textsc{DBLP} with core number estimate ratios less optimal that a given threshold (see Figure~\ref{fig:estimate_growth}).}\end{figure}
\begin{figure}[H]
	\centering
	\includegraphics[width=0.62\linewidth]{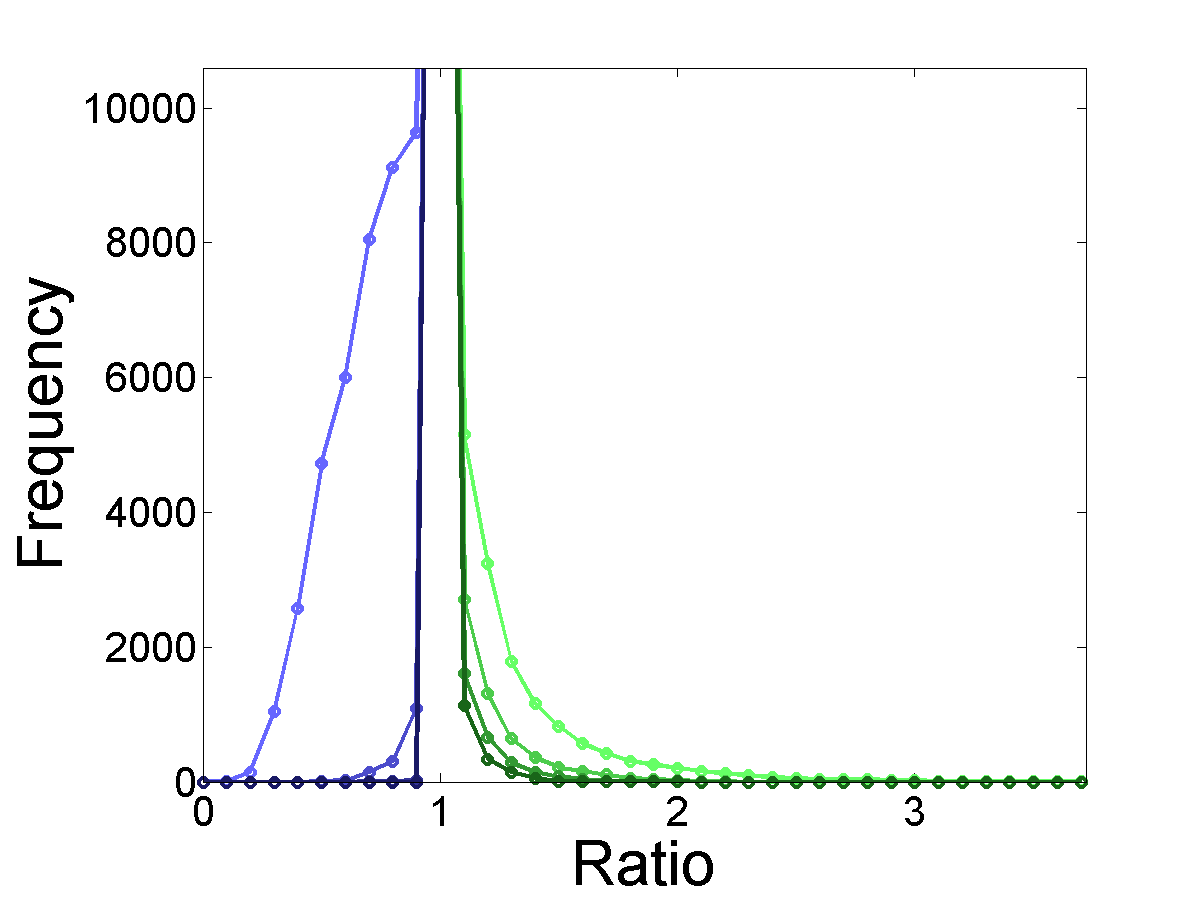}
\caption{Number of vertices in \textsc{Enron} with core number estimate ratios less optimal that a given threshold (see Figure~\ref{fig:estimate_growth}).}\end{figure}
\begin{figure}[H]
	\centering
	\includegraphics[width=0.62\linewidth]{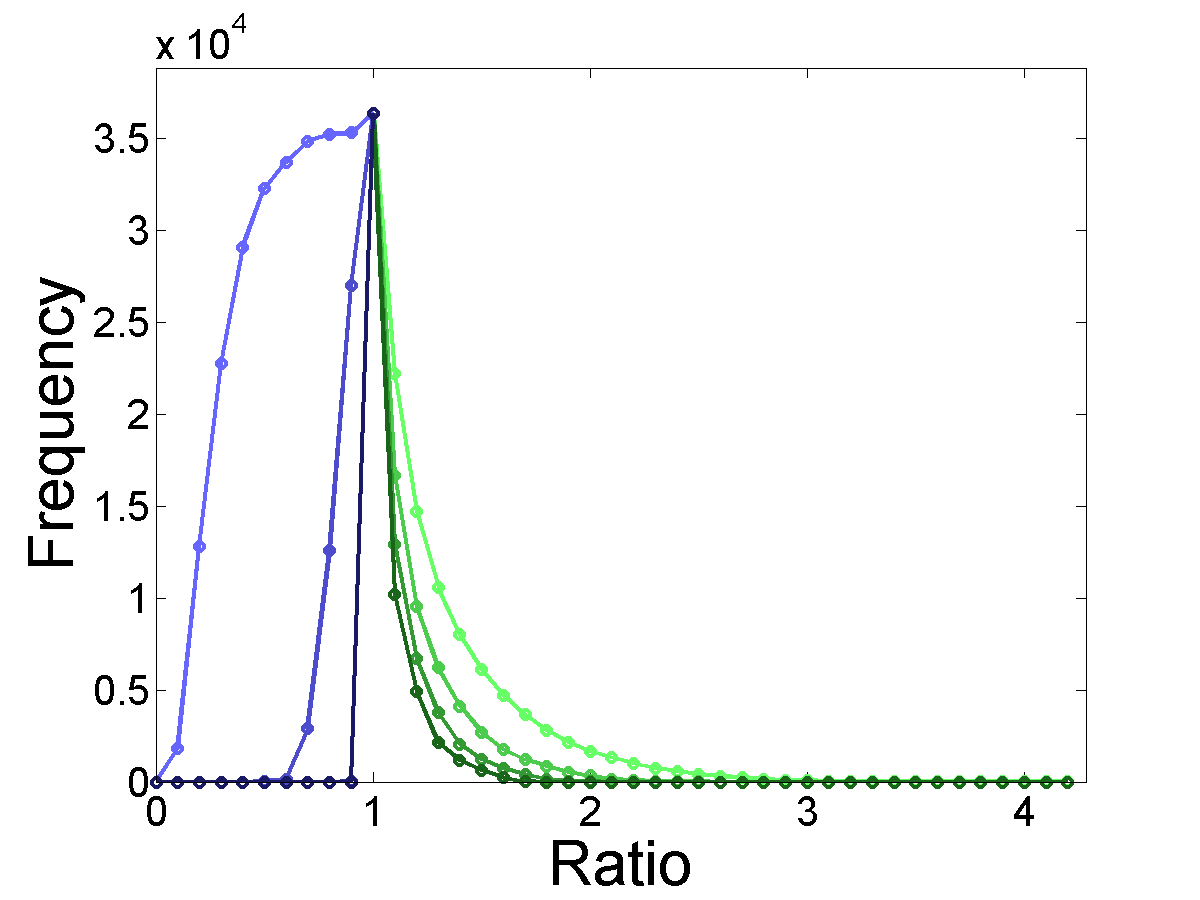}
\caption{Number of vertices in \textsc{Facebook} with core number estimate ratios less optimal that a given threshold (see Figure~\ref{fig:estimate_growth}).}\end{figure}
\begin{figure}[H]
	\centering
	\includegraphics[width=0.62\linewidth]{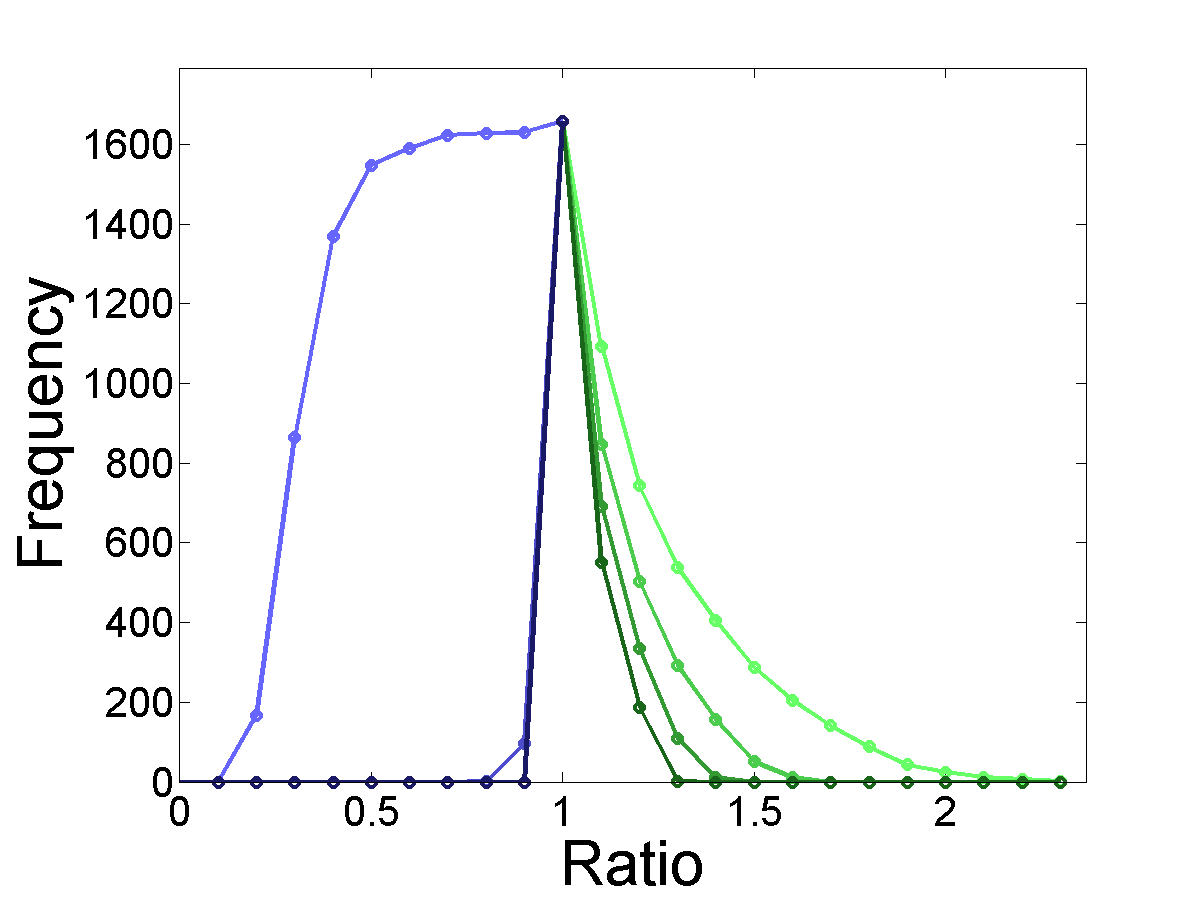}
\caption{Number of vertices in \textsc{Facebook2} with core number estimate ratios less optimal that a given threshold (see Figure~\ref{fig:estimate_growth}).}\end{figure}
\begin{figure}[H]
	\centering
	\includegraphics[width=0.62\linewidth]{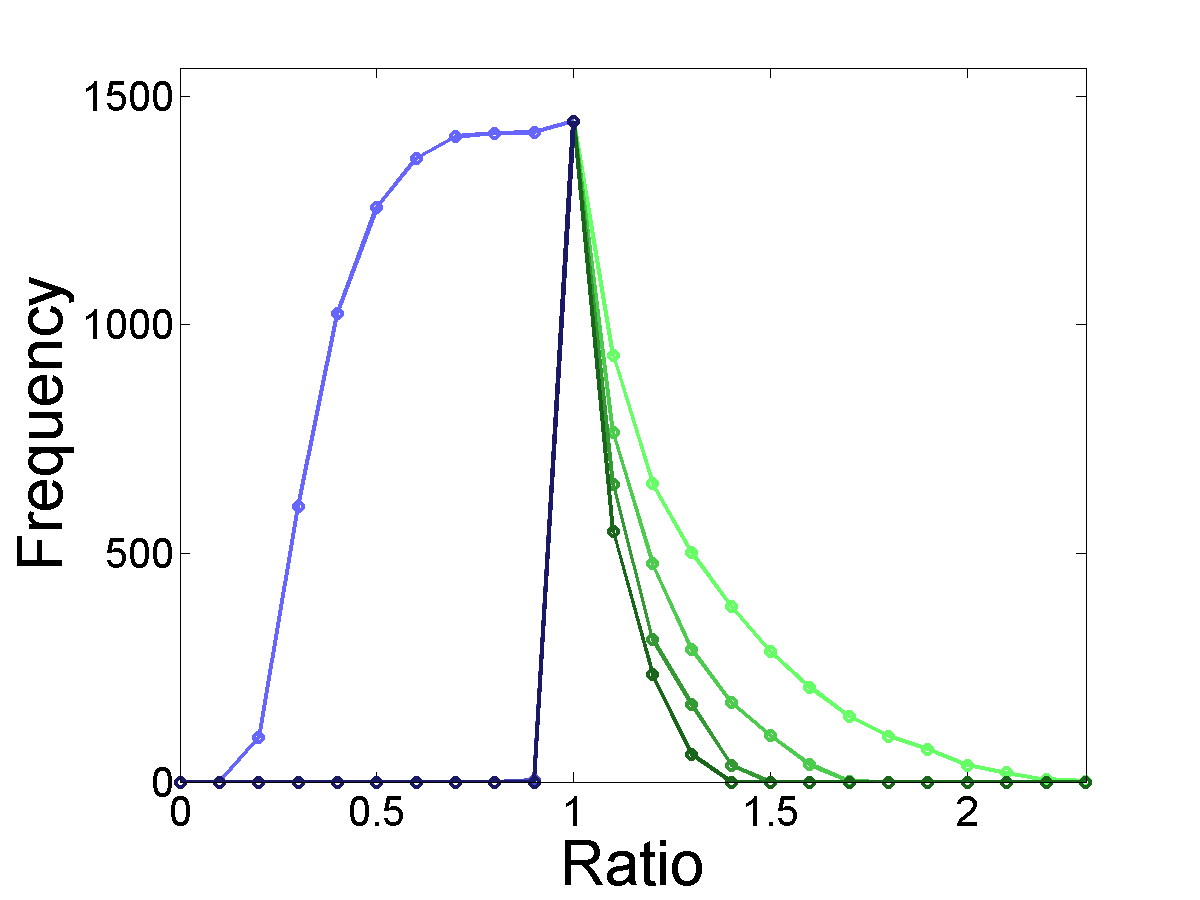}
\caption{Number of vertices in \textsc{Facebook3} with core number estimate ratios less optimal that a given threshold (see Figure~\ref{fig:estimate_growth}).}\end{figure}
\begin{figure}[H]
	\centering
	\includegraphics[width=0.62\linewidth]{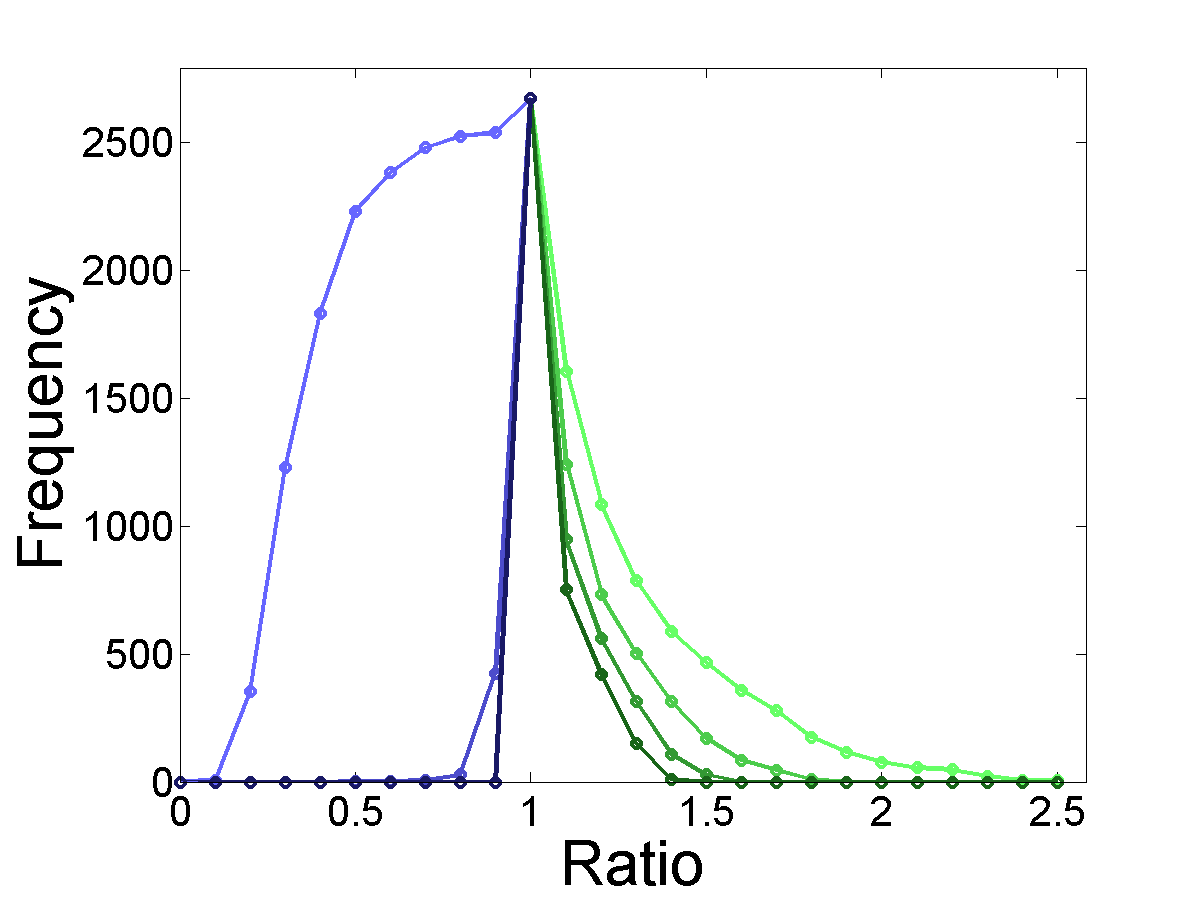}
\caption{Number of vertices in \textsc{Facebook4} with core number estimate ratios less optimal that a given threshold (see Figure~\ref{fig:estimate_growth}).}\end{figure}
\begin{figure}[H]
	\centering
	\includegraphics[width=0.62\linewidth]{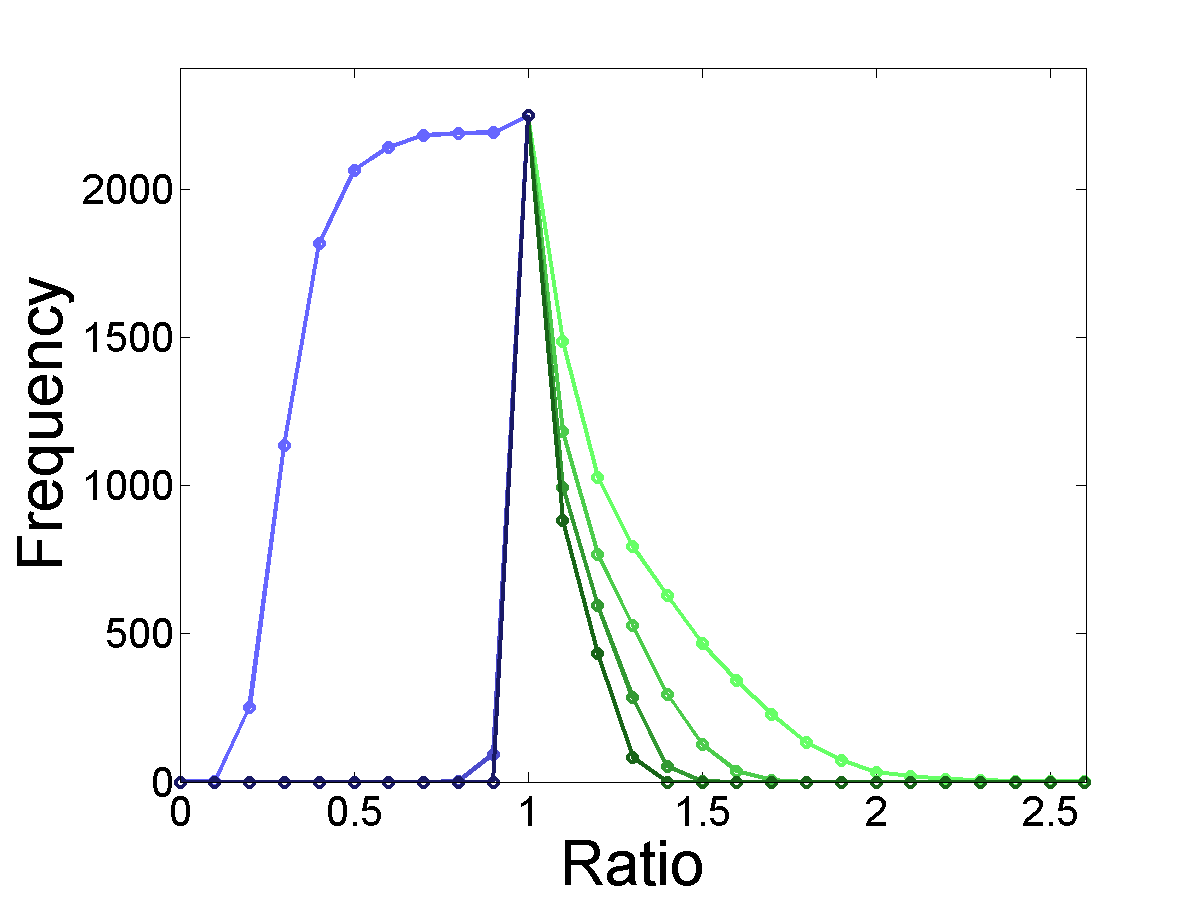}
\caption{Number of vertices in \textsc{Facebook5} with core number estimate ratios less optimal that a given threshold (see Figure~\ref{fig:estimate_growth}).}\end{figure}
\begin{figure}[H]
	\centering
	\includegraphics[width=0.62\linewidth]{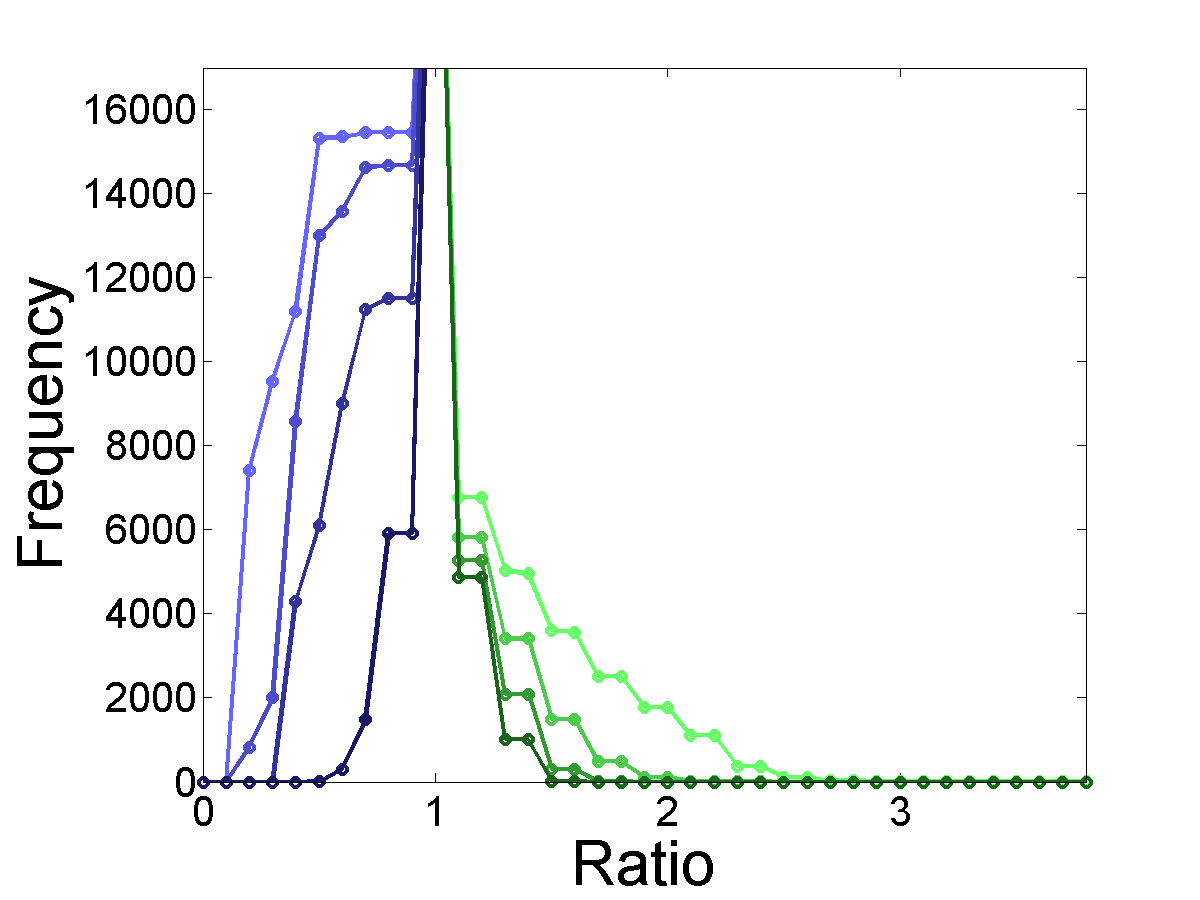}
\caption{Number of vertices in \textsc{Gnutella} with core number estimate ratios less optimal that a given threshold (see Figure~\ref{fig:estimate_growth}).}\end{figure}
\begin{figure}[H]
	\centering
	\includegraphics[width=0.62\linewidth]{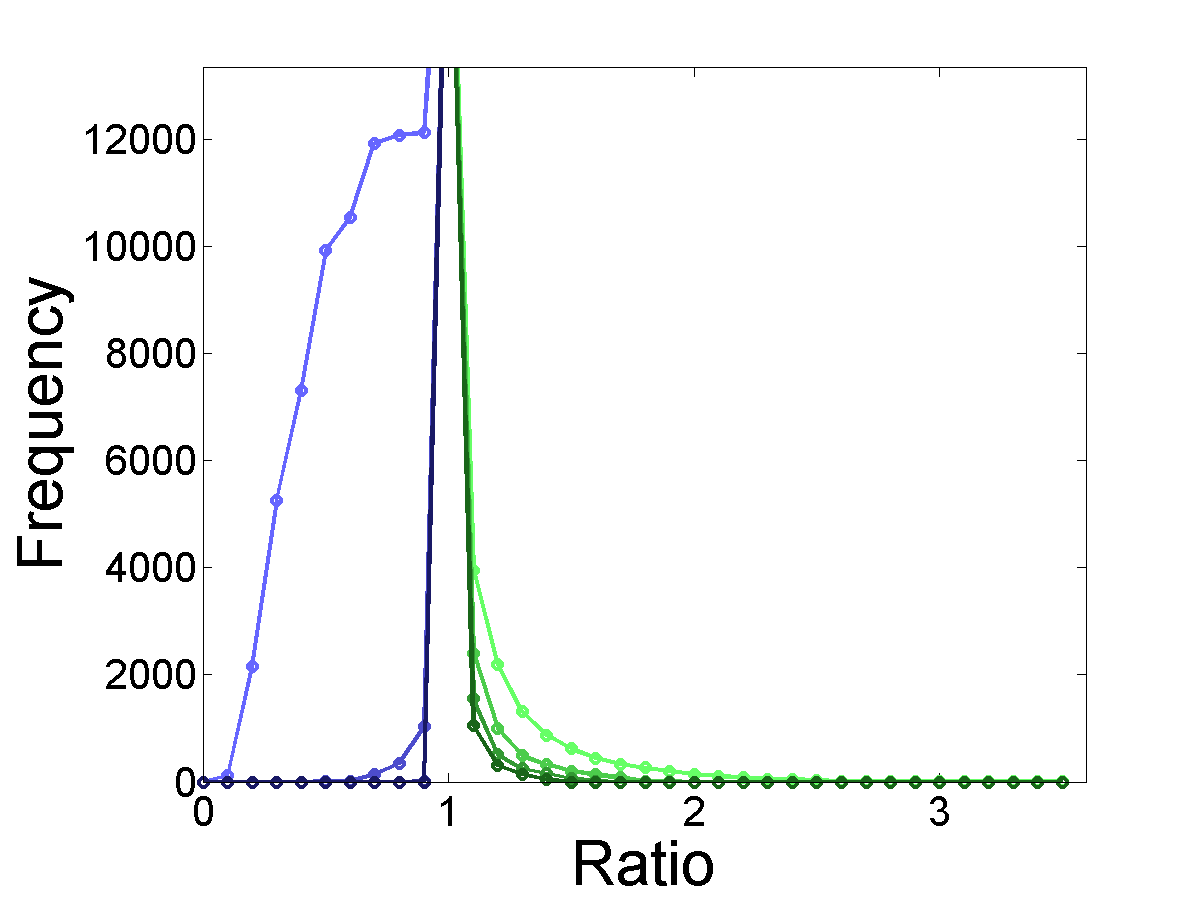}
\caption{Number of vertices in \textsc{H.~sapiens} with core number estimate ratios less optimal that a given threshold (see Figure~\ref{fig:estimate_growth}).}\end{figure}
\begin{figure}[H]
	\centering
	\includegraphics[width=0.62\linewidth]{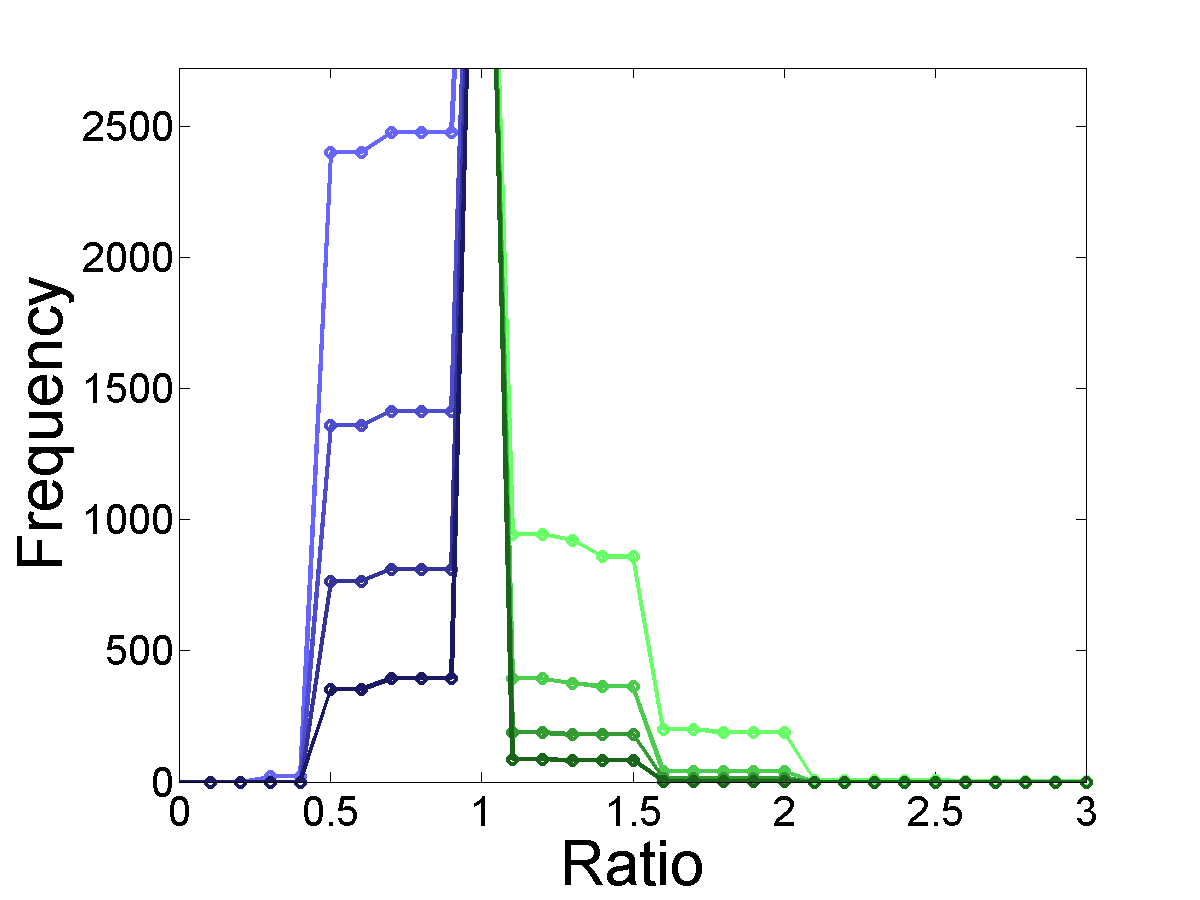}
\caption{Number of vertices in \textsc{WPG} with core number estimate ratios less optimal that a given threshold (see Figure~\ref{fig:estimate_growth}).}\end{figure}

\begin{figure}[H]
	\centering
	\includegraphics[width=0.62\linewidth]{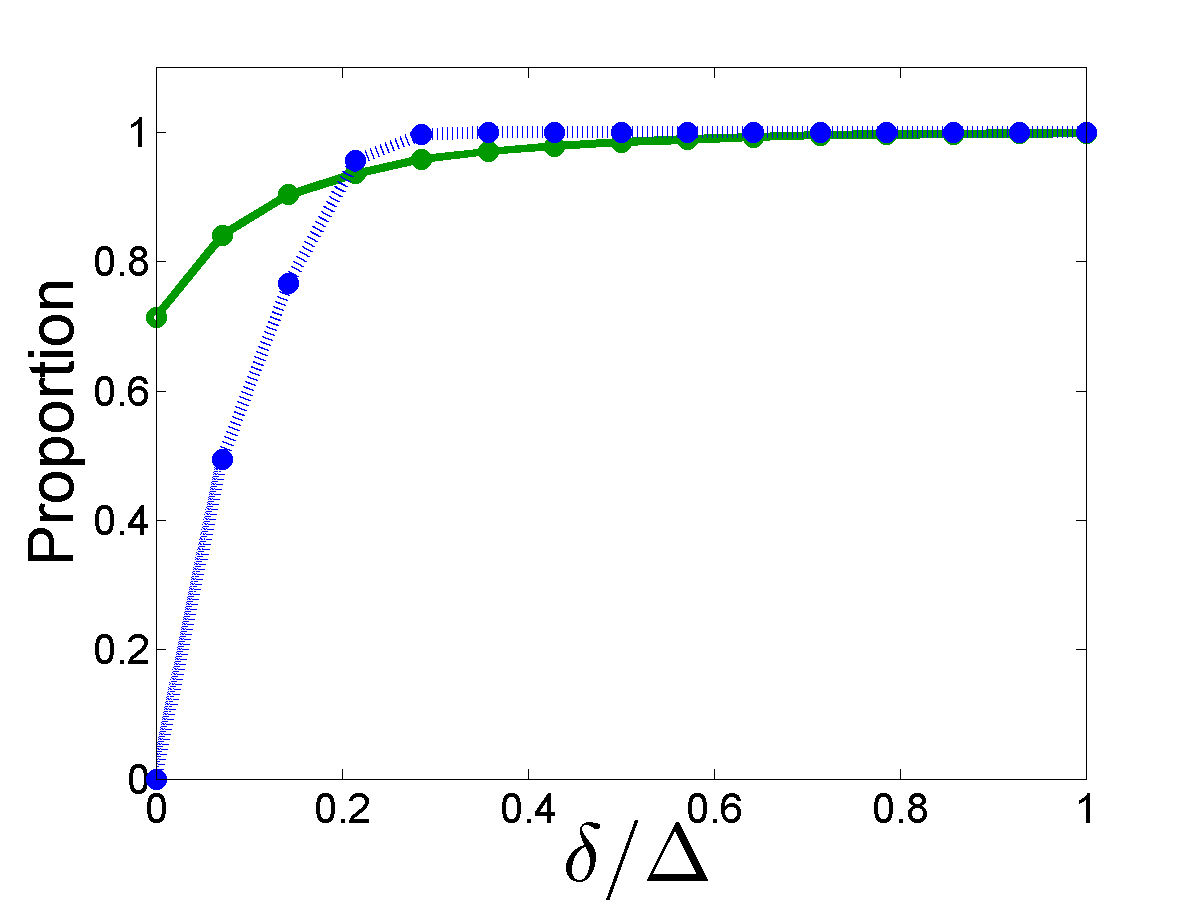}
\caption{Proportion of vertices in \textsc{A.~thaliana} with optimal core number estimate ratios (see Figure~\ref{fig:normalized_correct}).}\end{figure}
\begin{figure}[H]
	\centering
	\includegraphics[width=0.62\linewidth]{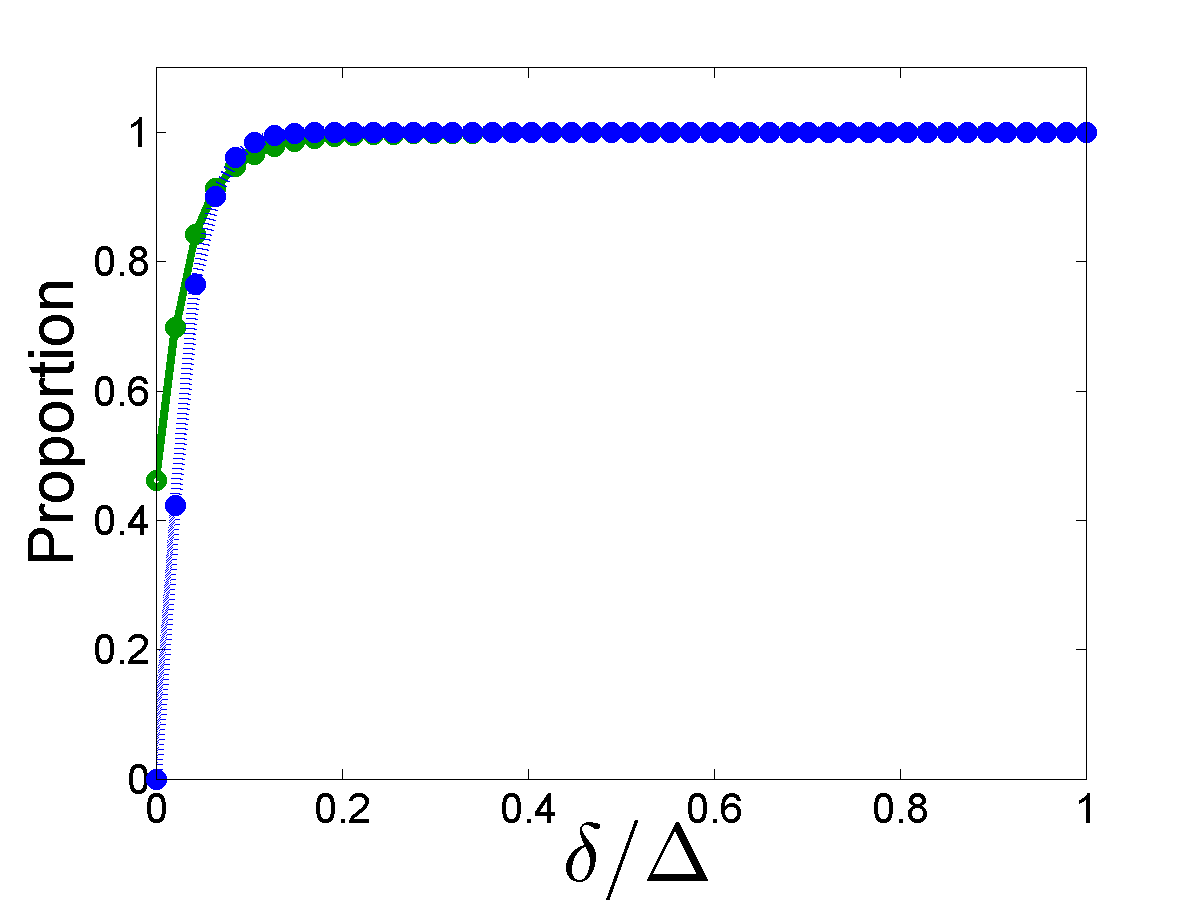}
\caption{Proportion of vertices in \textsc{Amazon} with optimal core number estimate ratios (see Figure~\ref{fig:normalized_correct}).}\end{figure}
\begin{figure}[H]
	\centering
	\includegraphics[width=0.62\linewidth]{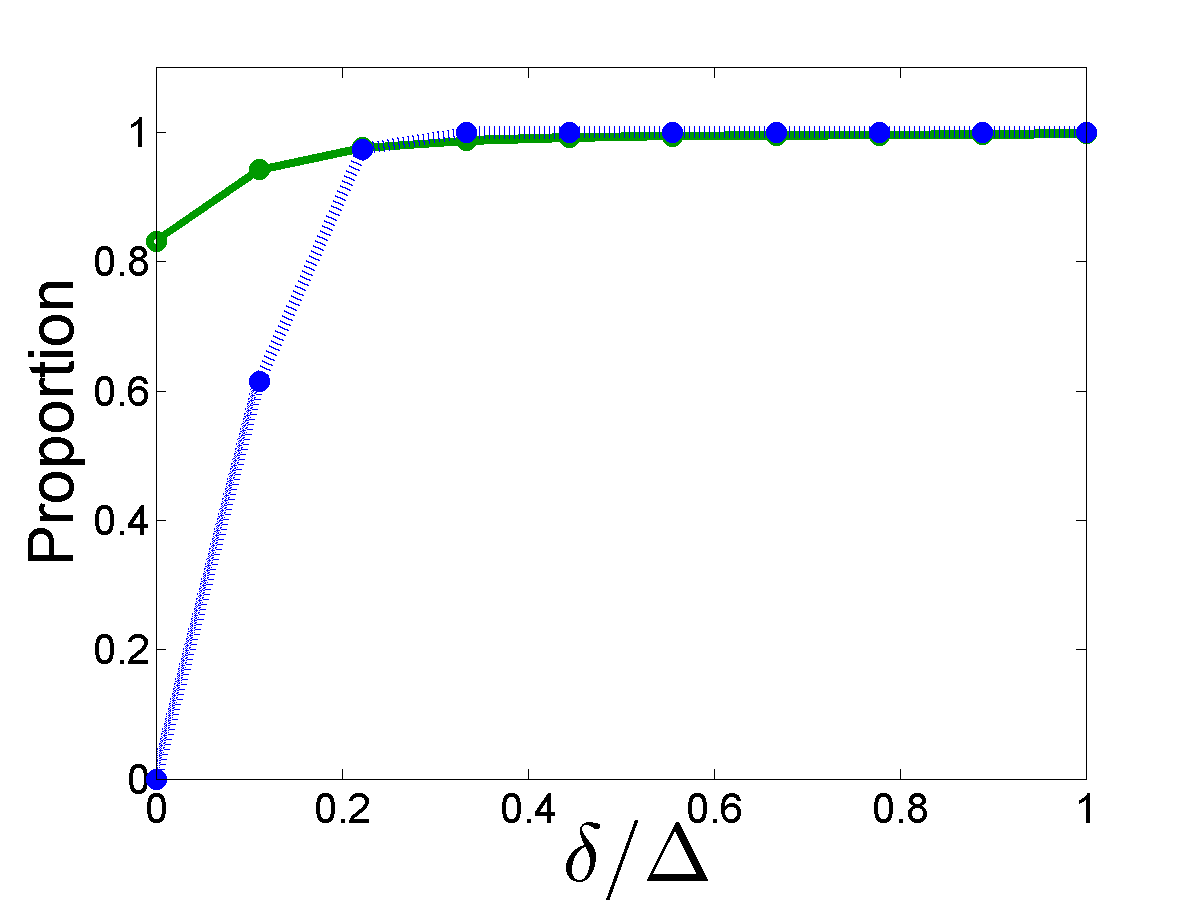}
\caption{Proportion of vertices in \textsc{AS} with optimal core number estimate ratios (see Figure~\ref{fig:normalized_correct}).}\end{figure}
\begin{figure}[H]
	\centering
	\includegraphics[width=0.62\linewidth]{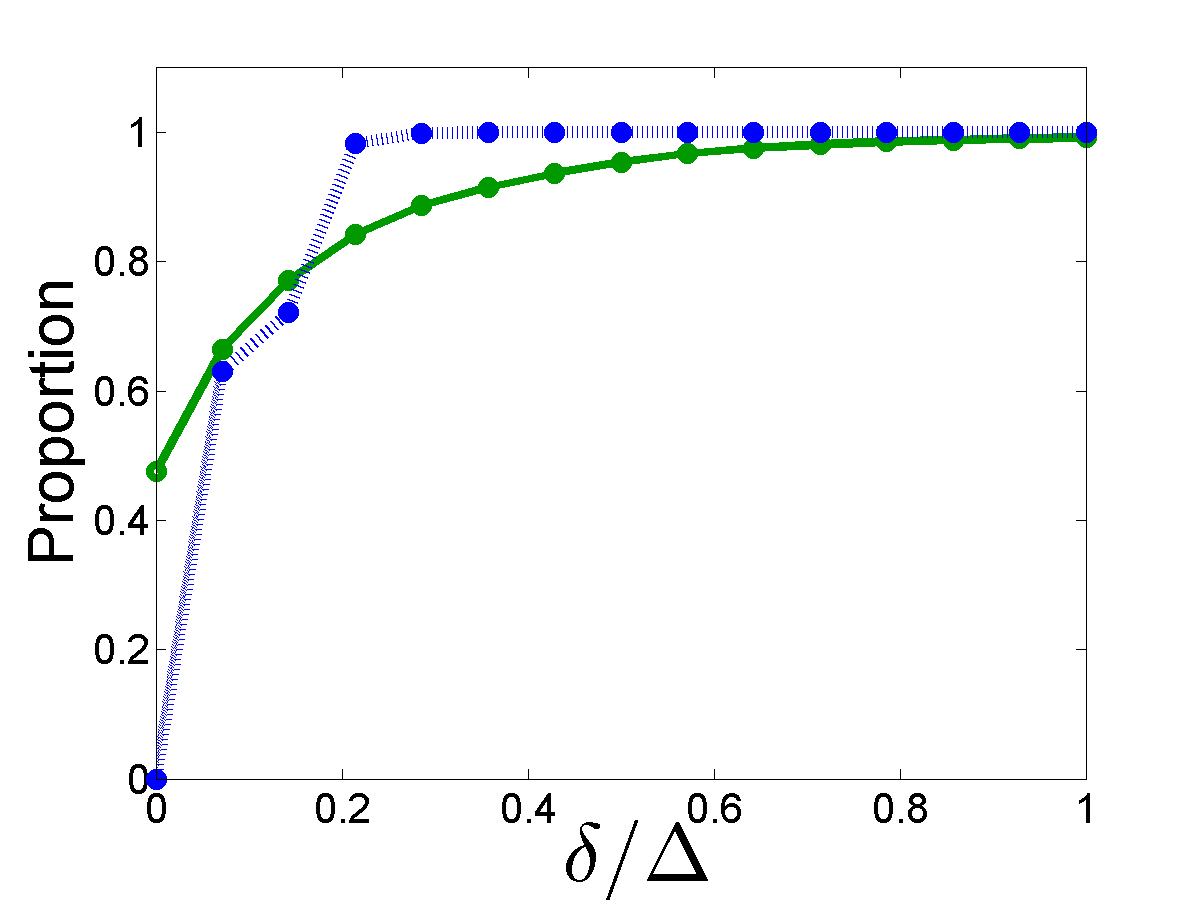}
\caption{Proportion of vertices in \textsc{ca-AstroPh} with optimal core number estimate ratios (see Figure~\ref{fig:normalized_correct}).}\end{figure}
\begin{figure}[H]
	\centering
	\includegraphics[width=0.62\linewidth]{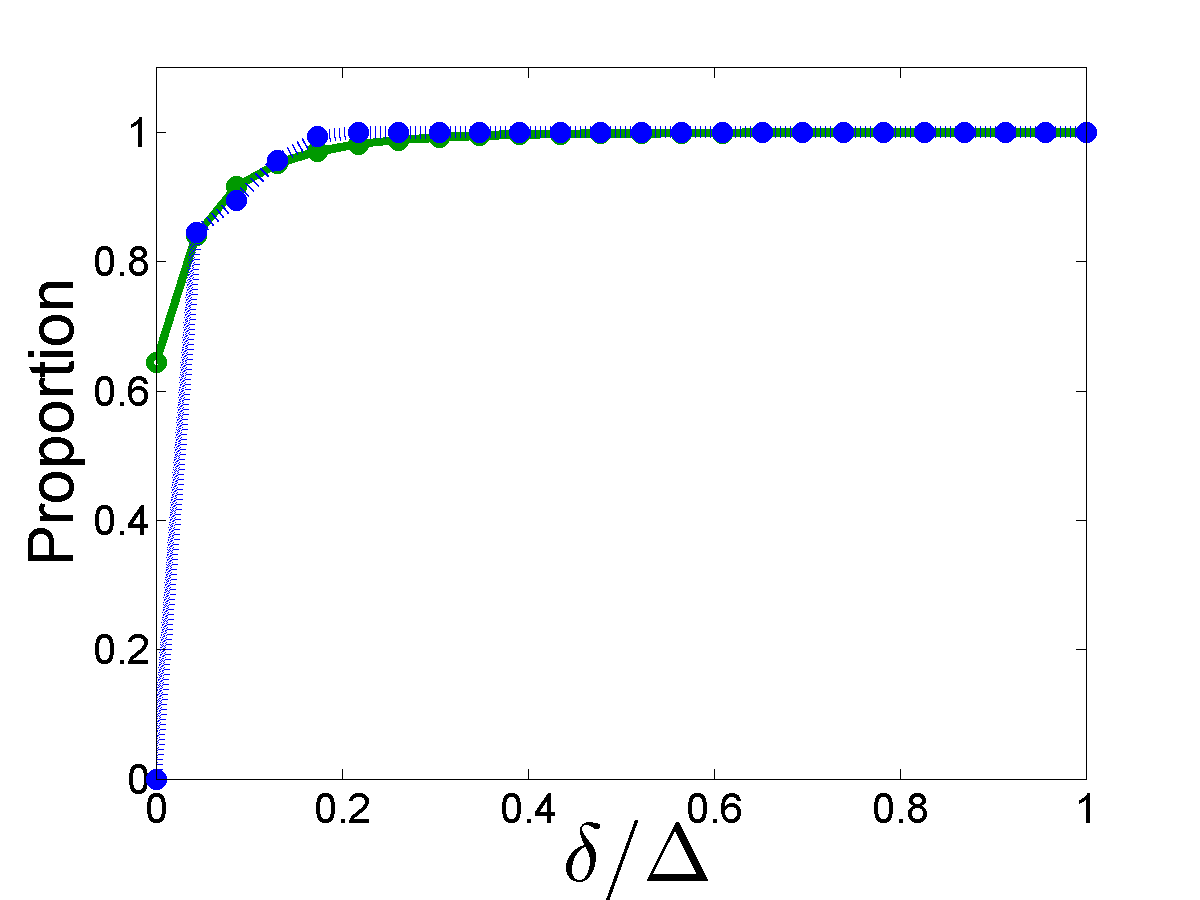}
\caption{Proportion of vertices in \textsc{DBLP} with optimal core number estimate ratios (see Figure~\ref{fig:normalized_correct}).}\end{figure}
\begin{figure}[H]
	\centering
	\includegraphics[width=0.62\linewidth]{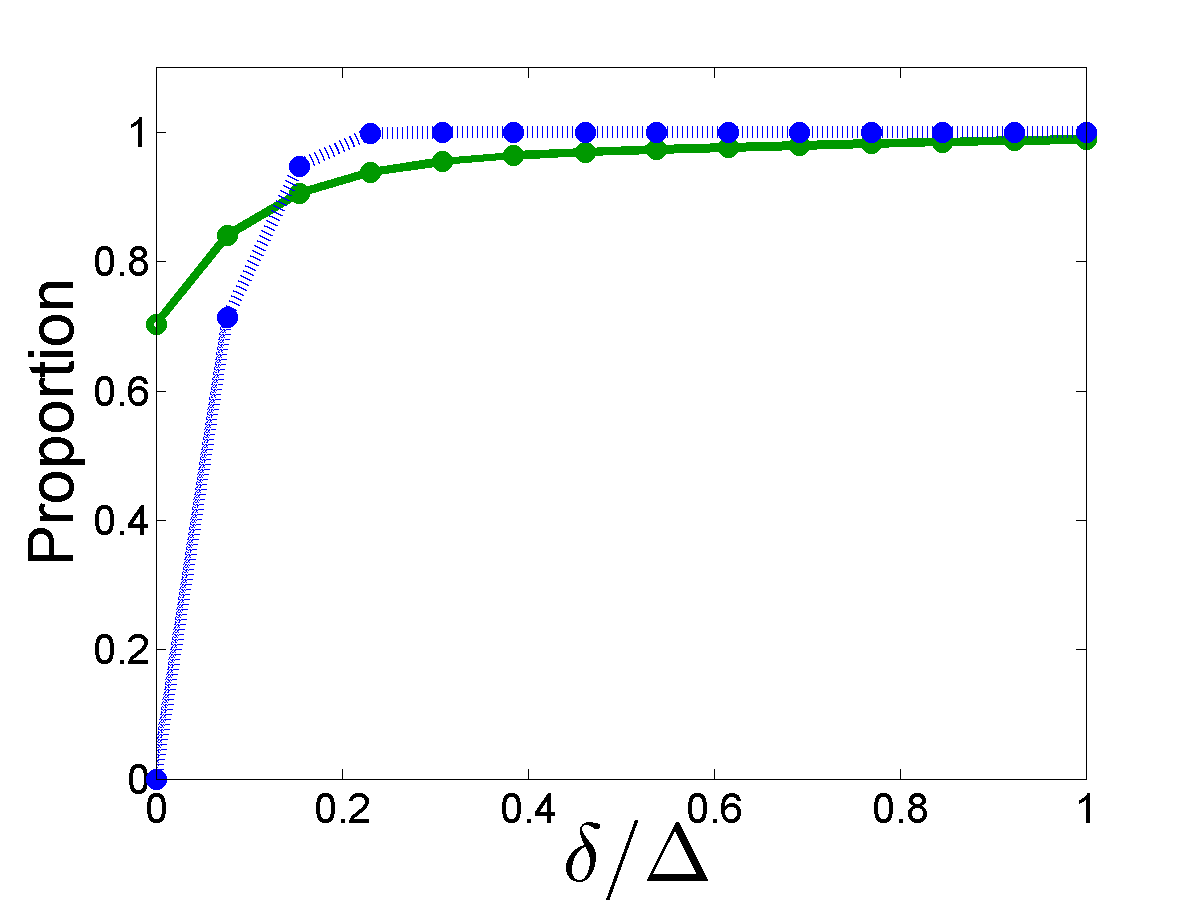}
\caption{Proportion of vertices in \textsc{Enron} with optimal core number estimate ratios (see Figure~\ref{fig:normalized_correct}).}\end{figure}
\begin{figure}[H]
	\centering
	\includegraphics[width=0.62\linewidth]{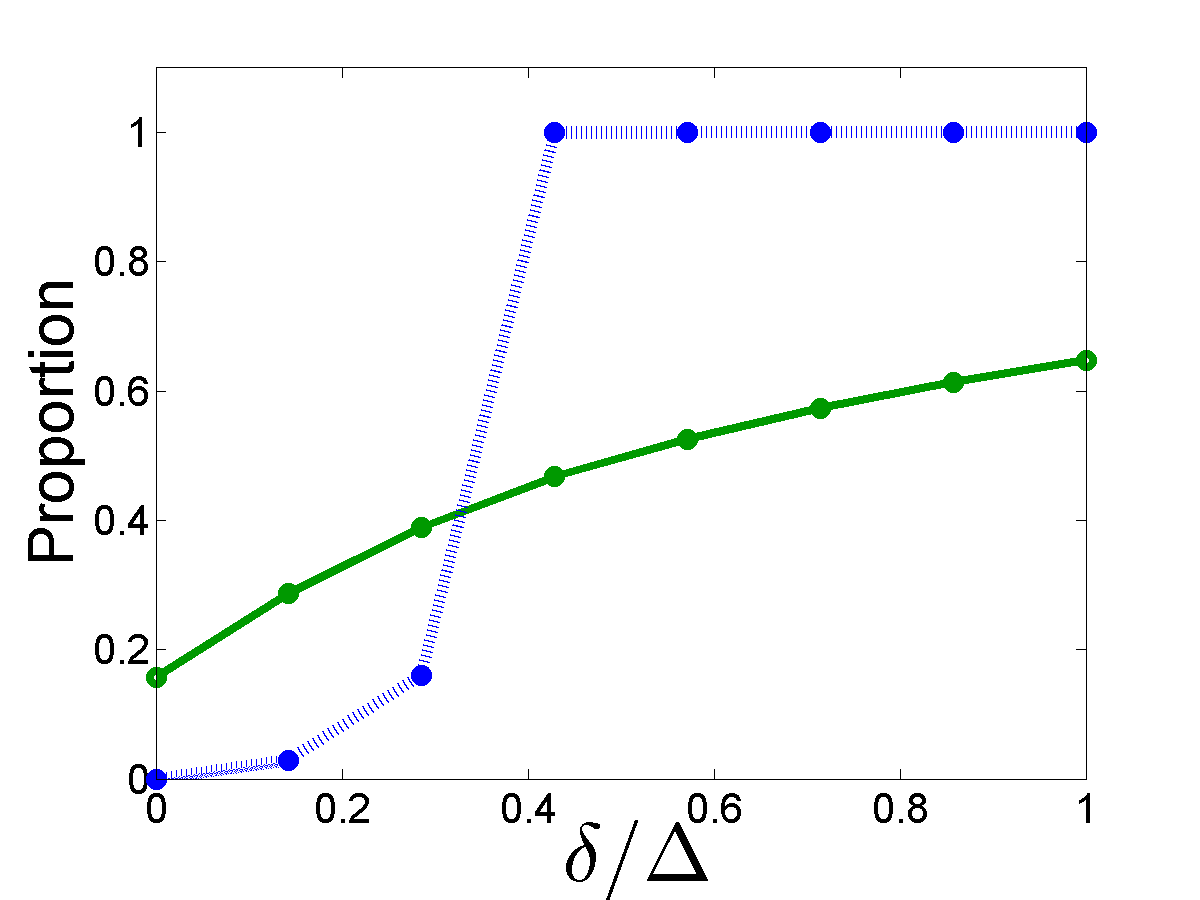}
\caption{Proportion of vertices in \textsc{Facebook} with optimal core number estimate ratios (see Figure~\ref{fig:normalized_correct}).}\end{figure}
\begin{figure}[H]
	\centering
	\includegraphics[width=0.62\linewidth]{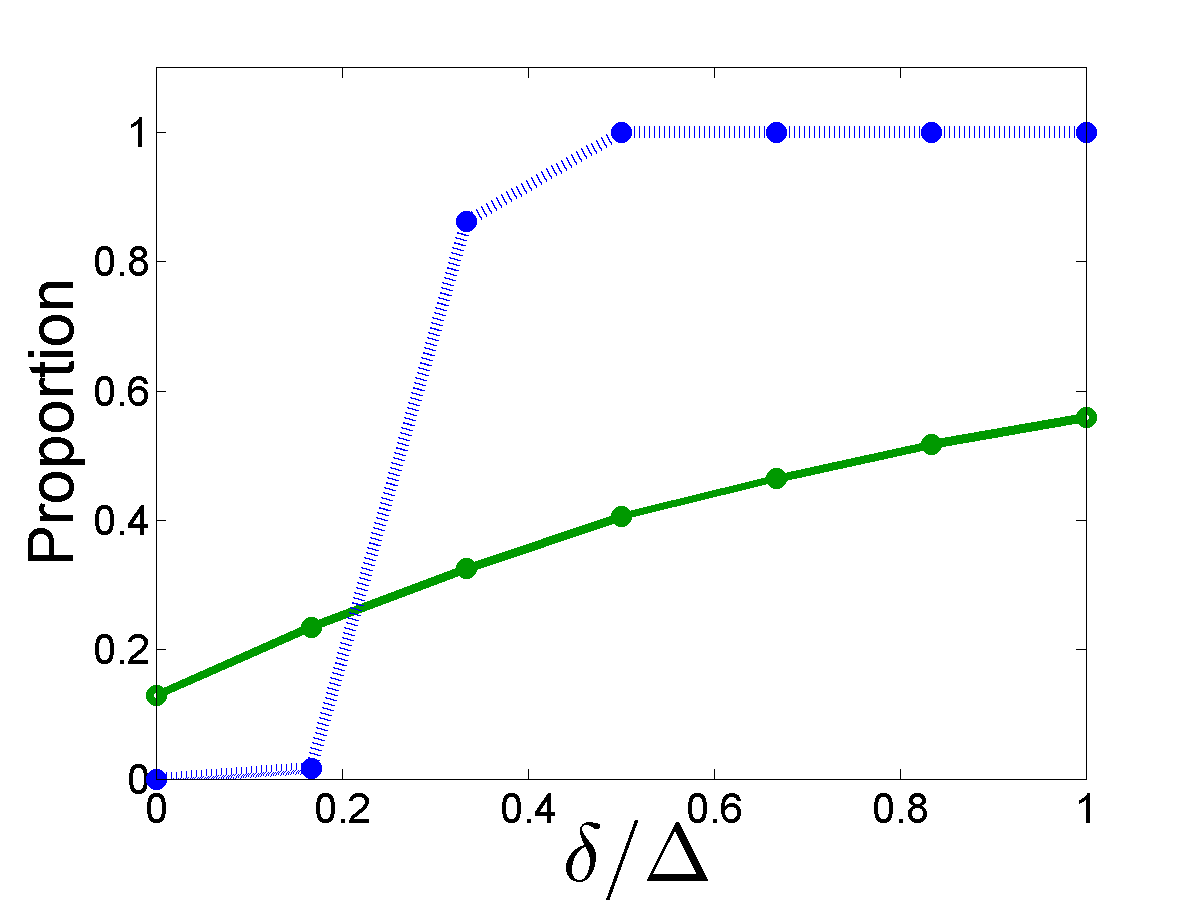}
\caption{Proportion of vertices in \textsc{Facebook2} with optimal core number estimate ratios (see Figure~\ref{fig:normalized_correct}).}\end{figure}
\begin{figure}[H]
	\centering
	\includegraphics[width=0.62\linewidth]{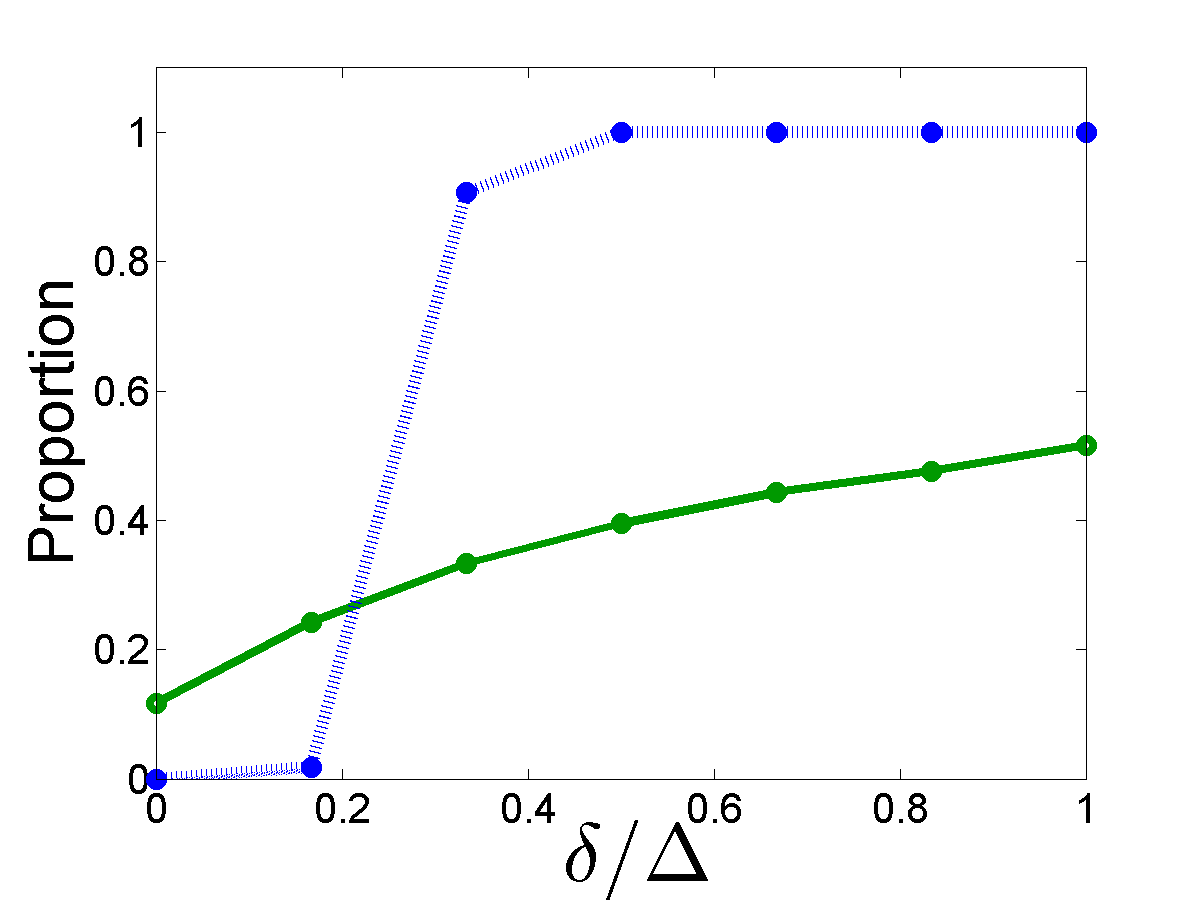}
\caption{Proportion of vertices in \textsc{Facebook3} with optimal core number estimate ratios (see Figure~\ref{fig:normalized_correct}).}\end{figure}
\begin{figure}[H]
	\centering
	\includegraphics[width=0.62\linewidth]{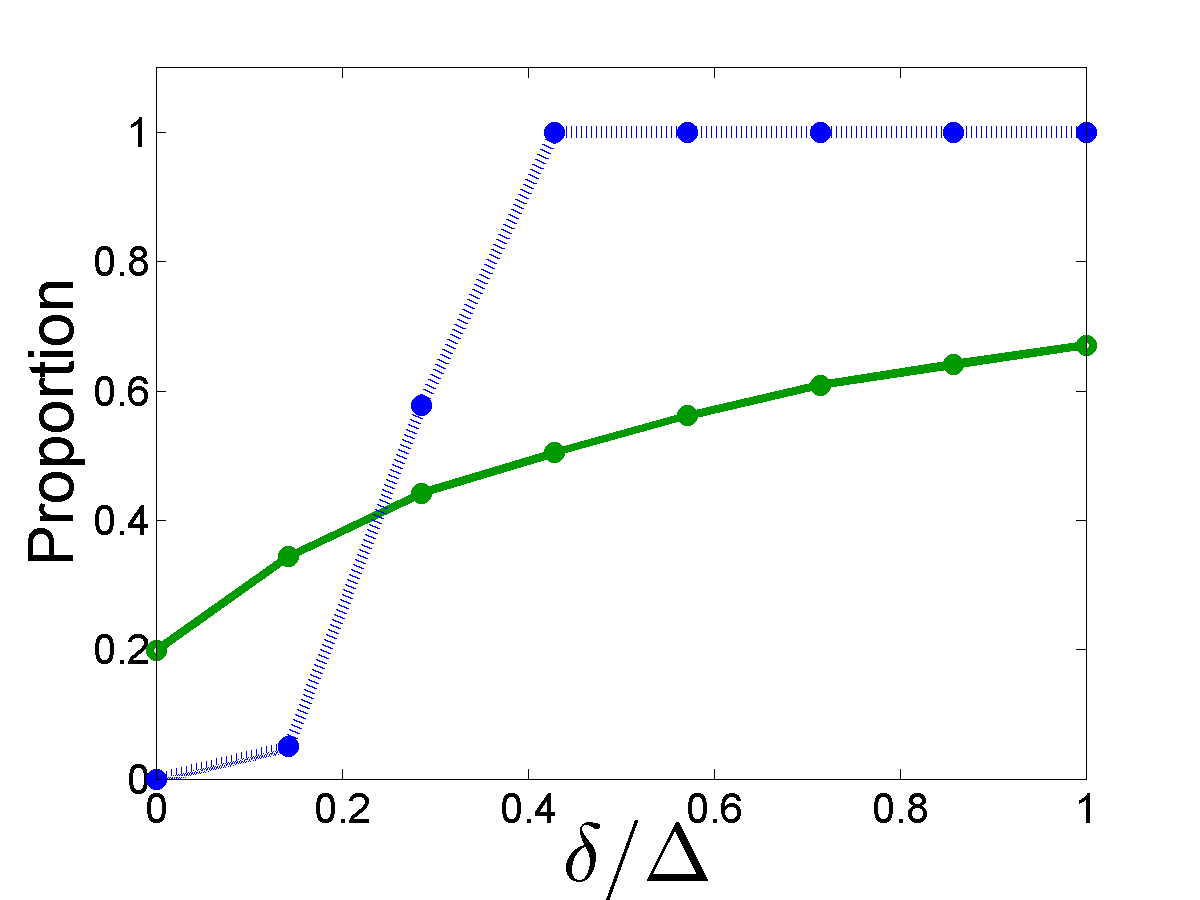}
\caption{Proportion of vertices in \textsc{Facebook4} with optimal core number estimate ratios (see Figure~\ref{fig:normalized_correct}).}\end{figure}
\begin{figure}[H]
	\centering
	\includegraphics[width=0.62\linewidth]{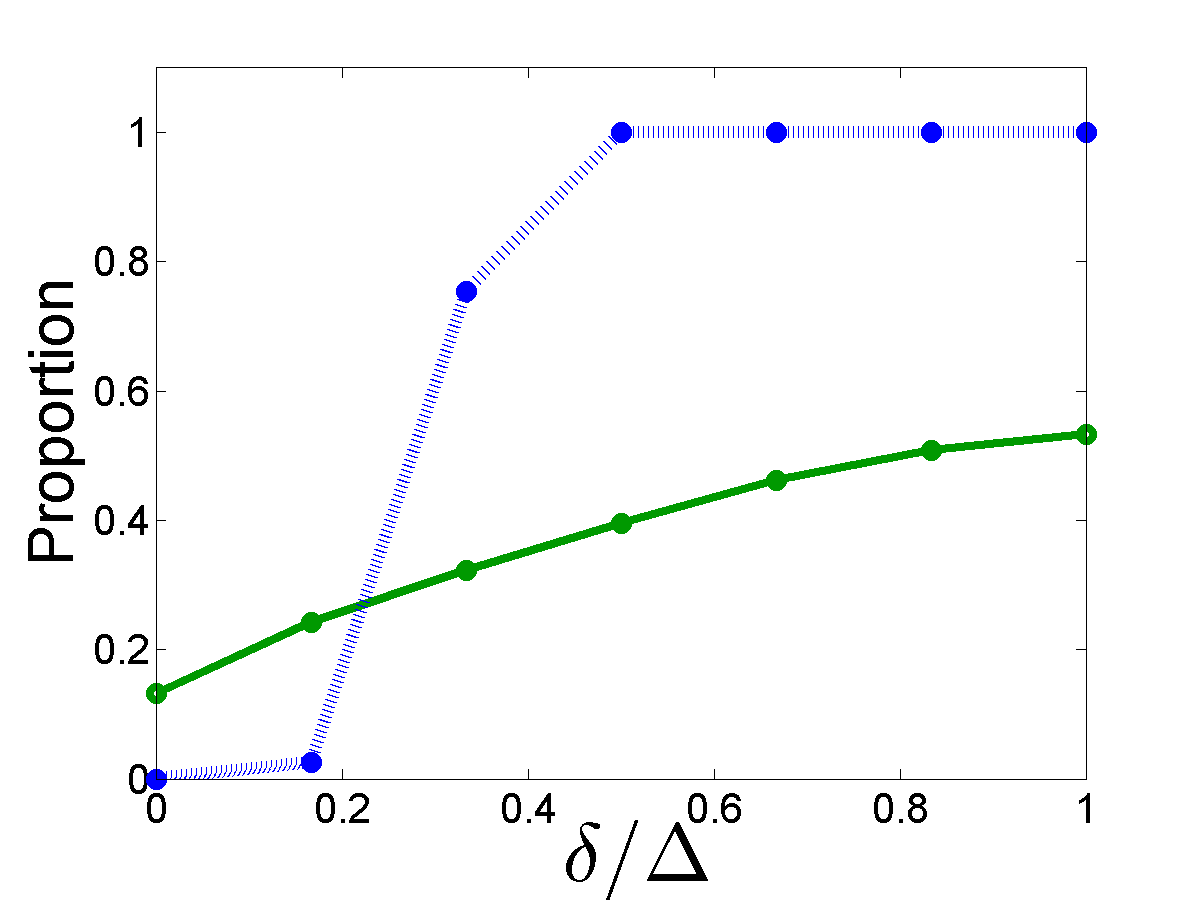}
\caption{Proportion of vertices in \textsc{Facebook5} with optimal core number estimate ratios (see Figure~\ref{fig:normalized_correct}).}\end{figure}
\begin{figure}[H]
	\centering
	\includegraphics[width=0.62\linewidth]{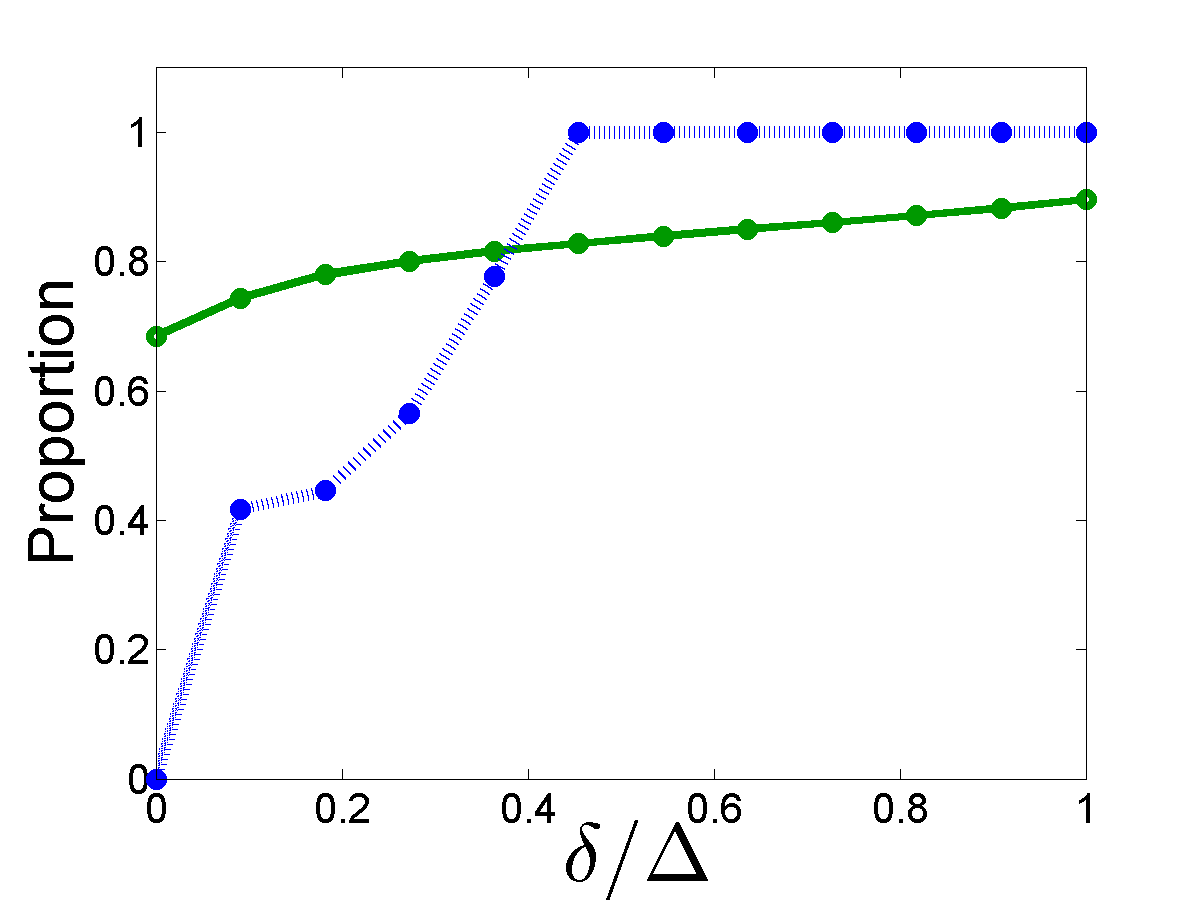}
\caption{Proportion of vertices in \textsc{Gnutella} with optimal core number estimate ratios (see Figure~\ref{fig:normalized_correct}).}\end{figure}
\begin{figure}[H]
	\centering
	\includegraphics[width=0.62\linewidth]{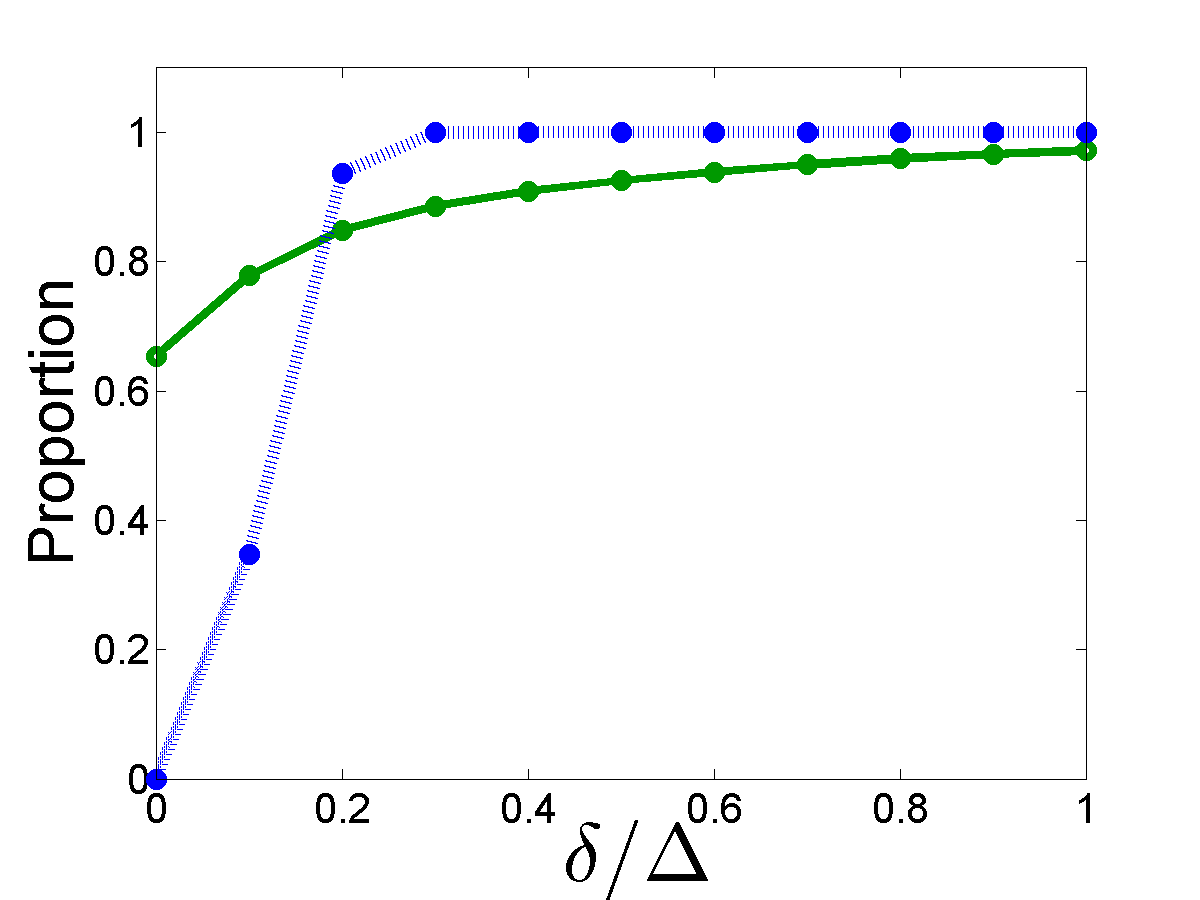}
\caption{Proportion of vertices in \textsc{H.~sapiens} with optimal core number estimate ratios (see Figure~\ref{fig:normalized_correct}).}\end{figure}
\begin{figure}[H]
	\centering
	\includegraphics[width=0.62\linewidth]{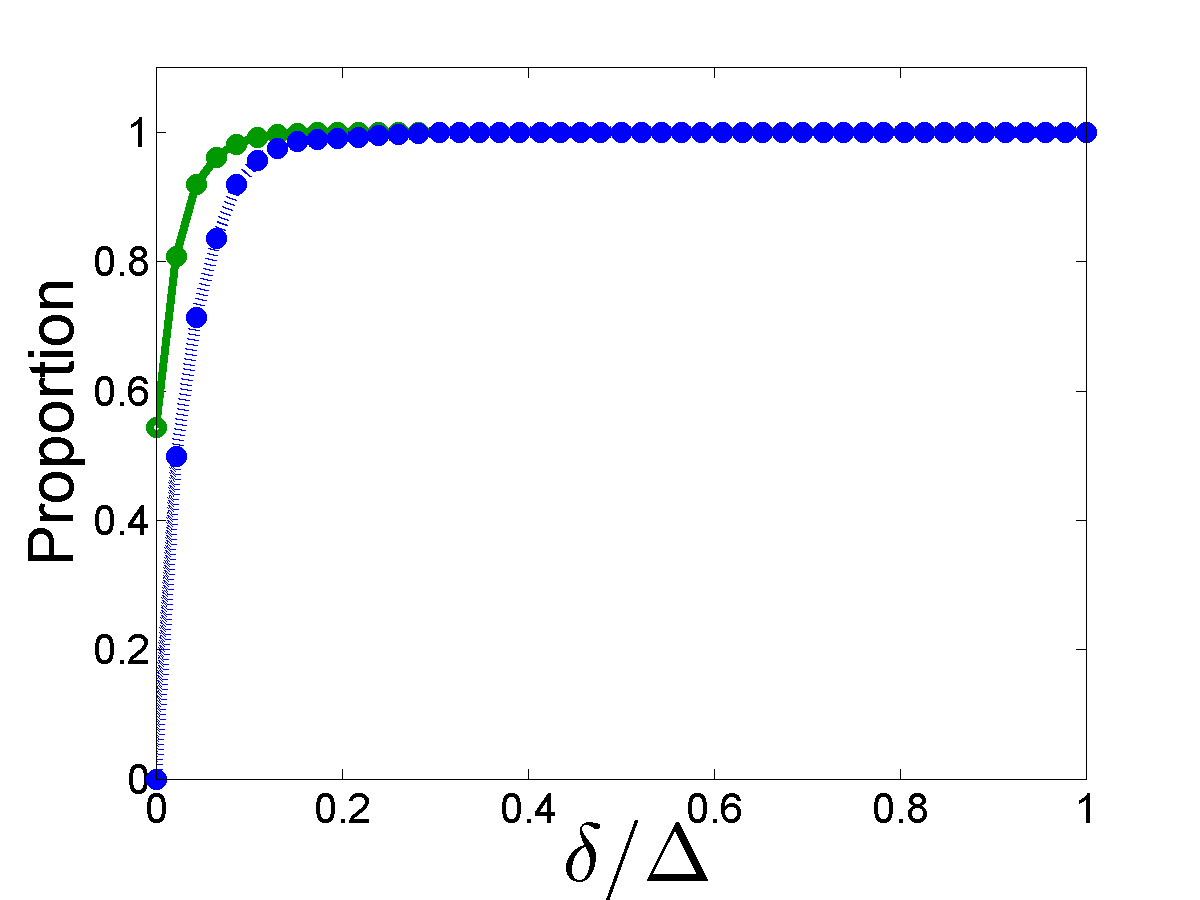}
\caption{Proportion of vertices in \textsc{WPG} with optimal core number estimate ratios (see Figure~\ref{fig:normalized_correct}).}\end{figure}
\begin{figure}[H]
	\centering
	\includegraphics[width=0.62\linewidth]{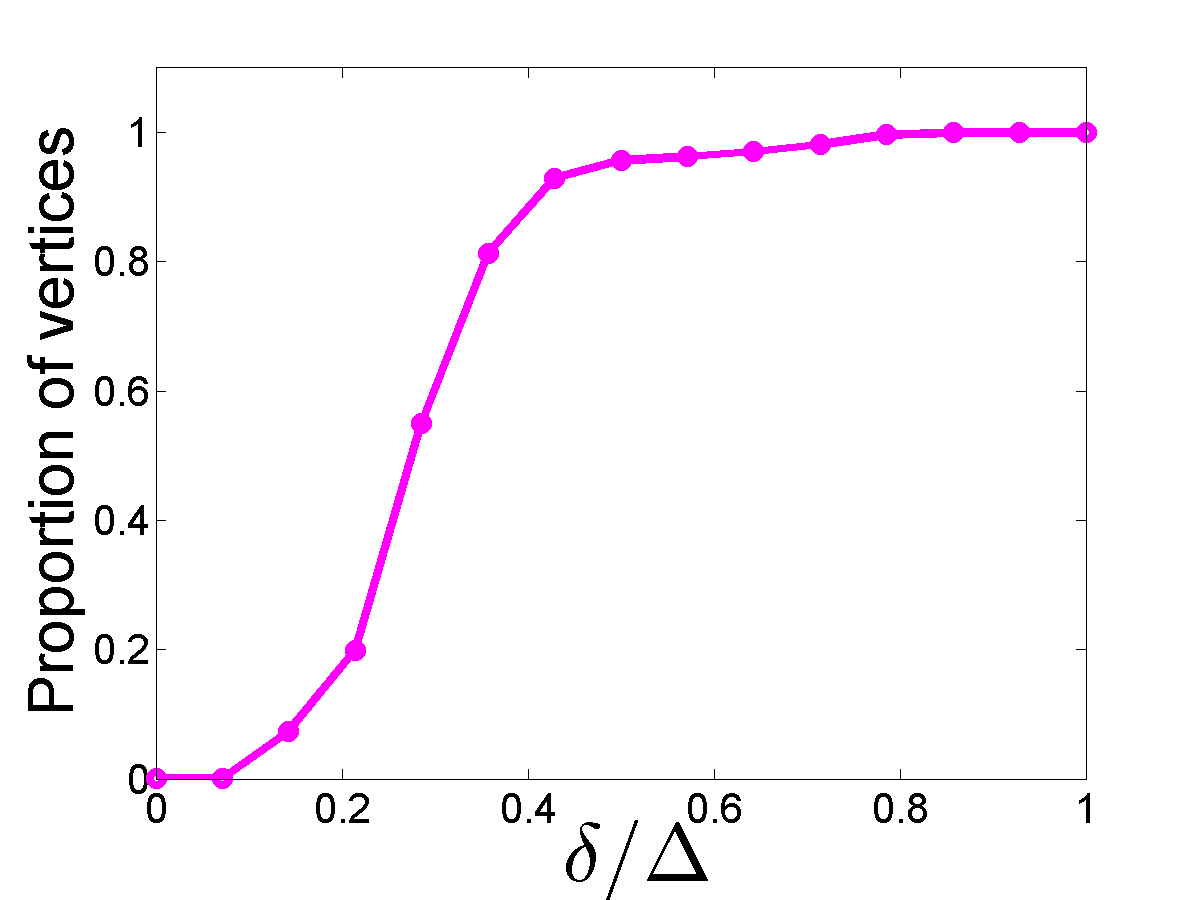}
\caption{Average proportion of vertices in $N_\delta$ of \textsc{A.~thaliana} as a function of $\delta/\Delta$ (see Figure~\ref{fig:normalized_neighborhood_full}).}\end{figure}
\begin{figure}[H]
	\centering
	\includegraphics[width=0.62\linewidth]{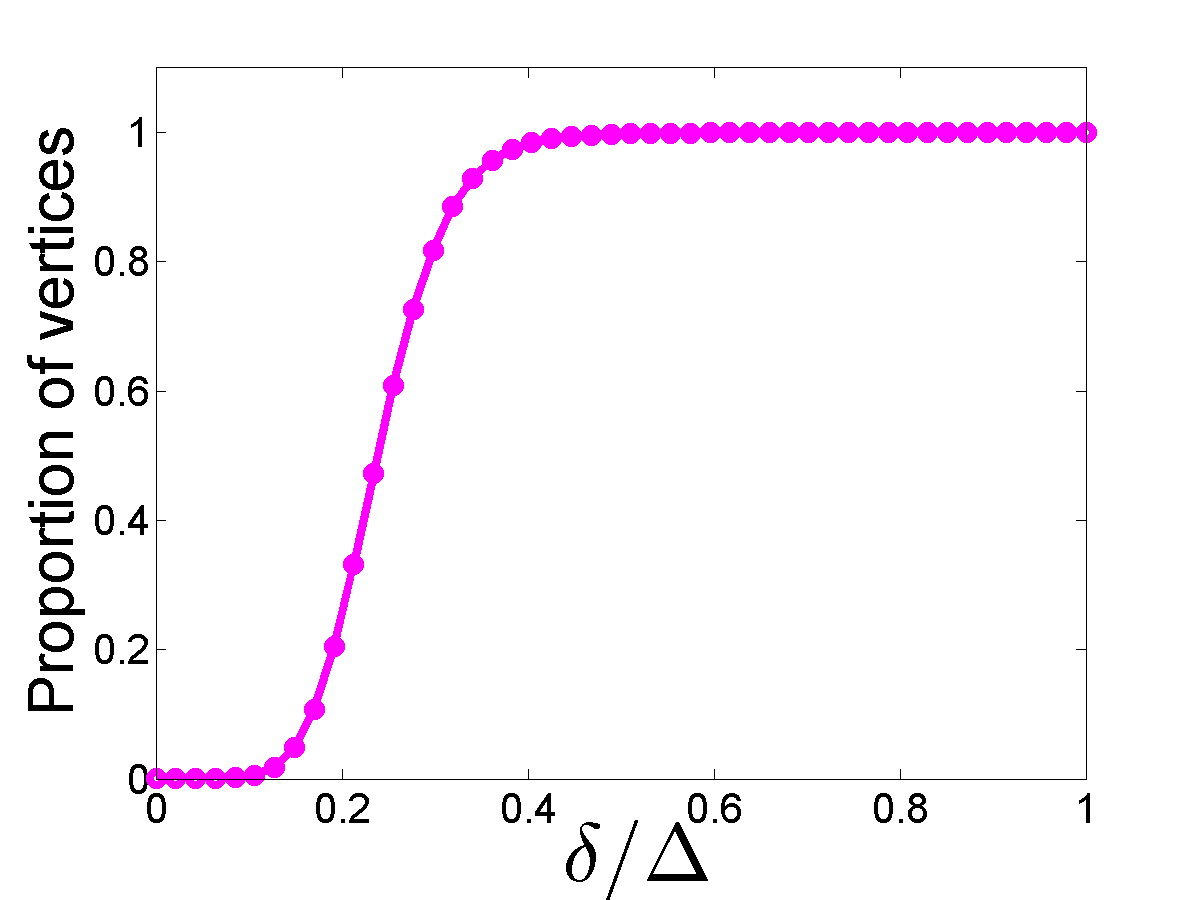}
\caption{Average proportion of vertices in $N_\delta$ of \textsc{Amazon} as a function of $\delta/\Delta$ (see Figure~\ref{fig:normalized_neighborhood_full}).}\end{figure}
\begin{figure}[H]
	\centering
	\includegraphics[width=0.62\linewidth]{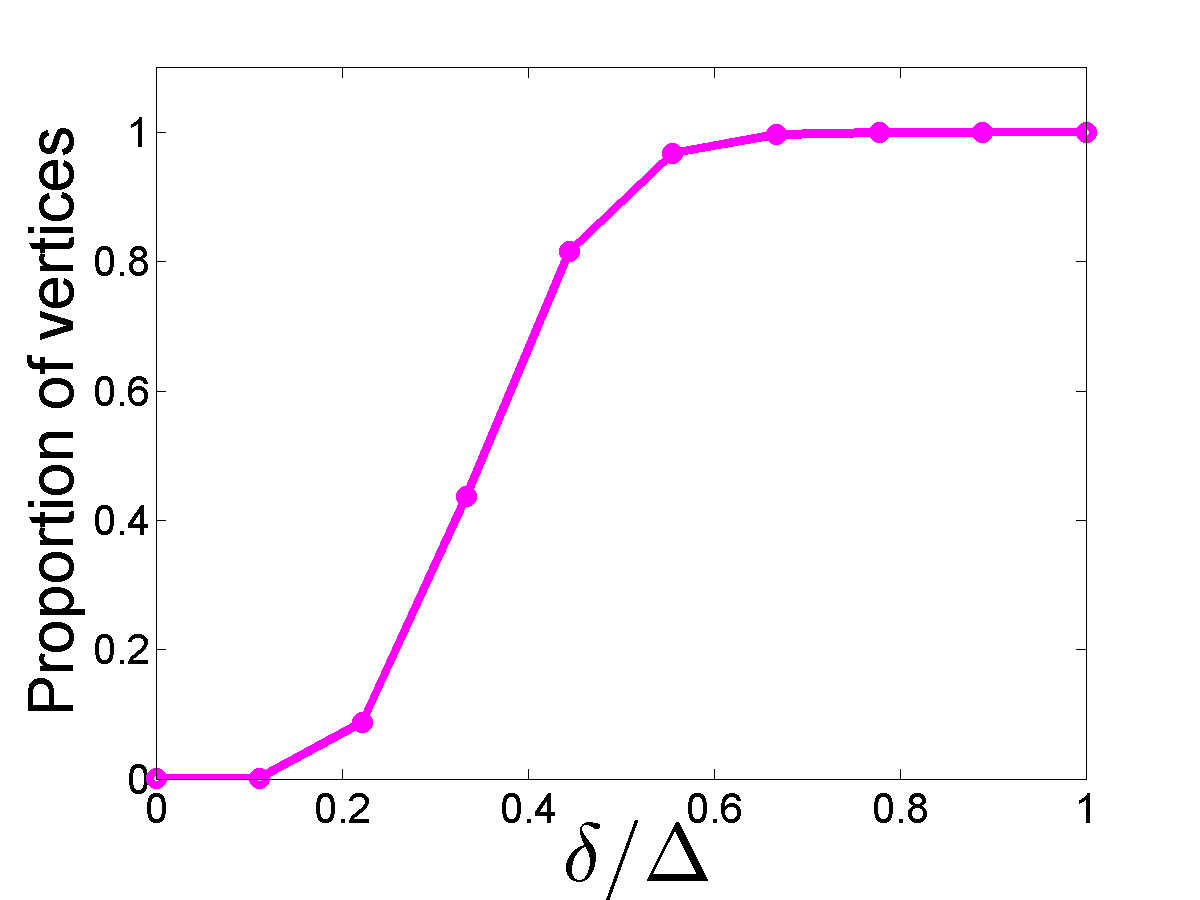}
\caption{Average proportion of vertices in $N_\delta$ of \textsc{AS} as a function of $\delta/\Delta$ (see Figure~\ref{fig:normalized_neighborhood_full}).}\end{figure}
\begin{figure}[H]
	\centering
	\includegraphics[width=0.62\linewidth]{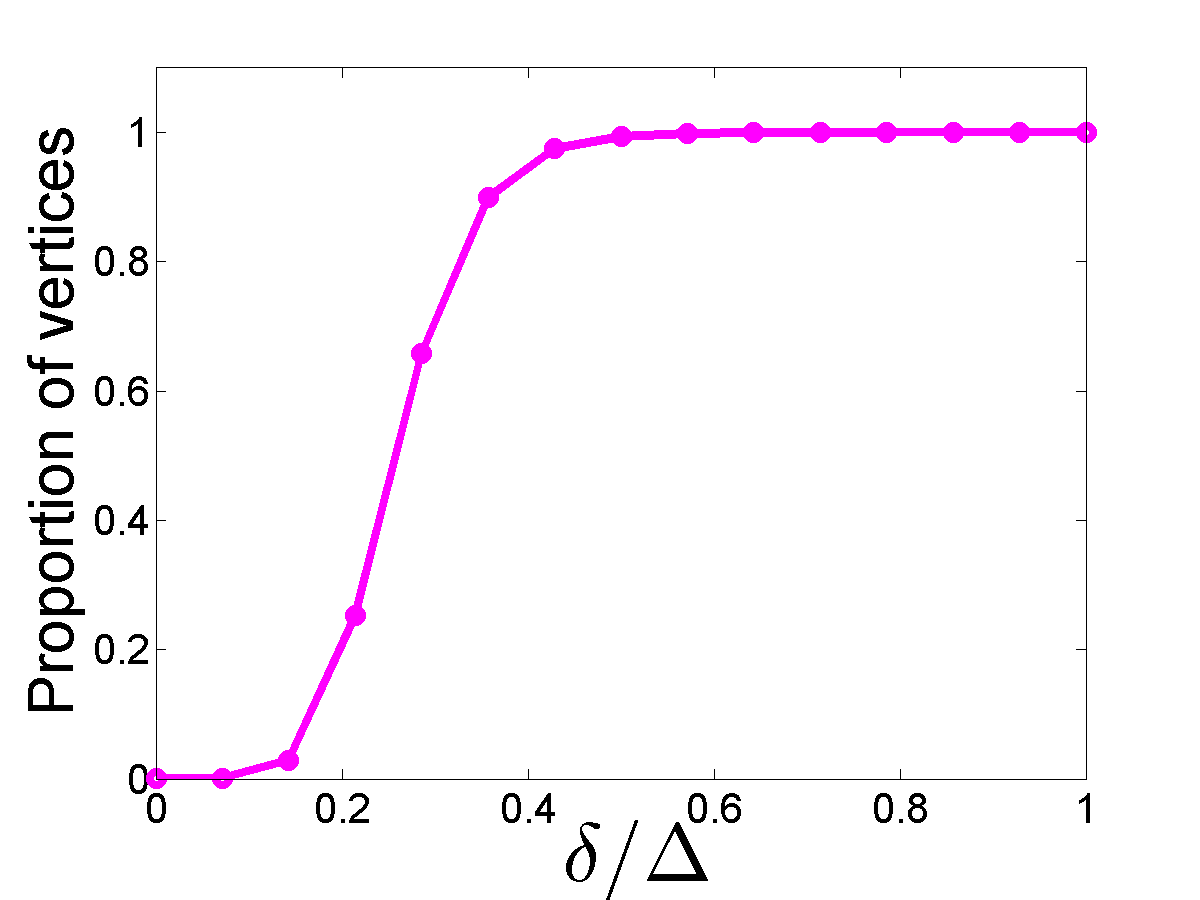}
\caption{Average proportion of vertices in $N_\delta$ of \textsc{ca-AstroPh} as a function of $\delta/\Delta$ (see Figure~\ref{fig:normalized_neighborhood_full}).}\end{figure}
\begin{figure}[H]
	\centering
	\includegraphics[width=0.62\linewidth]{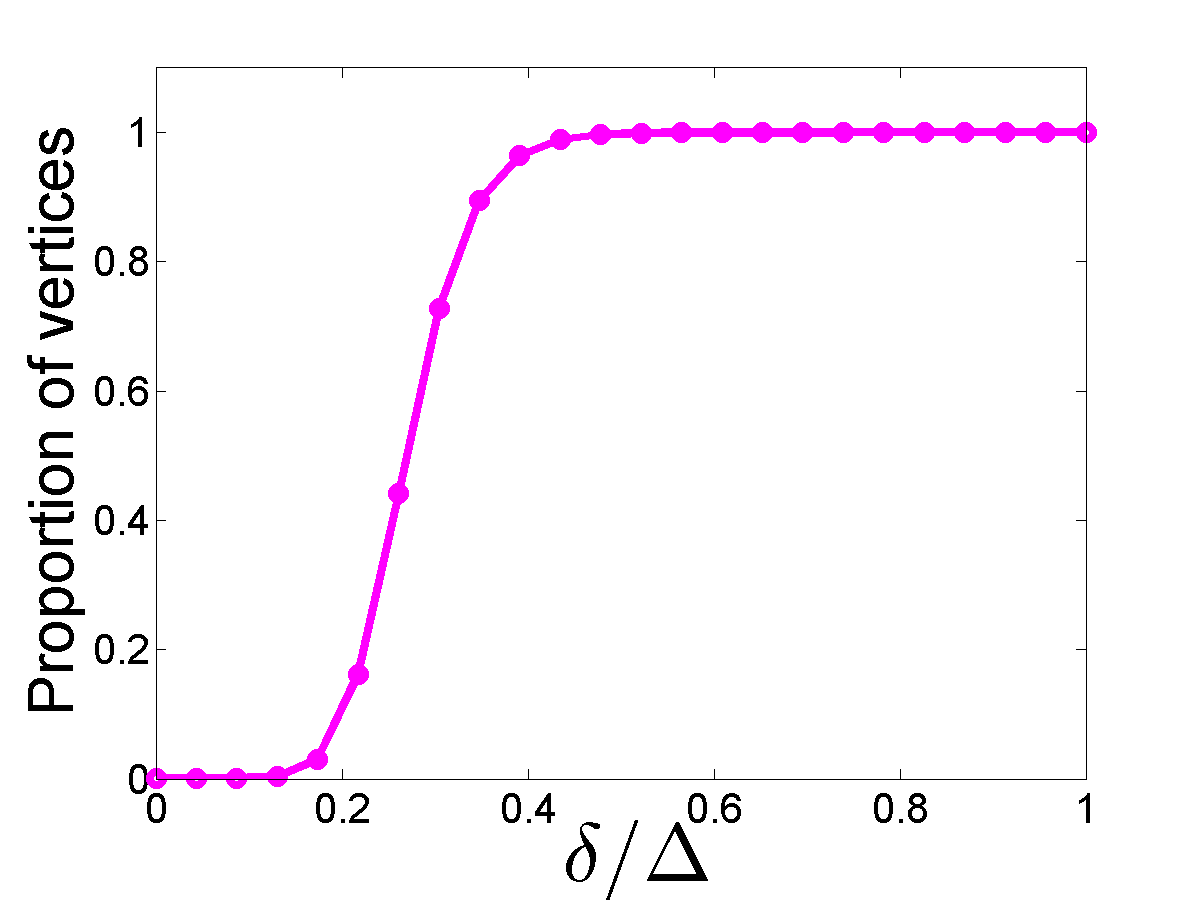}
\caption{Average proportion of vertices in $N_\delta$ of \textsc{DBLP} as a function of $\delta/\Delta$ (see Figure~\ref{fig:normalized_neighborhood_full}).}\end{figure}
\begin{figure}[H]
	\centering
	\includegraphics[width=0.62\linewidth]{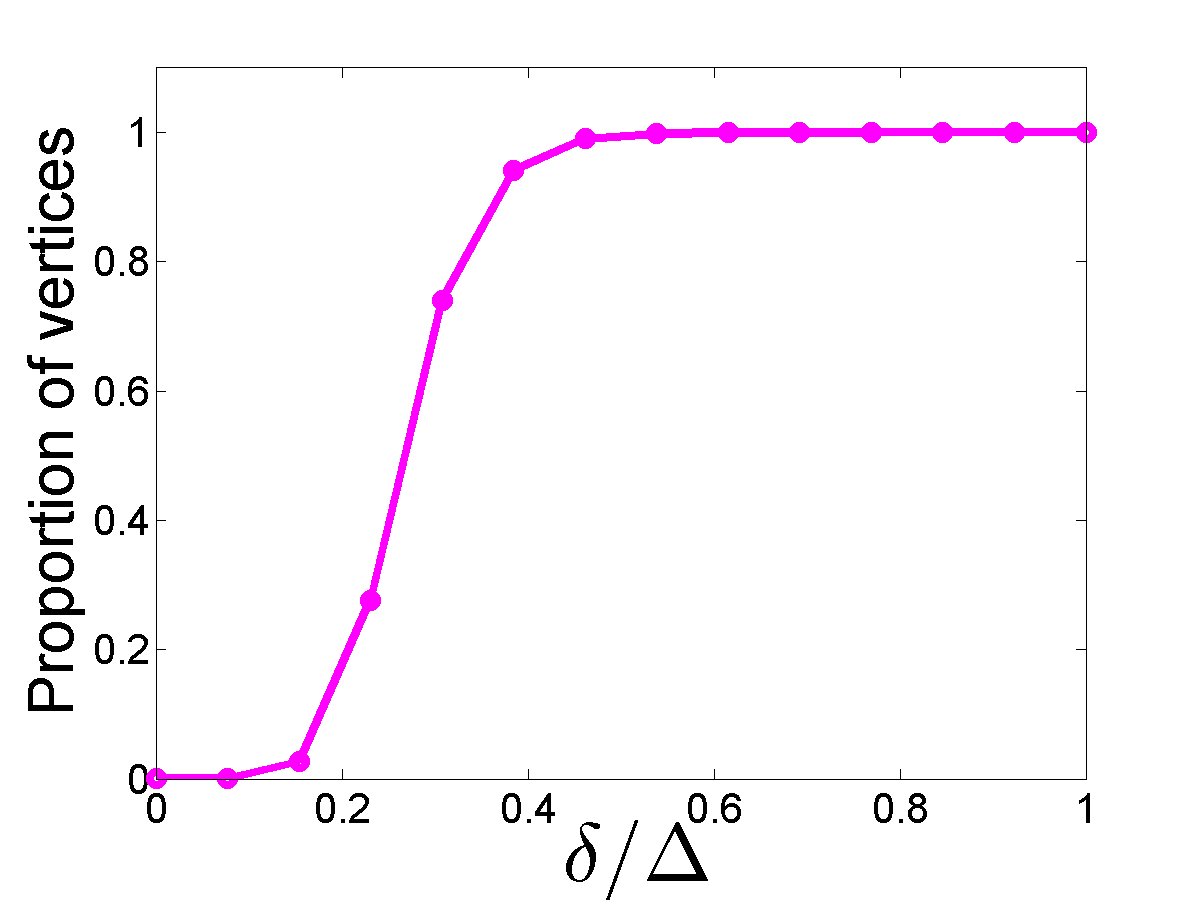}
\caption{Average proportion of vertices in $N_\delta$ of \textsc{Enron} as a function of $\delta/\Delta$ (see Figure~\ref{fig:normalized_neighborhood_full}).}\end{figure}
\begin{figure}[H]
	\centering
	\includegraphics[width=0.62\linewidth]{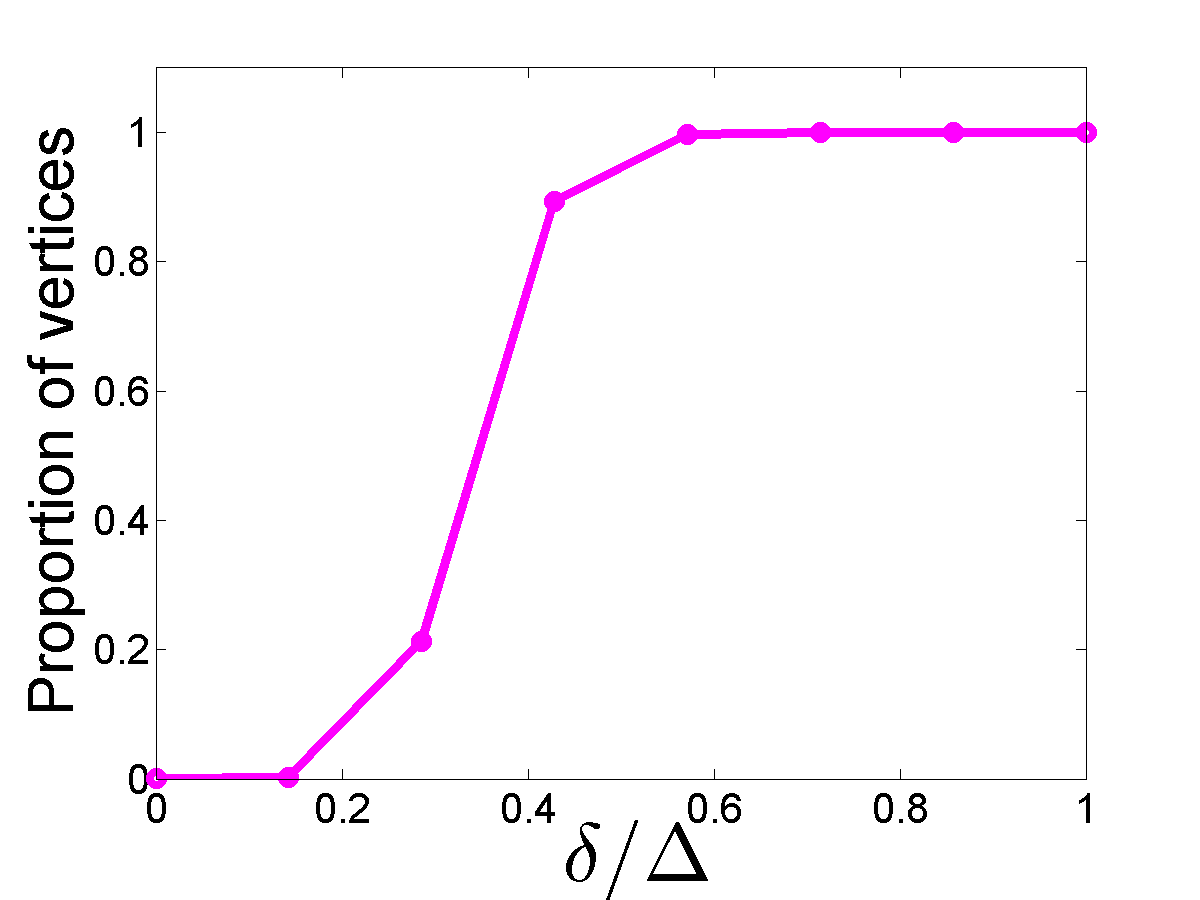}
\caption{Average proportion of vertices in $N_\delta$ of \textsc{Facebook} as a function of $\delta/\Delta$ (see Figure~\ref{fig:normalized_neighborhood_full}).}\end{figure}
\begin{figure}[H]
	\centering
	\includegraphics[width=0.62\linewidth]{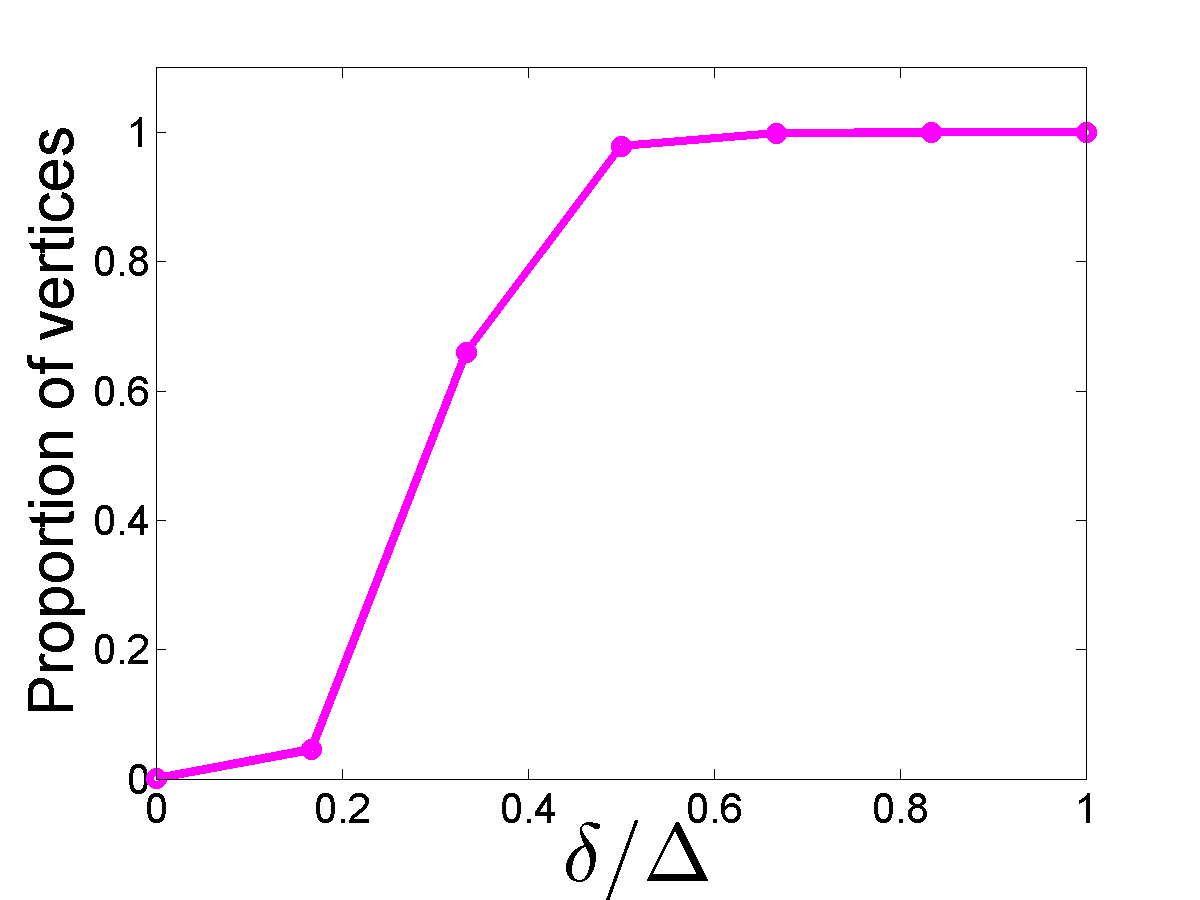}
\caption{Average proportion of vertices in $N_\delta$ of \textsc{Facebook2} as a function of $\delta/\Delta$ (see Figure~\ref{fig:normalized_neighborhood_full}).}\end{figure}
\begin{figure}[H]
	\centering
	\includegraphics[width=0.62\linewidth]{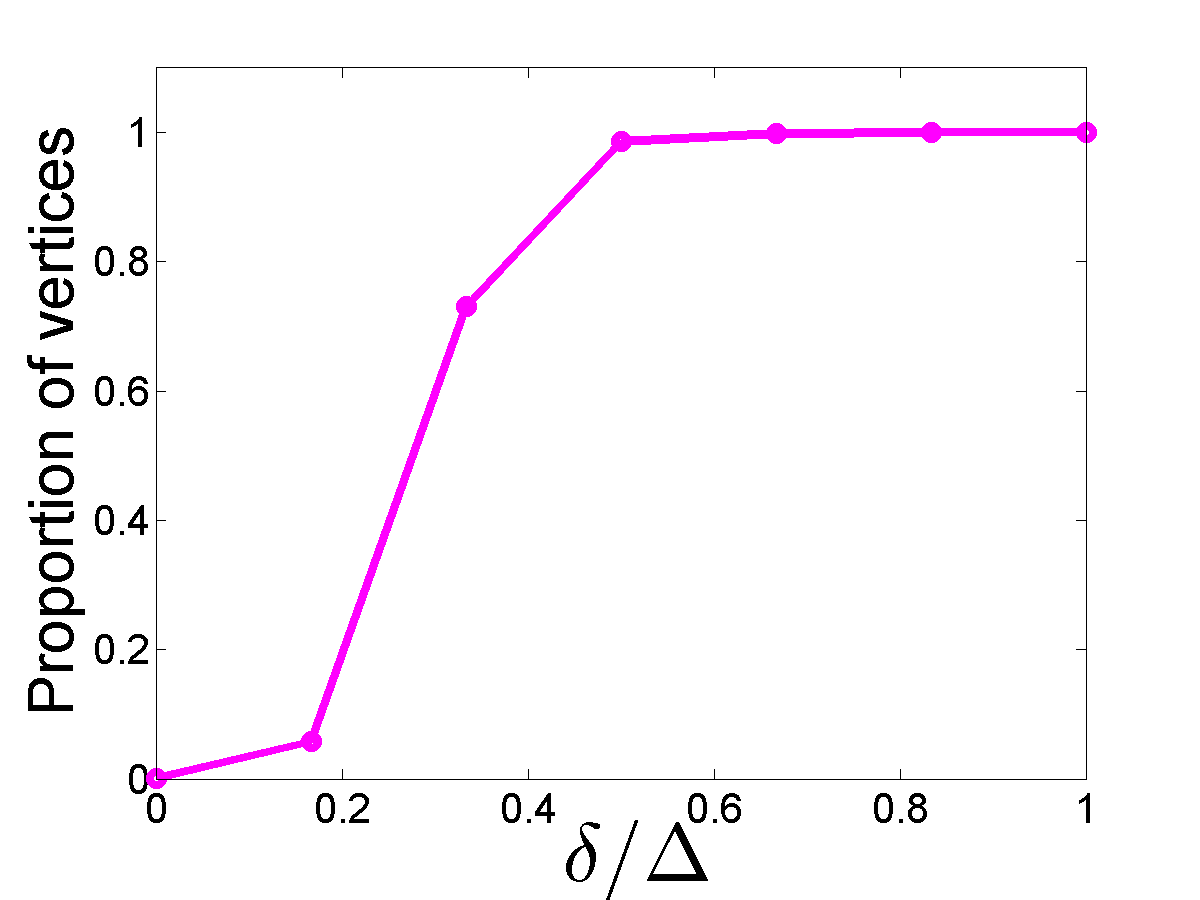}
\caption{Average proportion of vertices in $N_\delta$ of \textsc{Facebook3} as a function of $\delta/\Delta$ (see Figure~\ref{fig:normalized_neighborhood_full}).}\end{figure}
\begin{figure}[H]
	\centering
	\includegraphics[width=0.62\linewidth]{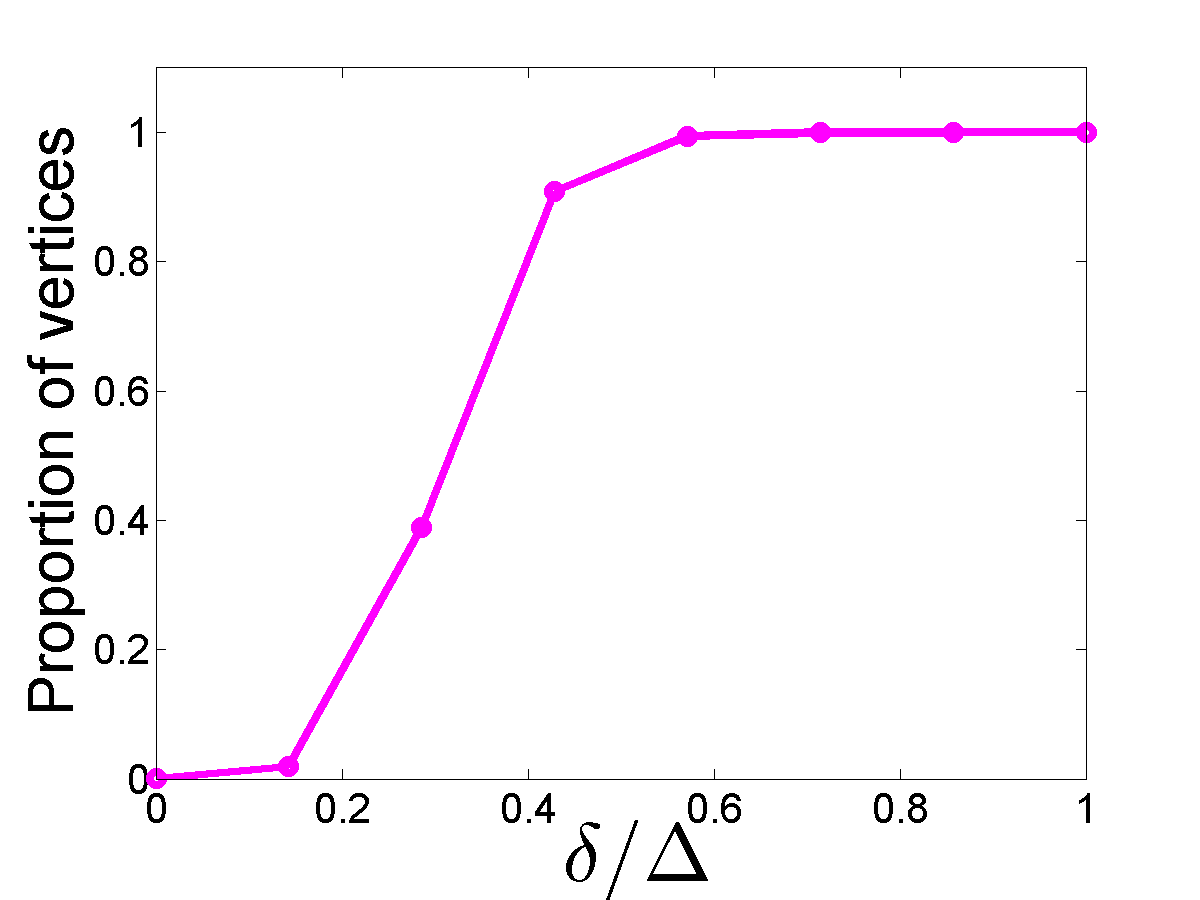}
\caption{Average proportion of vertices in $N_\delta$ of \textsc{Facebook4} as a function of $\delta/\Delta$ (see Figure~\ref{fig:normalized_neighborhood_full}).}\end{figure}
\begin{figure}[H]
	\centering
	\includegraphics[width=0.62\linewidth]{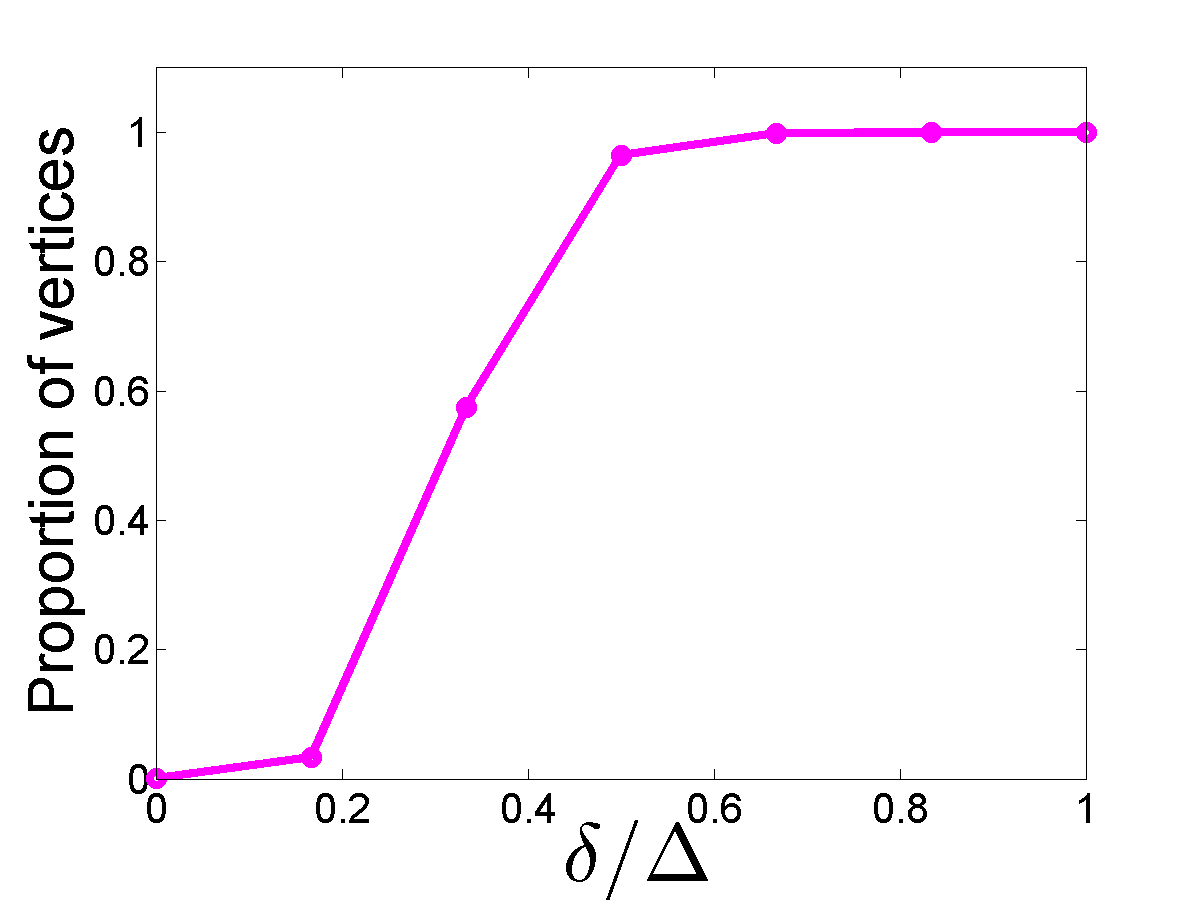}
\caption{Average proportion of vertices in $N_\delta$ of \textsc{Facebook5} as a function of $\delta/\Delta$ (see Figure~\ref{fig:normalized_neighborhood_full}).}\end{figure}
\begin{figure}[H]
	\centering
	\includegraphics[width=0.62\linewidth]{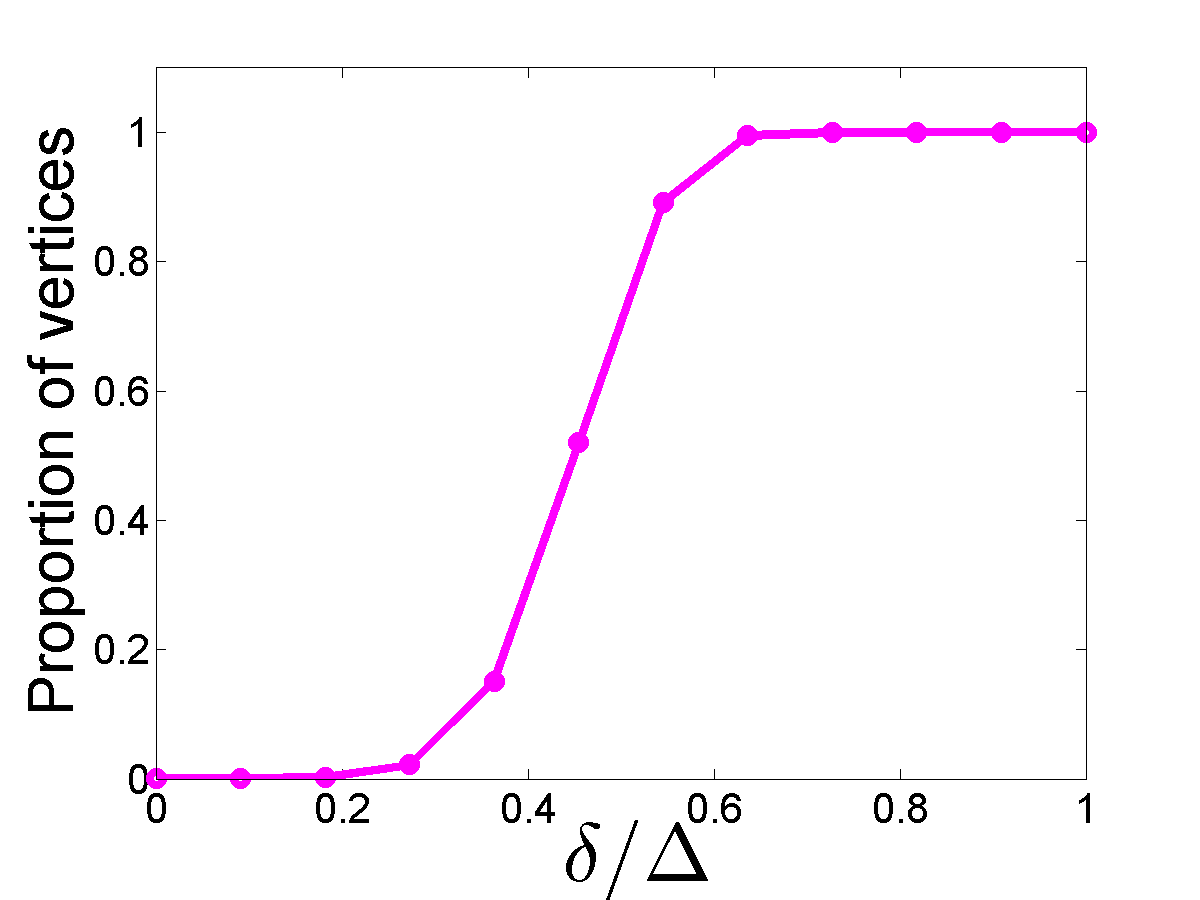}
\caption{Average proportion of vertices in $N_\delta$ of \textsc{Gnutella} as a function of $\delta/\Delta$ (see Figure~\ref{fig:normalized_neighborhood_full}).}\end{figure}
\begin{figure}[H]
	\centering
	\includegraphics[width=0.62\linewidth]{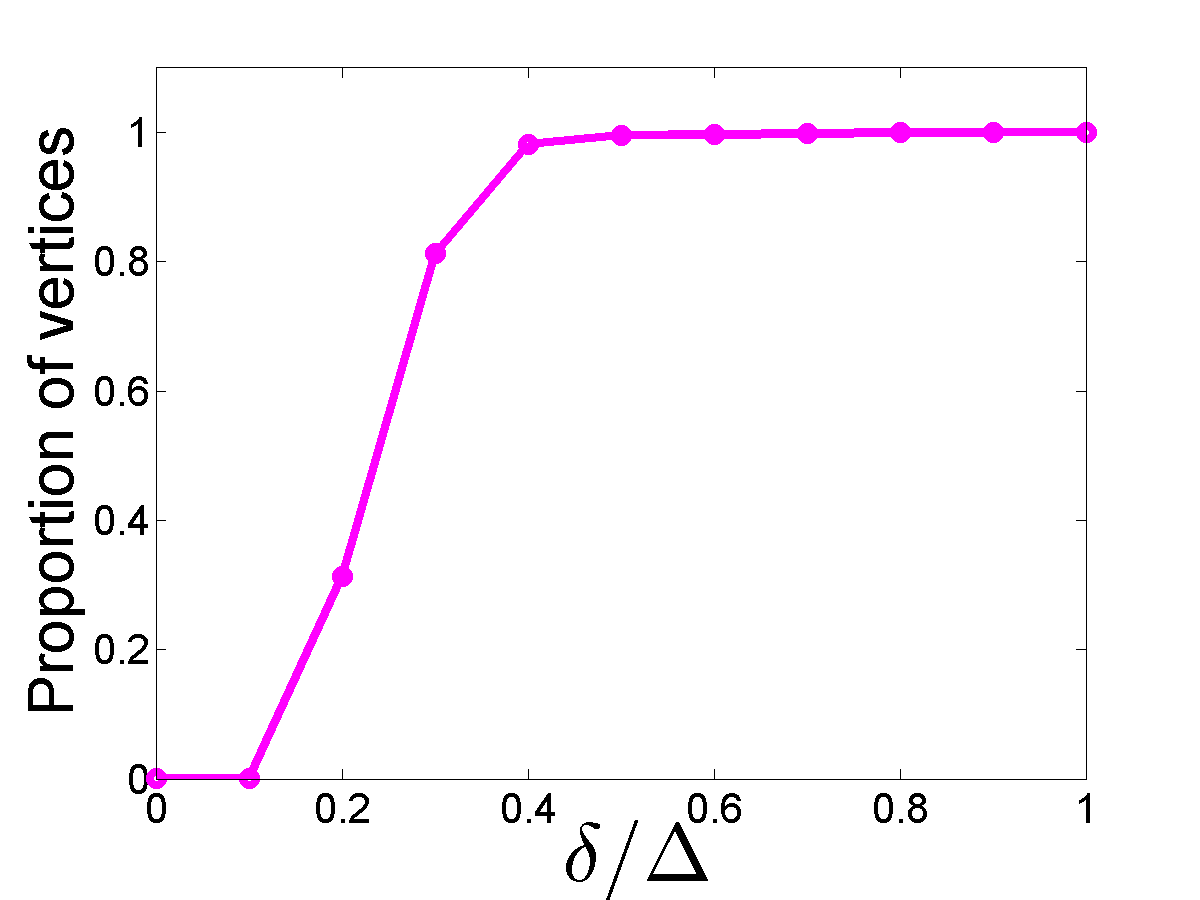}
\caption{Average proportion of vertices in $N_\delta$ of \textsc{H.~sapiens} as a function of $\delta/\Delta$ (see Figure~\ref{fig:normalized_neighborhood_full}).}\end{figure}
\begin{figure}[H]
	\centering
	\includegraphics[width=0.62\linewidth]{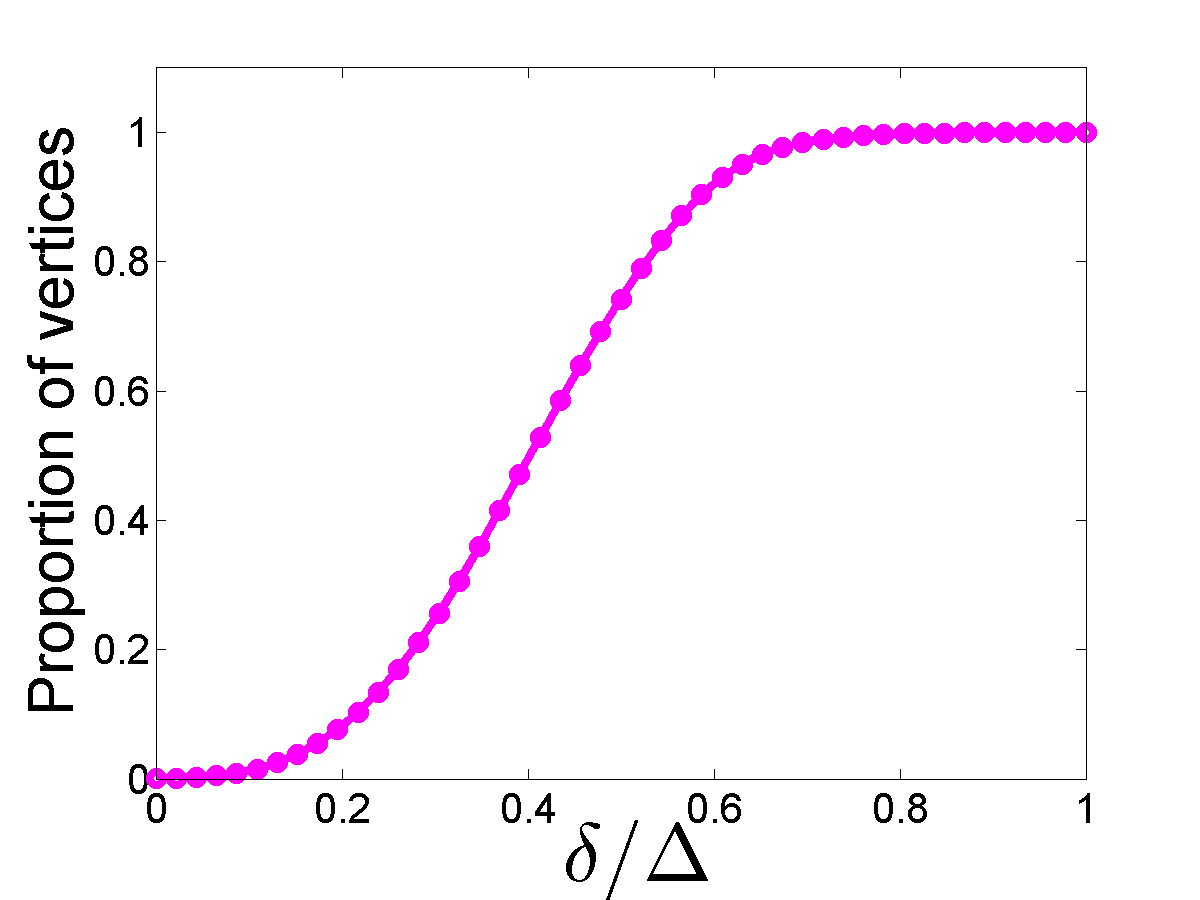}
\caption{Average proportion of vertices in $N_\delta$ of \textsc{WPG} as a function of $\delta/\Delta$ (see Figure~\ref{fig:normalized_neighborhood_full}).}\end{figure}


\end{document}